\documentclass[a4paper,12pt]{PhDThesisPSnPDF}

\ifsetCustomMargin
  \RequirePackage[left=37mm,right=30mm,top=35mm,bottom=30mm]{geometry}
  \setFancyHdr 
\fi

\ifsetCustomFont

\fi

\RequirePackage[labelsep=space,tableposition=top]{caption}

\usepackage{subcaption}

\usepackage{booktabs} 
\usepackage{multirow}

\usepackage{siunitx} 

\usepackage[noadjust]{cite}

\ifuseCustomBib

\RequirePackage[backend=biber, style=numeric-comp, sorting=none, firstinits=true, maxnames=999, maxalphanames=2, minalphanames=2, date=year, doi=false,isbn=false,url=false]{biblatex}

\renewbibmacro*{publisher+location+date}{%
  \printlist{publisher}
  \setunit*{\addcomma\space}
  \printtext[parens]{
    \printlist{location}%
    \setunit*{\addcomma\space}%
    \usebibmacro{date}%
  }
  \newunit}

\renewbibmacro*{volume+number+eid}{%
  \printfield{volume}%
  \setunit{\addcomma\space}%
  \printfield{eid}}

\DeclareFieldFormat[article,inbook,incollection,inproceedings,patent,thesis,unpublished]{volume}{\mkbibbold{#1}}

\renewbibmacro*{journal+issuetitle}{%
  \usebibmacro{journal}%
  \setunit*{\addspace}%
  \iffieldundef{series}
    {}
    {\newunit
     \printfield{series}%
     \setunit{\addspace}}%
  \usebibmacro{volume+number+eid}%
  \newunit}

\DeclareFieldFormat[article,inbook,incollection,inproceedings,patent,thesis,unpublished]{page}{#1}
\DeclareFieldFormat[article,inbook,incollection,inproceedings,patent,thesis,unpublished]{pages}{\mkfirstpage{#1}}

\renewbibmacro*{doi+eprint+url}{%
   \iftoggle{bbx:doi}
    {\usebibmacro{doi}}
    {}%
  \newunit\newblock
  \iftoggle{bbx:eprint}
    {\iffieldundef{doi}{\printfield{eprint}}{}}
    {}%
  \newunit\newblock
  \iftoggle{bbx:url}
    {\usebibmacro{url+urldate}}
    {}}

\DeclareBibliographyDriver{article}{%
  \usebibmacro{bibindex}%
  \usebibmacro{begentry}%
  \usebibmacro{author/translator+others}%
  \setunit{\labelnamepunct}\newblock
  \usebibmacro{title}%
  \newunit
  \printlist{language}%
  \newunit\newblock
  \usebibmacro{byauthor}%
  \newunit\newblock
  \usebibmacro{bytranslator+others}%
  \newunit\newblock
  \printfield{version}%
  \newunit\newblock
  \usebibmacro{journal+issuetitle}%
  \newunit
  \usebibmacro{byeditor+others}%
  \newunit
  \usebibmacro{note+pages}%
  \setunit{\addspace}
  \usebibmacro{issue+date}
  \setunit{\addcolon\space}
  \usebibmacro{issue}
  \newunit\newblock
  \iftoggle{bbx:isbn}
    {\printfield{issn}}
    {}%
  \newunit\newblock
  \usebibmacro{doi+eprint+url}%
  \newunit\newblock
  \usebibmacro{addendum+pubstate}%
  \setunit{\bibpagerefpunct}\newblock
  \usebibmacro{pageref}%
  \usebibmacro{finentry}}

\newbibmacro*{in}{\printtext{\bibstring{in}}\space}
\DeclareFieldFormat{booktitle}{\mkbibemph{#1\isdot}}
\DeclareBibliographyDriver{inproceedings}{%
  \usebibmacro{bibindex}%
  \usebibmacro{begentry}%
  \usebibmacro{author/translator+others}%
  \setunit{\printdelim{nametitledelim}}\newblock
  \usebibmacro{title}%
  \newunit
  \printlist{language}%
  \newunit\newblock
  \usebibmacro{byauthor}%
  \newunit\newblock
  \usebibmacro{in}%
  \usebibmacro{maintitle+booktitle}%
  \newunit\newblock
  \usebibmacro{event+venue+date}%
  \newunit\newblock
  \usebibmacro{byeditor+others}%
  \newunit\newblock
  \iffieldundef{maintitle}
    {\printfield{volume}%
     \printfield{part}}
    {}%
  \newunit
  \printfield{volumes}%
  \newunit\newblock
  \usebibmacro{series+number}%
  \newunit\newblock
  \printfield{note}%
  \newunit\newblock
  \printlist{organization}%
  \newunit
  \usebibmacro{chapter+pages}
  \newunit\newblock
  \usebibmacro{publisher+location+date}
  \newunit\newblock
  \iftoggle{bbx:isbn}
    {\printfield{isbn}}
    {}%
  \newunit\newblock
  \usebibmacro{doi+eprint+url}%
  \newunit\newblock
  \usebibmacro{addendum+pubstate}%
  \setunit{\bibpagerefpunct}\newblock
  \usebibmacro{pageref}%
  \newunit\newblock
  \iftoggle{bbx:related}
    {\usebibmacro{related:init}%
     \usebibmacro{related}}
    {}%
  \usebibmacro{finentry}}

\DeclareBibliographyDriver{thesis}{%
  \usebibmacro{bibindex}%
  \usebibmacro{begentry}%
  \usebibmacro{author}%
  \setunit{\printdelim{nametitledelim}}\newblock
  \usebibmacro{title}%
  \newunit
  \printlist{language}%
  \newunit\newblock
  \usebibmacro{byauthor}%
  \newunit\newblock
  \printfield{note}%
  \newunit\newblock
  \printfield{type}%
  \newunit
  \usebibmacro{institution+location+date}%
  \newunit\newblock
  \usebibmacro{chapter+pages}%
  \newunit
  \printfield{pagetotal}%
  \newunit\newblock
  \iftoggle{bbx:isbn}
    {\printfield{isbn}}
    {}%
  \newunit\newblock
  \usebibmacro{doi+eprint+url}%
  \newunit\newblock
  \usebibmacro{addendum+pubstate}%
  \setunit{\bibpagerefpunct}\newblock
  \usebibmacro{pageref}%
  \newunit\newblock
  \iftoggle{bbx:related}
    {\usebibmacro{related:init}%
     \usebibmacro{related}}
    {}%
  \usebibmacro{finentry}}

\renewbibmacro*{institution+location+date}{%
  \printlist{location}%
  \iflistundef{institution}
    {\setunit*{\addcomma\space}}
    {\setunit*{\addcolon\space}}%
  \printlist{institution}%
  \setunit*{\space}%
   \printtext[parens]{\usebibmacro{date}}%
  \newunit}

\fi

\usepackage{amsfonts}
\usepackage{graphicx}
\usepackage{amsmath}
\usepackage{amsthm}
\usepackage{bbm}
\usepackage{hyperref}
\newcommand{\oneover}[1]{\ensuremath{\frac{1}{#1}}}
\newcommand{\hil}{\ensuremath{\mathcal{H}}}
\newcommand{\states}[1]{\ensuremath{\mathfrak{S}(\mathcal{H}_{#1})}}
\newcommand{\ket}[1]{\ensuremath{|#1\rangle}}
\newcommand{\bra}[1]{\ensuremath{\langle #1|}}
\newcommand{\braket}[2]{\ensuremath{\langle #1 | #2 \rangle}}

\newcommand{\ketbra}[2]{\ensuremath{| #1 \rangle \langle #2 |}}
\newcommand{\dketbra}[1]{\ensuremath{| #1 \rangle \langle #1 |}}
\newcommand{\ave}[1]{\ensuremath{\left\langle#1\right\rangle}}
\newcommand{\abs}[1]{\ensuremath{\left|#1\right|}}
\newcommand{\norm}[1]{\ensuremath{\left\|#1\right\|}}
\newcommand{\tr}[1]{\ensuremath{\operatorname{Tr}\left[#1\right]}}
\newcommand{\trabs}[1]{\ensuremath{\operatorname{Tr}\left|#1\right|}}
\newcommand{\ptr}[2]{\ensuremath{\operatorname{Tr}_{#1}\left[#2\right]}}
\newcommand{\trd}[2]{\ensuremath{D\left(#1,#2\right)}}
\newcommand{\fid}[2]{\ensuremath{F\left(#1,#2\right)}}

\newcommand{\vecop}[1]{\ensuremath{\underline{\hat{#1}}}}
\newcommand{\comm}[2]{\ensuremath{\left[#1,#2\right]}}
\newcommand{\vecund}[1]{\ensuremath{\underline{#1}}}
\newcommand{\veclad}[1]{\ensuremath{\underline{\tilde{#1}}}}
\newcommand{\eqq}[1]{Eq.~(\ref{#1})}

\title{Decoding Protocols \\ for Classical Communication \\ on Quantum Channels}

\author{Matteo Rosati}

\university{Scuola Normale Superiore di Pisa}
\crest{\includegraphics[width=0.25\textwidth]{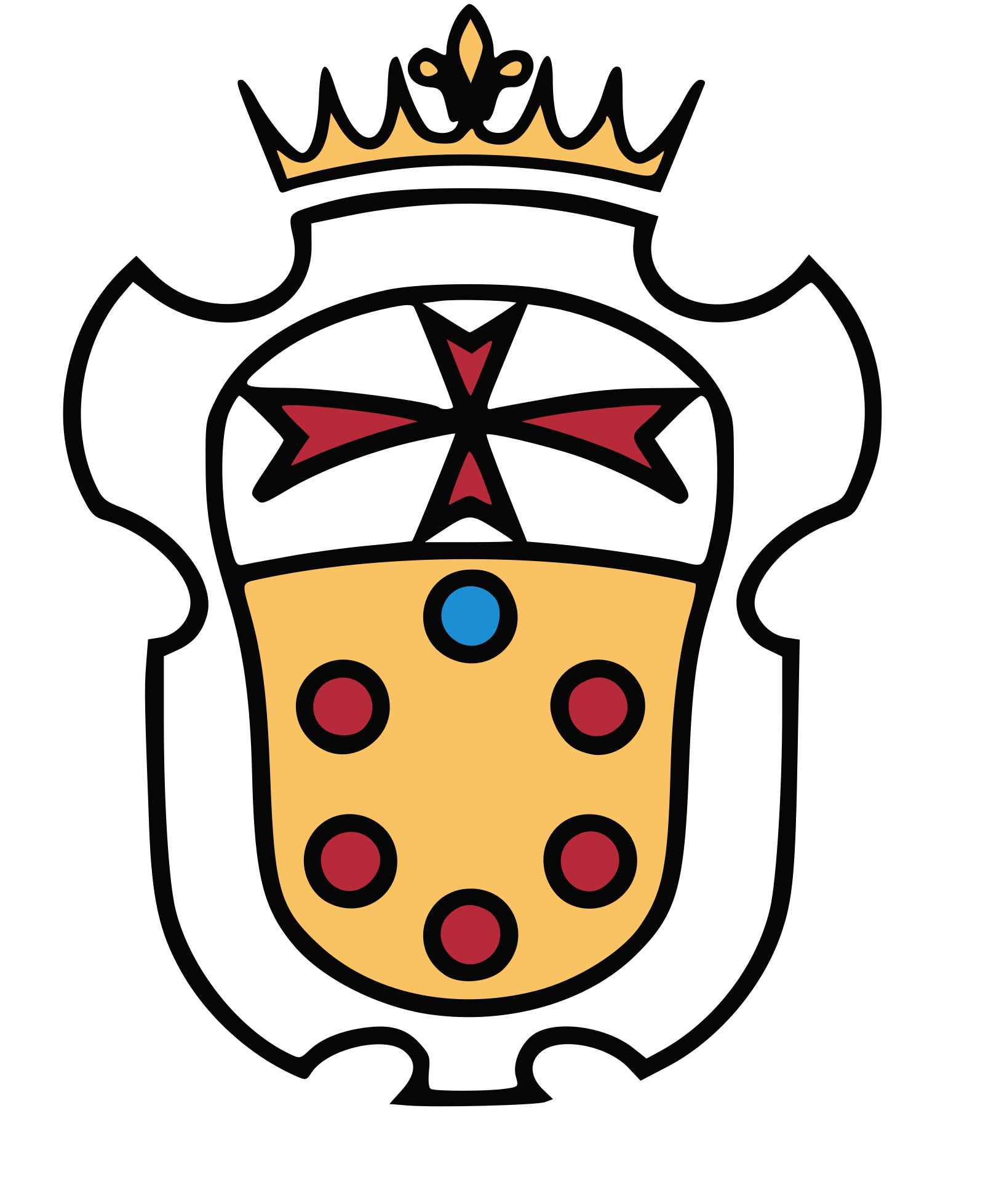}}

\supervisor{\textbf{Prof. V. Giovannetti\newline
\hspace*{2.45 cm} Dr. A. Mari\newline
}}

 \supervisorrole{\textbf{Supervisors: }}

\renewcommand{\submissiontext}{This thesis is submitted for the degree of}

\degreetitle{Doctor of Philosophy}

\subject{LaTeX} 

\begin{document}

\frontmatter

\maketitle


\begin{acknowledgements}      
Pursuing successful research means learning to cope with imperfection and the limits of human capabilities. The first difficulty is the lack of knowledge: a young researcher has limited experience of what has been done before and only time can fill this gap. But, as time passes, one realizes that it is impossible to keep track of all the news in one's field and even the best researchers seem to be experts only in some subfields, with limited knowledge of the neighbouring ones. However, what they lack in knowledge, they make up for in imagination. On a more fundamental level, the lack of knowledge is intrinsic in the activity we call research, because it takes place at a fragmented frontier of the human knowledge; hence a young researcher has to learn how to investigate the unknown, trying to fill the gaps encountered along the way. \\
The second difficulty is devising proper questions. This is a complex task because it requires one to gauge several competing factors: if the question is too general it may sound naive because of its huge scale; if the question is too specific it may be trivial or worse, it may still be difficult to answer but also of very limited interest to others; more subtly, a question may look relevant and tractable, but a true nightmare in technical terms. Moreover, there is no set of rules to follow for devising good questions because different people usually have different approaches and one can only hope to find one's own method through successive approximations. \\
The third difficulty is, of course, actually looking for the answer. It is well known that this usually requires a lot of persistence and passion for research: even standard methods may be painful to apply; one may get stuck at some point and need to follow different approaches; the results may be unsatisfactory. \\
However, all the difficulties are forgotten when each piece of the puzzle finally finds its right place; nothing can substitute the feeling of expectation and anxiety when opening a plot or reaching the end of a calculation. What was nature hiding? Will it work? How good will the result be? Then, even though it took entire months of research, even though maybe we did not answer our initial question, even though it may not be accepted by the journal we wanted and perhaps it is not going to change the world, we will have captured a new little piece of reality, colouring another small spot on the map of knowledge. \\
During the last three years, I was lucky to have people who could help me to successfully face these problems. I would like to thank Vittorio and Andrea for teaching me how to do research, for their patience and readiness to hear my ideas and doubts. I am extremely grateful to my girlfriend, Virginia, and my parents, Marco and Nicoletta, for their steady, intelligent and heartfelt support in good times and bad times. I would also like to thank and wish good luck to some colleagues that shared this adventure with me: Vasco, Stefano, Giacomo, Max, Marcello, Michael, Luigi and Federico, along with the whole research group at SNS. Eventually, I am thankful to all my friends, Edoardo above all, to my close relatives and to my teachers for having shaped my life so far, bringing me where I am now.
\end{acknowledgements}
\begin{abstract}
We study the problem of decoding classical information encoded on quantum states at
the output of a quantum channel, with particular focus on increasing the communication
rates towards the maximum allowed by Quantum Mechanics. \\
After a brief introduction
to the main theoretical formalism employed in the rest of the thesis, i.e., continuous-variable
Quantum Information Theory and Quantum Communication Theory, we consider
several decoding schemes. \\
First, we treat the problem from an abstract perspective,
presenting a method to decompose any quantum measurement into a sequence of easier
nested measurements through a binary-tree search. Furthermore we show that this
decomposition can be used to build a capacity-achieving decoding protocol for classical
communication on quantum channels and to solve the optimal discrimination of some sets
of quantum states. These results clarify the structure of optimal quantum measurements,
showing that it can be recast in a more operational and experimentally-oriented fashion.\\
Second, we consider a more practical approach and describe three receiver structures
for coherent states of the electromagnetic field with applications to single-mode state
discrimination and multi-mode decoding at the output of a quantum channel. We treat
the problem bearing in mind the technological limitations faced nowadays in the field
of optical communications: we evaluate the performance of general decoding schemes
based on such technology and report increased performance of two schemes, the first
one employing a non-Gaussian transformation and the second one employing a code
tailored to be read out easily by the most common detectors. Eventually we characterize
a large class of multi-mode adaptive receivers based on common technological resources,
obtaining a no-go theorem for their capacity. 
\end{abstract}


\tableofcontents





\mainmatter


\chapter{Introduction}


Quantum Communication~\cite{holevoBOOK,wildeBOOK} is a research field where Quantum Information Science may have its first promising applications, thanks to technological requirements somewhat easier than, e.g., Quantum Computation~\cite{nChuangBOOK}. Its study has begun before Quantum Information Science was actually born, stemming from the same core interest of Classical Information Theory~\cite{covThomBOOK}, i.e., the analysis of information transmission under the distinct requirements of reliability, resource-efficiency and security. Since information is transmitted through physical media that ultimately follow the laws of Quantum Mechanics, there rose an obvious interest in applying those laws to such practical task, specifically in the regime where quantum effects become relevant as compared to classical fluctuations. Following the idea of using optical pulses of an attenuated laser field, i.e., optical coherent states~\cite{glauber}, to transmit classical information, in the 1970's the pioneering theoretical works of Holevo~\cite{holevo1} and Helstrom~\cite{helstromBOOK} focused on two basic problems for the starting field, namely extending the Shannon Capacity Theorem~\cite{covThomBOOK} to quantum channels and the Classical Detection and Estimation Theory to quantum signals. Both these problems face a fundamental difference between the classical and quantum behaviour, intimately connected to the Quantum Uncertainty Principle.\\
 The first problem consists in determining the ultimate rate, or \textit{capacity}, at which classical information can be reliably transmitted on a quantum channel, i.e., without introducing errors in the process. While in the classical case the capacity of the channel is only bounded by the necessary presence of statistical noise in realistic descriptions of the communication media, in the quantum case even an ideal noiseless channel (with constrained input energy in the case of infinite-dimensional signals) has a finite capacity due to the presence of intrinsic quantum-mechanical noise. This line of research has succeeded in several purposes: establishing the ultimate capacity for the transmission of classical information on a quantum channel and demonstrating its achievability in the asymptotic limit of infinitely long sequences~\cite{holevoBOOK,holevo1,holevo2,holevo3,schumawest1,schumawest2,winter,oga1,oga2,oga3,hauswoot,oganaga,hayanaga,hayashi1,hayashi2,seq1,seq2,sen,polarWildeGuha,NOSTROHol}; clarifying the role of entanglement during encoding and decoding operations~\cite{hastings,bennetEntEnc}; generalizing the capacity to the settings where shared entanglement can be used as a resource~\cite{bennetEntAss1,bennetEntAss2} or quantum information has to be transmitted~\cite{lloydQuant,qCap1,qCap2}, preserving its coherence; giving an impulse to the development of Quantum Cryptography~\cite{crypto}. In particular for continuous-variables systems~\cite{RevGauss} it has been recently proved that a coherent-state encoding with amplitudes extracted at random from a Gaussian distribution achieves the ultimate communication capacity of phase-insensitive Gaussian channels~\cite{gaussOpt,maj1,maj2}, such as those modeling modern-day optical fiber and free-space communication~\cite{cavesDrum,LOSSY}. \\
 The second problem consists instead in recovering the information encoded on the quantum carriers by performing a parameter estimation or, in simpler cases, the \textit{optimal discrimination} of a set of quantum states with the maximum success probability of correctly identifying the message. Again, while the success probability in the classical case is only affected by the necessary imperfections of the measurement procedure, in the quantum case most sets of states can never be exactly distinguished due to their non-orthogonality in the Hilbert space of the system, even if ideal measurements are employed. This line of research has brought several contributions: establishing a general optimization strategy for the success probability of discrimination of an arbitrary set of quantum states~\cite{helstromBOOK,holevoDisc,YKL} through the identification of optimal measurement operators; deriving a closed expression in the case of two states and simplified methods in the case of multiple states with particular symmetries~\cite{multiHel1,multiHel2,opMeasCoh,bae,chiribella,symGuha}; generalizing the optimization to different approaches, characterized by other figures of merit~\cite{ivanovic,dieks,peres,chefBarn,mixedStrat,revDisc,genDisc,genDiscErr}, e.g., unambiguous discrimination. \\
  From a more practical perspective there has been a parallel effort to actually implement the optimal measurements described in theoretical works. In particular, regarding receivers for optical signals, in the 1970's Kennedy~\cite{kenDet,opKen,opGaussDet,NOSTRODisc} introduced the technique of signal nulling; Dolinar~\cite{dol,NOSTRODisc} refined it by means of signal splitting and conditional operations via fast feedback, devising a receiver that attains the maximum probability of success for discriminating any two coherent states in the asymptotic limit of infinite splitting steps~\cite{implProj,multicopy1,multicopy2,dolExp}. These receivers are particularly useful when employing low-energy signals as in long-distance communication, for which standard techniques such as homodyne and heterodyne detection~\cite{cavesDrum} do not perform well. More recently the problems of parameter estimation and state discrimination have given birth to the broader field of Quantum Metrology and Sensing~\cite{metroRev1,metroRev2}, also very promising for the practical application of Quantum Information Science. Interesting proposals include: Quantum Illumination to identify low-reflectivity objects in highly noisy environment~\cite{qIllum}, Quantum Reading to increase the performance of optical memory readout~\cite{qRead} and the prototypical example of Quantum Phase Estimation for interferometric purposes~\cite{phaseEst,ligoCaves}. \\
  
 In light of the previous discussion, it is clear that the intense theoretical work of the past has made the field of Quantum Communication ripe with interesting results that would allow for higher rates of information transmission on modern-day communication media. However it is not yet clear how to translate these theoretical advantages into practical technological advancements. The main obstacle in this direction seems to be the task of information decoding at the receiver end of the channel. Indeed all proofs of the achievability of the channel capacity so far rely on joint decoding operators~\cite{schumawest1,schumawest2,holevo2,holevo3,seq1,seq2,sen,polarWildeGuha,NOSTROHol}, which read out entire blocks of letters at the same time; these are very difficult to implement in practice~\cite{opticalFiber,deepSpace,seqCoh,guha1,guha2,banaszek,NOSTROHad,NOSTROAdap}. This problem is also closely related to state discrimination: the optimal codewords for communication constitute a set of multi-mode quantum states and a good receiver aims at extracting the information they encode by discriminating them. An additional difficulty is that these receivers should be realizable with technological resources available nowadays or in the near future, e.g., in the optical case, photodetectors and Gaussian operations such as beam splitters and squeezing devices~\cite{RevGauss,guhaDet,guhaDet1}.\\
 In this thesis we address the problem of decoding information at the output of a quantum channel from several perspectives. We start in Ch.~\ref{ch:Rev} with a review of the main results in Quantum Communication and State Discrimination related to our work; in Ch.~\ref{ch:Opt} we discuss a tree-search decomposition of quantum measurements and later on employ it to describe optimal decoding and discrimination protocols. Eventually in Ch.~\ref{ch:Imple} we consider optical implementations for coherent-state communication, proposing two structured single-mode and multi-mode receiver schemes and analyzing a large class of multi-mode adaptive decoders.

\chapter{Quantum Information Theory} \label{ch:Rev} 

    \graphicspath{{Chapter2/Figs/Raster/}{Chapter2/Figs/PDF/}{Chapter2/Figs/}}

In this Chapter we review the basic formalism of Quantum Information Theory, starting with a brief general introduction (Sec.~\ref{sec:basics}) and proceeding with particular attention to bosonic Gaussian systems (Sec.~\ref{sec:basicsGauss}), which are the main focus of the rest of the thesis. Then we introduce the field of Quantum Communication, describing in detail the transmission of classical information on quantum channels and the optimal solution for the important class of bosonic phase-insensitive Gaussian channels (Sec.~\ref{sec:commQ}). Eventually we consider the problem of State Discrimination, reviewing general optimization strategies and specific solutions, as well as CV implementations (Sec.~\ref{sec:disc}).

\section{Elements of Quantum Information Theory}\label{sec:basics}
\subsection{Quantum states}\label{quantumStates}
The state of a quantum system, such as a photon in an optical fiber or an atom trapped in an optical cavity, can be described by a bounded, linear, positive and unit-trace operator $\hat{\rho}$ on a complex Hilbert space $\mathcal{H}$ of dimension $d$. The simplest operators are pure states, represented by projectors $\dketbra{\psi}$ on unit vectors of the Hilbert space, $\ket{\psi}\in\mathcal{H}$.
The set of all quantum states, labelled by $\mathfrak{S}(\mathcal{H})$, is convex and the pure states are its extremal points, from which any other state can be obtained through convex combinations. Indeed pure states represent all those experimental settings where the observer has full knowledge of how the state of the system has been produced; on the other hand mixed states can be written as statistical mixtures of pure states and represent the cases of incomplete knowledge of the preparation procedure, when the state of the system can be one among several possibilities according to a given probability distribution. 
More generally, any mixed state $\hat{\rho}\in\mathfrak{S}(\mathcal{H})$ can be unravelled as a statistical mixture of other pure or mixed states in infinitely many ways: each unravelling is represented by an ensemble $\mathcal{E}=\{\hat{\rho}_{k},p_{k}\}_{k=1,\cdots,K}$ of states, $\hat{\rho}_{k}$, with an associated probability distribution, $p_{k}$, such that their average state is $\hat{\rho}=\sum_{k=1}^{K}p_{k}\hat{\rho}_{k}$. For example each state $\hat{\rho}$ of rank $r\leq d$ has a diagonal unraveling corresponding to its spectral decomposition, $\mathcal{E}_{diag}(\hat{\rho})=\{\dketbra{e_{k}(\hat{\rho})},\nu_{k}(\hat{\rho})\}_{k=1,\cdots,r}$, where $\{\ket{e_{k}(\hat{\rho})}\}$ form an orthonormal basis of the Hilbert space of the system. \\
For systems comprising several parts, the total Hilbert space is given by the tensor product of the Hilbert space of each constituent. Due to the different structure of such product with respect to the ordinary Cartesian one, a large amount of exotic states that exhibit stronger-than-classical correlations appears when considering composite systems. For example take the Hilbert space of a bipartite system $\mathcal{H}_{AB}=\mathcal{H}_{A}\otimes\mathcal{H}_{B}$: there are pure states that cannot be written as the tensor product of states of the single subsystems, e.g., $\ket{\psi}_{A}\otimes\ket{\phi}_{B}+\ket{\phi}_{A}\otimes\ket{\psi}_{B}$, as opposed to, e.g., $\ket{\psi}_{A}\otimes\ket{\phi}_{B}$. The former states are called entangled, while the latter ones are separable. More generally a quantum state of an $n$-partite system is called separable if it can be unravelled as an ensemble of product states given by the tensor product of quantum states on each subsystem, i.e., \begin{equation}\hat{\rho}_{sep}=\sum_{k=1}^{K}p_{k}\hat{\rho}_{k}^{(1)}\otimes\cdots\otimes\hat{\rho}_{k}^{(n)}\end{equation} with $p_{k}\geq0$, $\sum_{k}p_{k}=1$ and $\hat{\rho}_{k}^{(j)}$ states of the $j$-th subsystem; otherwise the state is entangled.\\
A second important feature of composite quantum systems is purification. Suppose we have a mixed state of a single system $\hat{\rho}_{A}\in\mathfrak{S}(\mathcal{H}_{A})$. Then it is always possible to purify it by considering an enlarged Hilbert space $\mathcal{H}_{AB}$ such that the total state of the composite system is pure but it is locally equivalent to the original mixed state, i.e., there exists $\dketbra{\rho}\in\mathcal{H}_{AB}$ pure such that its reduced state on $\mathcal{H}_{A}$, obtained by partial trace on $\mathcal{H}_{B}$, is $\ptr{B}{\dketbra{\rho}}\equiv\hat{\rho}_{A}$. In particular, a purification can always be obtained by considering the eigensystem of $\hat{\rho}_{A}=\sum_{k=1}^{r}\nu_{k}\dketbra{e_{k}}$ and an additional Hilbert space $\mathcal{H}_{B}$ isomorphic to $\mathcal{H}_{A}$; then it is easy to check that 
\begin{equation}\label{purification}\ket{\rho}_{AB}=\sum_{k=1}^{r}\sqrt{\nu_{k}}\ket{e_{k}}_{A}\otimes\ket{f_{k}}_{B}\end{equation} 
purifies $\hat{\rho}_{A}$, where $\{\ket{f_{k}}\}_{k=1,\cdots,d}$ is any local basis on $\mathcal{H}_{B}$. Conversely, a result known as Schmidt decomposition holds: any pure state of a bipartite system can be written as a superposition of the kind~\eqq{purification}, which employs only $d$ out of the $d^{2}$ elements constituting the orthonormal basis $\{\ket{e_{k}}_{A}\otimes\ket{f_{j}}_{B}\}_{k,j=1,\cdots,d}$ of $\mathcal{H}_{AB}$.\\
Eventually a useful tool for the characterization of quantum states is their Von Neumann (or simply quantum) entropy~\cite{nChuangBOOK}, defined for $\hat{\rho}_{A}\in\hil_{A}$ as 
\begin{equation}\label{VNEnt}
S(A)=S(\hat{\rho}_{A})=-\tr{\hat{\rho}_{A}\log\hat{\rho}_{A}}=-\sum_{k=1}^{r}\nu_{k}\log\nu_{k}=H(\{\nu_{k}\}_{k=1,\cdots,r}),
\end{equation} 
where as before $r\leq d$ is the rank of the operator $\hat{\rho}$, $\{\nu_{k}\}_{k=1,\cdots,r}$ its eigenvalues and $H(\cdot)$ the Shannon (or classical) entropy of a classical probability distribution~\cite{covThomBOOK}. Here and in the rest of the thesis we will take the $\log$ in base $2$. The entropy,~\eqq{VNEnt}, is always positive and it is  commonly interpreted as a measure of the lack of knowledge about a given state; correspondingly, it attains its minimum value on pure states, i.e., $S(\dketbra{\psi})=0$, and its maximum on the completely mixed state, i.e., $S(\hat{\mathbf{1}}/d)=\log d$. Several other properties of the quantum entropy and of its derived quantities will be discussed in Sec.~\ref{sec:commQ}. For now let us focus on bipartite systems and notice a first striking feature of quantumness: while the Shannon entropy cannot increase when discarding a subsystem, the Von Neumann entropy can and this happens only with entangled states. Indeed consider a pure state of a bipartite system $\hat{\rho}_{AB}=\dketbra{\psi}\in\mathcal{H}_{AB}$ and its reduced states on each subsystem, $\hat{\rho}_{A,B}$. Since the total state is pure, $S(\hat{\rho}_{AB})=0$; hence classical intuition would suggest that $S(\hat{\rho}_{A})=S(\hat{\rho}_{B})=0$ too, because full knowledge of the state of the total system implies full knowledge of the state of its single components. This is true indeed if the state is separable, i.e., $\ket{\psi}_{AB}=\ket{\phi}_{A}\otimes\ket{\phi'}_{B}$, but if the state is entangled it happens that the local entropies are non-zero. The most extreme example is obtained by considering the maximally entangled state $\ket{\psi}_{AB}=\sum_{k=1}^{d}\ket{e_{k}}_{A}\otimes\ket{e_{k}}_{B}/\sqrt{d}$, where \ket{e_{k}} are elements of local bases; the reduced states are maximally mixed, $\hat{\rho}_{A,B}=\hat{\mathbf{1}}_{A,B}/d$, thus have maximum entropy. Such a deviation from classical intuition is typical of quantum correlations between subsystems: global knowledge does not imply local knowledge when strong correlations are at play.

\subsection{Quantum operations}\label{quantumOps}
The transformations to which a quantum state can be subject are described by maps $\Phi:\mathcal{H}_{A}\rightarrow\mathcal{H}_{B}$ from the Hilbert space of the initial system to that of the final one, which in general may be different. Any physically realizable map has to satisfy three basic requisites:
\begin{enumerate}
\item Linearity (L), so that the transformation of a statistical mixture of states is equivalent to the same mixture of the single transformed components, i.e., \begin{equation}\Phi\left(\sum_{k=1}^{K}p_{k}\hat{\rho}_{k}\right)=\sum_{k=1}^{K}p_{k}\Phi(\hat{\rho}_{k});\end{equation}
\item Complete Positivity (CP), so that any positive operator is transformed into another positive operator even if only a subsystem is affected by the map, i.e., 
\begin{equation}\hat{\rho}_{AR}\geq 0 \Rightarrow (\Phi_{A}\otimes\mathcal{I}_{R})(\hat{\rho}_{AR})\geq 0,\end{equation}
where $\hat{\rho}_{AR}\in\mathcal{H}_{AR}$ is the state of a composite system and $\mathcal{I}$ is the identity map. Note that CP implies simple Positivity (P), i.e., the property that $\Phi$ preserves positivity when acting on the entire state of a system. Moreover CP is equivalent to P when considering separable states but not for entangled ones;
\item Trace-preserving property (T), so that any unit-trace operator is transformed into another unit-trace operator, i.e., \begin{equation}\tr{\hat{\rho}}=1 \Rightarrow \tr{\Phi(\hat{\rho})}=1.\end{equation}
\end{enumerate}
There are several ways to characterize a physical (LCPT) map, but two among them are most common. The first one, called Kraus representation, employs a set of Kraus operators $\{\hat{M}_{k}\}_{k=1,\cdots,K}$ normalized as $\sum_{k=1}^{K}\hat{M}_{k}^{\dag}\hat{M}_{k}=\hat{\mathbf{1}}$, so that the action of the map can be written as
\begin{equation}\label{kraus}
\Phi(\hat{\rho})=\sum_{k=1}^{K}\hat{M}_{k}\hat{\rho} \hat{M}_{k}^{\dag},
\end{equation}
where $\dag$ represents hermitian conjugation of an operator. For example unitary maps have a Kraus representation with a single unitary operator, i.e., $\hat{M}_{1}=\hat{U}$ with $\hat{U}^{\dag}\hat{U}=\hat{U}\hat{U}^{\dag}=\hat{\mathbf{1}}$, and describe reversible transformations of the Hilbert space of the system into itself when it is isolated from the rest of the universe. More generally, a realistic evolution has to consider interaction with some external system, usually called the environment. This gives rise to the Stinespring representation of quantum maps, which is a sort of purification applied to physical transformations instead of states: any physical map on a system can be described as a unitary map on a larger system. Indeed let us consider a system of interest $S$ in the state $\hat{\rho}_{S}\in\mathcal{H}_{S}$ and add a second environmental system $E$ of dimension $d_{E}$ that includes any degree of freedom relevant to the evolution of $S$. The initial state of $E$ can always be taken pure w.l.o.g., i.e., $\dketbra{\psi}_{E}\in\mathcal{H}_{E}$.  Then the composite Hilbert space of the system and the environment, $\mathcal{H}_{SE}$, is isolated and transforms reversibly under a unitary map with Kraus operator $\hat{U}_{SE}$. After the evolution the final state of the system can be recovered by partial trace, obtaining an overall quantum map
\begin{equation}\label{stinespring}
\Phi(\hat{\rho}_{S})=\ptr{E}{\hat{U}_{SE}\left(\hat{\rho}_{S}\otimes\dketbra{\psi}_{E}\right)\hat{U}_{SE}^{\dag}}.
\end{equation}
Both the Kraus and Stinespring representations are valid descriptions of a quantum map. They can be related by defining the isometry $\hat{V}_{S\rightarrow SE}=\hat{U}_{SE}\ket{\psi}_{E}$ and computing the trace in~\eqq{stinespring} with respect to an orthonormal basis $\{\ket{e_{k}}_{E}\}_{k=1,\cdots,d_{E}}$ of $\mathcal{H}_{E}$, so that $\Phi(\hat{\rho}_{S})=\sum_{k=1}^{d_{E}}{}_{E}\bra{e_{k}}\hat{V}\hat{\rho}_{S}\hat{V}^{\dag}\ket{e_{k}}_{E}$. Then the Kraus operators $\{\hat{M}_{k}\}_{k=1,\cdots,d_{E}}$ of $\Phi$ are connected to the Stinespring unitary $\hat{U}_{SE}$ and reference state $\ket{\psi}_{E}$ by
\begin{align}
&\hat{M}_{k}={}_{E}\bra{e_{k}}\hat{V}_{S\rightarrow SE}={}_{E}\bra{e_{k}}\hat{U}_{SE}\ket{\psi}_{E},\\
&\hat{U}_{SE}\ket{\psi}_{E}=\sum_{k=1}^{d_{E}}\hat{M}_{k}\otimes\ket{e_{k}}_{E}.
\end{align}Finally, one can define the dual $\Phi^{*}$ of a channel $\Phi$ with Kraus representation $\left\{\hat{M}_{k}\right\}_{k=1}^{K}$ as its equivalent in the Heisenberg representation, i.e., 
\begin{equation}\label{dualCh}
\tr{\hat{O}\Phi\left(\hat{\rho}\right)}=\tr{\Phi^{*}\left(\hat{O}\right)\hat{\rho}},
\end{equation}
so that the Kraus representation of $\Phi^{*}$ is trivially given by $\left\{\hat{M}^{\dag}_{k}\right\}_{k=1}^{K}$.\\

A peculiar kind of evolution is that taking place when a measurement is performed on the system. The most general quantum measurement is described by a Positive Operator-Valued Measurement (POVM), i.e., a set of positive operators $\{\hat{E}_{k}\}_{k=1,\cdots,K}$, with $\hat{E}_{k}\geq 0$ and $\sum_{k=1}^{K}\hat{E}_{k}=\hat{\mathbf{1}}$. When applying the POVM to a state $\hat{\rho}$ of the system, each of these operators represents one of $K$ possible outcomes of the measurement, occurring with probability
\begin{equation}\label{outcomeProb}
p(k|\hat{\rho})=\tr{\hat{E}_{k}\hat{\rho}}
\end{equation}
and transforming the state as
\begin{equation}\label{outcomeState}
\hat{\rho}_{k}=\frac{\sqrt{\hat{E}_{k}}\hat{\rho}\sqrt{\hat{E}_{k}}}{p(k|\hat{\rho})}.
\end{equation}
If the measurement outcome is not registered instead, the final state of the system will be an average of~\eqq{outcomeState} with weights~\eqq{outcomeProb}, i.e., $\hat{\rho}_{out}=\sum_{k=1}^{K}\sqrt{\hat{E}_{k}}\hat{\rho}\sqrt{\hat{E}_{k}}$, which corresponds to a quantum map of Kraus operators $\sqrt{\hat{E}_{k}}$ applied to the input state $\hat{\rho}$. \\
Let us note that, by choosing the measurement operators to be orthogonal projectors, i.e., $\hat{E}_{k}=\hat{\Pi}_{k}$ with $\hat{\Pi}_{k}\hat{\Pi}_{j}=\delta_{k,j}\hat{\Pi}_{k}$, the ordinary Von-Neumann projective measurement~\cite{nChuangBOOK} can be recovered. Then there exists a third kind of purification that regards measurements, such that a POVM on a system can always be described as a projective measurement on an enlarged system~\cite{nChuangBOOK}.\\

\subsection{Quantifying distance between quantum states}\label{subsec:distance}
We conclude this introductory section with an operational notion of distance between two quantum states $\hat{\rho},\hat{\sigma}\in\mathfrak{S}(\hil)$, which can be useful in several practical settings, e.g., evaluating the error in the approximate implementation of a quantum operation or estimating the overlap between two desired states produced by some physical process and thus their degree of distinguishability. The two most common distance-like quantities employed in Quantum Information are the trace-distance and the fidelity. The trace-distance is a quantum analogue of the Kolmogorov distance~\cite{nChuangBOOK} between two probability distributions and it is defined as
\begin{equation}\label{trd}
D(\hat{\rho},\hat{\sigma})=\oneover{2}\norm{\hat{\rho}-\hat{\sigma}}_{1}=\oneover{2}\trabs{\hat{\rho}-\hat{\sigma}},
\end{equation}
where $\norm{\cdot}_{1}$ is the trace-norm of an operator.
The trace-distance is a true distance since it is non-negative, symmetric and it satisfies the triangular inequality, i.e.,
\begin{equation}\label{triangTD}
\trd{\hat{\rho}}{\hat{\sigma}}\leq \trd{\hat{\rho}}{\hat{\tau}}+\trd{\hat{\tau}}{\hat{\sigma}}
\end{equation}
for all $\hat{\rho},\hat{\sigma},\hat{\tau}\in\states{}$. Other important properties are:
\begin{enumerate}
\item $\trd{\hat{\rho}}{\hat{\sigma}}\leq1$, with equality if and only if the states are orthogonal, i.e., $\hat{\rho}\hat{\sigma}=\hat{\sigma}\hat{\rho}=0$;
\item It can be written as an optimization over projection operators $\hat{\Pi}$, i.e.,
\begin{equation}
\trd{\hat{\rho}}{\hat{\sigma}}=\max_{\substack{0\leq\hat{\Pi}\leq\hat{\mathbf{1}}~s.t.\\ \hat{\Pi}^{2}=\hat{\Pi}}}\tr{\hat{\Pi}\left(\hat{\rho}-\hat{\sigma}\right)},
\end{equation}
or over arbitrary operators $\hat{\Lambda}$ with absolute value less than $\hat{\mathbf{1}}$, i.e.,
\begin{equation}\label{normLemma}
\trd{\hat{\rho}}{\hat{\sigma}}=\max_{-\hat{\mathbf{1}}\leq\hat{\Lambda}\leq\hat{\mathbf{1}}}\tr{\hat{\Lambda}\left(\hat{\rho}-\hat{\sigma}\right)}.
\end{equation}
Both properties are inherited from the trace-norm;
\item It is non-increasing under physical maps, i.e., for all $\Phi$ LCPT
\begin{equation}\label{trdNIncreasing}
\trd{\Phi\left(\hat{\rho}\right)}{\Phi\left(\hat{\sigma}\right)}\leq\trd{\hat{\rho}}{\hat{\sigma}},
\end{equation}
with equality if and only if $\Phi\equiv\mathcal{U}$ is reversible. For example this implies that states get closer after partial trace, which is a non-reversible physical evolution;
\item It is strongly convex, i.e., given two convex combinations of states $\hat{\rho}=\sum_{k=1}^{K}p_{k}\hat{\rho}_{k}$ and $\hat{\sigma}=\sum_{k=1}^{K}q_{k}\hat{\sigma}_{k}$ it holds
\begin{equation}
\trd{\hat{\rho}}{\hat{\sigma}}\leq\oneover{2}\sum_{k=1}^{K}\abs{p_{k}-q_{k}}+\sum_{k=1}^{K}p_{k}\trd{\hat{\rho}_{k}}{\hat{\sigma}_{k}},
\end{equation}
where the first term on the right-hand side is the Kolmogorov distance between the classical probability distributions forming the convex combinations. Let us note that this property implies other kinds of convexity, such as the ordinary one.
\end{enumerate}
In addition to these general properties there are some that are specifically employed in Quantum Communication Theory, see~\cite{winter,WINTERPHD,wildeBOOK,NOSTROHol}. Here we state three of them in the form of lemmas that will be used to obtain some results in Ch.~\ref{ch:Opt}. For the proofs refer to App.~\ref{app:measLemmas} or~\cite{NOSTROHol}. 
The first lemma is simply a generalization of the non-increasing property of the trace-distance, \eqq{trdNIncreasing}, to the action of single elements of a POVM on unnormalized states, i.e., density operators $\hat{\rho}$ such that $\tr{\hat{\rho}}\leq1$.
\newtheorem{lemma}{Lemma}
\begin{lemma}\label{contra}
  (Contractivity of trace-distance for POVM elements) Let $\hat{\rho},\hat{\sigma}$ be 
  unnormalized states and $\hat{E}$ a positive and 
  less-than-one operator.
  Then it holds
  \begin{equation}\label{contreq}
    \trd{\hat{E}\hat{\rho} \hat{E}}{\hat{E}\hat{\sigma} \hat{E}}\leq \trd{\hat{\rho}}{\hat{\sigma}}.
  \end{equation}
\end{lemma}
 The last two lemmas have been extensively employed in Quantum Shannon Theory~\cite{winter,WINTERPHD,wildeBOOK} when applied to ordinary density operators; here we prove their validity in the case of unnormalized states.
\begin{lemma}\label{appclose}
 (Measurement on approximately close states) Let $\hat{\rho},\hat{\sigma}$ be two unnormalized states. Let $\hat{E}$ be a positive and less-than-one 
 operator, i.e. $0\leq \hat{E}\leq\hat{\mathbf{1}}$. Then it holds
 \begin{equation}\label{appcloseq}
   \tr{\hat{E}\hat{\rho}}\geq\tr{\hat{E}\hat{\sigma}}-2D(\hat{\rho},\hat{\sigma}).
 \end{equation} 
\end{lemma}

\begin{lemma}\label{gentop}
  (Gentle operator) Let $\hat{\rho}$ be an unnormalized state and $\hat{E}$ a positive and less-than-one operator. Let also $\ave{\cdot}$ denote the average with respect to some probability distribution, which $\hat{\rho}$ and $\hat{E}$ may depend on. 
  Suppose that, for some $\epsilon>0$,
  \begin{equation}\label{gentopequno}
    \ave{\tr{\hat{E}\hat{\rho}}}\geq 1-\epsilon.
  \end{equation}
  Then it holds
  \begin{equation}\label{gentopeqdue}
    \ave{D\left(\sqrt{\hat{E}}\hat{\rho}\sqrt{\hat{E}},\hat{\rho}\right)}\leq \sqrt{\epsilon}.
  \end{equation}
\end{lemma}

Let us now discuss the fidelity of two states $\hat{\rho}, \hat{\sigma}$, defined as
\begin{equation}
F(\hat{\rho},\hat{\sigma})=\left(\norm{\sqrt{\hat{\rho}}\sqrt{\hat{\sigma}}}_{1}\right)^{2}=\tr{\sqrt{\sqrt{\hat{\rho}}\hat{\sigma}\sqrt{\hat{\rho}}}}^{2}
\end{equation}
The fidelity is positive and symmetric but, unlike the trace-distance, it does not directly satisfy the triangle inequality so it is not a proper distance. However it can be used to define a true distance measure. Apart from that, it gives a qualitative idea of the distance between two states and it shares many other properties with the trace-distance, though in a converse way:
\begin{enumerate}
\item $0\leq F(\hat{\rho},\hat{\sigma})\leq 1$ with equality on the left if and only if the states are orthogonal, while on the right if $\hat{\rho}=\hat{\sigma}$;
\item If one of the states is pure we have
\begin{equation}\label{fidpure}
\fid{\ket{\psi}}{\hat{\sigma}}=\bra{\psi}\hat{\sigma}\ket{\psi},
\end{equation}
which is equal to the states' overlap;
\item It can be written as an optimization of the overlap between arbitrary purifications of the two states (by Uhlmann's theorem~\cite{uhlmann}), i.e.,
\begin{equation}
\fid{\hat{\rho}_{S}}{\hat{\sigma}_{S}}=\max_{\ket{\rho}_{SR},\ket{\sigma}_{SR}}\fid{\ket{\rho}_{SR}}{\ket{\sigma}_{SR}};
\end{equation}
\item It is non-decreasing under physical maps, i.e.,
\begin{equation}\label{nondec}
\fid{\hat{\rho}}{\hat{\sigma}}\leq\fid{\Phi(\hat{\rho})}{\Phi(\hat{\sigma})},
\end{equation}
with equality if and only if the transformation is reversible;
\item It is strongly concave, i.e.,
\begin{equation}
\fid{\sum_{k=1}^{K}p_{k}\hat{\rho}_{k}}{\sum_{k=1}^{K}q_{k}\hat{\sigma}_{k}}\geq\sum_{k=1}^{K}\sqrt{p_{k}q_{k}}\fid{\hat{\rho}_{k}}{\hat{\sigma}_{k}}.
\end{equation}
\end{enumerate}
Eventually, the trace-distance and fidelity of two states are related as follows:
\begin{equation}
1-\sqrt{\fid{\hat{\rho}}{\hat{\sigma}}}\leq\trd{\hat{\rho}}{\hat{\sigma}}\leq\sqrt{1-\fid{\hat{\rho}}{\hat{\sigma}}}.
\end{equation}

\subsection{Qubit states}\label{subsec:qubits}
Until now we have made no specific example of quantum states and transformations. The easiest case is that of qubits, i.e., quantum states of a bidimensional Hilbert space, $\hil_{2}$. Any density operator $\hat{\rho}\in\states{2}$ can be written in terms of three real parameters defining a real vector $\vecund{r}_{\rho}$\footnote{Note that here and in the following we will drop the hat notation for qubit operators when they appear as indexes, for simplicity.} inside a three-dimensional unit sphere (the Bloch sphere), i.e. 
  \begin{equation}\label{blochStates}
  \hat{\rho}=\frac{\hat{\mathbf{1}}_{2}+\vecund{r}_{\rho}\cdot\vecund{\hat{\sigma}}}{2},
  \end{equation}
   where $\mathbf{1}_{2}$ is the identity operator of the space and  $\vecund{\hat{\sigma}} = (\hat{\sigma}_1,\hat{\sigma}_2,\hat{\sigma}_3)$ is the vector of Pauli matrices 
\begin{equation} 
\hat{\sigma}_{1}=\left(\begin{array}{cc}0&1\\1&0\end{array}\right), \;  \hat{\sigma}_{2}=\left(\begin{array}{cc}0&-i\\i&0\end{array}\right),\; \hat{\sigma}_{3}=\left(\begin{array}{cc}1&0\\0&-1\end{array}\right).
\end{equation} 
In particular, pure states are situated on the sphere's surface, i.e., $r_{\rho}=|\vecund{r}_{\rho}|=1$ for $\hat{\rho}=\dketbra{\psi}$, while the completely mixed state $\hat{\mathbf{1}}_{2}/2$ is at the origin. 
More generally, any hermitian operator $\hat{H}$ on $\hil_{2}$ can be expressed in terms of four real coefficients: a scalar $c_{H}$, which represents its normalization coefficient, and a vector $\vecund{r}_{H}$, which represents the operator in the Bloch space, i.e. 
 \begin{equation}\label{blochHerm}
  \hat{H}=c_{H}\hat{\mathbf{1}}_{2}+\vecund{r}_{H}\cdot\vecund{\hat{\sigma}},
 \end{equation} 
the trace of the operator being determined by $\operatorname{Tr}[\hat{H}]=2c_{H}$, while its eigenvalues by $\lambda_{H}^{(\pm)}=c_{H}\pm r_{H}$ with $r_H= |\vecund{r}_H|$. \\
Unitary transformations of qubits can be written in terms of the identity and Pauli matrices as
\begin{equation}
\hat{U}_{\vecund{n},\theta}=e^{i\theta \vecund{n}\cdot\vecund{\hat{\sigma}}}=\hat{\mathbf{1}}_{2}\cos\theta+i\vecund{n}\cdot\vecund{\hat{\sigma}}\sin\theta,
\end{equation}
i.e., a rotation of angle $\theta$ around the direction $\vecund{n}$ in Bloch space~\cite{nChuangBOOK}.
A generic qubit map $\Phi$ instead can be written in a canonical form due to King and Ruskai~\cite{ruskaiQubitMaps1}:
\begin{equation}
\Phi\left(\frac{\hat{\mathbf{1}}_{2}+\vecund{r}_{\rho}\cdot\vecund{\hat{\sigma}}}{2}\right)=\frac{\hat{\mathbf{1}}_{2}+(T\vecund{r}_{\rho}+\vecund{t})\cdot\vecund{\hat{\sigma}}}{2},
\end{equation}
so that it maps a qubit with Bloch vector $\vecund{r}$ into $T\vecund{r}+\vecund{t}$, with $T$ a square real matrix of order $3$ and $\vecund{t}\in\mathbb{R}^{3}$. Moreover $T$ can be always diagonalized via proper unitary transformations and several properties of $\Phi$ are determined by its three eigenvalues, e.g., complete positivity~\cite{ruskaiQubitMaps}. Some examples of typical qubit channels are discussed in~\cite{nChuangBOOK}.

\section{Bosonic Gaussian Systems}\label{sec:basicsGauss}
\subsection{The Hilbert space of a bosonic field}
Let us now focus on the case of Continuous-Variables (CV) systems~\cite{RevGauss,RevGauss2,RevGauss3,serafiniBOOK} that are represented by operators on an infinite-dimensional Hilbert space. In particular, we are interested in the description of excitations of the electromagnetic field, i.e., photons, that are commonly used to transfer information. Each mode of the field is represented by the state of a quantum harmonic oscillator of frequency $\nu=\omega/(2\pi)$ and hamiltonian $\hat{H}=\omega \hat{n}$\footnote{Here and in the rest of the thesis the measurement units will be chosen such that Planck's constant $\hbar=1$ w.l.o.g.}, where $\hat{n}$ is the photon-number operator. Its eigenstates are labelled by the number of photonic excitations and form the Fock basis $\{\ket{n}\}$ of the system. The photon-number operator can be written as $\hat{n}=\hat{a}^{\dag}\hat{a}$, the composition of the bosonic ladder operators $\hat{a}$, $\hat{a}^{\dag}$, whose action on Fock states is
\begin{equation}\label{ladder}
\begin{aligned}
&\hat{a}\ket{n}=\sqrt{n}\ket{n-1},\\
&\hat{a}^{\dag}\ket{n}=\sqrt{n+1}\ket{n+1},
\end{aligned}
\end{equation}
i.e., respectively to decrease and increase the number of photons by one. 
These satisfy the Bosonic Commutation Relations (BCR)
\begin{equation}
\left[\hat{a},\hat{a}^{\dag}\right]=\hat{a}\hat{a}^{\dag}-\hat{a}^{\dag}\hat{a}=1
\end{equation} 
and are related to the field position and momentum quadratures, respectively $\hat{q}$ and $\hat{p}$, by a unitary matrix $\gamma$:
\begin{align}\label{quadratures}
\left(\begin{array}{c}\hat{q}\\\hat{p}\end{array}\right)=\gamma\left(\begin{array}{c}\hat{a}\\\hat{a}^{\dag}\end{array}\right),\quad
\gamma=\frac{1}{\sqrt{2}}\left(\begin{array}{cc}1&1\\-i&i\end{array}\right).
\end{align}\\
More generally the Hilbert space of $N$ modes of the field is described by the tensor product of $N$ harmonic oscillators with $2N$ bosonic ladder operators $\{\hat{a}_{i}$, $\hat{a}_{i}^{\dag}\}_{j=1,\cdots,N}$. The latter can be arranged in separate column vectors $\vecop{a}=(\hat{a}_{1},\cdots,\hat{a}_{N})^{T}$ and $\vecop{a}^{*}=(\hat{a}_{1}^{\dag},\cdots,\hat{a}_{N}^{\dag})^{T}$, or in a common one $\underline{\hat{A}}=(\hat{a}_{1},\hat{a}_{1}^{\dag},\cdots,\hat{a}_{N},\hat{a}_{N}^{\dag})^{T}$ and satisfy the BCR
\begin{equation}\label{bcr}
\left[\hat{A}_{j},\hat{A}_{k}\right]=\Omega_{j k}, \quad \Omega=\bigoplus_{i=1}^{N}\omega=\bigoplus_{i=1}^{N}\left(\begin{array}{cc}0&1\\-1&0\end{array}\right).
\end{equation}
The matrix $\Omega$ is called symplectic form and it plays an important role in the description of CV systems, as will be shown in the following. It is orthogonal and antisymmetric, i.e., 
$\Omega^{-1}=\Omega^{T}=-\Omega$, 
and has the property
\begin{equation}\label{sympproperty}
M^{T}\Omega M=\bigoplus_{i=1}^{N}\operatorname{Det}(M_{i})\omega\quad\forall M=\bigoplus_{i=1}^{N}M_{i}.
\end{equation}
The vector of quadratures $\vecop{r}=(\hat{q}_{1},\hat{p}_{1},\cdots,\hat{q}_{N},\hat{p}_{N})^{T}$ is provided by a multi-mode extension of~\eqq{quadratures}, i.e.,
 \begin{equation}\label{switch}
 \vecop{r}=\Gamma\vecop{A},\quad \Gamma=\bigoplus_{i=1}^{N}\gamma,
 \end{equation} 
 and its components satisfy the Canonical Commutation Relations (CCR)
\begin{equation}\label{ccr}
\comm{\hat{r}_{j}}{\hat{r}_{k}}=\sum_{\ell,m}\Gamma_{j \ell} \Gamma_{k m} \comm{\hat{A}_{\ell}}{\hat{A}_{m}}=\left(\Gamma\Omega\Gamma^{T}\right)_{j k}=i\Omega_{j k},
\end{equation}
where we have used the property~\eqq{sympproperty} with $M=\Gamma$.\\
The most important class of CV states is that of Gaussian states, due to their simple theoretical description, physical implementation and manipulation. These are all the states $\hat{\rho}_{G}$ than can be completely described in terms of their first and second moments, respectively the vector of mean values $\vecund{m}\in\mathbb{R}^{2N}$ and the real $2N\times 2N$ positive symmetric covariance matrix~$V$:
\begin{equation}\label{moments}
\vecund{m}=\tr{\vecop{r} \hat{\rho}_{G}}, \quad V_{j k}=\oneover{2}\tr{(\hat{r}_{j}\hat{r}_{k}+\hat{r}_{k}\hat{r}_{j})\hat{\rho}_{G}}.
\end{equation}
Equivalently, by inverting the transformation~\eqq{switch} one can define complex-valued moments of the ladder operators $\vecop{A}$, which will be indicated with a $\tilde{\cdot}$, i.e., $\veclad{m}=\Gamma^{\dag}\vecund{m}$ and $\tilde{V}=\Gamma^{\dag}V\Gamma^{*}$. Better characterizations of Gaussian states can be given both in Hilbert space and phase space, starting from their definition in terms of first and second moments. \\
Before proceeding, let us note that the quantum uncertainty principle together with positivity of any quantum state imposes a necessary condition~\cite{gaussOpsStates}, stricter than positivity, on the covariance matrix of any quantum state: 
\begin{equation}\label{uncertainty}
\vecund{v}^{T}\left(V+\frac{i}{2}\Omega\right)\vecund{v}=\tr{(\vecund{v}\cdot\vecop{r})^{2}\hat{\rho}}\geq 0\quad\forall\vecund{v}\in\mathbb{R}^{2N}\quad\Leftrightarrow\quad V+\frac{i}{2}\Omega\geq0.
\end{equation}
In particular, this constraint is also sufficient to ensure that a Gaussian state is physical.

\subsection{Weyl operators and the phase-space representation}\label{phaseSpace}
Let us first analyze the phase-space description and define the multi-mode Weyl (or displacement) operator of real parameters $\vecund{r}\in\mathbb{R}^{2N}$ (latin letters) as 
\begin{equation}\label{dispreal}
\hat{D}(\vecund{r})=e^{i\vecund{r}^{T}\Omega^{T}\vecop{r}}.
\end{equation}
An alternative definition employs the ladder operators and the complex parameters $\vecund{\alpha}\in\mathbb{C}^{N}$ (greek letters):
\begin{equation}\label{dispcomplex}
\hat{D}(\vecund{\alpha})=\bigotimes_{i=1}^{N}e^{\alpha_{i}\hat{a}_{i}^{\dag}-\alpha_{i}^{*}\hat{a}_{i}}.
\end{equation}
The two representations of Eqs.~(\ref{dispreal},\ref{dispcomplex}) turn out to be equivalent upon setting $\veclad{m}=\veclad{m}_{\alpha}=(\alpha_{1},\alpha^{*}_{1},\cdots,\alpha_{N},\alpha^{*}_{N})$, so that 
\begin{equation}
\hat{D}(\vecund{\alpha})=e^{\veclad{m}^{T}\Omega\vecop{A}}=e^{\veclad{m}^{T}\Gamma^{*}\Omega\Gamma^{\dag}\vecop{r}}=e^{-i\vecund{r}^{T}\Omega\vecop{r}}=\hat{D}(\vecund{r}).
\end{equation}
Let us further note that the commutation relations, Eqs.~(\ref{bcr},\ref{ccr}), imply a composition rule for Weyl operators through the well-known Baker-Campbell-Hausdorff (BCH) relations~\cite{puriBOOK}
\begin{align}
\label{bchxyx}&e^{\hat{X}}\hat{Y}e^{-\hat{X}}=\hat{Y}+\sum_{j=1}^{\infty}\oneover{j!}\comm{\hat{X}}{\hat{Y}}_{j},\\
\label{bchxy}&e^{\hat{X}}e^{\hat{Y}}=e^{\hat{X}+\hat{Y}}e^{\oneover{2}\comm{\hat{X}}{\hat{Y}}}\quad\text{if } \comm{\hat{X}}{\hat{Y}}_{2}=0,
\end{align}
with
\begin{equation}
 \comm{\hat{X}}{\hat{Y}}_{j}=\begin{cases}\comm{\hat{X}}{\comm{\hat{X}}{\hat{Y}}_{j-1}}& \forall j\geq 2,\\
\comm{\hat{X}}{\hat{Y}}& \text{for } j=1.\end{cases}
\end{equation}
Indeed we can apply \eqref{bchxy} with $\hat{X}=i\vecund{r}^{T}\Omega^{T}\vecop{r}$ and $\hat{Y}=i\vecund{s}^{T}\Omega^{T}\vecop{r}$ to obtain 
\begin{equation}\label{dispcomposition}\begin{aligned}
\hat{D}(\vecund{r})\hat{D}(\vecund{s})&=\hat{D}(\vecund{r}+\vecund{s})e^{\frac{i}{2}\vecund{r}^{T}\Omega^{T}\vecund{s}},\\
\hat{D}(\vecund{\alpha})\hat{D}(\vecund{\beta})&=\hat{D}(\veclad{m}_{\alpha}+\veclad{s}_{\beta})e^{\oneover{2}\veclad{m}_{\alpha}^{T}\Omega\veclad{s}_{\beta}}=\bigotimes_{i=1}^{N}\hat{D}(\alpha_{i}+\beta_{i})e^{\oneover{2}(\alpha_{i}\beta_{i}^{*}-\alpha_{i}^{*}\beta_{i})}.
\end{aligned}\end{equation}
The Weyl operator is also called displacement operator because it displaces the quadrature and ladder operators when acting on them as a unitary transformation in Heisenberg representation; this can be seen by applying \eqref{bchxyx} with $\hat{X}=i\vecund{r}^{T}\Omega\vecop{r}$ and $\hat{Y}=\hat{r}_{i}$ for all $i=1,\cdots,N$ to obtain
\begin{equation}\label{dispaction}
\hat{D}^{\dag}(\vecund{r})\vecop{r}\hat{D}(\vecund{r})=\vecop{r}+\vecund{r},\quad \hat{D}^{\dag}(\vecund{\alpha})\vecop{a}\hat{D}(\vecund{\alpha})=\vecop{a}+\vecund{\alpha}.
\end{equation}
Thus the displacement operator can be used to define coherent states of the field when applied to a vacuum state with no photons:
$\hat{D}(\alpha)\ket{0}=\ket{\alpha}$. In the rest of this subsection we will refer to the one-mode case, since the multi-mode generalization is straightforward; for example the notation $\ket{\vecund{\alpha}}=\bigotimes_{i=1}^{N}\ket{\alpha_{i}}_{i}$ can be used for a tensor product of coherent states on several modes. These states are perhaps better known as eigenstates of the photon-number decreasing operator, i.e., $\hat{a}\ket{\alpha}=\alpha\ket{\alpha}$, and they describe the pulses produced by a laser~\cite{mandelBOOK}, often used to transmit information on optical fibers or in free space. Their decomposition in the Fock basis can be obtained by computing the overlaps $\braket{n}{\alpha}$ via \eqq{dispaction} and renormalizing the resulting state:
\begin{equation}\label{coherentstates}
\ket{\alpha}=e^{-\oneover{2}\abs{\alpha}^{2}}\sum_{n=0}^{\infty}\frac{\alpha^{n}}{\sqrt{n!}}\ket{n}.
\end{equation}
Using the latter expansion in the Fock basis it can be seen that coherent states form an overcomplete set of operators, i.e., they sum up to the identity operator but they have non-zero overlap,
\begin{align}
&\int \frac{d^{2}\alpha}{\pi}~\dketbra{\alpha}=\hat{\mathbf{1}},\\
&\braket{\alpha}{\beta}=e^{-\frac{\abs{\alpha-\beta}^{2}+\alpha\beta^{*}-\alpha^{*}\beta}{2}},\label{coherentNotOrth}
\end{align}
where the two-dimensional integral is over the real and imaginary parts of $\alpha\in\mathbb{C}$, equivalent to an integral over the corresponding real components of $\vecund{r}\in\mathbb{R}^{2}$ up to a factor of $1/2$. Hence any finite-trace operator can be represented as a coherent-state sum with weights given by its coherent-state averages, i.e.,
\begin{equation}\label{coherentdeco}
\hat{O}=\int \frac{d^{2}\alpha~d^{2}\beta}{\pi^{2}}~\bra{\alpha}\hat{O}\ket{\alpha}\ketbra{\alpha}{\beta}.
\end{equation} \\
The set of Weyl operators instead is complete and orthogonal, thus it forms a basis for the operators on the bosonic Hilbert space with respect to the usual Hilbert-Schmidt product, i.e., $\left(\hat{O},\hat{Q}\right)=\tr{\hat{O}^{\dag}\hat{Q}}$. Indeed by applying Eqs.~(\ref{dispaction},\ref{coherentdeco}) and the definition of the complex delta function, 
\begin{equation}
\delta(\alpha)=\int \frac{d^{2}\beta}{\pi^{2}}~e^{\alpha\beta^{*}-\alpha^{*}\beta}=2\delta^{2}(\vecund{r}_{\alpha}),
\end{equation}
the following orthogonality relations can be proved:
\begin{equation}\label{orthogonalitydisp}
\begin{aligned}
\tr{\hat{D}(-\vecund{r})\hat{D}(\vecund{s})}&=2\pi\delta^{2}(\vecund{r}-\vecund{s}),\\
\tr{\hat{D}(-\alpha)\hat{D}(\beta)}&=\pi\delta(\alpha-\beta).
\end{aligned}
\end{equation}
Then the representation of a state $\hat{\rho}$ in this basis has coefficients known as its characteristic function $\chi_{\hat{\rho}}$:
\begin{equation}\label{charfunction}
\begin{aligned}
\chi_{\hat{\rho}}(\vecund{r})&=\tr{\hat{\rho}\hat{D}(\vecund{r})}\leftrightarrow\hat{\rho}=\int\frac{d^{2}\vecund{r}}{2\pi}~\chi_{\hat{\rho}}(\vecund{r})\hat{D}(-\vecund{r}),\\
\chi_{\hat{\rho}}(\alpha)&=\tr{\hat{\rho}\hat{D}(\alpha)}\leftrightarrow\hat{\rho}=\int\frac{d^{2}\alpha}{\pi}~\chi_{\hat{\rho}}(\alpha)\hat{D}(-\alpha).\\
\end{aligned}
\end{equation}
The phase-space representation of $\hat{\rho}$ is obtained by taking the Fourier transform of its characteristic function times a scale factor,
\begin{equation}\label{wignereco}
\begin{aligned}
W_{\hat{\rho}}^{(u)}(\vecund{r})&=\int\frac{d^{2}\vecund{s}}{2\pi^{2}}~\chi_{\hat{\rho}}(\vecund{s})e^{\frac{u}{2}\abs{\vecund{s}}^{2}}e^{-i\vecund{s}^{T}\Omega^{T}\vecund{r}},\\
W_{\hat{\rho}}^{(u)}(\alpha)&=\int\frac{d^{2}\beta}{\pi^{2}}~\chi_{\hat{\rho}}(\beta)e^{\frac{u}{2}\abs{\beta}^{2}}e^{\alpha\beta^{*}-\alpha^{*}\beta},
\end{aligned}
\end{equation}
where the value of $u=0,\pm1$ determines the peculiar properties of the quasi-probability distribution $W^{(u)}_{\hat{\rho}}$, which is normalized but can take negative values. \\
For $u=1$ we obtain the Glauber-Sudarshan~\cite{glauber,sudarshan} P-representation, $W^{(1)}_{\hat{\rho}}(\alpha)=P(\alpha)$, which is employed to write any density matrix $\hat{\rho}$ as an integral over coherent states with possibly negative coefficients:
\begin{equation}
\hat{\rho}=\int d^{2}\alpha~P_{\hat{\rho}}(\alpha)\dketbra{\alpha}.
\end{equation}
For example a coherent state $\dketbra{\beta}$ has $P_{\beta}(\alpha)=\delta(\alpha-\beta)$ and in general a classical state is identified by a P-function not more singular than a delta function, while Fock states can be expressed in terms of derivatives thereof. \\
For $u=0$ instead we obtain the Wigner function~\cite{wigner} of the state, $W^{(0)}_{\hat{\rho}}(\vecund{r})=W_{\hat{\rho}}(\vecund{r})$, which is its most common representation in the phase-space $\mathbb{R}^{2}$ determined by the quadratures' mean values, $(q,p)$. This is a quasi-probability distribution that attains negative values on non-classical states. Moreover its marginal with respect to a given quadrature provides, up to a normalization factor, the probability distribution of a measurement of its conjugate quadrature, i.e.,
\begin{equation}\label{marginalwigner}
\oneover{2}\int dp~W_{\hat{\rho}}(q,p)= \bra{q}\hat{\rho}\ket{q},
\end{equation}
and this holds true for any couple of rotated quadratures. Moreover the operation \eqq{marginalwigner} can be performed in practice via a heterodyne measurement, as will be discussed later.\\
Finally for $u=-1$ we obtain the Husimi Q-function~\cite{husimi}, $W^{(-1)}_{\hat{\rho}}(\alpha)=Q_{\hat{\rho}}(\alpha)$, which is always positive hence a true probability distribution. Indeed it is easy to show that
\begin{equation}
Q(\alpha)=\oneover{\pi}\bra{\alpha}\hat{\rho}\ket{\alpha},
\end{equation}
i.e., the Husimi function represents the probability of measuring the coherent-state amplitude $\alpha$ via a heterodyne measurement, see \eqq{heterodyne} below. \\
Let us now consider Gaussian states, $\hat{\rho}_{G}$, characterized by their first and second moments, $\vecund{m}$, $V$. They can be equivalently defined as those states whose characteristic and Wigner functions are Gaussian, i.e., 
\begin{align}\label{gaussianchar}
\chi_{\hat{\rho}_{G}}(\vecund{r})&=e^{-\oneover{2}\vecund{r}^{T}\Omega^{T}V\Omega\vecund{r}+i\vecund{r}^{T}\Omega^{T}\vecund{m}},\\
\label{gaussianwigner}W_{\hat{\rho}_{G}}(\vecund{r})&=\frac{e^{-\oneover{2}(\vecund{r}-\vecund{m})^{T}V^{-1}(\vecund{r}-\vecund{m})}}{\pi\sqrt{\det{V}}}.
\end{align}
Eventually let us note a useful property of Gaussian states: the trace of their product is still a Gaussian function of their moments, i.e., if $\hat{\rho}_{G}$ and $\hat{\sigma}_{G}$ have moments respectively $\vecund{m}$, $V$ and $\vecund{m}'$, $V'$ then
\begin{equation}\label{gaussProduct}
\tr{\hat{\rho}_{G}\hat{\sigma}_{G}}=\int \frac{d^{2}\vecund{r}}{2\pi}~\chi_{\hat{\rho}_{G}}(\vecund{r})\chi_{\hat{\sigma}_{G}}(-\vecund{r})=\frac{e^{-\oneover{2}(\vecund{m}-\vecund{m}')^{T}(V+V')^{-1}(\vecund{m}-\vecund{m}')}}{\sqrt{\det(V+V')}},
\end{equation}
where we have used the relations Eqs.~(\ref{orthogonalitydisp},\ref{charfunction}) and performed a $2N$-dimensional Gaussian integral. More generally, the first equality above holds for any bosonic quantum state, not just Gaussian ones. 

\subsection{Gaussian unitaries and the Hilbert-space representation}\label{gaussUni}
Given the importance of Gaussian states, it is natural to define the class of transformations that preserve them as Gaussian too. In particular, Gaussian reversible transformations are represented by all those unitaries $\hat{U}_{G}=e^{-i\hat{H}_{q}}$ whose generating Hamiltonian is quadratic in the quadrature (or ladder) operators, i.e., 
\begin{equation}\label{quadraticHam}
\hat{H}_{q}=\vecop{r}^{T}H\vecop{r}+\vecund{h}^{T}\vecop{r},
\end{equation}
with $H$ a $2N\times2N$ real positive and symmetric matrix and $\vecund{h}\in\mathbb{R}^{2N}$. Indeed this form ensures that the quadrature operators transform in the Heisenberg representation into quadratic combinations of themselves, sending Gaussian states into Gaussian states. In particular, the linear part of the transformation is carried out by displacement operators, as already shown in \eqq{dispaction}, so that we can directly study the purely quadratic case with $\vecund{h}=0$ w.l.o.g. Let us then consider the action of a purely quadratic Hamiltonian on the quadrature operators: since $\comm{\vecop{r}^{T}\hat{H}\vecop{r}}{\vecop{r}}_{k}=(-2i\Omega H)^{k}\vecund{r}$, by applying \eqq{bchxyx} we have
\begin{equation}\label{similarity}
e^{\frac{i}{2}\vecop{r}^{T}\hat{H}\vecop{r}}~\hat{r}~e^{-\frac{i}{2}\vecop{r}^{T}\hat{H}\vecop{r}}=e^{\Omega H}\vecop{r},
\end{equation}
so that the purely quadratic evolution acts on the quadratures vector as a matrix $S=e^{\Omega H}$. Note that the first and second moments of an arbitrary state transform accordingly as $\vecund{r}\rightarrow S\vecund{r}$ and $V\rightarrow SVS^{T}$.\\
The matrix $S$ is symplectic, i.e., it leaves invariant the symplectic form \eqref{bcr},
\begin{equation}\label{symp}
S\Omega S^{T}=\Omega.
\end{equation}
Hence in order to understand the possible transformations generated by a purely quadratic evolution one needs to classify symplectic matrices; as shown in~\cite{sympGroup}, any real symplectic matrix $S$ of order $2N$ can be written in the Euler decomposition as the product 
\begin{equation}\label{sympMatDec}S=O_{1}DO_{2},\end{equation}
with $O_{1,2}$ symplectic orthogonal matrices and $D=\operatorname{diag}(d_{1},d_{1}^{-1},\cdots,d_{N},d_{N}^{-1})$ symplectic diagonal. The matrices $O_{1,2}$ have a corresponding unitary action, \eqq{similarity}, that is called ``passive'' since it preserves the total photon number, while for $D$ the unitary is called ``active'' since it does not preserve the total photon number. As for the hamiltonians generating these unitaries, they have a simple representation in terms of the ladder operators: passive Gaussian unitaries are given by hamiltonians with interactions of the kind $\hat{a}_{j}^{\dag}\hat{a}_{k}$, i.e., between photon-number-increasing and -decreasing operators, while active Gaussian ones are given by hamiltonians with interactions of the kind $\hat{a}_{j}^{\dag}\hat{a}_{k}^{\dag}$, i.e., between ladder operators of the same species. In the following we will discuss the most peculiar Gaussian unitary transformations, discussing their unitary representation with ladder operators and symplectic representation with respect to quadratures in phase-space:
\begin{enumerate}
\item The phase-shifter describes the phase change of a beam with respect to a reference one, usually implemented by sending the beam through a material with a refractive index different from that of the transmission medium. For a phase change of $\phi\in[0,2\pi)$ on a single mode the photon-number decreasing operator transforms as $\hat{a}\rightarrow\hat{a}e^{-i\phi}$, while the unitary and symplectic representations are, respectively,
\begin{equation}\label{ps}
\hat{U}_{p}(\phi)=e^{-i\phi\hat{a}^{\dag}\hat{a}},\quad S_{p}(\phi)=\left(\begin{array}{cc}
\cos\phi&\sin\phi\\
-\sin\phi&\cos\phi
\end{array}\right).
\end{equation}
In particular the symplectic representation is orthogonal, since the transformation is passive;

\item The single-mode squeezing describes the degenerate parametric down-conversion obtained by pumping a non-linear crystal with a laser, producing single-mode states that have enhanced fluctuations on one quadrature and damped fluctuations on its conjugate. For a squeezing amount of $r\in\mathbb{R}$ on a single mode the photon-number decreasing operator transforms as $\hat{a}\rightarrow\hat{a}\cosh r-\hat{a}^{\dag}\sinh r$, while the unitary and symplectic representations are, respectively,
\begin{equation}\label{1msq}
\hat{U}_{sq}(r)=e^{-\frac{r}{2}\left(\hat{a}^{\dag}{}^{2}-\hat{a}^{2}\right)},\quad S_{sq}(r)=\left(\begin{array}{cc}
e^{-r}&0\\
0&e^{r}
\end{array}\right).
\end{equation}
This is clearly an active operation, since the symplectic representation is not orthogonal;

\item The beam splitter describes the mixing of two beams colliding on different sides of a partially reflecting mirror. For a transmission coefficent $\cos^{2}\theta$, $\theta\in[0,2\pi)$, the photon-number decreasing operators of the two modes transform as \begin{equation}\hat{a}_{k}\rightarrow\hat{a}_{k}\cos\theta+(-1)^{k}\hat{a}_{\bar{k}}\sin \theta,\end{equation} with $k=0,1$ and $\bar{k}$ its mod-2 complementary, while the unitary and symplectic representations are, respectively,
\begin{equation}\label{bs}
\hat{U}_{bs}(\theta)=e^{-\theta\left(\hat{a}_{0}^{\dag}\hat{a}_{1}-\hat{a}_{0}\hat{a}_{1}^{\dag}\right)},\quad S_{bs}(\theta)=\left(\begin{array}{cc}
\mathbf{1}\cos\theta&\mathbf{1}\sin\theta\\
-\mathbf{1}\sin\theta&\mathbf{1}\cos\theta
\end{array}\right),
\end{equation}
where $\mathbf{1}$ indicates the identity matrix of order $2$. Once again, $S_{bs}$ is orthogonal and the beam splitter is a passive operation. The beam-splitter is called balanced when $\theta=\pi/4$. Let us further note that a single-mode displacement operation $\hat{D}(\beta)$, \eqq{dispcomplex}, can be implemented on an arbitrary state $\hat{\rho}$ of a system by combining its state with a strong local oscillator, i.e., a coherent state $\ket{\beta}$ with complex amplitude $\beta=b e^{i\phi}$ and intensity $b^{2}\gg\tr{\hat{n}\hat{\rho}}$, on a beam-splitter, \eqq{bs}, of low reflectivity, i.e., $\sin\theta\ll1$. Indeed it is easy to show~\cite{displacementLimit,displacementLimit1} that if $\sin\theta\rightarrow0$ with $\beta\sin\theta$ finite, the resulting operation after tracing out the local oscillator mode coincides with $\hat{D}(\alpha)$;

\item The two-mode squeezing describes the nondegenerate parametric down-conversion obtained by pumping a non-linear crystal with a laser, producing a couple of beams that are strongly correlated. For a squeezing amount of $r\in\mathbb{R}$ the photon-number decreasing operators of the two modes transform as\begin{equation}\hat{a}_{k}\rightarrow\hat{a}_{k}\cosh r+\hat{a}_{\bar{k}}^{\dag}\sinh r,\end{equation} while the unitary and symplectic representations are, respectively,
\begin{equation}\label{2sq}
\hat{U}_{2sq}(r)=e^{-r\left(\hat{a}_{0}^{\dag}\hat{a}_{1}^{\dag}-\hat{a}_{0}\hat{a}_{1}\right)},\quad S_{2sq}(r)=\left(\begin{array}{cc}
\mathbf{1}\cosh r&\sigma_{3}\sinh r\\
\sigma_{3}\sinh r&\mathbf{1}\cosh r
\end{array}\right),
\end{equation}
where $\sigma_{3}=\operatorname{diag}(1,-1)$ is one of the Pauli matrices that generate the SU(2) algebra~\cite{nChuangBOOK}. Let us note that the transformation is active, since the symplectic representation is not orthogonal. Let us furthermore note that the two-mode squeezing can be also produced from two single-mode squeezing operations by applying proper beam-splitter transformations before and afterwards, i.e., $\hat{U}_{2sq}(r)=\hat{U}_{bs}(\pi/4)\hat{U}_{sq}(r)^{\otimes 2}\hat{U}_{bs}(-\pi/4)$;

\item Any $N$-mode interferometer comprising beam-splitters and phase-shifters is described by a generic $N$-mode passive Gaussian unitary. Its action on the photon-number decreasing vector of operators is $\vecop{a}\rightarrow U\vecop{a}$, where $U$ is a unitary matrix of order $N$ defined by $N(N+1)/2$ complex parameters. A useful decomposition of any such matrix for implementation purposes employs a set of $N(N-1)/2$ two-mode passive Gaussian unitaries, i.e., compositions of beam-splitters, \eqq{bs}, and phase-shifters, \eqq{ps}, plus $N$ additional phase-shifters, see~\cite{zeilinger}.

\end{enumerate}
In light of the previous discussion, let us note that, through the Euler decomposition \eqref{sympMatDec} in phase-space, any Gaussian unitary has a corresponding unitary decomposition as
\begin{equation}\label{passActDeco}
\hat{U}_{S}=\hat{U}_{O_{1}}\left[\bigotimes_{i=1}^{N}\hat{U}_{sq}(\log d_{i})\right]\hat{U}_{O_{2}},
\end{equation} 
where $\hat{U}_{O_{1,2}}$ are $N$-mode interferometers and $\hat{U}_{sq}(\log d_{i})$ are single-mode squeezers of parameter determined by the diagonal matrix $D$. \\

Eventually let us discuss a third equivalent definition of Gaussian states as all the ground and thermal states of quadratic positive definite hamiltonians,   \eqq{quadraticHam}. Thanks to this definition a Gaussian state can be written compactly as
\begin{equation}\label{thermalAndGroundStates}
\hat{\rho}_{G}=\frac{e^{-\beta \hat{H}_{q}}}{\tr{e^{-\beta \hat{H}_{q}}}}, 
\end{equation}
with $\beta\geq0$ an inverse-temperature-like parameter. In particular, for $\beta=0$ we obtain the infinite-dimensional completely mixed state, while for $\beta\rightarrow\infty$ only the ground states of the $\hat{H}_{q}$ survive, describing all pure Gaussian states. \\
From the definition \eqref{thermalAndGroundStates} it is easy to obtain a diagonal decomposition: first of all, as already stated, the linear part of the Hamiltonian is equivalent to a displacement operation so that we can restrict to a purely quadratic $\hat{H}_{q}$. Next we compute the normal-mode decomposition of the Hamiltonian via the symplectic diagonalization~\cite{williamson} of $H$, valid for any $2N\times 2N$ positive-definite real matrix, as follows:
\begin{equation}\label{normalMode}
H=S^{T}\Lambda S,
\end{equation}
with $S$ symplectic and $\Lambda=\operatorname{diag}(\lambda_{1},\lambda_{1},\cdots,\lambda_{N},\lambda_{N})$. The symplectic eigenvalues $\{\lambda_{n}\}_{n=1}^{N}$ of $H$ are positive and doubly degenerate; they can be computed as the absolute value of the ordinary eigenvalues of the matrix $i\Omega H$. This decomposition identifies the normal modes of $\hat{H}_{q}$, i.e., those that are not mixed by its action, as linear combinations of the starting ones. Accordingly, the generic $N$-mode Gaussian state, \eqq{thermalAndGroundStates}, can be written as
\begin{equation}\label{almostDecomposed}
\hat{\rho}_{G}\propto\hat{D}^{\dag}(\vecund{r})\hat{U}^{\dag}_{S}\left[\bigotimes_{i=1}^{N}e^{-\frac{\beta}{2} \lambda_{i}\left(\hat{q}_{i}^{2}+\hat{p}_{i}^{2}\right)}\right]\hat{U}_{S}\hat{D}(\vecund{r}),
\end{equation}
where $\vecund{r}\in\mathbb{R}^{2N}$ is determined by the coefficients of the linear part of the hamiltonian, while $\hat{U}_{S}$ is the unitary representation of the symplectic matrix that diagonalizes the quadratic part of the hamiltonian, i.e., a generic Gaussian unitary itself decomposable as in \eqq{passActDeco}. Eventually each normal-mode component of \eqref{almostDecomposed} is the ground or thermal state of a harmonic oscillator of frequency $\lambda_{i}$, diagonal in its Fock basis $\dketbra{n_{i}}_{i}$, so that
\begin{equation}\label{gaussianNormalMode}
\hat{\rho}_{G}=\hat{D}^{\dag}(\vecund{r})\hat{U}^{\dag}_{S}\left[\bigotimes_{i=1}^{N}\left(1-e^{-\beta\lambda_{i}}\right)\sum_{n_{i}=0}^{\infty}e^{-\beta \lambda_{i}n_{i}}\dketbra{n_{i}}_{i}\right]\hat{U}_{S}\hat{D}(\vecund{r}).
\end{equation}
By direct computation of \eqq{moments} for the ladder operators, it is easy to show that the first and second moments of a generic Gaussian state, \eqq{gaussianNormalMode}, are
\begin{equation}\label{momentsDeco}
\vecund{m}=S\vecund{r},\quad V=S\left[\bigoplus_{i=1}^{N} \nu_{i}\mathbf{1}\right]S^{T},
\end{equation}
where $\vecund{r}$, $S$ are the displacement vector and symplectic matrix corresponding to the normal-mode decomposition of $\hat{H}_{q}$  and $\nu_{i}$ are the symplectic eigenvalues of $V_{i}$, connected to those of $H$, i.e., $\lambda_{i}$, as follows:
\begin{equation}
\nu_{i}=\frac{1+e^{-\beta \lambda_{i}}}{2(1-e^{-\beta \lambda_{i}})}.
\end{equation}
Note that the decomposition \eqq{momentsDeco} of the covariance matrix is once again a symplectic diagonalization with symplectic eigenvalues $\nu_{i}$.\\

The Hilbert-space representation, \eqq{gaussianNormalMode}, can be used to describe several basic Gaussian states:
\begin{enumerate}

\item Thermal states have the simplest kind of Hilbert-space representation, since they are diagonal in the Fock basis of each mode. They are defined in terms of their average photon-number $\bar{n}$ as, for a single-mode, 
\begin{equation}\label{therm}
\hat{\rho}_{th}(\bar{n})=\oneover{\bar{n}+1}\sum_{n=0}^{\infty}\left(\frac{\bar{n}}{\bar{n}+1}\right)^{n}\dketbra{n},
\end{equation}
hence constitute the basic building blocks of the decomposition \eqq{gaussianNormalMode}. Their mean value is clearly zero, while the covariance matrix is already diagonal with symplectic eigenvalues $\nu=\bar{n}+1/2$. In particular, for $\bar{n}=0$ we obtain the pure vacuum state with covariance matrix $V_{vac}=\mathbf{1}/2$, while for higher values of $\bar{n}$ we obtain a mixture of Fock states;

\item Coherent states, already discussed in Sec.~\ref{phaseSpace}, are obtained by displacing the vacuum by a complex amount $\alpha$. Hence they have the same covariance matrix of the vacuum but of course exhibit a non-zero mean-value $\vecund{m}_{coh}=(q_{\alpha},p_{\alpha})^{T}$;

\item Single-mode squeezed states are obtained by squeezing the vacuum by an amount $r$, i.e., $\ket{r}=\hat{U}_{sq}(r)\ket{0}$. They have zero mean-value and covariance matrix 
\begin{equation}\label{squeezedCov}
V_{sq}=S_{sq}(r)\frac{\mathbf{1}}{2}S_{sq}^{T}(r)=\frac{S(2r)}{2}
\end{equation} 
with unit symplectic eigenvalues, as for the vacuum. Note that the covariance matrix, \eqq{squeezedCov}, is still diagonal but the fluctuations of its $\hat{q}$ ($\hat{p}$) quadrature are exponentially enhanced (damped) with respect to the vacuum. This effect is inherently quantum and is of help in several settings such as metrological applications~\cite{ligoCaves}; in particular, in the limit $r\rightarrow\infty$ we obtain an infinitely squeezed state with asymptotically zero quantum noise on a given quadrature. A squeezed state lives in a subspace of the total Hilbert space spanned by even-photon-number Fock states; its weights can be computed by disentangling the SU(1,1) Lie algebra generated by the operators $\hat{a}^{2}$, $\hat{a}^{\dag}{}^{2}$ and $\hat{a}^{\dag}\hat{a}$, see~\cite{puriBOOK}, obtaining
\begin{equation}\label{squeezedstates}
\ket{r}=\oneover{\sqrt{\cosh r}}\sum_{n=0}^{\infty}\frac{\sqrt{(2 n)!}}{2^{n} n!}(-\tanh r)^{n}\ket{2n}=\sum_{n=0}^{\infty}c_{n}(r)\ket{2n};
\end{equation}

\item Two-mode squeezed or Einstein-Podolsky-Rosen (EPR) states are obtained by squeezing a pair of vacuum modes by an amount $r$, i.e., $\ket{r}_{AB}=\hat{U}_{2sq}(r)\ket{0,0}_{AB}$. They have a zero mean-value and covariance matrix $V_{2sq}=S_{2sq}(2r)/2$. Their Fock-state decomposition can be obtained as for the one-mode squeezed states by using the disentangling of the $SU(1,1)$ Lie algebra:
\begin{equation}\label{epr}
\ket{r}_{AB}=\sqrt{1-\left(\tanh r\right)^{2}}\sum_{n=0}^{\infty}\left(-\tanh r\right)^{n}\ket{n, n}_{AB},
\end{equation} 
which is a coherent superposition of couples of Fock states with the same photon-number on each subsystem.
Note that for $r=0$ a separable double-vacuum state is obtained, while the state becomes increasingly entangled at larger $r$ and for $r\rightarrow\infty$ reaches asymptotically the infinite-dimensional maximally entangled state, analogous to the finite-dimensional case discussed in Sec.~\ref{quantumStates}. Eventually, by partial trace it is easy to show that an EPR state,  \eqq{epr}, is locally equivalent to a thermal state,  \eqq{therm}, with an average photon-number $\bar{n}=(\cosh(2r)-1)/2$.
\end{enumerate}

\subsection{Gaussian channels and measurements}\label{gaussChMeas}
The class of Gaussian transformations studied in Sec.~\ref{gaussUni} is limited to the reversible case. More generally one can imagine transformations that preserve Gaussian states but are not reversible, i.e., quantum Gaussian channels. These can be described as all those quantum maps $\Phi$ that transform the first and second moments, respectively $\vecund{m}$ and $V$, of a N-mode Gaussian state as
\begin{equation}\label{gaussianCh}
\vecund{m}'=A\vecund{m}+\vecund{b},\quad V'=AVA^{T}+B,
\end{equation}
where $A$ and $B=B^{T}$ are positive $2N\times2N$ real matrices and $\vecund{b}\in\mathbb{R}^{2N}$. This kind of map can always be dilated to a Gaussian unitary dynamics acting on the system and on a Gaussian state of the environment; if such state is also pure, one has the Stinespring representation, \eqq{stinespring}. This unitary dilation helps characterizing the maps of the form \eqq{gaussianCh}; in particular it can be used to obtain a constraint~\cite{holevoBOOK} on the matrices $A$, $B$ that is related to the uncertainty relation, \eqq{uncertainty}, for the environment:
\begin{equation}\label{uncertaintyGC}
B+\frac{i}{2}(\Omega-A\Omega A^{T})\geq0.
\end{equation}
The latter is a necessary and sufficient condition for a map \eqq{gaussianCh} to be physical.
Let us note that the matrix $B$ can be associated to noise addition to the system, while $A$ represents some sort of interaction between the modes of the system. Hence the uncertainty relation \eqref{uncertaintyGC} states that not all kinds of transformations $A$ can be physically implemented noiselessly: when this is the case then $B=0$ and the matrix $A$ must be symplectic, corresponding to a Gaussian unitary; otherwise we call quantum-limited a channel that attains the equality in \eqq{uncertaintyGC} with the minimal amount of noise. As for the vector $\vecund{b}$, as usual it amounts to a displacement operation. More generally, any Gaussian channel can be put in a canonical form $\Phi_{c}$, with $\vecund{b}_{c}=0$ and $A_{c}$ block-diagonal, by applying proper Gaussian unitaries before and after it, i.e., $\Phi=\mathcal{V}\circ\Phi_{c}\circ\mathcal{U}$, and without changing its spectral properties, see~\cite{multimodeGC}.\\
Restricting to the case of one mode, both the interaction and noise matrices can be put in diagonal form and the uncertainty relation, \eqq{uncertaintyGC}, can be written explicitly as \begin{equation}\label{uncertainty1m}4\det B\geq(1-\det A)^{2}.\end{equation} 
Starting from the latter relation a full classification of one-mode Gaussian channels can be carried out~\cite{holevoBOOK,holevoSingleMode} but here we restrict to the most studied and relevant case for practical applications: phase-insensitive (PI) Gaussian channels, defined as those that are invariant with respect to phase-shift operations, i.e., 
\begin{equation}\label{PIG}
\mathcal{U}_{p}(\phi)\circ\Phi_{PI}=\Phi_{PI}\circ\mathcal{U}_{p}(\phi) 
\end{equation}
for all $\mathcal{U}_{p}(\phi)$ with unitary representation given by \eqq{ps}. It is straightforward to see that the phase-space representation of $\Phi_{PI}$ has $A$ and $B$ proportional to the identity: $A=\sqrt{\eta}\mathbf{1}$, $B=\sqrt{\tau}\mathbf{1}$ and the uncertainty relation,  \eqq{uncertainty1m}, becomes $2\sqrt{\tau}\geq|1-\eta|$; these channels transform both quadratures in the same way.
The simplest kinds of PI Gaussian channels are:
\begin{enumerate}
\item The attenuator or lossy channel $\mathcal{E}_{\eta,\bar{n}}$ reduces the intensity of the input signal by $\eta\in[0,1)$ while adding noise of parameter $\sqrt{\tau}=(1-\eta)(\bar{n}+1/2)$, $\bar{n}\geq0$. Its unitary dilation is given by a beam splitter interaction $\hat{U}_{bs}(\arccos\sqrt{\eta})$, \eqq{bs}, with an environmental thermal state $\hat{\rho}_{th}(\bar{n})$, \eqq{therm}, and it transforms the photon-number decreasing operator $\hat{a}$ of the system as
\begin{equation}\label{att}
\hat{a}\rightarrow\hat{a}\sqrt{\eta}+\hat{e}\sqrt{1-\eta},
\end{equation}
where $\hat{e}$ is the photon-number decreasing operator of the environment;
\item The amplifier channel $\mathcal{A}_{\eta,\bar{n}}$ increases the intensity of the input signal by $\eta\geq0$ while adding noise of parameter $\sqrt{\tau}=(\eta-1)(\bar{n}+1/2)$, $\bar{n}\geq0$. Its unitary dilation is given by a two-mode squeezing interaction $\hat{U}_{2sq}(\operatorname{arccosh}\sqrt{\eta})$, \eqq{2sq}, with an environmental thermal state $\hat{\rho}_{th}(\bar{n})$, \eqq{therm}, and it transforms the photon-number decreasing operator $\hat{a}$ of the system as
\begin{equation}\label{amp}
\hat{a}\rightarrow\hat{a}\sqrt{\eta}+\hat{e}^{\dag}\sqrt{\eta-1},
\end{equation}
where $\hat{e}^{\dag}$ is the photon-number increasing operator of the environment;
\item The noise-addition channel, obtained for $\eta=0$, is equivalent to an additive classical-Gaussian-noise channel of intensity $\tau\geq0$. Surprisingly, its unitary dilation requires a less straightforward interaction with an environmental two-mode squeezed state, see~\cite{holevoBOOK,holevoSingleMode}, and it transforms the photon-number decreasing operator $\hat{a}$ as 
\begin{equation}
\hat{a}\rightarrow\hat{a}+\alpha,
\end{equation}
where $\alpha$ is a complex random variable with Gaussian distribution of zero mean and variance $\tau$.
\end{enumerate}
Let us further note that the quantum-limited versions of the attenuator and amplifier channels described above can be obtained by choosing a vacuum environmental state, i.e., $\bar{n}=0$ and hence $\sqrt{\tau}=\sqrt{\tau_{min}}=|1-\eta|/2$. Moreover it is easy to show that any PI Gaussian channel can be written as the composition of a quantum-limited attenuator and a quantum-limited amplifier, see~\cite{gaussOpt}, i.e.,
\begin{equation}\label{ampattdeco} 
\Phi_{PI}=\mathcal{A}_{\kappa,0}\circ\mathcal{E}_{\eta,0}.
\end{equation}
Another important property of the attenuator and amplifier channels just described is that they are the dual, \eqq{dualCh}, of one another, i.e.,
\begin{equation}\label{duality}
\mathcal{A}^{*}_{\kappa,\bar{n}}=\kappa^{-1}\mathcal{E}_{\kappa^{-1},\bar{n}}
\end{equation}
for all $\kappa\geq1$ and vice versa, as can be easily shown by computing $\tr{\hat{\sigma}\mathcal{A}_{\kappa,\bar{n}}(\hat{\rho})}$ through \eqq{gaussProduct}, writing the transformation as a rescaling of the integration variable $\alpha$ times some noise and performing a change of variable $\alpha'=\sqrt{\kappa}\alpha$.

Eventually, let us conclude our brief presentation of Gaussian Quantum Information with a discussion of the most common kinds of measurements that can be performed in the lab. As a start, consider the photodetector (PDT)~\cite{mandelBOOK}, a device that destroys the incoming state and outputs a classical electric current of intensity proportional to that of the measured quantum state, i.e., it performs an approximate measurement of the photon-number operator $\hat{n}$. The typical errors that affect a PDT are non-unit efficiency, i.e., when it estimates a lower intensity than that of the actual signal because of various loss effects, and dark counts, i.e., when it outputs a non-zero current even though no signal was sent, due to the presence of noise. These errors can be mimicked by adding a generic PI lossy channel before the ideal PDT.  \\
By interpreting a non-zero current as the presence of one or more photons, one obtains an on/off detector (OOD) that ``clicks'' whenever the incoming signal is different from the vacuum state. An OOD can be modeled by a two-outcome POVM (see Sec.~\ref{quantumOps})
\begin{equation}\label{photodetector}
\mathcal{M}_{ood}=\left\{\hat{E}_{0}=\dketbra{0},\hat{E}_{1}=\hat{\mathbf{1}}-\hat{E}_{0}\right\}
\end{equation}
followed by tracing out the output state. More refined photon-number-resolving (PNR) detectors, which can in principle distinguish higher-photon-number states, have been studied and experimentally demonstrated in recent years but are still affected by low efficiency~\cite{pnr}. \\
As for Gaussian measurements, they are defined as those measurements whose outcome statistics is Gaussian-distributed if the input state was Gaussian. The basic ones are those of the dyne family~\cite{shapiroDyne,RevGauss2,RevGauss}:
\begin{enumerate}

\item Homodyne detection (HoD) performs a destructive measurement of a single arbitrary quadrature of the field, i.e., $\hat{q}_{\phi}=\hat{q}\cos\phi+\hat{p}\sin\phi$ for some $\phi\in[0,2\pi)$, thus it can be modeled by a continuous-outcome projective POVM 
\begin{equation}\label{homodyne}
\mathcal{M}_{hod}(\phi)=\left\{\hat{E}_{q_{\phi}}=\dketbra{q_{\phi}}\right\}_{q_{\phi}\in\mathbb{R}},
\end{equation}
where $\ket{q_{\phi}}$ are improper eigenstates of the corresponding quadrature operator, again followed by tracing out the system. The Fock representation of such eigenstates is
\begin{equation}\label{quadraturestates}
\ket{q_{\phi}}=\oneover{\pi^{1/4}}e^{-\frac{q_{\phi}^{2}}{2}}\sum_{n=0}^{\infty}\frac{H_{n}(q_{\phi})}{2^{n/2}\sqrt{n!}}e^{-i\phi}\ket{n},
\end{equation}
where $H_{n}(x)$ is an Hermite polynomial~\cite{mandelBOOK}. It can be realized by combining the state of the system with a strong local oscillator of phase $\phi$ (relative to the system) on a balanced beam-splitter, similarly to the the implementation of a displacement operation discussed in Sec.~\ref{gaussUni}, \eqq{bs}, then measuring both outputs with a PDT and finally subtracting the two classical currents thus obtained. In order to see this, apply the beam-splitter unitary $\hat{U}_{bs}(\pi/4)$ to the input modes $\hat{a}$ and $\hat{b}$, populated respectively by the state of the system, $\rho$, and a coherent state $\ket{b e^{i\phi}}$ with $b^{2}\gg\tr{\hat{n}_{a}\hat{\rho}}$. Then the final classical current will be proportional to a measurement of the observable 
\begin{equation}
\hat{O}_{hod}=\frac{\hat{a}'^{\dag}\hat{a}'+\hat{b}'^{\dag}\hat{b}'}{b\sqrt{2}}=\frac{\hat{a}^{\dag}\hat{b}+\hat{b}^{\dag}\hat{a}}{b\sqrt{2}},
\end{equation}
where the primed operators describe the output modes of the beam-splitter. The first moment of such operator on the total state of the system, $\rho\otimes\dketbra{b e^{i\phi}}$, is equal to $q_{\phi}$, while higher moments differ from those of the desired quadrature $\hat{q}_{\phi}$ for terms of order $\tr{\hat{n}_{a}\hat{\rho}}/b^{2}$ that vanish in the strong-local-oscillator limit;

\item Heterodyne detection (HeD) performs a destructive measurement of the photon-number-decreasing operator, $\hat{a}=(\hat{q}+i\hat{p})/\sqrt{2}$, thus it can be modeled by a continuous-outcome POVM
\begin{equation}\label{heterodyne}
\mathcal{M}_{hed}=\left\{\hat{E}_{\alpha}=\oneover{\pi}\dketbra{\alpha}\right\}_{\alpha\in\mathbb{C}},
\end{equation}
where $\ket{\alpha}$ are coherent states of the field. Note that, like HoD, the POVM is a set of complete projection operators but, unlike HoD, these are not orthogonal, see \eqq{coherentNotOrth}, so that the POVM \eqq{heterodyne} is not projective. The HeD can be realized by combining the state of the system with a vacuum state on a balanced beam-splitter, then performing two HoD's on the output modes, each for a different canonical conjugate quadrature $\hat{q}$, $\hat{p}$, and finally subtracting the two classical currents thus obtained. Proceeding as for the HoD, it is easy to show that the correct statistics is recovered. 
\end{enumerate}
Eventually, a generalization of the HoD implementation allows to perform projections on arbitrary Gaussian states by applying a proper Gaussian unitary before the measurement; noisier measurements are obtained instead by applying noisy Gaussian channels, see~\cite{giedkeCirac,eisertPlenio}.

\section{Quantum Communication Theory}\label{sec:commQ}

\subsection{Classical Information Theory}\label{subsec:commC}
The transmission of information in a classical setting has been well studied in the past, starting with the seminal works of Shannon~\cite{shannonSeminal,shannonSeminal1}. We review here the basic concepts and results that will be useful to understand its extension to the quantum setting; for a full review see~\cite{covThomBOOK,gallagerBOOK,wildeBOOK,nChuangBOOK}. \\
Let us consider a set of possible signals that the transmitter, Alice, wants to send to the receiver, Bob, using a classical communication line. Each signal sent through the channel is represented by a random variable $X$ whose value can be chosen among an input alphabet $\mathcal{X}=\{x\}$ with probability $p_{x}=P_{X}(x)$\footnote{Here and in the following we will use the compressed notation $p_{x}$ to indicate the value of the probability function $P_{X}$ at point $x$, switching back to the extended notation $P_{X}(x)$ only when deemed necessary for more clarity.}. Note that the dimension of the alphabet can be infinite, in which case the variable has a continuous probability distribution and sums over it become integrals; moreover several information-theoretic quantities may exhibit an infinite optimal value, so that infinite-dimensional alphabets are usually subject to some kind of constraint, e.g., on the second moment of the probability distribution. For this section we restrict to the easier finite-dimensional case. The signals received by Bob in general are different from those sent by Alice and they are represented by an output random variable $Y$ with values in $\mathcal{Y}=\{y\}$. Accordingly, the channel is represented by a conditional (or transition) probability distribution $p_{y|x}=P_{Y|X}(y|x)$ that describes how probable it is for a given realization $X=x$ of the input to transition to an output value $Y=y$. Hence the probability distribution of an outcome is completely determined by the source of the input signal, i.e., $p_{x}$, and the type of noise introduced by the channel, i.e., $p_{y|x}$, following the usual probability rules:
\begin{equation}
p_{y}=\sum_{x\in\mathcal{X}}p_{y,x}=\sum_{x\in\mathcal{X}}p_{y|x}p_{x},
\end{equation}
where we have introduced the joint probability distribution $p_{y,x}=p_{y|x}p_{x}$. \\
There are two basic quantities employed in the study of Classical Information Theory.
First, the Shannon (or classical) entropy of a random variable is 
\begin{equation}
H(X)=\sum_{x\in\mathcal{X}}h(p_{x})=-\sum_{x\in\mathcal{X}}p_{x}\log p_{x}
\end{equation}
and it measures the average degree of disorder or information content of its probability distribution. Note that we have implicitly defined the function $h(p)=-p\log p$ for $p\in[0,1]$, with $h(0)=0$. The Shannon entropy is minimum for a deterministic variable, i.e., $H(X)=0$ for $p_{x}=\delta_{x,x_{0}}$, while it is maximum for a completely random variable, i.e., $H(X)=\log\abs{\mathcal{X}}$ for $p_{x}=1/\abs{\mathcal{X}}$. Such randomness is considered a resource in several settings, because it is a means of encoding information. Given two (or more) random variables, one can define their joint entropy as the entropy of their joint probability distribution, i.e.,
\begin{equation}
H(X,Y)=\sum_{(x,y)\in\mathcal{X}\times\mathcal{Y}}h(p_{x,y}),
\end{equation}
and the conditional entropy of one variable with respect to the other as the entropy of their conditional probability distribution averaged over the conditioning variable, i.e.,
\begin{equation}\label{condCEnt}
H(Y|X)=\sum_{x\in\mathcal{X}}p_{x}H(Y|x)=\sum_{x\in\mathcal{X}}p_{x}\sum_{y\in\mathcal{Y}}h(p_{y|x})\equiv H(Y,X)-H(X),
\end{equation}
with a similar definition for $H(X|Y)$. The conditional entropy is non-negative by definition and it attains its minimum when the variables are perfectly correlated, i.e., $H(Y|X)=0$ for $p_{y|x}=\delta_{x,y}$ and hence $H(X,Y)=H(X)\equiv H(Y)$, while it is maximum for completely uncorrelated variables, i.e., $H(Y|X)=H(Y)$ for $p_{y|x}=p_{y}$ and hence $H(X,Y)=H(X)+H(Y)$.\\
Second, the mutual information of two random variables is
\begin{equation}\label{mCInfo}
I(X:Y)=H(X)+H(Y)-H(X,Y)=H(Y)-H(Y|X)=H(X)-H(X|Y),
\end{equation}
where we have used \eqq{condCEnt} to obtain the equivalent forms. This symmetric quantity measures the amount of information shared between the two random variables, i.e., their correlations. It attains its minimum for completely uncorrelated variables, i.e., $I(X:Y)=0$ for $H(Y|X)=H(Y)$, and its maximum for perfectly correlated ones, i.e., $I(X:Y)=H(Y)\equiv H(X)$ for $H(Y|X)=0$. Let us also define the mutual information of $X$ and $Y$ conditioned on a third random variable $Z$ as
\begin{equation}\label{condCMInfo}
I(X:Y|Z)=I(X:Y,Z)-I(X:Z)=H(Y|Z)-H(Y|X,Z).
\end{equation}
Here and in the rest of the thesis we will use the notation $x_{(i,j)}=x_{i},\cdots,x_{j}$, with $i,j$ positive integers, to indicate a sequence of $j-i+1$ quantities and $\vecund{x}=x_{(1,J)}$ if the sequence is complete, i.e., the index $i,j$ are respectively the minimum and the maximum allowed for the given quantities, e.g., $1$ and $J$ in this case. In the following we present some useful properties of the quantities just introduced:
\begin{enumerate}
\item Positivity of mutual information, i.e., $I(X:Y)\geq0$, provable through the positivity of relative entropy~\cite{covThomBOOK};
\item Conditioning does not increase entropy, i.e., $H(Y)\geq H(Y|X)$, which follows from the previous property;
\item Chain rule of entropy, i.e., 
\begin{equation}
H(X_{(1,N)})=\sum_{n=1}^{N}H(X_{n}|X_{(1,n-1)}),
\end{equation}
which easily follows from the definition of conditional entropy, \eqq{condCEnt};
\item Chain rule of mutual information, i.e., 
\begin{equation}\label{chainMInfo}
I(X_{(1,N)}:Y)=\sum_{n=1}^{N}I(X_{n}:Y|X_{(1,n-1)}),
\end{equation} 
which can be obtained by applying the definition \eqq{mCInfo} and repeatedly adding and subtracting conditional entropy terms $H(Y|X_{(1,n-1)})$ to reconstruct the terms in the sum;
\item Concavity of entropy, i.e., given $\bar{X}=\sum_{n=1}^{N}p_{n}X_{n}$ a convex combination of random variables $X_{n}$ with probability $p_{n}$ then \begin{equation}
\sum_{n=1}^{N}p_{n}H(X_{n})\leq H(\bar{X}).
\end{equation}
This can be proved using the convexity of relative entropy or by considering $p_{n}$ as the probability distribution of a discrete random variable and conditioning over it;
\item Concavity of mutual information in $p_{x}$ at fixed $p_{y|x}$, i.e., given $\bar{X}$ as in the previous property then 
\begin{equation}
\sum_{n=1}^{N}p_{n}I(X_{n}:Y)\leq I(\bar{X}:Y),
\end{equation} 
which can be proved by using the concavity of entropy;
\item Data-processing inequality, i.e., for three random variables forming a Markov chain, $X\rightarrow Y\rightarrow Z$, it holds $I(X:Y)\geq I(X:Z)$. We say that $X\rightarrow Y\rightarrow Z$ if and only if the variable $Z$ conditioned on $Y$ is uncorrelated to $X$, i.e., $p_{z|x,y}=p_{z|y}$. This is equivalent to $Z\rightarrow Y\rightarrow X$ and also to requiring that the joint distribution of $X$ and $Z$ conditioned on $Y$ is factorized, i.e., $p_{x,z|y}=p_{x|y}p_{z|y}$. This property can be easily proved by applying the definition \eqq{condCMInfo} twice, conditioning on $Y$ and $Z$ to obtain
\begin{equation}
I(X:Z|Y)+I(X:Y)=I(X:Y|Z)+I(X:Z).
\end{equation}
Then the first term on the left-hand side is zero for a Markov chain, while that on the right-hand side is non-negative and can be eliminated, obtaining a lower quantity and the desired relation. Note that by eliminating the second term on the right-hand side instead one concludes that conditioning does not increase the mutual information for a Markov chain.
\end{enumerate}

In its seminal contribution~\cite{shannonSeminal,shannonSeminal1}, Shannon introduced the entropy and mutual information to characterize the compression and transmission of information. His results are based on the methods of typicality and random coding, which we briefly describe here since they are employed in Quantum Communication Theory as well. Let us suppose that Alice chooses a sequence $\vecund{x}=x_{(1,N)}$ of $N$ letters  at random from $\mathcal{X}$ by sampling the probability distribution $P_{X}(x)$. We call the alphabet equipped with a probability distribution a classical source $\mathcal{S}=\{(x,p_{x}):x\in\mathcal{X},p_{x}=P_{X}(x)\}$ and we would like to know its ultimate information content in order to understand if such information can be reliably compressed or transmitted through a given channel. In the asymptotic limit of long sequences, i.e., $N\rightarrow\infty$, there is a set of them, called typical sequences, in which each letter of $\mathcal{X}$ appears approximately its expected number of times, i.e., $N p_{x}$. Moreover the number of such typical sequences or equivalently the dimension of the set is approximately $2^{NH(X)}$, as can be seen by expanding the multinomial coefficient through the Stirling formula. \\
More specifically, given a classical source and a sequence $\vecund{x}$ sampled from it with probability $p_{\vecund{x}}=\prod_{n=1}^{N}p_{x_{n}}$, we define the sample entropy of such sequence as its average information content per symbol, i.e.,
   \begin{equation}
      \bar{H}(\vecund{x})=-\oneover{N}\log{p_{\vecund{x}}}=-\oneover{N}\sum_{n=1}^N\log{p_{x_n}}.
    \end{equation}
Then the $\delta-$typical set of the source is the set of sequences whose sample entropy differs from the entropy of the source for less than a given quantity $\delta>0$:
 \begin{equation}\label{tSet}
 T_{\delta,N}\left(\mathcal{X}\right)=\left\{\vecund{x}\in\mathcal{X}^{N}:\abs{\bar{H}(\vecund{x})-H(X)}\leq\delta\right\}.
\end{equation}
The typical set has the following three main properties as $N\rightarrow\infty$ with $\epsilon$ positive and close to zero:
\begin{enumerate}
\item The probability of sampling a typical sequence from the source is high, i.e., 
\begin{equation}\label{typC1}
Pr\left\{\vecund{x}\in T_{\delta,N}(\mathcal{X})\right\}\geq1-\epsilon;
\end{equation}
\item The size of the typical set is exponential in $N$ with scaling coefficient close to the entropy, i.e., 
\begin{equation}\label{typC2}
(1-\epsilon)2^{N(H(X)-\delta)}\leq \abs{T_{\delta,N}(\mathcal{X})}\leq 2^{N(H(X)+\delta)};\end{equation}
\item The probability distribution of typical sequences is approximately uniform, i.e., 
\begin{equation}\label{typC3}
2^{-N(H(X)+\delta)}\leq p_{\vecund{x}}\leq 2^{-N(H(X)-\delta)}.
\end{equation}
\end{enumerate}
These properties hold for an i.i.d.~sampling of the source and can be derived as a consequence of the law of large numbers. In particular, it can be shown using the stronger Chernoff bound that the parameter $\epsilon$ above decreases exponentially in the length of the sequences, i.e., $\epsilon=O(e^{-N})$, see~\ref{app:chernoff}. Thanks to these properties it is possible to prove the first fundamental theorem by Shannon~\cite{shannonSeminal,shannonSeminal1}:
 \newtheorem{theorem}{Theorem}
\begin{theorem}\label{compressionCThm}(Shannon Compression)
Let $\mathcal{S}=\{(x,p_{x}): x\in\mathcal{X}, p_{x}=P_{X}(x)\}$ be an i.i.d.~source and let us sample from it sequences $\vecund{x}$ of $N$ letters. Then for all $\epsilon>0$ and $N$ sufficiently large there exists a compression protocol that maps such sequences to codewords of length $\lceil NR\rceil$ with $R\geq H(X)+\delta_{\epsilon}$; moreover the protocol can be inverted with a recovery error at most as large as $\epsilon$. Conversely, there exists no such asymptotically reversible map for $R<H(X)$.
\end{theorem}
 The proof of Theorem~\ref{compressionCThm} relies on the typicality properties, Eqs.~(\ref{typC1},\ref{typC2},\ref{typC3}), and it employs a compression protocol that associates to each typical sequence a unique binary sequence of $\sim NR$ bits and a reference one to non-typical sequences. Given the high probability of picking a typical sequence and the size of the typical set, the shortest compression length that allows for infinitesimal decoding errors is obtained for $R\simeq H(X)$, see~\cite{nChuangBOOK,covThomBOOK} for further details.\\

Let us note that the compression operation can be regarded as a problem of optimal transmission through an ideal noiseless channel, where the only objective of Alice is to reliably convey as much information as possible for each use of the channel by reducing the size of her messages without compromising their information content. 
For the more realistic case of transmission through noisy channels, Alice also needs to consider how the noise transforms her messages and possibly protect them by means of a clever encoding. 
Shannon's results in this case additionally rely on random coding: suppose that, out of all the sequences of fixed length $N$ that can be sampled from the source, Alice selects a subset of $\sim 2^{N R}$ typical sequences, \eqq{tSet}, that she calls the codebook $\mathcal{C}\subset\mathcal{X}^{N}$. These can be used for communication with a uniform probability distribution, as guaranteed by \eqq{typC3} for a transmission rate $R\leq H(X)$. The probability of selecting a specific codebook is given by the product of the probabilities of each sample sequence that it contains, i.e.,
\begin{equation}\label{randomCCoding}
P(\mathcal{C})=\prod_{\vecund{x}\in\mathcal{C}}p_{\vecund{x}}=\prod_{\vecund{x}\in\mathcal{C}}\prod_{n=1}^{N}p_{x_{n}}.
\end{equation}
In this way any quantity of relevance for the problem, e.g., the decoding error probability, can be first computed for a given codebook and then averaged over all codebooks; this procedure simplifies the calculations, allowing to reach a solution that holds on average. This in turn implies that at least for some codebooks the result must hold true. In the following we will use the symbol $\ave{\cdot}_{\mathcal{S}}$ to indicate the average over codebooks with probability distribution $P(\mathcal{C})$. Let us then fix a codebook $\mathcal{C}$ and let Alice transmit sequences $\vecund{x}\in\mathcal{C}$ with uniform probability; for each of them the output sequence of the channel is with high probability in the $\delta$-conditionally-typical set defined as
 \begin{equation}\label{ctSet}
 T_{\delta,N}\left(\mathcal{Y}|\vecund{x}\right)=\left\{\vecund{y}\in\mathcal{Y}^{N}:\abs{\bar{H}(\vecund{y}|\vecund{x})-H(Y|X)}\leq\delta\right\},
\end{equation}
with
\begin{equation}
  \bar{H}(\vecund{y}|\vecund{x})=-\oneover{N}\log{p_{\vecund{y}|\vecund{x}}}=-\oneover{N}\sum_{n=1}^N\log{p_{y_{n}|x_n}}
\end{equation}
for a memoryless channel without feedback, i.e., with transition probability factorized over several uses. The properties of such set mirror those of Eqs.~(\ref{typC1},\ref{typC2},\ref{typC3}) but some of them hold in this case only on average over the input conditioning sequences:
\begin{enumerate}
\item The average probability that the channel outputs a conditionally-typical sequence is high, i.e., 
\begin{equation}\label{ctypC1}
\sum_{\vecund{x}\in\mathcal{C}}p_{\vecund{x}}Pr\left\{\vecund{y}\in T_{\delta,N}\left(\mathcal{Y}|\vecund{x}\right)\right\}\geq1-\epsilon;
\end{equation}
\item The size of the conditionally-typical set is exponential in $N$ with scaling coefficient close to the conditional entropy, i.e., 
\begin{equation}
\begin{aligned}\label{ctypC2}
&\abs{T_{\delta,N}\left(\mathcal{Y}|\vecund{x}\right)}\leq 2^{N(H(Y|X)+\delta)},\\
&(1-\epsilon)2^{N(H(Y|X)-\delta)}\leq\sum_{\vecund{x}\in\mathcal{C}}p_{\vecund{x}}\abs{T_{\delta,N}\left(\mathcal{Y}|\vecund{x}\right)};
\end{aligned}
\end{equation}
\item The probability distribution of conditionally-typical sequences is approximately uniform, i.e., 
\begin{equation}\label{ctypC3}
2^{-N(H(Y|X)+\delta)}\leq p_{\vecund{y}|\vecund{x}}\leq 2^{-N(H(Y|X)-\delta)}.
\end{equation}
\end{enumerate}

Since an error-free transmission requires that distinct input codewords are transformed into distinct outputs, the size of the conditionally-typical set crucially determines the performance of the channel.  Let us now state the second Shannon theorem:
\begin{theorem}\label{shanCap} (Shannon capacity)
Let $\{(x,p_{x}): x\in\mathcal{X}, p_{x}=P_{X}(x)\}$ be an input source and $P_{Y|X}(y|x)$ the transition probability of a memoryless channel without feedback with output alphabet $\mathcal{Y}$. Define the capacity $C$ of such channel as the maximum rate of information transmission through the channel, i.e., the maximum number of bits sent per $N$ number of uses of the channel with average error probability asymptotically vanishing in $N$. Then the capacity is equal to the mutual information of $X$ and $Y$ optimized over all input probability distributions of the source, i.e.,
\begin{equation}
C_{Shan}=\max_{P_{X}(x)}I(X:Y).
\end{equation}
\end{theorem}
Note that the Shannon capacity is additive as a function of the channel, i.e., \begin{equation}C_{Shan}^{(M)}=M C_{Shan}\end{equation}
where the former is the capacity obtained by coding for $M>1$ copies of the channel with correlated inputs or equivalently by employing a source \begin{equation}\mathcal{S}^{M}=\{(x_{(1,M)},p_{x_{(1,M)}}):x_{(1,M)}\in\mathcal{X}^{M},p_{x_{(1,M)}}=P_{X_{(1,M)}}(x_{(1,M)})\}.\end{equation} 
The additivity is due to the fact that the channel transition probability factorizes over several uses, i.e., $p_{y_{(1,M)}|X_{(1,M)}}=\Pi_{m=1}^{M}p_{y_{m}|x_{m}}$. Moreover this result holds true even in some situations where the factorization does not take place, for example: in the presence of feedback to Alice, i.e., when the choice of the $n$-th input letter $x_{n}$ is influenced by the previous $n-1$ outcomes $y_{(1,n-1)}$; for certain channels with memory whose transition probability does not explicitly depend on the previous inputs, see Sec.~\ref{sec:adRec}.
As already said, the proof of the theorem relies on conditional typicality and random coding, showing that there exists a protocol with rate $R\leq C_{Shan}$ for which the average error reaches asymptotically zero in the codewords' length, while it tends to unit for any protocol with rate $R>C_{Shan}$.

\subsection{Classical Communication on Quantum Channels} \label{subsec:commQ}
Since communication exploits a physical medium, usually electromagnetic pulses, it is more than reasonable to ask what limits on the transmission of information are posed by quantum mechanics, in particular in those low-energy and long-distance regimes where the electromagnetic field exhibits a quantum behaviour. In this setting the information carriers and thus the channel are quantum states $\hat{\rho}$ and LCPT maps $\Phi$; moreover Alice and Bob can be interested in several transmission protocols: classical or quantum communication, the latter preserving coherences and superpositions; entanglement-assisted communication, where they share several couples of entangled states as a resource; private communication, where they want to hide information from a third malicious party. In the rest of the thesis we focus on the task of classical communication on quantum channels, which may appear trivial with respect to other tasks but is actually far from simple. For a thorough treatment see~\cite{wildeBOOK,holevoBOOK,holGiovRev}. \\
For classical communication on quantum channels, Alice wants to transmit her usual classical alphabet $\mathcal{X}=\{x\}$ with probability distribution $p_{x}$, but this time she encodes information on a set of quantum states, which we call $\mathcal{X}=\{\hat{\rho}_{x}\}$ with an overload of notation; similarly for the corresponding source 
\begin{equation}\label{qSource}
\mathcal{S}=\{(\hat{\rho}_{x},p_{x}):x\in\mathcal{X},p_{x}=P_{X}(x)\}.
\end{equation}
 After transmission through the channel the output states are $\Phi(\hat{\rho}_{x})$ and Bob performs a decoding protocol to obtain a set of classical outputs $\mathcal{Y}=\{y\}$ with probability distribution $p_{y}$ and estimate which input state was sent. Any such decoding protocol can be described by a POVM $\mathcal{M}=\{\hat{E}_{y}\}_{y\in\mathcal{Y}}$, Sec.~\ref{quantumOps}, which produces a classical transition probability $p_{y|x}=\tr{\hat{E}_{y}\Phi\left(\hat{\rho}_{x}\right)}$, so that it still holds $p_{y}=\sum_{x\in\mathcal{X}}p_{y|x}p_{x}$. The encoding and decoding procedures just defined represent respectively a classical-quantum and a quantum-classical channel, while the actual transmission through $\Phi$ is a quantum-quantum channel. Any physical transmission line can be represented by combining these three classes of channels. \\
The basic information-theoretic quantities employed in Quantum Communication Theory are the quantum entropy and the generalizations of mutual information. The quantum entropy of a density operator $\hat{\rho}$ was already defined in \eqq{VNEnt} as the classical entropy of the spectrum of the operator. Its joint and conditional versions are similarly defined, e.g., for a bipartite state $\hat{\rho}_{AB}\in\states{AB}$ they are 
\begin{align}
&S(AB)=S(\hat{\rho}_{AB}),\\
&S(A|B)=S(AB)-S(B).
\end{align}
The main properties of these quantum entropies are listed below:
\begin{enumerate}
\item $S(A)\in[0,\log d]$, with $d=\operatorname{dim}(\hil_{A})$, as already discussed in Sec.~\ref{sec:basics};
\item For a pure bipartite state $\hat{\rho}_{AB}=\dketbra{\Psi}_{AB}$ the local entropies are equal, i.e., $S(\hat{\rho}_{A})=S(\hat{\rho}_{B})$ as already discussed in Sec.~\ref{sec:basics};
\item The entropy is subadditive, i.e., for a bipartite system
\begin{equation}
S(AB)\leq S(A)+S(B),
\end{equation}
with equality if and only if the total state is separable, $\hat{\rho}_{AB}=\hat{\rho}_{A}\otimes\hat{\rho}_{B}$. This property can be proved from the positivity of relative entropy, see~\cite{wildeBOOK,nChuangBOOK,holevoBOOK};
\item The entropy is concave, i.e., 
\begin{equation}
S\left(\sum_{n=1}^{N}p_{n}\hat{\rho}_{n}\right)\geq\sum_{n=1}^{N}p_{n}S\left(\hat{\rho}_{n}\right),
\end{equation}
for any convex combination of states $\hat{\rho}_{n}$ with probabilities $p_{n}$. This property follows from subadditivity;
\item Any entropic quantity is invariant under unitary transformations of all their variables, e.g., $S(\hat{\rho})=S(U\rho U^{\dag})$ for $\hat{U}$ unitary;
\item Conditioning does not increase entropy, i.e., $S(A|B)\leq S(A)$, which also follows from the positivity of relative entropy;
\item Differently from the classical case, the quantum conditional entropy is not positive. In particular, the quantity can be negative when entangled states are employed, as discussed in Sec.~\ref{sec:basics}.
\end{enumerate}
Let us now describe quantum typical and conditionally-typical spaces, following~\cite{wildeBOOK}. Assume for simplicity and w.l.o.g. that the quantum source comprises the states already at the output of the channel (or equivalently set $\Phi=\mathcal{I}$) and take the spectral decomposition of the average state of the source:
\begin{equation}
\hat{\rho}_{\mathcal{S}}=\sum_{x\in\mathcal{X}}p_{x}\hat{\rho}_{x}=\sum_{j\in\mathcal{J}}\lambda_{j}\dketbra{e_{j}},
\end{equation}
where $\{\ket{e_{j}}\}_{j\in\mathcal{J}}$ form an orthonormal basis of the Hilbert space of the source and $\{\lambda_{j}\}_{j\in\mathcal{J}}$ are the eigenvalues of $\hat{\rho}_{\mathcal{S}}$. Accordingly, for $N$ uses of the source the average state can be written as
\begin{equation}
\hat{\rho}_{\mathcal{S}}^{\otimes N}=\sum_{\vecund{j}\in\mathcal{J}^{N}}\lambda_{\vecund{j}}\dketbra{e_{\vecund{j}}},
\end{equation}
where $\ket{e_{\vecund{j}}}=\Pi_{n=1}^{N}\ket{e_{j_{n}}}$ are eigenstate-sequences with eigenvalues $\lambda_{\vecund{j}}=\Pi_{n=1}^{N}\lambda_{j_{n}}$. Then the $\delta$-typical space of the source is the span of all those eigenstate-sequences whose classical label is $\delta$-typical, i.e.,
\begin{equation}\label{typSpace}
\hil_{\delta,N}(\mathcal{J})=\operatorname{span}\left\{\ket{e_{\vecund{j}}}:\vecund{j}\in T_{\delta,N}(\mathcal{J})\right\},
\end{equation}
and the projector on this space is given by
\begin{equation}\label{typProj}
  \hat{\Pi}=\sum_{\vecund{j}\in T_{\delta,N}}\dketbra{e_{\vecund{j}}}.
\end{equation}
Similar properties as for the classical typical set hold for this space, namely that for $\epsilon>0$ and $N$ sufficiently large:
\begin{enumerate}
 \item The average state $\hat{\rho}_{\mathcal{S}}^{\otimes N}$ resides with high probability in the $\delta-$typical space of the source, i.e., 
 \begin{equation}\label{totaleuno}
 \tr{\hat{\Pi}\hat{\rho}_{\mathcal{S}}^{\otimes N}}\geq1-\epsilon;
 \end{equation}
  \item The size of the $\delta-$typical space is exponential in $N$ with scaling coefficient close to the entropy of the average state of the source, i.e.,
  \begin{equation}\label{totaledue}
  (1-\epsilon)2^{N\left(S(\hat{\rho}_{\mathcal{S}})-\delta\right)}\leq\tr{\hat{\Pi}}\leq 2^{N\left(S(\hat{\rho}_{\mathcal{S}})+\delta\right)};
  \end{equation}
  \item The probability distribution of $\delta-$typical sequences is approximately uniform, i.e.,
\begin{equation}\label{totaletre}
2^{-N\left(S(\hat{\rho}_{\mathcal{S}})+\delta\right)}\hat{\Pi}\leq \hat{\Pi}\hat{\rho}_{\mathcal{S}}^{\otimes N}\hat{\Pi}\leq 2^{-N\left(S(\hat{\rho}_{\mathcal{S}})-\delta\right)}\hat{\Pi}.
\end{equation}
\end{enumerate}
These properties follow directly from their classical counterparts, Eq.~(\ref{typC1},\ref{typC2},\ref{typC3}), and in particular the exponential scaling of $\epsilon$ in $N$ still holds, see App.~\ref{app:chernoff}. \\
The random coding technique introduced in Sec.~\ref{sec:commQ} can be straightforwardly applied also to the quantum case: Alice selects a codebook of $\sim 2^{N R}$ classical typical codewords $\vecund{x}$ and the corresponding codebook of quantum codewords, \eqq{typSpace}, which we call $\mathcal{C}\subset\hil^{\otimes N}$ with another overload of notation. Once again these codewords will be sent with uniform probability for all rates $R\leq H(\hat{\rho}_{\mathcal{S}})$, as guaranteed by~\eqq{typC3}. The probability of a given code $\mathcal{C}$ is still $P(\mathcal{C})$ as given by \eqq{randomCCoding}.
Let us then consider a state $\hat{\rho}_{\vecund{x}}\in\mathcal{X}$ and take its spectral decomposition:
\begin{equation}
\hat{\rho}_{x}=\sum_{j\in\mathcal{J}_{x}}\lambda_{j|x}\dketbra{e_{j|x}},
\end{equation} 
with $\{\ket{e_{j|x}}\}_{j\in\mathcal{J}_{x}}$ an orthonormal basis and $\{\lambda_{j|x}\}_{k\in\mathcal{J}_{x}}$ the corresponding eigenvalues. For each classical sequence $\vecund{x}\in\mathcal{C}$, the corresponding sequence of quantum states can then be written as
\begin{equation}
\hat{\rho}_{\vecund{x}}=\bigotimes_{n=1}^{N}\hat{\rho}_{x_{n}}=\sum_{\vecund{j}\in\mathcal{J}_{\vecund{x}}}\lambda_{\vecund{j}|\vecund{x}}\dketbra{e_{\vecund{j}|\vecund{x}}},
\end{equation}
where $\mathcal{J}_{\vecund{x}}=\mathcal{J}_{x_{1}}\times\cdots\times\mathcal{J}_{x_{N}}$, while
\begin{equation}
\ket{e_{\vecund{j}|\vecund{x}}}=\bigotimes_{n=1}^{N}\ket{e_{j_{n}|x_{n}}},\quad\lambda_{\vecund{j}|\vecund{x}}=\prod_{n=1}^{N}\lambda_{j_{n}|x_{n}}
\end{equation}
are its eigenstate-sequences and eigenvalues.
Hence the $\delta-$conditionally typical space of the source conditioned on an input sequence $\hat{\rho}_{\vecund{x}}$ is defined as 
the span of all those eigenstate-sequences whose classical label is $\delta$-conditionally-typical, i.e., 
\begin{equation}\label{ctypSpace}
  \hil_{\delta,N}\left(\mathcal{J}|\vecund{x}\right)=\operatorname{span}\left\{\ket{e_{\vecund{j}|\vecund{x}}}:\vecund{j}\in T_{\delta,N}\left(\mathcal{J}|\vecund{x}\right)\right\},
\end{equation}
and the projector on this space is
\begin{equation}\label{ctypProj}
 \hat{\Pi}_{\vecund{x}}=\sum_{j_{\vecund{x}}\in T_{\delta,N}\left(\mathcal{J}|\vecund{x}\right)}\dketbra{e_{\vecund{j}|\vecund{x}}}.
\end{equation}
Similar properties as for the classical conditionally-typical set hold for this space, namely that for $\epsilon>0$ and $N$ sufficiently large:
\begin{enumerate}
 \item On average over input sequences, the state $\hat{\rho}_{\vecund{x}}$ resides with high probability in the $\delta$-conditionally-typical space of the source, i.e., 
 \begin{equation}\label{parzialeuno}
 \sum_{\vecund{x}\in\mathcal{C}}p_{\vecund{x}}\tr{\hat{\Pi}_{\vecund{x}}\hat{\rho}_{\vecund{x}}}\geq1-\epsilon;
 \end{equation}
  \item The size of the $\delta-$typical space is exponential in $N$ with scaling coefficient close to the average entropy of the source, i.e.,
  \begin{equation}\label{parzialedue}
  \begin{aligned}
  &\tr{\hat{\Pi}_{\vecund{x}}}\leq 2^{N\left( \sum_{x\in\mathcal{X}}p_{x}S(\hat{\rho}_{x})+\delta\right)},\\
  &(1-\epsilon)2^{N\left(\sum_{x\in\mathcal{X}}p_{x}S(\hat{\rho}_{x})-\delta\right)}\leq \sum_{\vecund{x}\in\mathcal{C}}p_{\vecund{x}}\tr{\hat{\Pi}_{\vecund{x}}};
  \end{aligned}
  \end{equation}
  \item The probability distribution of $\delta$-conditionally-typical sequences is approximately uniform, i.e.,
\begin{equation}\label{parzialetre}
2^{-N\left(\sum_{x\in\mathcal{X}}p_{x}S(\hat{\rho}_{x})+\delta\right)}\hat{\Pi}_{\vecund{x}}\leq \hat{\Pi}_{\vecund{x}}\hat{\rho}_{\vecund{x}}\hat{\Pi}_{\vecund{x}}\leq 2^{-N\left(\sum_{x\in\mathcal{X}}p_{x}S(\hat{\rho}_{x})-\delta\right)}\hat{\Pi}_{\vecund{x}}.
\end{equation}
\end{enumerate}
 Let us finally note that the conditionally-typical spaces of different codewords are in general not orthogonal, i.e., $\hat{\Pi}_{\vecund{x}}\hat{\Pi}_{\vecund{x}'}\neq0$, since they are built using spectral decompositions of non-orthogonal operators. Nevertheless their overlap is small and tends to zero as $N\rightarrow\infty$. \\
Based on the methods above, Holevo~\cite{holevo1,holevo2,holevo3} and Schumacher and Westmoreland~\cite{schumawest1,schumawest2} (HSW), computed an expression for the classical capacity of quantum channels. Given a quantum channel $\Phi$, there are several ways it can be employed to transmit classical information with $N$ uses, depending on whether the input codewords $\hat{\rho}_{\vecund{x}}$ and the output POVM elements $\hat{E}_{\vecund{y}}$ are factorized (classical) over the $N$ uses or not (quantum). In any of these cases, the capacity is computed by maximizing the classical mutual information of the input and output probability distributions, but if the states and/or the measurements are classical such optimization can be carried out over single-mode sources $\{(\hat{\rho}_{x},p_{x})\}$ and/or measurements $\{\hat{E}_{y}\}$. Taking into account this fact, there are four classical capacities that can be defined:
\begin{enumerate}
\item The classical-classical capacity, i.e.,
\begin{equation}\label{cQQ}
C_{cc}(\Phi)=\lim_{N\rightarrow\infty}\oneover{N}\max_{\{(\hat{\rho}_{x},p_{x})\},\{\hat{E}_{y}\}}I(X_{(1,N)}:Y_{(1,N)})=\max_{\{\hat{\rho}_{x}\},\{\hat{E}_{y}\}}C_{Shan};
\end{equation}
\item The classical-quantum capacity, i.e.,
\begin{equation}
C_{cq}(\Phi)=\lim_{N\rightarrow\infty}\oneover{N}\max_{\{(\hat{\rho}_{x},p_{x})\},\{\hat{E}_{\vecund{y}}\}}I(X_{(1,N)}:Y_{(1,N)})=\lim_{N\rightarrow\infty}\oneover{N}\max_{\{\hat{\rho}_{x}\},\{\hat{E}_{\vecund{y}}\}}C_{Shan}^{(N)};
\end{equation}
\item The quantum-classical capacity, i.e.,
\begin{equation}
C_{qc}(\Phi)=\lim_{N\rightarrow\infty}\oneover{N}\max_{\{(\hat{\rho}_{\vecund{x}},p_{\vecund{x}})\},\{\hat{E}_{y}\}}I(X_{(1,N)}:Y_{(1,N)})=\lim_{N\rightarrow\infty}\oneover{N}\max_{\{\hat{\rho}_{\vecund{x}}\},\{\hat{E}_{y}\}}C_{Shan}^{(N)};
\end{equation}
\item The quantum-quantum capacity, i.e.,
\begin{equation}
C_{qq}(\Phi)=\lim_{N\rightarrow\infty}\oneover{N}\max_{\{(\hat{\rho}_{\vecund{x}},p_{\vecund{x}})\},\{\hat{E}_{\vecund{y}}\}}I(X_{(1,N)}:Y_{(1,N)})=\lim_{N\rightarrow\infty}\oneover{N}\max_{\{\hat{\rho}_{\vecund{x}}\},\{\hat{E}_{\vecund{y}}\}}C_{Shan}^{(N)}.
\end{equation}
\end{enumerate}
Let us further note that $C_{qq}$ is the highest among the previous capacities and the true ultimate rate, while $C_{cc}$ is the lowest. However, there is no clear ordering between $C_{cq}$ and $C_{qc}$. In particular, not much is known about the latter, while the former is used to compute $C_{qq}$ as discussed in the following capacity theorem.
\begin{theorem}\label{hsw}
(HSW) The classical capacity of a quantum channel $\Phi$ can be computed as 
\begin{equation}\label{cUlt}
C(\Phi)=C_{qq}(\Phi)=\lim_{N\rightarrow\infty}\frac{C_{\chi}(\Phi^{\otimes N})}{N},
\end{equation}
where 
\begin{equation}\label{cChi}
C_{\chi}(\Phi)=C_{cq}(\Phi)=\max_{\{(\hat{\rho}_{x},p_{x})\}}\chi\left(\{\left(\Phi\left(\hat{\rho}_{x}\right),p_{x}\right)\}\right)
\end{equation}
and the $\chi$-information (or Holevo information) of a quantum source $\mathcal{S}$, \eqq{qSource}, is defined as
\begin{equation}\label{holevoInfo}
\chi(\mathcal{S})=S\left(\sum_{x\in\mathcal{X}}p_{x}\hat{\rho}_{x}\right)-\sum_{x\in\mathcal{X}}p_{x}S\left(\hat{\rho}_{x}\right).
\end{equation}
\end{theorem}
Clearly, the distinct definitions Eqs.~(\ref{cUlt},\ref{cChi}) are necessary due to the possible super-additivity of the channel capacity with respect to the input states, which has been proved to be true in general by Hastings~\cite{hastings}. However there are several specific cases where additivity holds, e.g., PI Gaussian channels for which the long-standing additivity conjecture has been recently solved by Giovannetti et al.~\cite{gaussOpt}, see also~\cite{maj1,maj2}. The latter result can be stated as follows:
\begin{theorem}\label{cPIG}
The classical capacity of a phase-insensitive quantum Gaussian channel, \eqq{PIG}, of parameters $\eta$ and $\sqrt{\tau}=\abs{1-\eta}(\bar{n}+1/2)$, with maximum input average-mode-energy $E$ is additive and equal to
\begin{equation}\label{ccapPI}
C(\Phi_{PI})=C_{\chi}(\Phi_{PI})=g(e(E))-g(e(0)),
\end{equation}
where 
\begin{equation}
g(\bar{n})=(1+\bar{n})\log(1+\bar{n})-\bar{n}\log \bar{n}
\end{equation}
is the entropy of a Gaussian thermal state $\hat{\rho}(\bar{n})$, while 
\begin{equation}
e(E)=\eta E+\max(0,\eta-1)+\bar{n}\abs{\eta-1}
\end{equation} 
is the maximum output average-mode-energy of the channel. Moreover this optimal value is achieved by a phase-insensitive Gaussian input source of average energy $E$, i.e.,
\begin{equation}\label{opSource}
\mathcal{S}_{PI}(E)=\left\{\left(\dketbra{\alpha},p_{\alpha}\right):\alpha\in\mathbb{C},p_{\alpha}=\frac{e^{-\frac{\abs{\alpha}^{2}}{E}}}{\pi E}\right\}.
\end{equation}
\end{theorem}
Coming back to the HSW Theorem~\ref{hsw}, its original proof can be found in~\cite{holevo1,holevo2,holevo3,schumawest1,schumawest2,wildeBOOK}. Here we are particularly interested in the direct part of the theorem that proves the achievability of any rate up to the Holevo information by exhibiting a decoding protocol based on typical projectors, Eqs.~(\ref{typProj},\ref{ctypProj}). Indeed these projectors are entangled over several uses of the channel, hence they require the ability to perform joint measurements on multiple output modes. This fact can be understood as a super-additivity of the channel capacity with respect to the output measurements and no counter-example to that has been found so far, though several alternative achievability proofs have been carried out~\cite{schumawest1,schumawest2,holevo2,holevo3,seq1,seq2,sen,polarWildeGuha,wildeguha1,wildeRen2,wildeRen1,wildeHayden,NOSTROHol}. In the following we briefly discuss these proofs and the measurement procedures they employ and in Sec.~\ref{sec:bisDec} we will present a new one based on~\cite{NOSTROHol}.\\

A proof of the direct part of Theorem~\ref{hsw} is usually carried out by providing an explicit decoding POVM 
\begin{equation}
\mathcal{M}_{dec}=\left\{\hat{E}_{\vecund{x}}\right\}_{\vecund{x}\in\mathcal{C}}\cup\left\{\hat{E}_{err}\right\}
\end{equation}
 with $\abs{\mathcal{C}}+1$ elements, i.e., one outcome for each possible input codeword, $\vecund{y}=\vecund{x}\in\mathcal{C}$, and, if necessary, an additional error outcome, $\vecund{y}=err$. Then the average error probability of the protocol with a given code $\mathcal{C}$ is computed:
 \begin{equation}\label{pErrGen}
 P_{err}(\mathcal{C})=1-2^{-NR}\sum_{\vecund{x}\in\mathcal{C}}p_{\vecund{x}|\vecund{x}}=1-2^{-NR}\sum_{\vecund{x}\in\mathcal{C}}\tr{\hat{E}_{\vecund{x}}\hat{\rho}_{\vecund{x}}},
 \end{equation}
where we have used the definition~\eqq{outcomeProb} of the outcome probability and assumed again w.l.o.g. that $\hat{\rho}_{\vecund{x}}$ are the states at the output of the channel. The proof consists in showing that the error probability~\eqq{pErrGen} approaches zero as the codewords' length $N\rightarrow\infty$ for all rates $R\leq \chi(\mathcal{S})$. As already discussed in Secs.~\ref{subsec:commC},~\ref{subsec:commQ}, a convergence on average over all codes is sufficient. Hence different capacity-achieving protocols are characterized by different decoding POVM structures. The protocols known so far are:
\begin{enumerate}
\item The square-root measurement (SRM) (or pretty good measurement) that was used in the original proof~\cite{holevo2,schumawest1,schumawest2} and relies on the use of typical projectors, Eqs.~(\ref{typProj},\ref{ctypProj}). Each POVM element is proportional to the conditionally-typical projector of the corresponding codeword. However these projectors are non-orthogonal and they would form an overcomplete POVM, i.e., summing up to more than the identity; hence each element is renormalized by their total sum. The resulting POVM is
\begin{equation}\label{srm}
\mathcal{M}_{srm}=\left\{\left(\sum_{\vecund{x}\in\mathcal{C}}\hat{\Pi}_{\vecund{x}}\right)^{-1/2}\hat{\Pi}_{\vecund{x}}\left(\sum_{\vecund{x}\in\mathcal{C}}\hat{\Pi}_{\vecund{x}}\right)^{-1/2}\right\}_{\vecund{x}\in\mathcal{C}},
\end{equation}
where the inverse of an operator is intended as a pseudoinverse, i.e., equal to the ordinary inverse on its support and zero on its kernel. Note that in this case there is no need for an error outcome. The SRM has a rather pragmatic construction but it is far from easy to envisage a feasible implementation of it;
\item The sequential measurement (SM) that was proposed by Giovannetti et al.~\cite{seq1,seq2} (see also~\cite{sen}) and relies on typical projectors too, though with a clear operational structure inspired by the works~\cite{oga1,oga2,oga3,oganaga,hayanaga} related to quantum hypothesis testing~\cite{qhyptest1,qhyptest2}. The codewords $\vecund{x}$ are first labelled according to a certain ordering, i.e., $\vecund{x}^{(\ell)}$ with $\ell=1,\cdots,\abs{\mathcal{C}}$. Then the received state is subject to a binary-outcome projective measurement which distinguishes the conditionally-typical space of the first codeword from its complementary. If the result is inside the conditionally-typical set, then Bob declares that the first codeword was identified, otherwise he repeats the same procedure with the conditionally-typical space of the second codeword and so on in the decided order, until either he declares that some codeword was identified or he has unsuccessfully probed all codewords. In the latter case Bob declares an error. Accordingly the POVM elements of the SM are
\begin{align}\label{seqSucEle}
&\hat{E}_{\vecund{x}^{(\ell)}}=\abs{\hat{\Pi}_{\vecund{x}^{(\ell)}}\hat{\Xi}_{\vecund{x}^{(\ell-1)}}\cdots\hat{\Xi}_{\vecund{x}^{(1)}}}^{2}~\forall\ell=1,\cdots,\abs{\mathcal{C}},\\ \label{seqErrEle}
&\hat{E}_{err}=\hat{\mathbf{1}}-\sum_{\forall\ell=1,\cdots,\abs{\mathcal{C}}}\hat{E}_{\vecund{x}^{(\ell)}},
\end{align}
where we have defined the complementary of a projector as $\hat{\Xi}=\hat{\mathbf{1}}-\hat{\Pi}$, which is itself a projector, and the absolute value of an operator as $\abs{\hat{O}}=\sqrt{\hat{O}^{\dag}\hat{O}}$. The downside of the SM is that it requires to perform several steps, of the order of the number of codewords $2^{NR}$, and still it can fail in some cases. A design for an explicit and structured optical receiver of this kind was proposed in~\cite{wildeguha1,wildeguha2} and, for a lossy bosonic channel, \eqq{att}, it was shown that a sequential decoder can be built with gaussian displacement operators and vacuum-or-not measurements~\cite{sen,qhyptest1,qhyptest2,LOSSY};
\item The successive cancellation decoder (SCD) together with a polar code (PC) that were proposed by Wilde and Guha~\cite{polarWildeGuha,wildeguha1,wildeRen2,wildeRen1,wildeHayden} following a proposal by Arikan~\cite{arikan} for classical channels. This communication scheme relies on the phenomenon of channel polarization, i.e., when a set of channels, after a sequence of simple SWAP and CNOT operations, can be divided in two subsets such that: the channels of the first kind have mutual information approaching the channel capacity, while those of the second kind have mutual information approaching zero asymptotically in the number of channels employed. By coding on the good channels and sending frozen bits on the other ones, faithful transmission is accomplished. However the SCD requires a series of successive joint measurements over channel outputs, each aimed at recovering the output of a single channel. Hence it is affected by a jointness requirement quite similar to those of the SRM and SM, though on somewhat smaller sub-sequences of the codewords~\cite{wildeHayden}. Nevertheless it is less demanding than the SM since it requires a number of steps equal to the codewords' length $N$ and it still has a clear operational meaning. The optimal measurements composing the SCD are Helstrom-optimal projectors, \eqq{helstroOptimalProj}, on the positive and negative support of operators related to the output codewords.
\end{enumerate}
Eventually, let us note that all these capacity-achieving decoding protocols (including the one presented in Sec.~\ref{sec:bisDec}) require the ability to perform joint measurements on multi-mode quantum states and hence are far from easy to implement with current technology. For example, the practical use of the Gaussian source \eqq{opSource} described above is limited by the difficulty in implementing  the optimal detection scheme that is able to read its codewords efficiently. This motivates the search for alternative ways of encoding classical messages into the channel which, while being possibly not as efficient as the Gaussian one, will guarantee nevertheless better performances with readout strategies which are easier to implement.  The proposal of ~\ref{sec:had} is one of such attempts. 

\section{Quantum State Discrimination}\label{sec:disc}
\subsection{Minimum-error discrimination}\label{subsec:discTheo}
The problem of state discrimination stems from a similar setting to that of communication. Indeed it can be stated too in terms of a random variable $X$ with values in a certain set $\mathcal{X}$ with a certain probability distribution $P_{X}(x)$: given an instance of $X$ we want to estimate its true value. The set can comprise classical values $x$ or quantum states $\hat{\rho}_{x}$ of arbitrary dimension and structure, e.g., they can be single-mode or multi-mode states. In the following we will discuss only the case of a discrete random variable, since the techniques developed in the estimation of a continuous variable are quite different from the discrete case, see~\cite{helstromBOOK,metroRev1,metroRev2}. While in the classical setting the main difficulty in discrimination arises from an imperfect knowledge of the system and a noisy measurement procedure, in the quantum case there is a more fundamental reason why a generic set of states cannot be perfectly distinguished: the uncertainty principle. Indeed if we were able to decide with certainty which one of two non-orthogonal quantum states has been prepared, we would then be able to perform perfect cloning, which in turn would allow to measure the value of two non-commuting observables starting from a single copy of a quantum state. \\
A quantum discrimination procedure is also determined by a POVM describing the measurement to be performed, with at least one outcome for each possible state. If we also allow for an error outcome it is even possible to achieve perfect discrimination; this procedure is called unambiguous discrimination (UD) and it consists in accepting only those results that give a certain identification of the state, while discarding all other uncertain results as errors, see~\cite{ivanovic,dieks,peres,chefBarn,revDisc} for further details. In the following instead we will be interested with the standard minimum-error state discrimination (MED) with one outcome for each input state, i.e., $y=x\in\mathcal{X}$, where the figure of merit is the error probability of the measurement protocol $\mathcal{M}=\{\hat{E}_{x}\}_{x\in\mathcal{X}}$  in discriminating the source of $M$ quantum states $\mathcal{S}$, as given by \eqq{qSource}:
\begin{equation}\label{errProbMed}
P_{err}(\mathcal{S},\mathcal{M})=1-\sum_{x\in\mathcal{X}}p_{x} p_{x|x}=1-\sum_{x\in\mathcal{X}}p_{x}\tr{\hat{E}_{x}\hat{\rho}_{x}},
\end{equation}
analogous to that defined in the context of communication, \eqq{pErrGen}.
We are interested in optimizing the error probability, \eqq{errProbMed}, over all POVM choices, finding its minimum value and the measurement that attains it, which in general is not unique~\cite{helstromBOOK}. Hence for a given source both the problems of discrimination and communication can be defined by optimizing a proper figure of merit, i.e., the error probability and the mutual information respectively. Moreover these two quantities can be related via Fano's inequality as employed in the converse Shannon theorem, showing that a small error probability implies a rate close to the capacity, see~\cite{covThomBOOK}.\\
The seminal works of Holevo~\cite{holevoDisc} and Yuen, Kennedy, Lax~\cite{YKL} (HYKL), derived a set of necessary and sufficient conditions for the optimal POVM; in particular~\cite{YKL} showed how to state the MED problem as a semidefinite program~\cite{sdp1,sdp2}, a popular optimization technique that can be easily carried out numerically; for a full treatment see~\cite{helstromBOOK,bae,eldarSdp}. In particular the optimization problem can be stated as
\begin{align}
&\max_{\mathcal{M}}\left[1-P_{err}(\mathcal{S},\mathcal{M})\right]\\
&\text{subject to } \hat{E}_{x}\geq0~\forall x\in\mathcal{X},\quad \sum_{x\in\mathcal{X}}\hat{E}_{x}=\hat{\mathbf{1}},
\end{align}
while its dual as
\begin{align}
&\min_{\hat{Q}}\tr{\hat{Q}}\\
&\text{subject to } \hat{Q}\geq p_{x}\hat{\rho}_{x}~\forall x\in\mathcal{X},
\end{align}
where $\hat{Q}$ is an hermitian operator on the Hilbert space of the system. Moreover this couple of conjugated problems is strongly dual, i.e., they have the same solution 
\begin{equation}\label{pErrOp}
\tr{\hat{Q}_{op}}=1-\mathbb{P}_{err}(\mathcal{S})=1-P_{err}(\mathcal{S},\mathcal{M}_{op}).
\end{equation}
Hence the easier dual problem can be solved numerically for a given set of states $\mathcal{S}$; its solution directly provides the value of the minimum error probability while the optimal POVM is found by solving the HYKL conditions
\begin{equation}
\hat{E}_{x}^{op}\left(\hat{Q}_{op}-p_{x}\hat{\rho}_{x}\right)=\left(\hat{Q}_{op}-p_{x}\hat{\rho}_{x}\right)\hat{E}_{x}^{op}=0~\forall{x}\in\mathcal{X},
\end{equation}
which are necessary and sufficient given a solution $\hat{Q}_{op}$ of the dual problem.
Using this statement of the problem, Eldar et al.~\cite{eldarSdp} were able to generalize a result by Kennedy~\cite{kenProj} and showed that each element $\hat{E}_{x}^{op}$ of the optimal POVM has rank at most as large as the rank of the corresponding state $\hat{\rho}_{x}$. This in particular implies that for a set of pure states $\hat{\rho}_{x}=\dketbra{\psi_{x}}$, the optimal POVM elements are rank-one non-orthogonal projectors; they are also orthogonal if the states are linearly independent, i.e., their number is not greater than the dimension of the Hilbert space, $M\leq d$. More generally, it is clear that the optimal POVM need only have support on the subspace of the whole Hilbert space spanned by the states of the source. \\

The SDP approach to MED has been also used by Bae~\cite{bae} as a starting point for a geometrical treatment of the problem, see also~\cite{qubits,threeQubits}. Such treatment is particularly successful in the qubit case, where an easy geometrical representation of the Hilbert space is available, i.e., the Bloch sphere, Sec.~\ref{subsec:qubits}. The main result of~\cite{bae} in the qubit case states that the optimal success probability of discrimination of a set $\mathcal{S}_{eq}^{(M)}=\{\hat{\rho}_{\ell}/M\}_{\ell=0,\cdots,M-1}$ of $M$ \textit{equiprobable} qubit states can be found by: i) considering the geometric figure determined by the weighted states in the Bloch sphere, i.e., their polytope of vertices $\{\vecund{r}_{\hat{\rho}_{\ell}}/M\}$; ii) finding the polytope similar to the latter and that is also maximal in the Bloch sphere; iii) computing the ratio $R$ between the original and the maximal polytope. Then the optimal success probability is 
\begin{equation}\label{baeDisc}
\mathbb{P}_{succ}\left(\mathcal{S}_{eq}^{(M)}\right)=\frac{1}{M}+R.
\end{equation} \\
Another important case where the solution can be found analytically is that of a cyclic-symmetric set of pure states~\cite{helstromBOOK,multiHel1,multiHel2}, described by a source
\begin{equation}
\mathcal{S}_{U}=\left\{(\ket{\psi_{\ell}},p_{\ell}):\ket{\psi_{\ell}}=\hat{U}^{\ell}\ket{\psi_{0}}, p_{\ell}=\oneover{M}~\forall \ell=0,\cdots,M-1\right\},
\end{equation}
where we have adopted a specific ordering of the states, i.e., $x=x^{(\ell)}$. These states are equiprobable and they can be obtained by repeatedly applying the same unitary operation $\hat{U}$ to a given reference state $\ket{\psi_{0}}$. The optimal POVM can be shown to have the same cyclic symmetry with rank-one projectors $\hat{E}_{\ell}^{op}=\dketbra{\phi_{\ell}}$ for $\ell=0,\cdots,M-1$, where
\begin{equation}
\quad\ket{\phi_{\ell}}=\hat{U}^{\ell}\ket{\phi_{0}}, \quad \ket{\phi_{0}}=\sum_{k=0}^{K-1}\lambda_{k}^{-1/2}e^{-i\theta_{k}\ell}\braket{d_{k}}{\psi_{0}}\ket{d_{k}},
\end{equation}
and $\{\ket{d_{k}}\}_{k=0}^{K-1}$ is the common eigenbasis of the symmetry operator $\hat{U}$ and of the average state of the source $\hat{\rho}_{\mathcal{S}}$, with eigenvalues respectively $\lambda_{k}$ and $e^{-i\theta_{k}}/M$. The minimum error probability attained by this measurement then is
\begin{equation}\label{cycSym}
\mathbb{P}_{err}\left(\mathcal{S}_{U}\right)=1-\abs{\sum_{k=0}^{K-1}\lambda_{k}^{-1/2}\abs{\braket{d_{k}}{\psi_{0}}}^{2}}^{2}.
\end{equation}
Let us note that the number of outcomes of the optimal measurement is bounded by the dimension of the Hilbert space of the system, i.e., $K\leq d$; hence if the number of states is greater than the dimension of the Hilbert space, only $d$ of them can be truly discriminated. This treatment can be generalized to sets with other kinds of symmetries, see~~\cite{bae,chiribella,symGuha}.\\
Eventually, in the case of binary discrimination, i.e., when the source comprises only two states, $\hat{\rho}_{0,1}$, it is easy to show~\cite{holevoBOOK,nChuangBOOK,opMeasCoh} that the optimal POVM comprises two orthogonal rank-one projectors 
\begin{equation}\label{helstroOptimalProj}
\hat{\Pi}_{+}^{op}=\hat{\Pi}_{\{\hat{\rho}_{0}-\hat{\rho}_{1}\geq 0\}}, \quad \hat{\Pi}_{-}^{op}=\hat{\mathbf{1}}-\hat{\Pi}_{+}^{op},
\end{equation}
respectively on the positive and negative support of the operator obtained by taking the difference of the two states. The corresponding optimal error probability is given by the well-known Helstrom bound~\cite{helstromBOOK}:
\begin{equation}\label{helBound}
\mathbb{P}_{err}^{(hel)}(\{(\hat{\rho}_{\ell},p_{\ell})\}_{\ell=0,1})=\oneover{2}-\trd{\hat{\rho}_{0}}{\hat{\rho}_{1}},
\end{equation}
where we have used the trace-distance of the two states, \eqq{trd}. In the case of pure states, the latter formula reduces to
\begin{equation}
\mathbb{P}_{err}^{(hel)}(\{(\ket{\psi_{\ell}},p_{\ell})\}_{n=0,1})=\oneover{2}\left(1-\sqrt{1-4p_{0}p_{1}\abs{\braket{\psi_{0}}{\psi_{1}}}^{2}}\right).
\end{equation}
In Sec.~\ref{sec:opDisc} we discuss a new method to compute the optimal POVM and error probability based on~\cite{NOSTROTheoDisc}.

\subsection{Implementation strategies}\label{subsec:discPract}
The problem of state discrimination has key importance in classical communication on quantum channels. For example, when employing free-space or optical-fiber links over long distances, the attenuation effect is dominating over other kinds of noise and hence the channels are easily modeled by a lossy bosonic channel, \eqq{att}. As stated by Theorem~\ref{cPIG}, a code based on coherent states is optimal in this case. Such states are largely overlapping at low intensity, \eqq{coherentNotOrth}, and it is thus extremely important to design receivers that discriminate them as efficiently as possible and may be employed in principle to increase the information transmission rate.\\
Let us now consider binary coherent-state MED, i.e., $\mathcal{S}_{\pm\alpha}=\left\{\left(\ket{\pm\alpha},\oneover{2}\right)\right\}$, where we restrict to the equiprobable case for simplicity. As in many other cases, it turns out that the optimal POVM achieving the Helstrom bound in this case is highly non-linear~\cite{opMeasCoh} and practically impossible to implement with current technology. The easiest decoder is the standard homodyne receiver (HoR)~\cite{RevGauss,parisDet}, obtained by performing homodyne detection, Eq.~(\ref{homodyne}), of the canonical quadrature $\hat{q}$ and assigning positive (negative) outcomes to $\ket{\alpha}$ ($\ket{-\alpha}$), i.e.
\begin{equation}\label{hor}
\mathcal{M}_{hor}=\left\{\hat{E}_{-}=\int_{q\leq0}dq\dketbra{q},\hat{E}_{+}=\hat{\mathbf{1}}-\hat{E}_{-}\right\}.
\end{equation}
Its error probability can be computed by employing the Fock representations of coherent and quadrature states, Eqs.~(\ref{coherentstates},\ref{quadraturestates}), and the expression for the generating function of Hermite polynomials, see~\cite{parisDet}, and it is given by
\begin{equation}
P_{err}(\mathcal{S}_{\pm\alpha},\mathcal{M}_{hor})=\oneover{2}\left(1-\operatorname{erf}\left(q_{\alpha}\right)\right),
\end{equation} 
where $\operatorname{erf}(x)=\pi^{-1/2}\int_{\abs{y}\leq x}dy\;e^{-y^{2}}$ is the well-known error function. The HoR, \eqq{hor}, has been shown by Takeoka and Sasaki \cite{opGaussDet} to be optimal among all Gaussian measurements for binary coherent-state MED, and it approximates well the optimal measurement at high energy but it is sub-optimal in the low-energy region that we are mainly interested in. \\
A first realistic, yet sub-optimal, non-Gaussian receiver was proposed by Kennedy \cite{kenDet}: it employs a coherent displacement operation, \eqq{dispcomplex}, that perfectly nulls one of the two possible signals (or the most favored one if they are not equiprobable), followed by OOD, \eqq{photodetector}. If the detector registers no photon the result is interpreted as a successful identification of the nulled signal, else if the detector clicks the other signal is chosen, i.e.,
This receiver captures the main ingredient in practical coherent-state discrimination, i.e. signal nulling, which reduces the errors committed when interpreting a no-click result from the OOD. Nevertheless at low intensity values its performance is lower than that of conventional dyne detectors.
Better results, which surpass the homodyne detection also at low intensity values, can be obtained by employing the optimized Kennedy scheme~\cite{opKen}, an imperfect nulling technique where the displacement is arbitrary, i.e.,
\begin{equation}\label{kennedy}
\mathcal{M}_{ken}(\beta)=\left\{\hat{E}_{-}=\dketbra{\beta},\hat{E}_{+}=\hat{\mathbf{1}}-\hat{E}_{-}\right\},
\end{equation} 
and it is chosen so as to minimize the error probability of the protocol, i.e.,
\begin{equation}
P_{err}(\mathcal{S}_{\pm\alpha},\mathcal{M}_{ken}(\beta))=\oneover{2}\left(1+\abs{\braket{\beta}{-\alpha}}^{2}-\abs{\braket{\beta}{+\alpha}}^{2}\right),
\end{equation}
which depends on the difference between the vacuum-overlaps of the two displaced states~\cite{opKen,marq} because $\abs{\braket{\beta}{\alpha}}=\abs{\braket{0}{\alpha-\beta}}$. 
Further improvements can be obtained by embedding the above techniques into a multiplexing procedure~ \cite{implProj,implProj2,sychLeuchs,multicopy1,multicopy2,genDisc,genDiscErr,guhaDol}  along the line first suggested by Dolinar~\cite{dol}: here the received coherent is at first split in $N$ lower-intensity copies via a passive interferometer, Sec.~\ref{gaussUni}; the copies are then individually probed (say via an optimized Kennedy detection) in a feedforward-adaptive routine where the parameters of the $n$-th detection are determined by the outcomes of the previous $n-1$ ones. This approach ensures a reduction of the error probability as $N$ increases so that, assuming perfect OODs, it allows for the saturation of the Helstrom bound in the asymptotic limit of infinitely many iterations. The Dolinar receiver has been recently demonstrated in a proof-of-principle experiment~\cite{dolExp}, although the fast feedforward that it requires is still technologically demanding and it can be seen as its main limitation. Variations of the easier optimized Kennedy scheme have thus been discussed~\cite{opGaussDet,wittmanKen,tsujinoKen,NOSTRODisc} and in particular~\cite{NOSTRODisc} will be presented in Sec.~\ref{sec:cohDisc}. In any case the Dolinar receiver seems to be the most robust to typical detection errors, see~\cite{gereErr,marq}. Moreover the Dolinar receiver is easy to generalize to sets of larger numbers of states but it seems difficult that its generalizations may be able to achieve the optimum multi-state error probability~\cite{bondurant,izumiMulti,fujiwaraMulti,becerraMulti,marq,guhaDet}, so that in general the addition of some highly non-Gaussian operation may be needed~\cite{guhaOpQuantum}.\\
Eventually, as already noted, we may expect that a good discrimination technique performs well also for coherent-state communication, though in general it will not achieve the channel capacity. The relation between these two tasks has been little investigated~\cite{multiHel1,guhaDet} because it is in general a complex one, e.g., the detection problem relevant for communication would need to consider multi-mode codewords. However realistic decoders for coherent-state communication have to take into account the fact that ultimately a single-mode detection involving PDTs will be involved and the only way to perform a multi-mode measurement in this way is to perform some multi-mode interaction before detection. In Ch.~\ref{ch:Imple} we address both the problems from an implementation-oriented perspective based on~\cite{NOSTRODisc,NOSTROHad,NOSTROAdap}.

\chapter{Optimal communication and state discrimination with binary-tree-search protocols}\label{ch:Opt}

\ifpdf
    \graphicspath{{Chapter3/Figs/Raster/}{Chapter3/Figs/PDF/}{Chapter3/Figs/}}
\else
    \graphicspath{{Chapter3/Figs/Vector/}{Chapter3/Figs/}}
\fi

In this chapter we consider the problems of optimal communication, Sec.~\ref{sec:commQ}, and state discrimination, Sec.~\ref{sec:disc}, in whole generality. In Sec.~\ref{sec:bisMeas} we describe an operational decomposition of any POVM, Sec.~\ref{quantumOps}, in terms of binary-outcome ones through a tree-like structure. In Sec.~\ref{sec:bisDec} we apply the decomposition to classical communication on quantum channels, obtaining a new capacity-achieving decoding protocol that requires only a number of steps equal to the length of the codewords. In Sec.~\ref{sec:opDisc} instead we apply the same decomposition to quantum state discrimination, obtaining an expression of the error probability for sets of three or four states that can be optimized numerically; we also discuss the specific case of qubits, where the optimization is particularly straightforward. This chapter is based on~\cite{NOSTROHol,NOSTROTheoDisc}.

\section{A binary-tree-search decomposition for quantum measurements}\label{sec:bisMeas}
In this section we prove that any quantum measurement with an arbitrary number of outcomes can be decomposed into a sequence of nested measurements with binary outcomes, where the previous results determine the choice of successive measurements.
At variance with a previous proof~\cite{previousBisDec}, ours does not make use of the spectral decomposition of the initial measurement operators; we present it here in a form adapted to the main purpose of the following sections.\\
Let us suppose we want to perform a quantum measurement with $M$ possible outcomes: it can be expressed in general as a POVM $\mathcal{M}^{(M)}=\left\{\hat{E}_{\ell}\right\}_{\ell=0}^{M-1}$, see Sec~\ref{quantumOps}. This expression can be interpreted as a one-shot measurement with several possible results and its practical realization may often be very hard. 
On the other hand we could restrict to performing only measurements with two outcomes, as described by \textit{binary} POVMs: $\mathcal{B}\equiv\mathcal{M}^{(2)}=\left\{\hat{B}_{0},\hat{B}_{1}\right\}$. This may be useful when limited technological capabilities or specific theoretical requirements constrain the  number of allowed outcomes and the complexity of our measurement. 
It is then natural to ask whether this smaller set of resources is sufficient to describe a general quantum measurement. We answer positively by showing that the more general $M$-outcome formalism can be broken up into several binary steps and interpreted as a sequence of nested POVMs with two outcomes, trading a \textit{one-shot}, \textit{multiple-outcome} measurement for a \textit{multiple-step}, \textit{yes-no} measurement.\\
The nested POVM can be expressed in terms of \textit{conditional binary} POVMs 
\begin{equation}
\mathcal{B}_{\vecund{k}}=\left\{\hat{B}_{\vecund{k},0},\hat{B}_{\vecund{k},1}\right\},
\end{equation} 
each complete by itself, to be applied only if a specific string $\vecund{k}$ of previous results is obtained. 
 For example for $M=4$ the nested POVM can be realized in two steps and written compactly as the collection of three binary POVMs: 
  \begin{equation}
  \mathcal{N}^{(4)}=\left\{\mathcal{B}^{(2)}_{0},\mathcal{B}^{(2)}_{1}\right\}\circ\left\{\mathcal{B}^{(1)}\right\},
  \end{equation}
   properly composed as follows and shown in Fig.~\ref{schema}. The measurement starts by applying the first-step binary POVM $\mathcal{B}^{(1)}=\left\{\hat{B}^{(1)}_{k_{1}}\right\}_{k_{1}=0,1}$ then, depending on its outcome $k_{1}$, it selects $\mathcal{B}^{(2)}_{k_{1}}=\left\{\hat{B}^{(2)}_{k_{1},k_{2}}\right\}_{k_{2}=0,1}$ among the two POVMs available in the second-step collection $\left\{\mathcal{B}^{(2)}_{0},\mathcal{B}^{(2)}_{1}\right\}$. Eventually, the chosen second-step POVM is applied, receiving an outcome $k_{2}$. The total outcome is a string of two bits, i.e., $k_{1},k_{2}$, whose value identifies one of  four possible outcomes, as desired. Suppose now to apply this measurement on a state $\hat{\hat{\rho}}$ of some physical system: if the first-step outcome is $k_{1}=0$, the resulting unnormalized evolved state is $\sqrt{\hat{B}^{(1)}_{0}}\hat{\rho}\sqrt{\hat{B}^{(1)}_{0}}$; if then the second-step outcome is $k_{2}=0$, the final unnormalized state of the system is $\sqrt{\hat{B}^{(2)}_{0,0}}\sqrt{\hat{B}^{(1)}_{0}}\hat{\rho}\sqrt{\hat{B}^{(1)}_{0}}\sqrt{\hat{B}^{(2)}_{0,0}}$. This means that the nested POVM has a more explicit representation as 
 \begin{equation} 
 \mathcal{N}^{(4)}=\left\{\hat{F}_{k_{1},k_{2}}=\left|\sqrt{\hat{B}^{(2)}_{k_{1},k_{2}}}\sqrt{\hat{B}^{(1)}_{k_{1}}}\right|^{2}\right\}_{k_{1},k_{2}=0,1}.
 \end{equation} 
 \begin{figure}[t!]
\center \includegraphics[trim={0 8cm 0 10cm},clip,scale=.43]{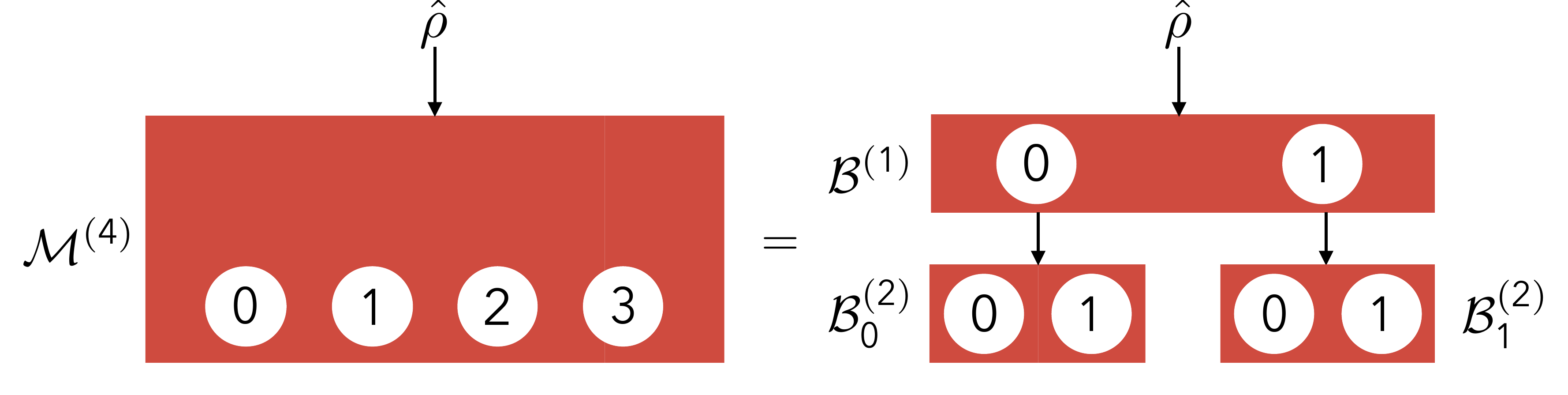}\caption{Schematic depiction of the nested decomposition for $M=4$, explicitly discussed in the text. Any four-outcome measurement $\mathcal{M}^{(4)}$ acting on a state $\hat{\rho}$ is equivalent to the concatenation of two-outcome measurements: the first-step one $\mathcal{B}^{(1)}$, with result $k_{1}=0,1$, and the second-step ones $\mathcal{B}^{(2)}_{k_{1}}$, which are mutually exclusive and applied only if the corresponding first outcome $k_{1}$ was obtained. }
 \label{schema}
 \end{figure}\\
More generally, we can define a nested POVM $\mathcal{N}^{(M)}$ of order $M=2^{u_{\tiny{F}}}$ as 
\begin{equation}\label{nested}
\begin{aligned}
\mathcal{N}^{(M)}&=\left\{\mathcal{B}^{(u_{\tiny{F}})}_{k_{(1,u_{\tiny{F}}-1)}}\right\}_{k_{1},\cdots,k_{u_{\tiny{F}}-1}=0,1}\circ\cdots\circ\left\{\mathcal{B}^{(1)}\right\}\\
&=\left\{\hat{F}_{k_{(1,u_{\tiny{F}})}}=\left|\sqrt{\hat{B}^{(u_{\tiny{F}})}_{k_{(1,u_{\tiny{F}})}}} \cdots\sqrt{\hat{B}^{(1)}_{k_{1}}}\right|^{2}\right\}_{k_{1},\cdots,k_{u_{\tiny{F}}}=0,1},
\end{aligned}
\end{equation}
i.e., the collection of $2^{u_{\tiny{F}}}-1$ binary POVMs $\mathcal{B}^{(u)}_{k_{(1,u-1)}}$, for all previous outcomes $k_{(1,u-1)}$ at a given step $u$ and all steps $u=1,\cdots,u_{\tiny{F}}$. We can certify that $\mathcal{N}^{(M)}$ so constructed actually is a POVM by checking positivity and completeness of its elements $\hat{F}_{k_{(1,u_{\tiny{F}})}}$: the former requirement is trivial, while the latter follows from the fact that each binary POVM is complete, as can be easily shown by employing the completeness of each binary POVM $\mathcal{B}^{(u)}_{k_{(1,u-1)}}$ at each step $u$. Indeed we can start by summing over the last bit $k_{u_{\tiny{F}}}=0,1$, coupling elements that differ only for its value, i.e., $\hat{F}_{k_{(1,u_{\tiny{F}}-1)},0}$ and $\hat{F}_{k_{(1,u_{\tiny{F}}-1)},1}$. 
These are made of the same sequence of operators apart from the most interior ones, $\hat{B}^{(u_{\tiny{F}})}_{k_{(1,u_{\tiny{F}}-1)},0}$ and $\hat{B}^{(u_{\tiny{F}})}_{k_{(1,u_{\tiny{F}}-1)},1}$, which are instead two different elements of the same binary POVM $\mathcal{B}^{(u_{\tiny{F}})}_{k_{(1,u_{\tiny{F}}-1)}}$, thus satisfy a completeness relation and their sum can be simplified. The same procedure is then applied recursively on previous bits as follows:
\begin{equation}
\begin{aligned}
\sum_{k_{1},\cdots,k_{u_{\tiny{F}}}}&\hat{F}_{k_{(1,u_{\tiny{F}})}}=\sum_{k_{1},\cdots,k_{u_{\tiny{F}}-1}}\left(\hat{F}_{k_{(1,u_{\tiny{F}}-1)},0}+\hat{F}_{k_{(1,u_{\tiny{F}}-1)},1}\right)\\
&=\sum_{k_{1},\cdots,k_{u_{\tiny{F}}-1}}\sqrt{\hat{B}^{(1)}_{k_{1}}}\cdots\sqrt{\hat{B}^{(u_{\tiny{F}}-1)}_{k_{(1,u_{\tiny{F}}-1)}}}\left(\hat{B}^{(u_{\tiny{F}})}_{k_{(1,u_{\tiny{F}}-1)},0}+\hat{B}^{(u_{\tiny{F}})}_{k_{(1,u_{\tiny{F}}-1)},1}\right)\sqrt{\hat{B}^{(u_{\tiny{F}}-1)}_{k_{(1,u_{\tiny{F}}-1)}}}\cdots\sqrt{\hat{B}^{(1)}_{k_{1}}}\\
&=\sum_{k_{1},\cdots,k_{u_{\tiny{F}}-1}}\sqrt{\hat{B}^{(1)}_{k_{1}}}\cdots\sqrt{\hat{B}^{(u_{\tiny{F}}-2)}_{k_{(1,u_{\tiny{F}}-2)}}}\hat{B}^{(u_{\tiny{F}}-1)}_{k_{(1,u_{\tiny{F}}-1)}}\sqrt{\hat{B}^{(u_{\tiny{F}}-2)}_{k_{(1,u_{\tiny{F}}-2)}}}\cdots\sqrt{\hat{B}^{(1)}_{k_{1}}}\\
&=\sum_{k_{1},\cdots,k_{u_{\tiny{F}}-2}}\sqrt{\hat{B}^{(1)}_{k_{1}}}\cdots\sqrt{\hat{B}^{(u_{\tiny{F}}-2)}_{k_{(1,u_{\tiny{F}}-2)}}}\left(\hat{B}^{(u_{\tiny{F}}-1)}_{k_{(1,u_{\tiny{F}}-2)},0}+\hat{B}^{(u_{\tiny{F}}-1)}_{k_{(1,u_{\tiny{F}}-2)},1}\right)\sqrt{\hat{B}^{(u_{\tiny{F}}-2)}_{k_{(1,u_{\tiny{F}}-2)}}}\cdots\sqrt{\hat{B}^{(1)}_{k_{1}}}\\
&=\cdots=\sum_{k_{1}}\sqrt{\hat{B}^{(1)}_{k_{1}}}\left(\hat{B}^{(2)}_{k_{1},0}+\hat{B}^{(2)}_{k_{1},1}\right) \sqrt{\hat{B}^{(1)}_{k_{1}}} =\hat{B}^{(1)}_{0}+\hat{B}^{(1)}_{1}=\hat{\mathbf{1}}.\label{complete}
\end{aligned}
\end{equation}
We note that the previous result does not change if instead of employing complete binary POVMs, we relax to weak completeness, as defined in Sec.~\ref{sec:bisMeas}. Indeed in this case it still holds 
\begin{equation}
\sqrt{\hat{B}^{(u-1)}_{k_{(1,u-1)}}}\left(\hat{B}^{(u)}_{k_{(1,u-1)},0}+\hat{B}^{(u)}_{k_{(1,u-1)},1}\right)\sqrt{\hat{B}^{(u-1)}_{k_{(1,u-1)}}}=\hat{B}^{(u-1)}_{k_{(1,u-1)}}
\end{equation} 
and the equalities in \eqq{complete} are unchanged.\\

In light of the previous discussion we can now state the main theorem:
\begin{theorem}\label{decomposition}(Binary-tree decomposition)
Any $M$-outcome POVM $\mathcal{M}^{(M)}=\{\hat{E}_{j}\}_{j=0}^{M-1}$ is equivalent to a nested POVM $\mathcal{N}^{(\tilde{M})}$, $\tilde{M}=2^{u_{\tiny{F}}}$, as in ~\eqq{nested}, composed exclusively of binary POVMs $\mathcal{B}^{(u)}_{k_{(1,u-1)}}$, with a total number of steps $u_{\tiny{F}}$ equal to:
\begin{enumerate}
\item $\log M$, if $M$ is a power of $2$;
\item $\lceil\log M\rceil$ otherwise, where $\lceil\cdot\rceil$ is the ceiling function, equal to the smallest integer following the argument.
\end{enumerate}
 \end{theorem}
 \begin{proof}
 Consider the first case above, i.e., $M=2^{u_{\tiny{F}}}\equiv\tilde{M}$. 
We start by providing a binary representation of the labels $\ell$ of the initial POVM $\mathcal{M}^{(M)}$, i.e., we define $\hat{E}_{k_{(1,u)}}\equiv \hat{E}_{\ell^{(k)}}$, with $\ell^{(k)}=\sum_{u=1}^{u_{\tiny{F}}}2^{u-1}k_{u}$. In order to prove the theorem we have to show that by combining the elements of the initial $M$-outcome POVM $\mathcal{M}^{(M)}$ one can always define a set of binary POVMs $\mathcal{B}^{(u)}_{k_{(1,u-1)}}$, for all $k_{1},\cdots,k_{u-1}=0,1$ and $u=1,\cdots,u_{\tiny{F}}$, such that: i) their nested composition is a POVM of the form $\mathcal{N}^{(M)}$, \eqq{nested}; ii) the elements of the latter are equal to the elements of $\mathcal{M}^{(M)}$.  \\
First of all we construct the binary elements at each step $u$ by taking the sum of all the elements $\hat{E}_{k_{(1,u)},k_{(u+1,u_{\tiny{F}})}}$ with a fixed value of the first $u$ bits, then renormalizing it by all previous binary elements, as in a Square Root Measurement~\cite{schumawest1,schumawest2}. For example define the elements of the first-step POVM $\mathcal{B}^{(1)}$ as
 \begin{equation}\label{first}
 \hat{B}^{(1)}_{k_{1}}=\sum_{k_{(2,u_{\tiny{F}})}}\hat{E}_{k_{1},k_{(2,u_{\tiny{F}})}},
 \end{equation}
 for each value of the outcome $k_{1}=0,1$; here and in the following we omit the obvious set of values assumed by the sum index. Being a sum of positive operators, the elements so defined are themselves positive; moreover their sum equals the sum of all the elements of $\mathcal{M}^{(M)}$, implying that they are complete. At the second step define the elements of the two possible POVMs $\mathcal{B}^{(2)}_{k_{1}}$ as
  \begin{equation}\label{second}
 \hat{B}^{(2)}_{k_{(1,2)}}=\left(\hat{B}^{(1)}_{k_{1}}\right)^{-1/2}\sum_{k_{(3,u_{\tiny{F}})}}\hat{E}_{k_{(1,2)},k_{(3,u_{\tiny{F}})}}\left(\hat{B}^{(1)}_{k_{1}}\right)^{-1/2},
 \end{equation}
 where the inverse of an operator is its pseudoinverse.  Also in this case the defined elements  are positive by construction, but they are not complete. Indeed it is easy to show, employing the definition \eqq{first}, that $\hat{B}^{(2)}_{k_{1},0}+\hat{B}^{(2)}_{k_{1},1}=\hat{\mathbf{1}}_{k_{1}}$. Here $\hat{\mathbf{1}}_{k_{1}}$ is the projector on the support of the previous-outcome operator, $\hat{B}^{(1)}_{k_{1}}$, which may have a non-trivial kernel, so that in general it holds $\hat{\mathbf{1}}_{k_{1}}\leq\hat{\mathbf{1}}$. This problem may be overcome easily by redefining the POVM elements as $\hat{\tilde{B}}^{(2)}_{k_{(1,2)}}=\hat{B}^{(2)}_{k_{(1,2)}}\oplus\left(\hat{\mathbf{1}}-\hat{\mathbf{1}}_{k_{1}}\right)/2$, i.e., trivially expanding the support of those already defined in \eqq{second}, so that $\hat{\tilde{B}}^{(2)}_{k_{1},0}+\hat{\tilde{B}}^{(2)}_{k_{1},1}=\hat{\mathbf{1}}_{k_{1}}\oplus\left(\hat{\mathbf{1}}-\hat{\mathbf{1}}_{k_{1}}\right)=\hat{\mathbf{1}}$. This operation is trivial because, in the construction \eqq{nested} of the nested POVM, the operators $\hat{B}^{(2)}_{k_{(1,2)}}$ always act after the operator $\hat{B}^{(1)}_{k_{1}}$, so that the value of the former outside the support of the latter is completely irrelevant. In other words, completeness of the binary POVMs is not necessary for the definition of $\mathcal{N}^{(M)}$ as a proper POVM; it is sufficient to ask for \textit{weak completeness}, i.e., that $\mathcal{B}^{(u)}_{k_{(1,u-1)}}$ is complete on the support of the operator preceding it in the decomposition, $\hat{B}^{(u-1)}_{k_{(1,u-1)}}$. \\
Generalizing the previous discussion, at the $u$-th step we can define the elements of the $2^{u-1}$ possible POVMs $\mathcal{B}^{(u)}_{k_{(1,u-1)}}$ as 
   \begin{equation}\label{elements}
\hat{B}^{(u)}_{k_{(1,u)}}=\abs{\left(\sum_{k_{(u+1,u_{\tiny{F}})}}\hat{E}_{k_{(1,u)},k_{(u+1,u_{\tiny{F}})}}\right)^{1/2}\left(\hat{B}^{(1)}_{k_{1}}\right)^{-1/2}\cdots\left(\hat{B}^{(u-1)}_{k_{(1,u-1)}}\right)^{-1/2}}^{2}.
 \end{equation}
These elements are positive by construction and they satisfy the weak completeness relation $\hat{B}^{(u)}_{k_{(1,u-1),0}}+\hat{B}^{(u)}_{k_{(1,u-1),1}}=\hat{\mathbf{1}}_{k_{(1,u-1)}}$, which is sufficient to define the POVM $\mathcal{N}^{(M)}$. Hence we are left to show that, when combining the binary elements~\eqq{elements} as in~\eqq{nested}, the elements $\hat{F}_{k_{(1,u_{\tiny{F}})}}$ so constructed are equal to the original $\hat{E}_{k_{(1,u_{\tiny{F}})}}$. Indeed let us evaluate~\eqq{elements} for $u=u_{\tiny{F}}$, i.e., at the last step, noting that the sum contains only one term:
\begin{equation}\label{last}
 \hat{B}^{(u_{\tiny{F}})}_{k_{(1,u_{\tiny{F}})}}=\abs{\left(\hat{E}_{k_{(1,u_{\tiny{F}})}}\right)^{1/2}\left(\hat{B}^{(1)}_{k_{1}}\right)^{-1/2}\cdots\left(\hat{B}^{(u_{\tiny{F}}-1)}_{k_{(1,u_{\tiny{F}}-1)}}\right)^{-1/2}}^{2}.
\end{equation}
Let us then successively invert the outer square roots on the left-hand side of the equation exactly $u_{\tiny{F}}-1$ times, to obtain the relation 
\begin{equation} 
\hat{E}_{k_{(1,u_{\tiny{F}})}}=\left|\sqrt{\hat{B}^{(u_{\tiny{F}})}_{k_{(1,u_{\tiny{F}})}}} \cdots\sqrt{\hat{B}^{(1)}_{k_{1}}}\right|^{2}\equiv \hat{F}_{k_{(1,u_{\tiny{F}})}},
\end{equation} 
 which demonstrates that we can recover the initial POVM with the procedure outlined above. This completes the proof when $M$ is an exact power of $2$. \\
 If this is not the case, it means that $\log_{2}M$ is not an integer and it suffices to consider the nested decomposition for the next higher integer, i.e., set $u_{\tiny{F}}=\lceil\log_{2}M\rceil$, $\tilde{M}=2^{u_{\tiny{F}}}$. Let us then trivially expand the initial $M$-outcome POVM to a $\tilde{M}$-outcome one as  
\begin{equation} 
\mathcal{M}^{(\tilde{M})}=\mathcal{M}^{(M)}\cup\left\{\hat{E}_{k_{(1,u_{\tiny{F}})}}=0, \forall j^{(k)}>M-1\right\},
\end{equation} 
 by adding $\tilde{M}-M$ null elements. The nested decomposition $\mathcal{N}^{(\tilde{M})}$ equivalent to $\mathcal{M}^{(\tilde{M})}$ can be computed again by Eqs.~(\ref{nested},\ref{elements}) and it comprises $\tilde{M}-M$ null elements too. If we isolate these elements from the rest we obtain a decomposition \begin{equation}
  \mathcal{N}^{(\tilde{M})}=\mathcal{N}^{(M)}\cup\left\{\hat{F}_{k_{(1,u_{\tiny{F}})}}=0, \forall j^{(k)}>M-1\right\},
  \end{equation}  
  where $\mathcal{N}^{(M)}$ can be interpreted as a nested representation of the initial POVM $\mathcal{M}^{(M)}$.  
  \end{proof}

\section{Achieving the Holevo bound via a binary-tree-search decoding protocol}\label{sec:bisDec}
In this section we present a new decoding protocol to realize transmission of classical information on a quantum channel at asymptotically maximum capacity, achieving the Holevo bound, Theorem~\ref{hsw}, and thus the optimal communication rate, see Sec.~\ref{subsec:commQ}. At variance with previous proposals this scheme recovers the message bit by bit, making use of a nested series of ``yes-no'' measurements organized as in Sec.~\ref{sec:bisMeas} and thus determining which codeword was sent in $\log M$ steps, $M$ being the number of codewords, with an exponential advantage in the length of the codewords. We stress however that, being our individual binary measurements explicitly many-body operations, it is still not clear how such advantage could be translated in a decoding scheme which is efficient from the computational point of view, i.e., in terms of the number of quantum gates one has to apply to the received string of quantum information carriers. A similar problem arises also in the case of polar codes, see ~\cite{wildeHayden}, and it is caused by the lack of an explicit implementation of the ``{atomic}'' steps involved in the two protocols, i.e., the ``yes-no'' binary measurements for the present method and the Helstrom measurement for polar coding. Still, our method enlarges the number of decoding protocols that are known to be capacity-achieving, increasing the possibility of finding an implementation. In this respect it is also worth noticing that the proposed scheme exhibits the nontrivial advantage of gaining a bit of information at each step of the procedure, a feature which may be extremely appealing when dealing with faulty decoders, as it  allows  partial identification of the transmitted message even in the presence of subsequent detection failures. As in several previous works on the subject, in our derivation we heavily rely on the structure of typical projectors, Sec.~\ref{subsec:commQ}.

\subsection{Structure of the protocol}\label{bisProt}
 Let us now introduce our decoding protocol (Fig. \ref{figure1}) which, given a quantum codeword $\hat{\rho}_{\vecund{x}^{(\ell)}}$ extracted from a quantum code $\mathcal{C}$ of size $M=2^{NR}$ generated by sampling $N$-long sequences from the source $\mathcal{S}$, \eqq{qSource}, tries to identify it by using a binary-tree-search method. The measurement process comprises ${u}_{\mbox{\tiny{\tiny{F}}}}= NR$ nested detection events, each aimed to recover one bit of information from the transmitted signal. 
As a preliminary step, Bob assigns an ordering of the codewords $\mathcal{C}$ as in Sec.~\ref{sec:bisMeas}, identifying each of them with a unique string $\vecund{k}$ of ${u}_{\mbox{\tiny{\tiny{F}}}}$ bits, e.g., by providing a binary representation of their label so that $\hat{\rho}_{\vecund{x}^{(\ell)}}\equiv\hat{\rho}_{\vecund{x}^{(\vecund{k})}}\equiv\hat{\rho}_{\vecund{k}}\in\mathcal{C}$ for each $\ell=\ell^{(\vecund{k})}\in\{0,\dots,M-1\}$ and the probability of each codeword is
\begin{equation}\label{cwProb}
p_{\vecund{k}}=p_{\vecund{x}^{(\vecund{k})}}=\Pi_{n=1}^{N}p_{x_{n}}
\end{equation}
 In particular the first bit of the string $\vecund{k}$  identifies two distinct subsets of ${\mathcal{C}}$ containing each $M/2$ codewords: the subset 
    ${\mathcal{C}}^{(1)}_0$ formed by the codewords whose corresponding strings start with $k_1=0$, and the subset  ${\mathcal{C}}^{(1)}_{1}$ characterized by those for which instead $k_1=1$. The second bit of the string $\vecund{k}$  is then used to  further halve ${\mathcal{C}}^{(1)}_{0}$  and ${\mathcal{C}}^{(1)}_{1}$. Specifically 
for  $k_1=0,1$, ${\mathcal{C}}^{(1)}_{k_1}$ is split into  the sub-subsets ${\mathcal{C}}^{(2)}_{k_1,k_2=0}$ and  ${\mathcal{C}}^{(2)}_{k_1,k_2=1}$  which includes the $M/4$ codewords whose bits strings have $k_1$ as first bit and $k_2=0$ and $k_2=1$ as second bit, respectively.  Proceeding along the same line 
Bob  hence identifies a hierarchy  of subsets organized in ${u}_{\mbox{\tiny{\tiny{F}}}}$ sets, 
the $u$-th one being composed  by  $2^u$ disjoint subsets 
\begin{equation}
{\mathcal{C}}^{(u)}_{k_{(1,u)}}=\left\{\hat{\rho}_{k_{(1,u)},k_{(u+1,u_{\tiny{F}})}}\right\}_{k_{u+1},\cdots,k_{u_{\tiny{\tiny{F}}}}=0,1},
\end{equation} 
each containing $2^{{u}_{\mbox{\tiny{\tiny{F}}}}-u} = M/2^u$ codewords. By construction    for all $u\in\{ 1, \cdots, {u}_{\mbox{\tiny{\tiny{F}}}}\}$
they  fulfill the identities
\begin{align} 
&{\mathcal{C}}^{(u)}_{k_{(1,u-1)},0}\; \bigcap \;   {\mathcal{C}}^{(u)}_{k_{(1,u-1)},1} \; =\; \emptyset \;, \label{inter} \\
&{\mathcal{C}}^{(u)}_{k_{(1,u-1)},0}\; \bigcup \; {\mathcal{C}}^{(u)}_{k_{(1,u-1)},1}\; = \; {\mathcal{C}}^{(u-1)}_{k_{(1,u-1)}} \;, \label{complete1} 
\end{align} 
and the completeness relation 
\begin{equation} 
{\mathcal{C}} = \bigcup_{k_1, k_2, \cdots, k_{u_{\tiny{F}}} \in \{0,1\}} {\mathcal{C}}^{(u)}_{k_{(1,u_{\tiny{F}})}}\;.
\end{equation}

\begin{figure}[t!]
	\center
	\includegraphics[trim={0 0cm 0 0cm},clip,scale=.43]{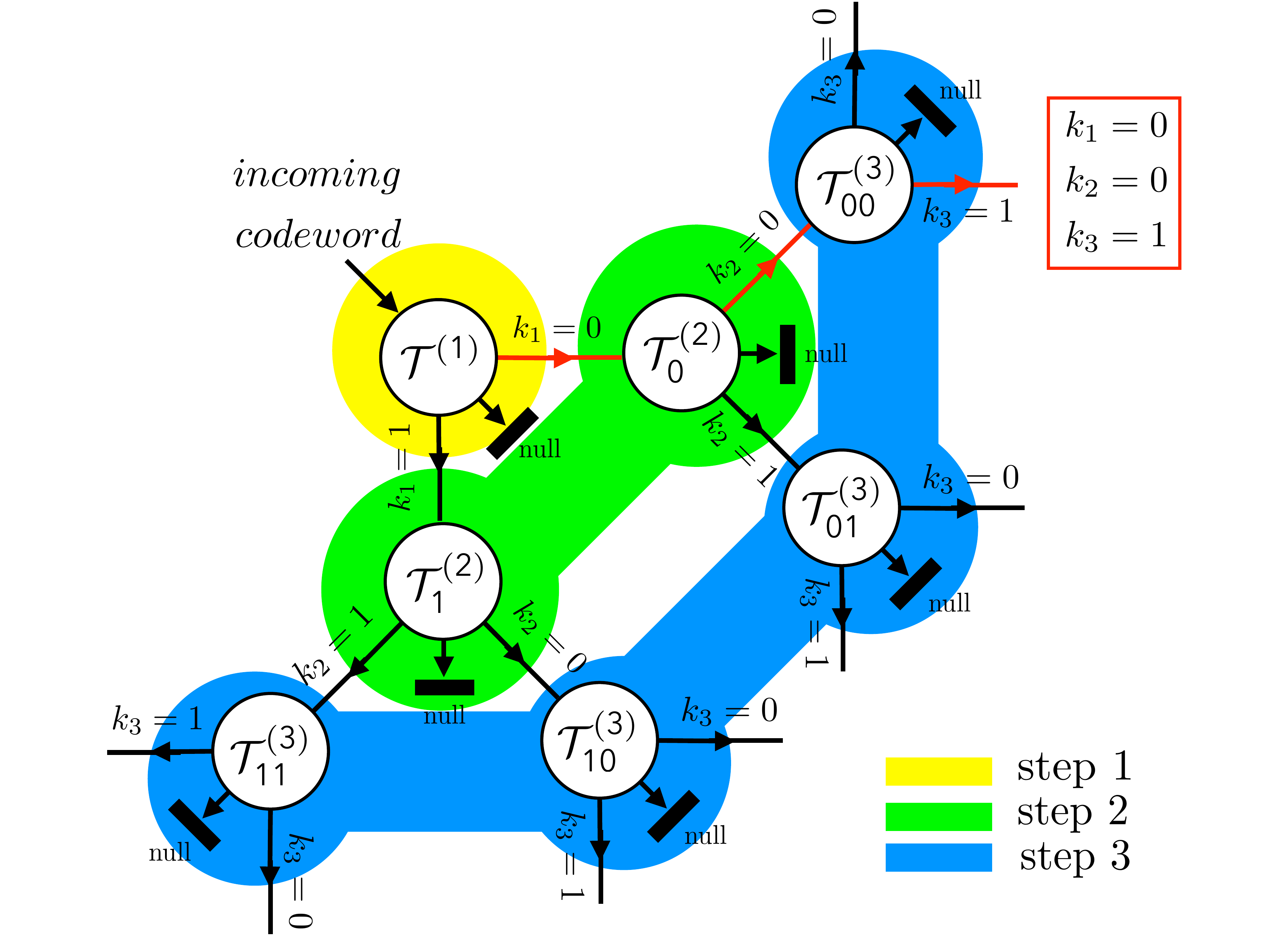}
	\caption{Schematic representation of the binary-tree-search decoding procedure. It consists in a sequence of adaptive measurements which are performed in series of 
	 ${u}_{\mbox{\tiny{\tiny{F}}}}$ concatenated steps, each being characterized by a POVM (the white circles) which admits three possible outcomes: two being associated respectively to 
	 the identification of the corresponding bit as  $0$ or $1$, and one, the {\it null} outcome, associated with the event where no decision can be made on the value of the bit.  
	The POVM to be performed at the $u$-th step depends upon the value of the bit obtained at the previous ones: for instance at the step number 2 Bob will perform
	either the POVM $\mathcal{T}_0^{(2)}$ or the POVM $\mathcal{T}_1^{(2)}$ depending on the value of $k_1$ he has obtained at the first step of the procedure, while at the step
	number $3$ Bob will perform  the POVMs $\mathcal{T}_{00}^{(3)}$, $\mathcal{T}_{01}^{(2)}$, $\mathcal{T}_{10}^{(3)}$, or $\mathcal{T}_{11}^{(2)}$ depending on the values of $k_1$ and $k_2$ obtained in the previous two steps.   The figure refers to the case of  ${u}_{\mbox{\tiny{\tiny{F}}}}=3$, the redline representing the trajectory which yields Bob to assign the binary string $\vecund{k}=(0,0,1)$ to the received codeword.   }
	\label{figure1}
\end{figure}
Note that, as in Sec.~\ref{sec:commQ}, we overload the notation $\mathcal{C}^{(u)}_{k_{(1,u-1)},0}$ to indicate both a set of classical indices $\vecund{k}$ and their quantum counterparts $\hat{\rho}_{\vecund{k}}$.

To recover which codeword Alice is transmitting,  Bob performs a nested sequence of ${u}_{\mbox{\tiny{\tiny{F}}}}$ ternary-outcome measurements organized as shown in Fig.~\ref{figure1}. 
The first of these measures is aimed to determine the value of the first bit $k_1$ of the bit string associated with the transmitted codeword, i.e., it allows Bob to determine whether the codeword is in the subset  ${\mathcal{C}}^{(1)}_0$ or in the subset  ${\mathcal{C}}^{(1)}_1$. The exact form of the measurement will be assigned in the following sections where three alternative examples of the scheme will be discussed in details: for the moment it is sufficient to observe that it can be described, in the same spirit of Sec.~\ref{sec:bisMeas}, by a POVM $\mathcal{T}^{(1)}$  
 of  elements $\hat{T}_{0}^{(1)}$, $\hat{T}_{1}^{(1)}$ associated respectively to the outcomes $k_1=0$ and $k_1=1$, plus an error term $\hat{T}_{err}^{(1)} = \hat{\mathbf{1}} -\hat{T}_{0}^{(1)} -\hat{T}_{1}^{(1)}$ associated with the case in which no decision can be made on the value of $k_1$: if this event occurs Bob simply declares failure of the decoding procedure and stops the protocol (in the first implementation of the scheme we discuss  in Sec.~\ref{EX0}  this element is not present, which is equivalent to setting $\hat{T}_{err}^{(1)}=0$ and establishes exactly the binary-tree structure of Sec.~\ref{sec:bisMeas}). 
 Once $k_1$ has been determined, Bob proceeds with the second step of the protocol aimed
 to recover the value of the bit $k_2$ of the transmitted codeword. To this purpose, conditioned on the value of $k_1=0,1$ obtained in the previous step, Bob performs now a new POVM $\mathcal{T}_{k_1}^{(2)}$  aimed to determine whether the received codeword belongs to ${\mathcal{C}}^{(2)}_{k_1,0}$ or to ${\mathcal{C}}^{(2)}_{k_1,1}$.
Also  $\mathcal{T}_{k_1}^{(2)}$ is  characterized by three elements: $\hat{T}_{k_1,0}^{(2)}$, $\hat{T}_{k_1,1}^{(2)}$ corresponding to the cases $k_2=0$ and $k_2=1$ respectively, and  $\hat{T}_{k_1,err}^{(2)} = \hat{\mathbf{1}} -\hat{T}_{k_1,0}^{(2)} -\hat{T}_{k_1,1}^{(2)}$ corresponding to the failure event. 
  The procedure iterates  until Bob either gets an error event or recovers all the ${u}_{\mbox{\tiny{\tiny{F}}}}$ bits which identify the transmitted codeword.   Specifically, assuming that no failures have occurred in the first $u-1$ steps yielding the values $k_1$, $k_2$, $\cdots$, $k_{u-1}$ for the associated bits,  at the $u$-th step Bob 
  performs on the system a POVM 
  \begin{equation}\label{tOutcome}
  \mathcal{T}_{k_{(1,u-1)}}^{(u)}=\left\{\hat{T}_{k_{(1,u-1)}, 0}^{(u)},\hat{T}_{k_{(1,u-1)}, 1}^{(u)},\hat{T}_{k_{(1,u-1)}, err}^{(u)}\right\}
  \end{equation}
   to decide whether the received codeword belongs
  to the set  ${\mathcal{C}}_{k_{(1,u-1)}, 0}^{(u)}$ or  ${\mathcal{C}}_{k_{(1,u-1)}, 1}^{(u)}$. \\
The construction above gives rise to a nested POVM similar to \eqq{nested}, though with an additional error element:
\begin{equation}\label{nestedT}
\begin{aligned}
\tilde{\mathcal{N}}^{(M)}&=\left\{\mathcal{T}^{(u_{\tiny{F}})}_{k_{(1,u_{\tiny{F}}-1)}}\right\}_{k_{1},\cdots,k_{u_{\tiny{F}}-1}=0,1}\circ\cdots\circ\left\{\mathcal{T}^{(1)}\right\}\\
&=\left\{\hat{G}_{\vecund{k}}=\left|\sqrt{\hat{T}^{(u_{\tiny{F}})}_{k_{(1,u_{\tiny{F}})}}} \cdots\sqrt{\hat{T}^{(1)}_{k_{1}}}\right|^{2}\right\}_{\vecund{k}\in\mathcal{C}}\bigcup\left\{\hat{G}_{err}=\hat{\mathbf{1}}-\sum_{\vecund{k}}\hat{G}_{\vecund{k}\in\mathcal{C}}\right\}.
\end{aligned}
\end{equation}
The probability of correctly recovering a given string of bits $\vecund{k}$ when measuring an input state $\hat{\rho}_{\vecund{k}'}\in \mathcal{C}$ can now be expressed as in \eqq{outcomeProb}:   
\begin{equation}  \label{fff111} 
  p_{\vecund{k}| \vecund{k}'} = \tr{\hat{G}_{\vecund{k}}\hat{\rho}_{\vecund{k}'}}.
  \end{equation} 
   Hence the success probability of the procedure, complementary to the error probability \eqq{pErrGen}, is 
  \begin{equation}    \label{psucc}
   P_{succ}(\mathcal{C})=\oneover{M}\sum_{\vecund{k}\in\mathcal{C}}  p_{\vecund{k}| \vecund{k}}=\oneover{M}\sum_{\vecund{k}\in\mathcal{C}}\tr{\hat{G}_{\vecund{k}}\hat{\rho}_{\vecund{k}}}. 
  \end{equation}

 \subsection{Achievability proof}  \label{succprobA}
 In this subsection we are going to show that, as long as the rate is $R\leq\max_{\mathcal{S}}\chi(\mathcal{S})$ as detailed in Theorem~\ref{hsw}, 
 it is possible to assign the ternary-outcome POVMs, \eqq{tOutcome} in such a way that, in the limit of large $N$, the success probability \eqq{psucc} of the nested decoding protocol \eqq{nestedT}, converges asymptotically to 1 when averaged over all possible codes generated by the output source $\mathcal{S}$, i.e.,
   \begin{equation}\label{dopolemma12323}
 \lim_{N \rightarrow \infty} \ave{P_{succ}(\mathcal{C})}_{\mathcal{S}}=1,
    \end{equation}
    where we recall that the average $\ave{\cdot}$ is with respect to the probability distribution of codes, \eqq{randomCCoding}.
   Accordingly the corresponding average error probability asymptotically nullifies, 
   proving hence that binary-tree-search decoding procedures can be used to  saturate the Holevo bound. 
In order to achieve this goal we start by presenting a sufficient condition  on $\hat{T}^{(u)}_{k_{(1,u)}}$ which, if fulfilled, would yield the limit \eqq{dopolemma12323} independently of the value of $R$, see Theorem~\ref{MainTh}. Subsequently we show that  for all
rates $R$  respecting the Holevo Bound, Theorem~\ref{hsw}, we can indeed fulfill such sufficient condition. This is done by presenting three independent choices of  the operators $\hat{T}^{(u)}_{k_{(1,u)}}$, corresponding to three different ways of constructing the binary-tree-search scheme: via orthogonal projections (see Sec.~\ref{EX0}); via SRM detections (see Sec.~\ref{EX00}); and via SM detections (see Sec.~\ref{EX1}). 

\begin{theorem}\label{MainTh}(Sufficient Capacity-Achieving Condition for Nested Decoding) 
For  $N$ integer let $\mathcal{C}$ be a quantum code formed by $M=2^{NR}$ quantum codewords of length $N$ sampled from the source $\mathcal{S}$, \eqq{qSource}, 
and a binary-tree-search protocol defined by the nested POVM $\tilde{\mathcal{N}}^{(M)}$, \eqq{nestedT} characterized by the
ternary-outcome POVMs $\mathcal{T}^{(u)}_{k_{(1,u-1)}}$, \eqq{tOutcome}.
The corresponding success probability, \eqq{psucc}, converges to one as in \eqq{dopolemma12323} on average over all codes according to \eqq{randomCCoding} if, for all $k_{u}=0,1$ and $u\in \{ 1, \cdots, u_F\}$,  one of the following conditions is fulfilled
\item{i)} 
\begin{align}
\ave{\operatorname{Tr}\left[\hat{T}^{(u)}_{k_{(1,u)}}\hat{\rho}_{k_{(1,u)},k_{(u+1,u_{\tiny{F}})}}\right]}_{\mathcal{S}}\geq 1-\epsilon(M);\label{IMPOIMPO}
\end{align}
\item{ii)} 
\begin{align}
\ave{\operatorname{Tr}\left[\hat{T}^{(u)}_{k_{(1,u)}}\hat{\Pi}\hat{\rho}_{k_{(1,u)},k_{(u+1,u_{\tiny{F}})}}\hat{\Pi}\right]}_{\mathcal{S}}\geq 1-\epsilon(M),\label{IMPOIMPO11}
\end{align}
with $\epsilon(M)>0$ being a function that decreases asymptotically to zero faster than $N^{-2}$ as $N\rightarrow\infty$ and $\hat{\Pi}$ being the projector on the typical subspace of the average state of the source, \eqq{typProj}.
\end{theorem} 

\begin{proof} 
We start by directly proving  that Eq.~(\ref{IMPOIMPO}) is a sufficient condition for Eq.~(\ref{dopolemma12323}), part i) of the theorem.  Part ii) of the theorem is then obtained by showing that 
Eq.~(\ref{IMPOIMPO11}) implies Eq.~(\ref{IMPOIMPO}). 

{\it Part i)}: The success probability that an element $\hat{\rho}_{\vecund{k}} \in \mathcal{C}$ will be correctly decoded by the nested POVM can be computed as in \eqq{fff111} by setting $\vecund{k}'=\vecund{k}$.
  To put a bound on the average value of this quantity over the possible codes, we 
 observe that
  \begin{align}
   \ave{p_{\vecund{k}|\vecund{k}}}_{\mathcal{S}}
     &= 
    \ave{\tr{    {\hat{T}_{{u}_{\mbox{\tiny{\tiny{F}}}}}}^{2}  \; \hat{T}_{{u}_{\mbox{\tiny{\tiny{F}}}}-1}
  \; \cdots  \; \hat{T}_{{1}}\; 
    \hat{\rho}_{\vecund{k}} \; \hat{T}_{{1}}\; 
 \cdots\;\hat{T}_{{u}_{\mbox{\tiny{\tiny{F}}}}-1}
    }}_{\cal S}  \nonumber \\ 
    &\geq\ave{\tr{   \hat{T}^2_{{u}_{\mbox{\tiny{\tiny{F}}}}}    \hat{\rho}_{\vecund{k}} }-2\trd{\hat{T}_{{u}_{\mbox{\tiny{\tiny{F}}}}-1}\ldots \hat{T}_1\hat{\rho}_{\vecund{k} }\hat{T}_1\ldots 
\hat{T}_{{u}_{\mbox{\tiny{\tiny{F}}}}-1}}{\hat{\rho}_{\vecund{k}}}}_{\cal S}, \label{EQU1}
    \end{align}
    where we introduced the compact notation $\hat{T}_u=\sqrt{\hat{T}^{(u)}_{k_{(1,u)}}}$ and applied Lemma~\ref{appclose} with 
  $\hat{E}=\hat{T}^2_{{u}_{\mbox{\tiny{\tiny{F}}}}} $,
 $\hat{\rho}=  \hat{T}_{{u}_{\mbox{\tiny{\tiny{F}}}}-1}
  \; \cdots  \; \hat{T}_{{1}}\; 
    \hat{\rho}_{\vecund{k}} \; \hat{T}_{{1}}\; 
 \cdots\;\hat{T}_{{u}_{\mbox{\tiny{\tiny{F}}}}-1}$, 
    and $\hat{\sigma}= \hat{\rho}_{\vecund{k}}$.
        By use of the triangular inequality, \eqq{triangTD}, we also observe that 
    \begin{eqnarray}
  \label{ufmenouno}   
  && \trd{\hat{T}_{{u}_{\mbox{\tiny{\tiny{F}}}}-1}\ldots \hat{T}_1\hat{\rho}_{\vecund{k} }\hat{T}_1\ldots 
\hat{T}_{{u}_{\mbox{\tiny{\tiny{F}}}}-1}}{\hat{\rho}_{\vecund{k} }}
\nonumber\\ 
&& \qquad  \leq \trd{\hat{T}_{{u}_{\mbox{\tiny{\tiny{F}}}}-1}\ldots \hat{T}_1\hat{\rho}_{\vecund{k} }\hat{T}_1\ldots 
\hat{T}_{{u}_{\mbox{\tiny{\tiny{F}}}}-1}}{\hat{T}_{{u}_{\mbox{\tiny{\tiny{F}}}}-1}\hat{\rho}_{\vecund{k} }\hat{T}_{{u}_{\mbox{\tiny{\tiny{F}}}}-1}} + \trd{\hat{T}_{{u}_{\mbox{\tiny{\tiny{F}}}}-1}\hat{\rho}_{\vecund{k} }\hat{T}_{{u}_{\mbox{\tiny{\tiny{F}}}}-1}}{\hat{\rho}_{\vecund{k} }}\nonumber \\
&&\qquad  \leq \trd{\hat{T}_{{u}_{\mbox{\tiny{\tiny{F}}}}-2}\ldots \hat{T}_1\hat{\rho}_{\vecund{k} }\hat{T}_1\ldots 
\hat{T}_{{u}_{\mbox{\tiny{\tiny{F}}}}-2}}{\hat{T}_{{u}_{\mbox{\tiny{\tiny{F}}}}-2}\hat{\rho}_{\vecund{k} }\hat{T}_{{u}_{\mbox{\tiny{\tiny{F}}}}-2}} + \trd{\hat{T}_{{u}_{\mbox{\tiny{\tiny{F}}}}-1}\hat{\rho}_{\vecund{k} }\hat{T}_{{u}_{\mbox{\tiny{\tiny{F}}}}-2}}{\hat{\rho}_{\vecund{k} }} \nonumber \\
&&\qquad  \leq 
\sum_{u=1}^{{u}_{\mbox{\tiny{\tiny{F}}}}-1}\trd{\hat{T}_u\hat{\rho}_{\vecund{k}}\hat{T}_u}{\hat{\rho}_{\vecund{k}}},
    \end{eqnarray}
 where the second inequality follows from Lemma \ref{contra} while the third one by direct iteration of the previous passages.  
Replaced into Eq.~(\ref{EQU1})   this finally yields 
    \begin{align}\label{dopolemma}
  \ave{p_{\vecund{k}|\vecund{k}}}_{\mathcal{S}}
\geq\ave{\tr{   \hat{T}^2_{{u}_{\mbox{\tiny{\tiny{F}}}}}    \hat{\rho}_{\vecund{k}} }}_{\cal S}
 -2 \sum_{u=1}^{{u}_{\mbox{\tiny{\tiny{F}}}}-1}\ave{\trd{\hat{T}_u\hat{\rho}_{\vecund{k}}\hat{T}_u}{\hat{\rho}_{\vecund{k}}}}_{\cal S} .
    \end{align}
Assume now that Eq.~(\ref{IMPOIMPO})  holds. Accordingly, for all $k_{u}=0,1$ and $u\in\{ 1,\cdots, {u}_{\mbox{\tiny{\tiny{F}}}}\}$ we have 
\begin{equation} 
\ave{\tr{\hat{T}_u^2{\hat{\rho}}_{\vecund{k}}}}_{\cal S}\geq 1- \epsilon(M), 
\end{equation}
with $\epsilon(M)$ a positive function which goes to zero faster than $N^{-2}$. 
Then thanks to Lemma~\ref{gentop} we can write 
  \begin{align}\label{dopolemma1}
 \ave{p_{\vecund{k}|\vecund{k}}}_{\mathcal{S}}\geq 1- \epsilon(M) -2 n R \sqrt{\epsilon(M)} ,
    \end{align}
which forces  $\ave{p_{\vecund{k}|\vecund{k}}}_{\mathcal{S}}$  to converge to 1 as $N\rightarrow \infty$. This shows that   Eq.~(\ref{IMPOIMPO}) is indeed a  sufficient condition for \eqq{dopolemma12323}. 

{\it Part ii):} To prove  that  Eq.~(\ref{IMPOIMPO11}) is a sufficient condition for \eqq{dopolemma12323} we invoke 
 Lemma \ref{appclose} with $\hat{E}=\hat{T}_{u}^{2}$, $\hat{\rho}=\hat{\rho}_{\vecund{k}}$ 
  and $\hat{\sigma}=\hat{\Pi}\hat{\rho}_{\vecund{k}}\hat{\Pi}$ obtaining the following inequality 
  \begin{equation} \label{misuraSmoothed}
 \ave{\tr{\hat{T}_u^2\hat{\rho}_{\vecund{k}}}}_{\cal S}\geq\ave{\tr{\hat{T}_u^2 \hat{\Pi} \hat{\rho}_{\vecund{k}}\hat{\Pi}} -{2D\left(\hat{\rho}_{\vecund{k}},\hat{\Pi}\hat{\rho}_{\vecund{k}}\hat{\Pi}\right)}}_{\cal S}.
  \end{equation}
From \eqq{totaleuno} we also know that 
        for $n$ sufficiently large and $\epsilon_1=O(e^{-N})$ one has 
     \begin{equation}
       \ave{\tr{\hat{\Pi}\hat{\rho}_{\vecund{k} }}}_{\cal S}=\tr{\hat{\Pi}\ave{\hat{\rho}_{\vecund{k} }}_{\cal S}}=\tr{\hat{\Pi}\hat{\rho}^{\otimes N}}\geq 1-\epsilon_1,\label{genttotaluno}
     \end{equation}
     where we used the fact that the average over the codes of $\hat{\rho}_{\vecund{k}}$ corresponds to the average over all possible $N$-long codewords $\vecund{x}\in\mathcal{X}^{N}$ with respect to their joint probability distribution, \eqq{cwProb}, i.e., 
     \begin{equation}\label{aveTrans}
       \ave{\hat{\rho}_{\vecund{k}}}_{\cal S} =  \sum_{\vecund{x}\in\mathcal{X}^{N}}p_{\vecund{x}}\hat{\rho}_{\vecund{x}}=\hat{\rho}^{\otimes 
       N}.
     \end{equation}
     Accordingly, via Lemma \ref{gentop} we can conclude that 
\begin{equation} 
 \ave{D\left(\hat{\rho}_{\vecund{k}},\hat{\Pi}\hat{\rho}_{\vecund{k}}\hat{\Pi}\right)}_{\cal S}\leq\sqrt{\epsilon_1},\label{genttotaldue}
     \end{equation}
which inserted in  Eq.~(\ref{misuraSmoothed}) finally yields  
  \begin{equation}
    \ave{\tr{\hat{T}_u^2\hat{\rho}_{\vecund{k}}}}_{\cal S}\geq\ave{\tr{\hat{T}_u^2 \hat{\Pi} \hat{\rho}_{\vecund{k}}\hat{\Pi}}}_{\cal S} -2\sqrt{\epsilon_1}.
   \label{smoonosmoo}
  \end{equation}
The latter inequality implies \eqq{IMPOIMPO} thanks to \eqq{IMPOIMPO11} and hence also Eq.~(\ref{dopolemma12323}), via part i) of the theorem.
\end{proof}

 Thanks to Theorem~\ref{MainTh} we can now prove that specific binary-tree-search protocols allow one to attain the Holevo bound, by 
 showing that, for all rates $R$ respecting the Holevo bound, it is possible to identify a nested POVM $\tilde{\mathcal{N}}^{(M)}$ fulfilling Eq.~(\ref{IMPOIMPO11}). 
 Ideally one way of building the POVMs  $\mathcal{T}_{k_{(1,u-1)}}^{(u)}$ which define the binary-tree-search decoding procedure, would be to
  identify its elements $\hat{T}_{k_{(1,u-1)}, 0}^{(u)}$, $\hat{T}_{k_{(1,u-1)}, 1}^{(u)}$ with  the projectors on the subspaces spanned by the codewords of the sets $\mathcal{C}^{(u)}_{k_{(1,u-1)}, 0}$ and $\mathcal{C}^{(u)}_{k_{(1,u-1)}, 1}$ respectively. However this is not possible due to the fact that such spaces are in general not orthogonal, although we expect conditionally-typical spaces of different codewords of the source 
  to be disjoint in the long $N$ limit: some kind of regularization is hence necessary. 
 In the following we shall present three alternative, yet asymptotically equivalent,  ways to realize this:  
the first one makes use of orthogonal projections on subspaces identified by treating  asymmetrically the set $\mathcal{C}^{(u)}_{k_{(1,u-1)}, 0}$ and $\mathcal{C}^{(u)}_{k_{(1,u-1)}, 1}$, 
the second one is based on the SRM construction, \eqq{srm}, and finally the third makes use of the SM, Eqs.~(\ref{seqSucEle},\ref{seqErrEle}). 

 \subsubsection{Method  1: orthogonal projections}\label{EX0}
Consider the set $\mathcal{C}^{(u)}_{k_{(1,u-1)}, 0}$. To each of its codewords $\hat{\rho}_{\vecund{k}}$  we can associate a conditionally-typical space, \eqq{ctypSpace}, 
\begin{equation}
\hil_{\delta,N}\left(\mathcal{J}|\vecund{x}^{(\vecund{k})}\right)\equiv\hil_{\delta,N}\left(\mathcal{J}|\vecund{k}\right),
\end{equation}  
and a corresponding projector, \eqq{ctypProj}, 
\begin{equation}
\hat{\Pi}_{\vecund{x}^{(\vecund{k})}}\equiv\hat{\Pi}_{\vecund{k}},
\end{equation}
 along the lines detailed in Sec.~\ref{subsec:commQ}. Next we construct the space ${\cal H}^{(u)}_{k_{(1,u-1)}, 0}$ spanned by the vectors  which can be written as a direct sum of the elements of the conditionally-typical spaces of all the codewords in $\mathcal{C}^{(u)}_{k_{(1,u-1)}, 0}$, i.e., 
\begin{equation} 
{\cal H}^{(u)}_{k_{(1,u-1)}, 0} = \bigoplus_{\vecund{k} \in \mathcal{C}^{(u)}_{k_{(1,u-1)}, 0}} \hil_{\delta,N}\left(\mathcal{J}|\vecund{k}\right).
\end{equation}  
 By construction it follows that each of the conditionally-typical spaces associated to  $\mathcal{C}^{(u)}_{k_{(1,u-1)}, 0}$ are proper
 subspaces of ${\cal H}^{(u)}_{k_{(1,u-1)}, 0}$. Accordingly, indicating with $\hat{\Pi}^{(u)}_{k_{(1,u-1)}, 0}$ the projector on ${\cal H}^{(u)}_{k_{(1,u-1)}, 0}$
 we have that 
 \begin{equation} \label{NEWIMPO} 
 \hat{\Pi}^{(u)}_{k_{(1,u-1)}, 0} \geq \hat{\Pi}_{\vecund{k}},
 \end{equation} 
for all $\vecund{k} \in \mathcal{C}^{(u)}_{k_{(1,u-1)}, 0}$. Also, due to the partial overlapping among the $\hil_{\delta,N}\left(\mathcal{J}|\vecund{k}\right)$ of $\mathcal{C}^{(u)}_{k_{(1,u-1)}, 0}$
the sum of  the associated $\hat{\Pi}_{\vecund{k}}$  will in general be larger than $\hat{\Pi}^{(u)}_{k_{(1,u-1)}, 0}$, i.e., 
\begin{equation} 
\hat{\Pi}^{(u)}_{k_{(1,u-1)}, 0} \leq  \sum_{\vecund{k} \in \mathcal{C}^{(u)}_{k_{(1,u-1)}, 0}}  \hat{\Pi}_{\vecund{k}}. \label{NEWIMPO1} 
 \end{equation} 
Let us then define the orthogonal-projections method for the binary-tree-search POVM by identifying $\hat{T}_{k_{(1,u-1)}, 0}^{(u)}$ with $\hat{\Pi}^{(u)}_{k_{(1,u-1)}, 0}$ and 
$\hat{T}_{k_{(1,u-1)}, 1}^{(u)}$ with its complementary counterpart, i.e.
\begin{align}
\hat{T}_{k_{(1,u-1)}, 0}^{(u)} &= \hat{\Pi}^{(u)}_{k_{(1,u-1)}, 0} ,\label{orthOne} \\
\hat{T}_{k_{(1,u-1)}, 1}^{(u)} &=  \hat{\Xi}^{(u)}_{k_{(1,u-1)}, 0} = \hat{\mathbf 1} - \hat{\Pi}^{(u)}_{k_{(1,u-1)}, 0}. \label{orthTwo}
\end{align} 
A couple of remarks are mandatory:
\begin{enumerate}
\item Notice that  $\hat{T}_{k_{(1,u-1)}, 1}^{(u)}$ does  not coincide with the projector $\hat{\Pi}^{(u)}_{k_{(1,u-1)}, 1}$ on the subspace ${\cal H}^{(u)}_{k_{(1,u-1)}, 1}$
formed by the direct sum of the typical subspaces   $\hil_{\delta,N}\left(\mathcal{J}|\vecund{k}\right)$ associated with $\mathcal{C}^{(u)}_{k_{(1,u-1)}, 1}$.
Notice also that, due to the partial overlapping of the typical subspaces of different codewords, in general we can neither  establish
 an inequality similar to Eq.~(\ref{NEWIMPO}) which links  $\hat{T}_{k_{(1,u-1)}, 1}^{(u)}$  and the $ \hat{\Pi}_{\vecund{k}}$ of ${\cal H}^{(u)}_{k_{(1,u-1)}, 1}$,
nor fix an ordering between $\hat{T}_{k_{(1,u-1)}, 1}^{(u)}$ and 
$\hat{\Pi}^{(u)}_{k_{(1,u-1)}, 1}$;
\item By construction the scheme we are analyzing here does not include the possibility of the error event described in Sec.~\ref{bisProt}, so that the nested POVM \eqq{nestedT} coincides with the general one \eqq{nested} analyzed in Sec.~\ref{sec:bisMeas}. Indeed in this case we have $\hat{T}_{k_{(1,u-1)}, err}^{(u)}=0$, thus the ternary-outcome POVM is a binary-outcome one $\mathcal{T}_{k_{(1,u-1)}}^{(u)}\equiv\mathcal{B}_{k_{(1,u-1)}}^{(u)}$ and it is also projective. 
\end{enumerate}

From Theorem~\ref{MainTh} 
 the asymptotic attainability of the Holevo bound with this procedure  can be established by 
   showing that  Eq.~(\ref{IMPOIMPO11}) holds, i.e.,  
\begin{lemma}
For all rates $R$ satisfying the Holevo bound, Theorem~\ref{hsw},
the binary-tree-search nested POVM, \eqq{nested}, with binary-outcome elements $\hat{B}_{k_{(1,u)}}^{(u)}$
defined as in Eqs.~(\ref{orthOne},\ref{orthTwo}) fulfills the sufficient condition Eq.~(\ref{IMPOIMPO11}).
\end{lemma}  
\begin{proof}
   Consider first the case with $k_u=0$. Given then a generic codeword $\hat{\rho}_{\vecund{k}}\in\mathcal{C}^{(u)}_{k_{(1,u-1)}, 0}$
 we can write 
 \begin{equation}\label{IMPOIMPO111}
 \begin{aligned}  
 \ave{\tr{\hat{T}_{k_{(1,u-1)}, 0}^{(u)} \hat{\Pi}\hat{\rho}_{\vecund{k}}\hat{\Pi}}}_{\cal S}&=\ave{\tr{\hat{\Pi}_{k_{(1,u-1)}, 0}^{(u)}\hat{\Pi}\hat{\rho}_{\vecund{k}}\hat{\Pi}}}_{\cal S}\\
 &\geq \ave{\tr{\hat{\Pi}_{\vecund{k}}  \hat{\Pi}\hat{\rho}_{\vecund{k}}\hat{\Pi}}}_{\cal S} = \sum_{\vecund{x}\in\mathcal{X}^{N}} p_{\vecund{x}}  \tr{\hat{\Pi}_{\vecund{x}}\hat{\Pi}{\hat{\rho}}_{\vecund{x}}\hat{\Pi}}\\
 &\geq 
 \sum_{\vecund{x}\in\mathcal{X}^{N}} p_{\vecund{x}} \;\Big(  {\tr{\hat{\Pi}_{\vecund{x} }\hat{\rho}_{\vecund{x}}}} -2 {\trd{\hat{\Pi}{\hat{\rho}}_{\vecund{x}}\hat{\Pi}}{\hat{\rho}_{\vecund{x}}}}\Big) 
      \geq 1 - \epsilon_2 - 2\sqrt{\epsilon_1},
\end{aligned}
\end{equation}
where the first inequality follows from \eqq{NEWIMPO},  the second equality from the fact that the average over the codes corresponds to the average over all possible $N$-long codewords $\vecund{x}\in\mathcal{X}^{N}$ with respect to their joint probability distribution, \eqq{cwProb}, i.e., 
\begin{equation} \label{prima} 
\ave{\hat{\Pi}_{\vecund{k}}  \hat{\Pi}\hat{\rho}_{\vecund{k}}\hat{\Pi}}_{\cal S} = \sum_{\vecund{x}\in\mathcal{X}^{N}} p_{\vecund{x}} \hat{\Pi}_{\vecund{x}}  \hat{\Pi}\hat{\rho}_{\vecund{x}}\hat{\Pi},
\end{equation} 
the second inequality from applying Lemma \ref{appclose} with 
    $\hat{E}=\hat{\Pi}_{\vecund{x}}$, $\hat{\rho}=\hat{\Pi}\hat{\rho}_{\vecund{x}}\hat{\Pi}$ and $\hat{\sigma}=\hat{\rho}_{\vecund{x}}$, while  
the third inequality  follows  both from the high probability of projecting 
    codeword $\hat{\rho}_{\vecund{x}}$ on its conditionally-typical space \eqq{parzialeuno}
    and from the same concept for the average codeword, together with Lemma \ref{gentop} (as in Eqs.~(\ref{genttotaluno},\ref{genttotaldue})),
    the parameters $\epsilon_1$ and $\epsilon_2$ being both exponentially small in $N$. 
   Hence Eq.~(\ref{IMPOIMPO111}) proves that Eq.~(\ref{IMPOIMPO11}) applies at least for the sets  $\mathcal{C}^{(u)}_{k_{(1,u)}}$  with $k_u=0$. \\
  Take next   $k_u=1$ and  a generic codeword $\hat{\rho}_{\vecund{k}}\in\mathcal{C}^{(u)}_{k_{(1,u-1)}, 1}$. In this case we have 
 \begin{equation}  \label{IMPOIMPO112pre}
 \begin{aligned}
  \ave{\tr{\hat{T}_{k_{(1,u-1)}, 1}^{(u)}\;  \hat{\Pi}\hat{\rho}_{\vecund{k}}\hat{\Pi}}}_{\cal S}&= \ave{\tr{ \hat{\Pi}\hat{\rho}_{\vecund{k}}\hat{\Pi}}}_{\cal S}  -\ave{\tr{\hat{\Pi}_{k_{(1,u-1)}, 0}^{(u)}\;  \hat{\Pi}\hat{\rho}_{\vecund{k}}\hat{\Pi}}}_{\cal S} \\
&\geq \ave{\tr{ \hat{\Pi}\hat{\rho}_{\vecund{k}}\hat{\Pi}}}_{\cal S}  - \sum_{\vecund{k}'\neq \vecund{k}}    \ave{\tr{\hat{\Pi}_{\vecund{k}'}\;  \hat{\Pi}\hat{\rho}_{\vecund{k}}\hat{\Pi}}}_{\cal S}
\end{aligned}
\end{equation}
where the inequality follows from Eq.~(\ref{NEWIMPO1}) by  adding all the remaining  terms $\hat{\Pi}_{\vecund{k}'}$ associated with codewords different from $\hat{\rho}_{\vecund{k}}$.
Observe then that Eqs.~(\ref{totaleuno},\ref{aveTrans}) imply
\begin{equation} \label{prima1} 
\ave{\tr{ \hat{\Pi}\hat{\rho}_{\vecund{k}}\hat{\Pi}}}_{\cal S} = \tr{\hat{\Pi} \hat{\rho}^{\otimes N} } \geq 1 -\epsilon_1,
\end{equation} 
Furthermore  for each term of the sum on the RHS of Eq.~(\ref{IMPOIMPO112pre}) we have
\begin{equation}\label{pdiversoj1}
    \begin{aligned} 
    \ave{\tr{  \hat{\Pi}_{\vecund{k}'}  \hat{\Pi}\hat{\rho}_{\vecund{k}}\hat{\Pi} }}_{\cal S}&= \sum_{\vecund{x}, \vecund{x}' \in\mathcal{X}^{N}} 
       p_{\vecund{x}}\;  p_{\vecund{x}'} \;{\tr{ \hat{\Pi}_{\vecund{x}'} \hat{\Pi}\hat{\rho}_{\vecund{x}}\hat{\Pi} }} = \sum_{\vecund{x}'\in\mathcal{X}^{N}} p_{\vecund{x}'}\; \tr{ 
      \hat{\Pi}_{\vecund{x}'} \hat{\Pi}\hat{\rho}^{\otimes N}\hat{\Pi} }\\
      &\leq \norm{\hat{\Pi}\hat{\rho}^{\otimes N}\hat{\Pi}}_\infty 
     \sum_{\vecund{x}'\in\mathcal{X}^{N}} p_{\vecund{x}'} \; {\tr{\hat{\Pi}_{\vecund{x}'}}} \leq 2^{-N\left(S(\rho)-\delta\right)} \; 2^{N\left(\sum_{x\in\mathcal{X}}p_{x}S(\rho_x)+\delta\right)}\\
     &=2^{-N\left(\chi(\mathcal{S})-2\delta\right)},
    \end{aligned}
    \end{equation}
    where the second inequality follows from typical subspaces' properties (\ref{totaletre},\ref{parzialedue}), $\norm{\cdot}_{\infty}$ is the infinite-norm of an operator, equal to its maximum eigenvalue, and $\chi(\mathcal{S})$ is the Holevo information \eqq{holevoInfo} of the source ${\cal S}$. 
  Replacing Eq.~(\ref{prima1}) and Eq.~(\ref{pdiversoj1}) into Eq.~(\ref{IMPOIMPO112pre}) we arrive hence to 
   \begin{equation}  \label{IMPOIMPO112}
  \ave{\tr{\hat{T}_{k_{(1,u-1)}, 1}^{(u)}\;  \hat{\Pi}\hat{\rho}_{\vecund{k}}\hat{\Pi}}}_{\cal S}
\geq 1 -\epsilon_1 - 2^{N R}~2^{-N\left(\chi(\mathcal{S})-2\delta\right)} \;, 
\end{equation}
which shows that as long as  the rate  $R$ respects the  Holevo bound, i.e., 
  \begin{equation} 
     R<\chi(\mathcal{S})-2\delta,  \label{IMPO0} 
     \end{equation} 
    for some $\delta >0$,  Eq.~(\ref{IMPOIMPO11}) applies also for  the  sets  $\mathcal{C}^{(u)}_{k_{(1,u)}}$ with $k_u=1$.  
\end{proof}
The inequalities Eq.~(\ref{IMPOIMPO111}) and Eq.~(\ref{IMPOIMPO112}) prove that under the constraint Eq.~(\ref{IMPO0}) the proposed implementation of the binary-tree-search decoding scheme asymptotically attains the Holevo bound, yielding
an average error probability which converges to zero in the limit of $N\rightarrow \infty$. 

 \subsubsection{Method 2: via SRM detections} \label{EX00}
   An alternative way to implement the binary-tree-search 
    protocol is to employ ternary-outcome POVMs, \eqq{tOutcome}, inspired by the SRM, \eqq{srm}.
   For each set   $\mathcal{C}^{(u)}_{k_{(1,u)}}$ define the positive operator 
    \begin{equation}
      \hat{S}^{(u)}_{k_{(1,u)}}=\sum_{\vecund{k}\in\mathcal{C}^{(u)}_{k_{(1,u)}}} \hat{\Pi}_{\vecund{k}}, 
    \end{equation}
    i.e., the sum of the conditionally-typical projectors of all the codewords in that set. From the 
    non-orthogonality of conditionally-typical projectors and the completeness property \eqq{complete1} 
    it follows 
    \begin{equation}
      \hat{S}^{(u-1)}_{k_{(1,u-1)}}=\hat{S}^{(u)}_{k_{(1,u-1)},0}+\hat{S}^{(u)}_{k_{(1,u-1)},1}\geq\hat{\mathbf{1}}_{k_{(1,u-1)}}.
    \end{equation}
    Thus we can build the $u-$th measurement to decide whether the word belongs to $\mathcal{C}^{(u)}_{k_{(1,u-1)}, 0}$ or $\mathcal{C}^{(u)}_{k_{(1,u-1)}, 1}$  by using the sum operators for these two sets and renormalizing
    by the sum operator for $\mathcal{C}^{(u-1)}_{k_{(1,u-1)}}$, which contains both of them at the previous step:
    \begin{equation}\label{PGMdets}
\hat{T}^{(u)}_{k_{(1,u)}}=\left(\hat{S}^{(u-1)}_{k_{(1,u-1)}}\right)^{-1/2} \hat{S}^{(u)}_{k_{(1,u)}}\left(
\hat{S}^{(u-1)}_{k_{(1,u-1)}}
\right)^{-1/2}.
       \end{equation}
    In this way we obtain a proper measurement, which again collapses to a binary-outcome one, i.e., $\mathcal{T}^{(u)}_{k_{(1,u-1)}}\equiv\mathcal{B}^{(u)}_{k_{(1,u-1)}}$, so that the total binary-tree-search nested POVM \eqq{nestedT} is equivalent to \eqq{nested} also in this case.
  Let us now show that for all $R$ fulfilling the Holevo bound, the operators \eqq{PGMdets} satisfy the sufficient condition Eq.~(\ref{IMPOIMPO11}) of Theorem~\ref{MainTh}. 

\begin{lemma}
For all rates $R$ satisfying the Holevo bound, Theorem~\ref{hsw},
the binary-tree-search nested POVM, \eqq{nested}, with binary-outcome elements $\hat{B}_{k_{(1,u)}}^{(u)}$
defined as in Eq.~(\ref{PGMdets}) fulfills the sufficient condition Eq.~(\ref{IMPOIMPO11}).
\end{lemma}  
\begin{proof}
Observe that 
\begin{equation}\label{PGMred}
  \begin{aligned}
    \ave{\tr{ \hat{T}^{(u)}_{k_{(1,u)}}     \hat{\Pi}\hat{\rho}_{\vecund{k}}\hat{\Pi}  }}_{\cal S} 
      &\geq \ave{\tr{\left(\hat{S}^{(u-1)}_{k_{(1,u-1)}}\right)^{-1/2} \hat{\Pi}_{\vecund{k}}  \left(\hat{S}^{(u-1)}_{k_{(1,u-1)}}\right)^{-1/2}
    \hat{\Pi}\hat{\rho}_{\vecund{k}}\hat{\Pi}   }}_{\cal S} \\
    &= \ave{\tr{\hat{\Lambda}_{\vecund{k}} \hat{\Pi}\hat{\rho}_{\vecund{k}}\hat{\Pi} }}_{\cal S}=\ave{p_{\vecund{k}|\vecund{k}}^{(u-1)}}_{\cal S},
    \end{aligned} 
    \end{equation}
    where $p_{\vecund{k}|\vecund{k}}^{(u-1)}$ is the average success probability of correctly recovering the codeword $\hat{\rho}_{\vecund{k}}$
        from the set  $\mathcal{C}^{(u-1)}_{k_{(1,u-1)}}$ 
    when using a SRM strategy and $\hat{\Lambda}_{\vecund{k}}$ is the corresponding POVM 
    element.
    Accordingly  we can bound each term by exploiting the efficiency of the SRM protocol. 
       Specifically, we employ the Hayashi-Nagaoka inequality~\cite{hayanaga} 
       \begin{equation}
         \hat{\mathbf{1}}-\hat{\Lambda}_{\vecund{k}}\leq 2~\hat{\Xi}_{\vecund{k}} + 4 \sum_{\vecund{k}'\neq\vecund{k}}\hat{\Pi}_{\vecund{k}'},
       \end{equation}
       where $\hat{\Pi}_{\vecund{k}}$, $\hat{\Xi}_{\vecund{k}}$ are the usual conditionally-typical projectors and their complementaries, to write the average success probability of \eqq{PGMred} as 
\begin{equation}
 \begin{aligned}
   \ave{p_{\vecund{k}|\vecund{k}}^{(u-1)}}_{\cal S} &\geq \ave{\tr{\hat{\Pi}\hat{\rho}_{\vecund{k}}\hat{\Pi} }}_{\cal S} - 2\ave{\tr{\hat{\Xi}_{\vecund{k}} \hat{\Pi}\hat{\rho}_{\vecund{k}}\hat{\Pi} } 
   }_{\cal S} - 4\sum_{\vecund{k}'\neq\vecund{k}}\ave{\hat{\Pi}_{\vecund{k}'} \hat{\Pi}\hat{\rho}_{\vecund{k}}\hat{\Pi} }_{\cal 
   S}\\
   &\geq 1- \epsilon_1 - 2 (\epsilon_2+2\sqrt{\epsilon_1}) - 4~2^{NR}~2^{-N\left(\chi(\mathcal{S})-2\delta\right)},
 \end{aligned}
 \end{equation}
 where the last inequality follows from Eqs.~(\ref{genttotaluno},\ref{pdiversoj1}) 
 and from the fact that
 \begin{equation}\label{qugualej}
 \begin{aligned}
\ave{\tr{\hat{\Xi}_{\vecund{k}} \hat{\Pi}\hat{\rho}_{\vecund{k}}\hat{\Pi} } 
   }_{\cal S}&= \ave{\tr{\hat{\Pi}\hat{\rho}_{\vecund{k}}\hat{\Pi}}}_{\cal S} -\ave{\tr{\hat{\Pi}_{\vecund{k}} \hat{\Pi}\hat{\rho}_{\vecund{k}}\hat{\Pi} } 
   }_{\cal S}\\
   & =1-  \sum_{\vecund{x}\in\mathcal{X}^{N}} p_{\vecund{x}} {\tr{\hat{\Pi}_{\vecund{x}}\hat{\Pi}\hat{\rho}_{\vecund{x}}\hat{\Pi} }} \leq \epsilon_2+ 2\sqrt{\epsilon_1},
    \end{aligned}
    \end{equation}
  which is derived as in \eqq{IMPOIMPO111}.
 Similarly to what we observed in Eq.~(\ref{IMPOIMPO112}), it then follows that if the rate $R$ fulfills the constraint Eq.~(\ref{IMPO0}) 
    for some $\delta >0$, then  for  $N$ sufficiently large one has that, for all $k_{u}$ and $u$, 
    \begin{equation}
      \ave{\tr{ \hat{T}^{(u)}_{k_{(1,u)}}     \hat{\Pi}\hat{\rho}_{\vecund{k}}\hat{\Pi}  }}_{\cal S}  \geq   \ave{\tr{\hat{\Lambda}_{\vecund{k}} \hat{\Pi}\hat{\rho}_{\vecund{k}}\hat{\Pi} }}_{\cal S}
     \geq  1-\epsilon_3,\label{last}
    \end{equation}
    with $\epsilon_3=O(e^{-N})$ exponentially small in $N$, hence showing that 
   Eq.~ (\ref{IMPOIMPO11}) is satisfied by the selected operators.  
\end{proof}

  \subsubsection{Method 3: via SM detections}\label{EX1}
 Another way to construct valid operators for the POVMs \eqq{tOutcome} necessary to implement the binary-tree-search scheme, is to 
 make use of the SM, Eqs.~(\ref{seqSucEle},\ref{seqErrEle}), for each set $\mathcal{C}^{(u)}_{k_{(1,u-1)}}$, but without gaining knowledge about the result of this subroutine. Accordingly, the single-step outcome operators will be implemented as a black box that applies the SM decoding scheme to the set of codewords belonging to $\mathcal{C}^{(u)}_{k_{(1,u-1)}}$ and may also output an error, see Fig.~\ref{figure2}. The resulting
setting is clearly redundant, since the vast majority of information gathered via the SM is simply neglected in the process.  Also, the same
procedure is iterated at every step, hence increasing the chances of deteriorating the transmitted codeword.
Still, as we shall see in the following, the scheme is efficient enough to allow for the saturation of the Holevo bound. 
\begin{figure}[t!]
	\centering
	\includegraphics[trim={0 0cm 0 0cm},clip,scale=.43]{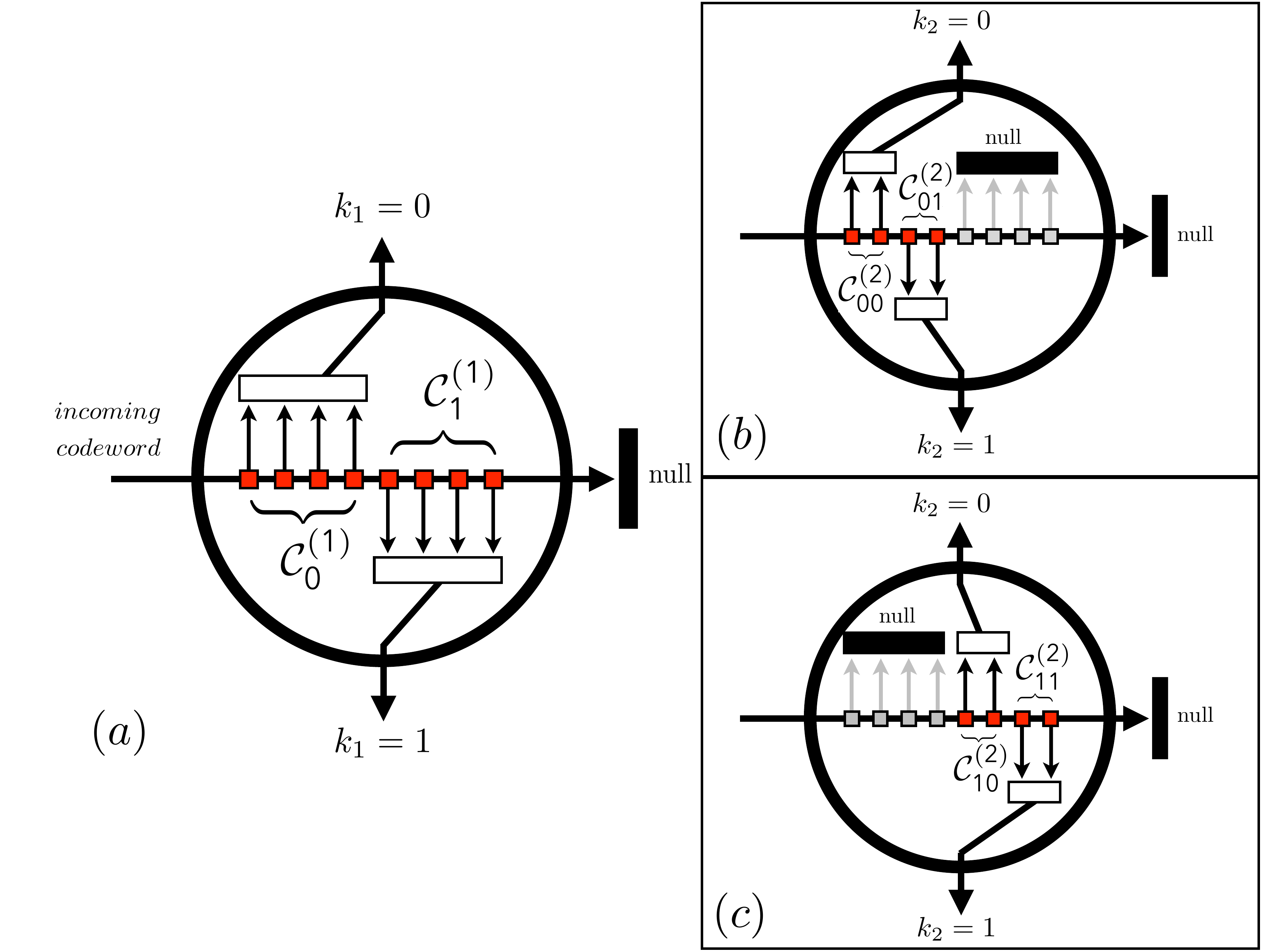}
	\caption{Schematic representation of the first- and second-step  POVMs $\mathcal{T}^{(1)}$, $\mathcal{T}_0^{(2)}$, and $\mathcal{T}_1^{(2)}$  in terms of the SM decoding procedure (little square elements) for $M=8$ codewords. Panel (a): implementation of $\mathcal{T}^{(1)}$. The red color of the square blocks indicates that all the  elements of the sequential decoding POVM are active: their outcomes are used to determine whether the
	incoming codeword belongs to the set  $\mathcal{C}_0^{(1)}$ (first four codewords), or to the set  $\mathcal{C}_1^{(1)}$ (last four codewords) fixing the value of $k_1$. The rectangular elements of the figure indicate that no other information is extracted from the outcomes of the sequential measurement. Panel (b):  implementation of $\mathcal{T}_0^{(2)}$ which discriminates
	between the sets $\mathcal{C}_{00}^{(2)}$ and $\mathcal{C}_{01}^{(2)}$.
	This element operates on the state emerging from the port $k_1=0$ of $\mathcal{T}^{(1)}$, see Fig.~\ref{figure1}. As indicated by the color, only the  first elements of the sequential POVM are active, while the outputs of the remaining ones are equivalent to the error result. Panel (c): implementation of $\mathcal{T}_1^{(2)}$ which discriminates among  $\mathcal{C}_{10}^{(2)}$ and $\mathcal{C}_{11}^{(2)}$. }
	\label{figure2}
\end{figure}
 In order to formalize this construction, for each quantum  codeword $\hat{\rho}_{\vecund{k}}\in \mathcal{C}$ we consider the corresponding element $\hat{E}_{\vecund{x}^{(\vecund{k})}}\equiv\hat{E}_{\vecund{k}}$ of the SM, \eqq{seqSucEle}, and define the elements of $\mathcal{T}^{(u)}_{k_{(1,u-1)}}$
as 
\begin{align}\label{sm01}
 \hat{T}^{(u)}_{k_{(1,u-1)},k_{u}}&=\sum_{\vecund{k}\in\mathcal{C}^{(u)}_{k_{(1,u-1)},k_{u}}} \hat{E}_{\vecund{k}},\quad\text{for }k_{u}=0,1,\\
 \label{smErr}  \hat{T}^{(u)}_{k_{(1,u-1)},err}&=\hat{\mathbf{1}}- \hat{T}^{(u)}_{k_{(1,u-1)},0}- \hat{T}^{(u)}_{k_{(1,u-1)},1}.
  \end{align}
Hence in this case the ternary-outcome POVMs do not reduce to binary-outcome ones. 
   
   \begin{lemma}
For all rates $R$ satisfying the Holevo bound, Theorem~\ref{hsw},
the binary-tree-search nested POVM, \eqq{nestedT}, with ternary-outcome elements $\hat{T}_{k_{(1,u)}}^{(u)}$
defined as in Eqs.~(\ref{sm01},\ref{smErr}) fulfills the sufficient condition Eq.~(\ref{IMPOIMPO11}).
\end{lemma}  
   \begin{proof}
   First observe that
   \begin{equation} \label{IMPO1}
    \ave{\tr{ \hat{T}^{(u)}_{k_{(1,u)}}
    \hat{\Pi}\hat{\rho}_{\vecund{k}}\hat{\Pi}}}_{\cal S} 
      =\sum_{\vecund{k}'\in\mathcal{C}^{(u)}_{k_{(1,u)}}} 
 \ave{\tr{\hat{E}_{\vecund{k}'}  \; \hat{\Pi}\hat{\rho}_{\vecund{k}'}\hat{\Pi}}}_{\cal S}
      \geq \ave{\tr{\hat{E}_{\vecund{k}} 
      \hat{\Pi}\hat{\rho}_{\vecund{k}}\hat{\Pi}}}_{\cal S} ,
    \end{equation}
    where in the last term we recognize the average success probability of the SM protocol for codeword $\vecund{k}$,
    computed on the unnormalized state $\hat{\Pi}\hat{\rho}_{\vecund{k}}\hat{\Pi}$. 
Accordingly we can bound each term by exploiting the efficiency of the SM protocol.
      Specifically, we need to apply the following lemma proved by Sen~\cite{sen}:
      \begin{lemma}\label{Sen} 
(Sen's Lemma) Let $\hat{\rho}$ be an unnormalized density operator and  
  $\hat{\Pi}_1,\ldots,\hat{\Pi}_n$ orthogonal projectors on 
  subspaces of its Hilbert space and $\hat{\Xi}_i$ their complementaries for all $i=1,\cdots,n$. 
  Then 
  \begin{equation}
    \tr{\abs{\hat{\Pi}_n\ldots \hat{\Pi}_1}^{2}\hat{\rho} }\geq 
    \tr{\hat{\rho}}-2\sqrt{\sum_{i=1}^n\tr{ \hat{\Xi}_i \hat{\rho}}}.
  \end{equation}
  \end{lemma}
  The previous lemma together with the concavity of the square root function implies: 
  \begin{equation}
    \begin{aligned}
    \langle\operatorname{Tr}[\hat{E}_{\vecund{k}} 
      \hat{\Pi}\hat{\rho}_{\vecund{k}}&\hat{\Pi}]\rangle_{\cal S}=\ave{\tr{\abs{\hat{\Pi}_{\vecund{k}^{(\ell)}}\hat{\Xi}_{\vecund{k}^{(\ell -1)}}\ldots \hat{\Xi}_{\vecund{k}^{(1)}}}^{2} \hat{\Pi}\hat{\rho}_{\vecund{k}}\hat{\Pi} }}_{\cal S}\\
      &\geq\ave{\tr{ \hat{\Pi}\hat{\rho}_{\vecund{k}}\hat{\Pi} }}_{\cal S}-2\sqrt{\ave{{\tr{ \hat{\Xi}_{\vecund{k}} \hat{\Pi}\hat{\rho}_{\vecund{k}}\hat{\Pi} }}}_{\cal S} +\sum_{\vecund{k}' \neq \vecund{k} }  \ave{ \tr{\hat{\Pi}_{\vecund{k}'} \hat{\Pi}\hat{\rho}_{\vecund{k}'}\hat{\Pi} } }_{\cal S}  },
      \end{aligned}
      \end{equation}
where the equality is obtained by writing explicitly the ordering labels $\ell$ necessary to describe a SM POVM element, while the inequality by also adding under the square root all the terms $\hat{\Pi}_{\vecund{k}^{(\ell')}}$ with $\ell'>\ell$. Now, the term outside the square-root can be treated as in \eqq{aveTrans}; for the first term under square-root we can simply apply \eqq{qugualej}; for each term of the sum under square-root we can instead use the inequality \eqq{pdiversoj1}. Therefore we can write 
      \begin{equation}
      \ave{\tr{\hat{E}_{\vecund{k}} 
      \hat{\Pi}\hat{\rho}_{\vecund{k}}\hat{\Pi}}}_{\cal S} \geq 
      1-\epsilon_1-2\sqrt{\epsilon_2+2\sqrt{\epsilon_1}+2^{NR}~2^{-N\left(\chi(\mathcal{S})-2\delta\right)}},
    \end{equation}
  which via Eq.~(\ref{IMPO1}) implies again that, for rates $R$ fulfilling Eq.~(\ref{IMPO0}) 
    for some $\delta >0$ and $N$ sufficiently large,
    \begin{equation}
   \ave{\tr{ \hat{T}^{(u)}_{k_{(1,u)}}
    \hat{\Pi}\hat{\rho}_{\vecund{k}}\hat{\Pi}}}_{\cal S} 
 \geq \ave{\tr{\hat{E}_{\vecund{k}} 
      \hat{\Pi}\hat{\rho}_{\vecund{k}}\hat{\Pi}}}_{\cal S}
 \geq  1-\epsilon_3,\label{last}
    \end{equation}
    for all $k_{u}$ and $u$,
    with $\epsilon_3=O(e^{-N})$ exponentially small in $N$.
       This proves Eq.~(\ref{IMPOIMPO11}) for the SM nested detections as well.
    \end{proof}
 
 \section{Optimal Quantum State Discrimination via a binary-tree-search measurement}\label{sec:opDisc}
In this section we employ the binary-tree-search decomposition of any POVM, Sec.~\ref{sec:bisMeas}, to compute the optimal success probability of discrimination of $M$ arbitrary quantum states, a problem presented in Sec.~\ref{sec:disc}. In this way the optimization of the whole POVM can be broken into several steps, optimizing first the binary-outcome POVMs at the deepest nesting level $u=u_{F}$ and then moving on to those at higher levels. We obtain an analytical expression for the maximum success probability after the first optimization step and examine its form for the specific case of $M=3,4$. Restricting to the qubit case, we are able to provide a compact expression for the success probability of any set of $M=3,4$ qubit states, whose numerical optimization is straightforward; the results thus obtained highlight some lesser-known features of the discrimination problem~\cite{qubits,bae,threeQubits}.

\subsection{The success probability}\label{StateDisc}
Let us suppose we want to discriminate the states belonging to a quantum source $\mathcal{S}^{(M)}$ of $M$ states, as given by \eqq{qSource}. The average success probability of discriminating the states $\hat{\rho}_{\ell}\in\mathcal{S}^{(M)}$ with a $M$-outcome POVM $\mathcal{M}^{(M)}$ can be computed as in \eqq{errProbMed}. We are particularly interested in the optimal success probability 
\begin{equation}
\mathbb{P}_{succ}(\mathcal{S}^{(M)})=1-\mathbb{P}_{err}(\mathcal{S}^{(M)}),
\end{equation}
obtained by optimizing over all measurements and complementary to the optimal error probability $\mathbb{P}_{err}$, \eqq{pErrOp}. \\
Following Sec.~\ref{sec:bisMeas}, we can always decompose the discrimination POVM into a sequence of nested binary POVMs, i.e., $\mathcal{M}^{(M)}\equiv\mathcal{N}^{(M)}$, \eqq{nested}, writing the success probability as 
\begin{equation}\label{probNest}
P_{succ}\left(\mathcal{S}^{(M)},\mathcal{N}^{(M)}\right)=\sum_{k_{(1,u_{\tiny{F}})}}\operatorname{Tr}\left[\hat{F}_{\vecund{k}}\hat{\sigma}_{\vecund{k}}\right]=\sum_{k_{(1,u_{\tiny{F}})}}\operatorname{Tr}\left[\left|\sqrt{\hat{B}^{(u_{\tiny{F}})}_{k_{(1,u_{\tiny{F}})}}} \cdots\sqrt{\hat{B}^{(1)}_{k_{1}}}\right|^{2}\hat{\sigma}_{\vecund{k}}\right],
\end{equation}
where we have introduced the binary representation $\vecund{k}=\vecund{k}^{(\ell)}$ for the labels $\ell$ of the states and measurement operators, the compact notation $\hat{\sigma}_{\vecund{k}}=p_{\vecund{k}}\hat{\rho}_{\vecund{k}}$ for the weighted states and employed the definition \eqq{nested} for the elements of the nested POVM. Moreover the sum is intended to be over all bits $k_{1},\cdots,k_{u_{\tiny{F}}}=0,1$.
This decomposition is interesting because it establishes a relation between the discrimination probability of a given source of states and that of its subsets of smaller size. Let us indeed suppose that the first measurement is successful, i.e., that an outcome $k_{1}$ occurs if one of the states $\hat{\rho}_{k_{1},k_{(2,u_{\tiny{F}})}}$ with that value of the first bit was present. This happens with probability 
\begin{equation}
P_{succ}(k_{1})=\sum_{k_{(2,u_{\tiny{F}})}}\operatorname{Tr}\left[\hat{B}^{(1)}_{k_{1}}\hat{\sigma}_{k_{1},k_{(2,u_{\tiny{F}})}}\right].
\end{equation}
 In this case the possible weighted states after the measurement are 
 \begin{equation}
  \hat{\tau}_{k_{1},k_{(2,u_{\tiny{F}})}}=\oneover{P_{succ}(k_{1})}\sqrt{\hat{B}^{(1)}_{k_{1}}}\hat{\sigma}_{k_{1},k_{(2,u_{\tiny{F}})}}\sqrt{\hat{B}^{(1)}_{k_{1}}},
  \end{equation} 
  forming a set of size $M/2$: $\mathcal{S}^{(M/2)}_{k_{1}}=\left\{\hat{\tau}_{k_{1},k_{(2,u_{\tiny{F}})}}\right\}_{k_{2},\cdots,k_{u_{\tiny{F}}}=0,1}$. Moreover the collection of remaining measurements can be seen as a nested POVM of order $M/2$: 
  \begin{equation}
   \mathcal{N}^{(M/2)}_{k_{1}}=\left\{\mathcal{B}^{(u_{\tiny{F}})}_{k_{1},k_{(2,u_{\tiny{F}}-1)}}\right\}_{k_{2},\cdots,k_{u_{\tiny{F}}-1}=0,1}\circ\cdots\circ\left\{\mathcal{B}^{(2)}_{k_{1}}\right\}. 
   \end{equation}
Hence we can easily rewrite the probability of discriminating the set of states $\mathcal{S}^{(M)}$ with the POVM $\mathcal{N}^{(M)}$, ~\eqq{probNest}, as the probability of discriminating the set $\mathcal{S}^{(M/2)}_{k_{1}}$ with the POVM $\mathcal{N}^{(M/2)}_{k_{1}}$ if the first measurement had an outcome $k_{1}$, averaged over all values of $k_{1}=0,1$:
\begin{equation}\label{rec}
\begin{aligned}
P_{succ}\Big(\mathcal{S}^{(M)},\mathcal{N}^{(M)}\Big)&=\sum_{k_{1},k_{(2,u_{\tiny{F}})}}P_{succ}(k_{1})\operatorname{Tr}\Bigg[\left|\sqrt{\hat{B}^{(u_{\tiny{F}})}_{k_{1},k_{(2,u_{\tiny{F}})}}}\cdots\sqrt{\hat{B}^{(2)}_{k_{1},k_{2}}}\right|^{2}\hat{\tau}_{k_{1},k_{(2,u_{\tiny{F}})}}\Bigg]\\
&=\sum_{k_{1}}P_{succ}(k_{1})P_{succ}\left(\mathcal{S}^{(M/2)}_{k_{1}},\mathcal{N}^{(M/2)}_{k_{1}}\right).
\end{aligned}
\end{equation}
This expression suggests a recursive optimization: if the optimal discrimination problem is solved for any set  $\mathcal{S}^{(M/2)}$ of a fixed size $M/2$, possibly restricting to a specific Hilbert space, e.g., qubits or continuous-variable states, then the result can be plugged into \eqq{rec} to obtain an expression for the discrimination probability of a set of double the size, which depends on a single couple of measurement operators, i.e., the first-step binary POVM $\mathcal{B}^{(1)}$. However, if a general solution for the discrimination of a smaller set of states is not available, the problem remains hard, since when optimizing $P_{succ}\left(\mathcal{S}^{(M/2)}_{k_{1}},\mathcal{N}^{(M/2)}_{k_{1}}\right)$ we still has to take into account the dependence of the states of $\mathcal{S}^{(M/2)}_{k_{1}}$ on the first-step POVM, which is itself subject to optimization afterwards, thus making the states arbitrary.\\
Fortunately the first step of the recursion has a well-known solution, \eqq{helBound}: 
\begin{equation}
 \mathbb{P}_{succ}\left(\mathcal{S}^{(2)}\right)=\frac{1+\norm{\hat{\sigma}_{0}-\hat{\sigma}_{1}}_{1}}{2}.
\end{equation} 
Then by plugging this expression into the optimization of ~\eqq{rec} for $M=4$ states we can write:
\begin{equation}
\begin{aligned}
\mathbb{P}_{succ}\Big(&\mathcal{S}^{(4)}\Big)=\max_{\mathcal{B}^{(1)}}\sum_{k_{1}}\oneover{2}P_{succ}(k_{1})\left(1+\norm{\hat{\tau}_{k_{1},0}-\hat{\tau}_{k_{1},1}}_{1}\right)\\
&=\max_{\mathcal{B}^{(1)}}\sum_{k_{1}}\oneover{2}\left(\operatorname{Tr}\left[\hat{B}^{(1)}_{k_{1}}\left(\hat{\sigma}_{k_{1},0}+\hat{\sigma}_{k_{1},1}\right)\right]+\norm{\sqrt{\hat{B}^{(1)}_{k_{1}}}\left(\hat{\sigma}_{k_{1},0}-\hat{\sigma}_{k_{1},1}\right)\sqrt{\hat{B}^{(1)}_{k_{1}}}}_{1}\right).
\end{aligned}
\end{equation}
The latter formula can be expressed more compactly by introducing the function
\begin{equation} \label{DEFFQABC}
\mathcal{\tiny{F}}_{\hat{Q}}\left(\hat{A},\hat{B},\hat{C}\right)=\operatorname{Tr}\left[\hat{Q}\hat{A}+\left|\sqrt{\hat{Q}}\hat{B}\sqrt{\hat{Q}}\right|+\left|\sqrt{\hat{\mathbf{1}}-\hat{Q}} \hat{C} \sqrt{\hat{\mathbf{1}}-\hat{Q}}\right|\right],
\end{equation} 
where $\hat{Q}$ is a positive and less-than-one operator, while the arguments $\hat{A},\hat{B},\hat{C}$ are hermitian operators, and its maximum over $\hat{Q}$, i.e., 
\begin{equation} 
\mathcal{\tiny{F}}\left(\hat{A},\hat{B},\hat{C}\right)=\max_{\hat{\mathbf{1}}\geq \hat{Q}\geq0}\mathcal{\tiny{F}}_{\hat{Q}}\left(\hat{A},\hat{B},\hat{C}\right). \label{DEFFABC}
\end{equation}
  Setting $\hat{B}^{(1)}_{0}=\hat{Q}$ and $\hat{B}^{(1)}_{1}=\hat{\mathbf{1}}-\hat{Q}$, we obtain:
 \begin{equation}
 \mathbb{P}_{succ}\Big(\mathcal{S}^{(4)}\Big)=\frac{p_{1,0}+p_{1,1}}{2}+\mathcal{\tiny{F}}\left(\hat{A}^{(4)},\hat{B}^{(4)},\hat{C}^{(4)}\right),
 \label{prob4}
 \end{equation} 
 with 
 \begin{align}
 \hat{A}^{(4)}&=\frac{\hat{\sigma}_{0,0}+\hat{\sigma}_{0,1}-\hat{\sigma}_{1,0}-\hat{\sigma}_{1,1}}{2},\\ 
 \hat{B}^{(4)}&=\frac{\hat{\sigma}_{0,0}-\hat{\sigma}_{0,1}}{2},\\
\hat{C}^{(4)}&=\frac{\hat{\sigma}_{1,0}-\hat{\sigma}_{1,1}}{2}. 
\end{align}
Similarly for $M=3$ states we have:
\begin{equation}
\mathbb{P}_{succ}\Big(\mathcal{S}^{(3)}\Big)=p_{1,0}+\mathcal{\tiny{F}}\left(\hat{A}^{(3)},\hat{B}^{(3)},\hat{C}^{(3)}\right),\label{prob3}
\end{equation}
with 
\begin{align}
\hat{A}^{(3)}&=\frac{\hat{\sigma}_{0,0}+\hat{\sigma}_{0,1}}{2}-\hat{\sigma}_{1,0},\\
\hat{B}^{(3)}&=\hat{B}^{(4)},\\
\hat{C}^{(3)}&=0.
\end{align}
Thus the optimal discrimination problem of $M=3,4$ states has been reduced to the evaluation of the function $\mathcal{\tiny{F}}$, which requires an optimization over a single operator $\hat{Q}$.\\
As already discussed, if the problem of \eqq{prob4} were to be solved exactly for any set of states $\mathcal{S}^{(4)}$, then the result could be plugged into ~\eqq{rec}, obtaining an expression for the optimal discrimination probability of $M=8$ states dependent only on the first binary POVM. Unfortunately a solution of Eqs.~(\ref{prob4},\ref{prob3}) can be found only in some specific cases, listed below and discussed in detail in Appendix~\ref{app:casesOpSol}. In the following we write the positive and negative part of an operator $\hat{X}$ as $\hat{X}_{\pm}=(\hat{X}\pm\abs{\hat{X}})/2$ and it holds $\hat{X}_{+}\geq0$ and $\hat{X}_{-}\leq0$ by definition; moreover the support of $\hat{X}$ is written as $\operatorname{supp}(\hat{X})$.
 \newtheorem{proposition}{Proposition}
\begin{proposition}\label{cases}
The value of the function $\mathcal{\tiny{F}}\left(\hat{A},\hat{B},\hat{C}\right)$ of Eq.~(\ref{DEFFABC}) is 
\begin{equation}\label{value} 
\mathcal{\tiny{F}}\left(\hat{A},\hat{B},\hat{C}\right)=\tr{\left(\hat{A}+\abs{\hat{B}}-\abs{\hat{C}}\right)_{+}}+\norm{\hat{C}}_{1}, 
\end{equation} 
in any of the following cases: 
\begin{enumerate} 
\item the operators $\hat{B}$ and $\hat{C}$ have support respectively on the positive and negative support of $\hat{A}$, i.e., $\operatorname{supp}(\hat{B})\subseteq\operatorname{supp}(\hat{A}_{+})$ and $\operatorname{supp}(\hat{C})\subseteq\operatorname{supp}(\hat{A}_{-})$;
\item $\hat{B}$ and $\hat{C}$ have a definite sign, i.e., $\hat{B}\geq0$ or $\hat{B}\leq0$ and $\hat{C}\geq0$ or $\hat{C}\leq0$; 
\item $\hat{A}$, $\hat{B}$ and $\hat{C}$ all commute with each other, i.e., $[\hat{A},\hat{B}]=[\hat{A},\hat{C}]=[\hat{B},\hat{C}]=0$.
\end{enumerate}
 \end{proposition}
 \newtheorem{remark}{Remark}
 \begin{remark}\label{remark1}
 In the first case of Proposition~\ref{cases} the expression \eqq{value} can be simplified as
 \begin{equation}
 \mathcal{\tiny{F}}\left(\hat{A},\hat{B},\hat{C}\right)=\operatorname{Tr}\left[\hat{A}_{+}\right]+\norm{\hat{B}}_{1}+\norm{\hat{C}}_{1}.
 \end{equation}
 \end{remark}
 The proofs of Proposition~\ref{cases} and Remark~\ref{remark1} can be found in Appendix~\ref{app:casesOpSol}.
\begin{remark}\label{recRem}
The optimal success probability is invariant under exchange of the states, i.e., under relabelling of the indices $k_{1},k_{2}$ in our case $M=3,4$. Hence it can happen that the conditions listed in Proposition~\ref{cases} are valid only for $\hat{A}, \hat{B}, \hat{C}$ given by a specific ordering of the states.
\end{remark}
The previous remark implies that, when checking whether a set of states satisfies the conditions of Proposition~\ref{cases} or not, one has to consider all possible sets of $\hat{A}, \hat{B}, \hat{C}$ obtainable by different orderings of the states, not only the conventional one employed in Eqs.~(\ref{prob4},\ref{prob3}). Alternatively, one can apply this symmetry under exchange of the states to obtain recursive relations for $\mathcal{\tiny{F}}\left(\hat{A},\hat{B},\hat{C}\right)$, e.g., for $M=3$ and by exchanging $(0,0)\leftrightarrow(1,0)$, it holds
\begin{equation}
\mathcal{\tiny{F}}(\hat{A},\hat{B},0)=\frac{\mathcal{\tiny{F}}(-3\hat{B}-\hat{A},\hat{B}-\hat{A},0)}{2}+\tr{\hat{A}+\hat{B}};\label{rec3}
\end{equation}
then Proposition~\ref{cases} holds on the right-hand side of \eqq{rec3} when $\hat{B}'=\hat{B}-\hat{A}$ has a definite sign, but the latter is simply $\hat{B}'=(\hat{\sigma}_{1,0}-\hat{\sigma}_{0,1})/2$, an expression of the operator $\hat{B}$ for the new ordering of the states. 
 \begin{remark}
 In all the cases listed in Proposition~\ref{cases}, with the conventional ordering of the states of Eqs.~(\ref{prob4},\ref{prob3}), the optimal success probabilities for the discrimination of $N=3,4$ states become
 \begin{align}
 &\mathbb{P}_{succ}\Big(\mathcal{S}^{(3)}\Big)=p_{1,0}+\oneover{2}\operatorname{Tr}\left[\left(\hat{\sigma}_{0,0}+\hat{\sigma}_{0,1}-2\hat{\sigma}_{1,0}+|\hat{\sigma}_{0,0}-\hat{\sigma}_{0,1}|\right)_{+}\right],\label{p3}\\
 &\begin{aligned} \label{p4}
 \mathbb{P}_{succ}\Big(\mathcal{S}^{(4)}\Big)&=\oneover{2}\operatorname{Tr}\left[\left(\hat{\sigma}_{0,0}+\hat{\sigma}_{0,1}-\hat{\sigma}_{1,0}-\hat{\sigma}_{1,1}+|\hat{\sigma}_{0,0}-\hat{\sigma}_{0,1}|-|\hat{\sigma}_{1,0}-\hat{\sigma}_{1,1}|\right)_{+}\right]\\
 &+\frac{p_{1,0}+p_{1,1}+\norm{\hat{\sigma}_{1,0}-\hat{\sigma}_{1,1}}_{1}}{2}.
\end{aligned}
 \end{align}
 \end{remark}

\subsection{The case of three or four qubits}\label{Qubits}
Let us now analyze the probability expressions obtained in the previous subsection in the case of $M=3,4$ \textit{qubit} states. Indeed, since Eqs.~(\ref{prob4},\ref{prob3}) seem not to be solvable analytically for generic sets of states, it is interesting to tackle the problem by choosing the simplest possible Hilbert space for the measured system, i.e., the qubit space $\mathcal{H}_{2}$, see Sec.~\ref{subsec:qubits}. 
Employing the Bloch representation for hermitian operators, \eqq{blochHerm}, we can rewrite the function $\mathcal{\tiny{F}}_{\hat{Q}}\left(\hat{A},\hat{B},\hat{C}\right)$ as (see Appendix~\ref{app:qubitOp} for details) 
\begin{equation}\label{fQub}
\begin{aligned}
\mathcal{\tiny{F}}^{(2)}_{Q}\Big(\hat{A},\hat{B}&,\hat{C}\Big)=2\Bigg(\sqrt{(c_{Q}c_{B}+\vecund{r}_{Q}\cdot\vecund{r}_{B})^{2}+\left(\left(r_{B}\right)^{2}-\left(c_{B}\right)^{2}\right)\left(\left(c_{Q}\right)^{2}-\left(r_{Q}\right)^{2}\right)}\\
&+\sqrt{((1-c_{Q})c_{C}-\vecund{r}_{Q}\cdot\vecund{r}_{C})^{2}+\left(\left(r_{C}\right)^{2}-\left(c_{C}\right)^{2}\right)\left(\left(1-c_{Q}\right)^{2}-\left(r_{Q}\right)^{2}\right)}\\
&+c_{Q}c_{A}+\vecund{r}_{Q}\cdot\vecund{r}_{A}\Bigg),
\end{aligned}
\end{equation}
 when $\hat{B}$ and $\hat{C}$ do not have a definite sign, or
\begin{equation}\label{defSign}
\mathcal{\tiny{F}}^{(2,d)}_{Q}\left(\hat{A},\hat{B},\hat{C}\right)\Big|_{\mathcal{H}_{2}}=2\left(c_{Q}c_{A+|B|-|C|}+\vecund{r}_{Q}\cdot\vecund{r}_{A+|B|-|C|}+c_{|C|}\right),
\end{equation}
when both $\hat{B}$ and $\hat{C}$ have a definite sign.\\
For each source of $M=3,4$ qubit states to discriminate, the operators $\hat{A}, \hat{B}, \hat{C}$, i.e., their coefficients $c$ and $\vecund{r}$, are fixed and the optimization of Eqs.~(\ref{fQub}, \ref{defSign}) is to be carried out only over $\hat{Q}$, i.e., on its coefficients $c_{Q}$ and $\vecund{r}_{Q}$ under the constraints
\begin{equation}\label{constraintsQ} 
1\geq c_{Q}\geq0, \qquad r_{Q}\leq\min[c_{Q},1-c_{Q}], 
\end{equation} 
which ensure that $0\leq\hat{Q}\leq\hat{\mathbf{1}}$.
In particular, for $M=3$ states $\hat{C}=0$ and one can show that $c_{Q}+r_{Q}=\lambda^{(+)}_{Q}=1$ is optimal. Moreover the optimal $\vecund{r}_{Q}$ lies on the plane of $\vecund{r}_{A}$ and $\vecund{r}_{B}$, so that it can be defined in terms of its norm $r_{Q}$ and a single angle $\phi_{Q}$ as $\vecund{r}_{Q}\cdot\vecund{r}_{A}=r_{Q}r_{A}\cos\phi_{Q}$. Then in this case we only need to optimize two parameters, namely $c_{Q}$ and $\phi_{Q}$. For $M=4$ instead there are no further simplifications and one has to optimize four parameters under the constraints \eqq{constraintsQ}.\\
 Let us now consider the case in which \eqq{defSign} is valid, i.e., $\hat{B}$ and $\hat{C}$ have a definite sign: it satisfies condition $2$ of Proposition~\ref{cases}, thus the optimization of \eqq{defSign} must match the expression \eqq{value}. Indeed it can also shown by direct analytical optimization that in this case 
 \begin{equation}
 \mathcal{\tiny{F}}^{(2,d)}(\hat{A},\hat{B},\hat{C})=|c_{A+|B|-|C|}|+c_{A+|B|-|C|}+2c_{C}
 \end{equation} 
 and the optimal value of $\hat{Q}$ is $r_{Q}=0$ and $c_{Q}=\theta(c_{A+|B|-|C|})$, with $\theta(\cdot)$ the step function, valued $1$ when its argument is positive and zero otherwise. \\
 As for the other case, in which \eqq{fQub} is valid, unfortunately the function cannot be completely optimized analytically. Nevertheless its numerical optimization is straightforward and thus, together with Eqs.~(\ref{prob4}, \ref{prob3}), it provides a convenient method to obtain the optimal success probability of discrimination and the optimal measurement operators for $M=3,4$ qubit states. 
\begin{figure}[t!]\center
\includegraphics[trim={0 1.8cm 0 1.8cm},clip,scale=.43]{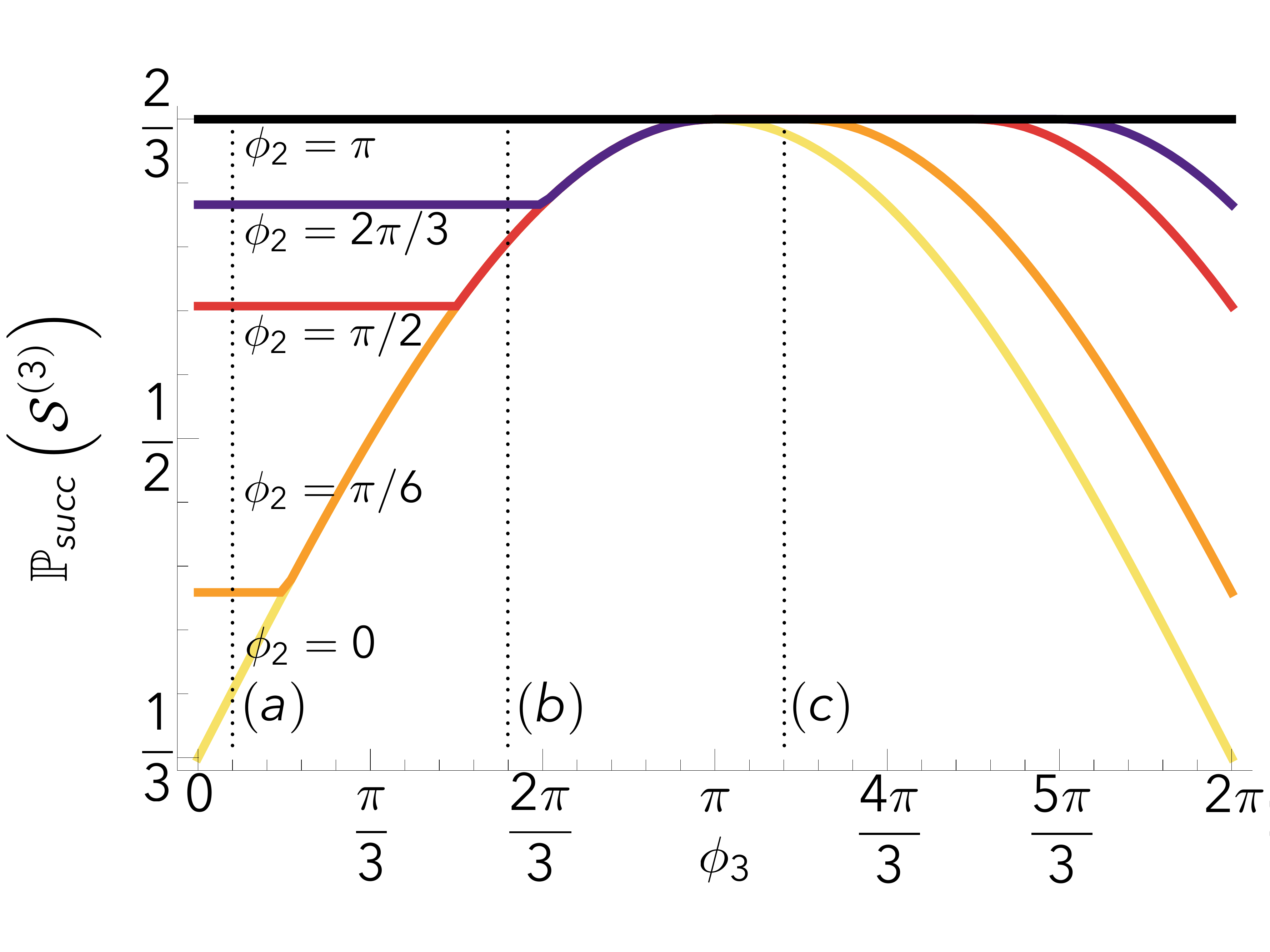}
\caption{Plot of $\mathbb{P}_{succ}\left(\mathcal{S}^{(3)}\right)$, the optimal success probability \eqq{prob3} for a set of three equiprobable pure qubit states, identified by the Bloch vectors $\vecund{r}_{\rho_{1}}=(1,0,0)$, $\vecund{r}_{\rho_{2}}=(\cos\phi_{2},\sin\phi_{2},0)$ and $\vecund{r}_{\rho_{3}}=(\cos\phi_{3},\sin\phi_{3},0)$, as a function of $\phi_{3}$ for several values of $\phi_{2}=0, \pi/6, \pi/2, 2\pi/3, \pi$ (respectively from yellow/light-gray to black). The results are obtained by numerical optimization of \eqq{fQub} over $c_{Q}$, $\vecund{r}_{Q}$. Observe that, for all values of $\phi_{2}$, there is a range of values of $\phi_{3}$ where $\mathbb{P}_{succ}\left(\mathcal{S}^{(3)}\right)$ attains the maximum allowed for non-orthogonal states, i.e., the same as for symmetric states. Outside of this range the quantity decreases, reaching a constant minimum for $\phi_{3}\leq\phi_{2}$ (see text and Fig.~\ref{f5} for an explanation). The cases $\phi_{2}=\pi/6, 2\pi/3$ are explicitly depicted in Fig.~\ref{f5} for three values of $\phi_{3}$ identified by the labelled dotted lines.}\label{f4}
\end{figure} 
As an example let us consider $M=3$ \textit{equiprobable pure} qubit states situated on the $(x,y)$ plane of the Bloch sphere, i.e., a combination of $\hat{\mathbf{1}}_{2}$, $\hat{\sigma}_{1}$ and $\hat{\sigma}_{2}$; this is a simple choice for the sake of clarity, but we stress that no additional optimization difficulties are met when considering non-equiprobable and mixed states. Let us fix the first state to be on the $x$ axis, i.e., $\vecund{r}_{\rho_{1}}=(1,0,0)$, without loss of generality. Then we can study the optimal success probability by varying the angles of the other two vectors with respect to the first one: $\vecund{r}_{\rho_{2}}=(\cos\phi_{2},\sin\phi_{2},0)$ and $\vecund{r}_{\rho_{3}}=(\cos\phi_{3},\sin\phi_{3},0)$. If the states are also symmetrically distributed at constant angles along the circumference, i.e., $\phi_{2}=\phi_{3}=2\pi/3$, the result is well-known~\cite{helstromBOOK}: $\mathbb{P}_{succ}\left(\mathcal{S}_{sym}^{(3)}\right)=2/3$, which is the maximum success probability of discrimination for any three equiprobable qubit states (because no other configuration can achieve a lower average state-overlap than this one). 
For more general states the results are shown in Fig.~\ref{f4} where we plot the optimal success probability, \eqq{prob3}, computed by numerical optimization of \eqq{fQub} over $c_{Q}$, $\vecund{r}_{Q}$, as a function of the third angle $\phi_{3}$ and for several choices of $\phi_{2}$. It can be seen that, for all values of $\phi_{2}$, there is a range of values of $\phi_{3}$ for which $\mathbb{P}_{succ}\left(\mathcal{S}^{(3)}\right)$ is equal to the maximum value of $2/3$, even though the states are not symmetrically distributed on the circumference. In other words there is a wide class of states that can be discriminated with performance as good as if they were symmetric. Outside of this range, whose width depends on $\phi_{2}$, the value of  $\mathbb{P}_{succ}\left(\mathcal{S}^{(3)}\right)$ decreases and it reaches a constant minimum when $\phi_{3}\leq\phi_{2}$. \\

The results of Fig.~\ref{f4} can be explained thanks to the geometrical treatment of~\cite{bae}, briefly discussed in Sec.~\ref{sec:disc}; in particular, \eqq{baeDisc} applies in this case. Let us then plot the states in the Bloch sphere and study the polytope formed by their vertices, as done in Fig.~\ref{f5} for two values of $\phi_{2}$ used in Fig.~\ref{f4} and $\phi_{3}$ corresponding to the dotted lines labelled $a, b, c$ in Fig.~\ref{f4}. If the polytope determined by the qubits, usually a triangle, contains the origin, then the optimal success probability is still maximum, i.e., $\mathbb{P}_{succ}\left(\mathcal{S}_{\supseteq0}^{(3)}\right)\equiv\mathbb{P}_{succ}\left(\mathcal{S}_{sym}^{(3)}\right)$, as in Fig.~\ref{f5}$c$. Indeed in this case the polytope formed by the states is already maximal in the Bloch sphere and $R=1/3$, as for the symmetric set. On the other hand, if the polytope does not contain the origin, the optimal success probability is strictly lower than the maximum one, i.e.,  $\mathbb{P}_{succ}\left(\mathcal{S}_{\nsupseteq0}^{(3)}\right)<\mathbb{P}_{succ}\left(\mathcal{S}_{sym}^{(3)}\right)$, as in Fig.~\ref{f5}$a,b$. Indeed in this second case the polytope formed by the states is not maximal and it can be expanded until its largest side matches a diameter of the circumference, so that $R<1/3$.
As for the region $\phi_{3}\leq\phi_{2}$ where $\mathbb{P}_{succ}\left(\mathcal{S}_{\nsupseteq0}^{(3)}\right)$ is minimum and constant (as in Fig.~\ref{f5}$a,b$ top and $a$ bottom), it can be explained by observing that in this case the largest side of the triangle determined by the states, which in turn determines $R$, is always the one that connects $\vecund{r}_{\rho_{1}}$ and $\vecund{r}_{\rho_{2}}$ independently of $\vecund{r}_{\rho_{3}}$. Since $\vecund{r}_{\rho_{2}}-\vecund{r}_{\rho_{1}}$ is constant for constant $\phi_{2}$, $R$ is constant too in this case. 
\begin{figure}[t!]\center
\includegraphics[trim={0 1cm 0 .8cm},clip,scale=.43]{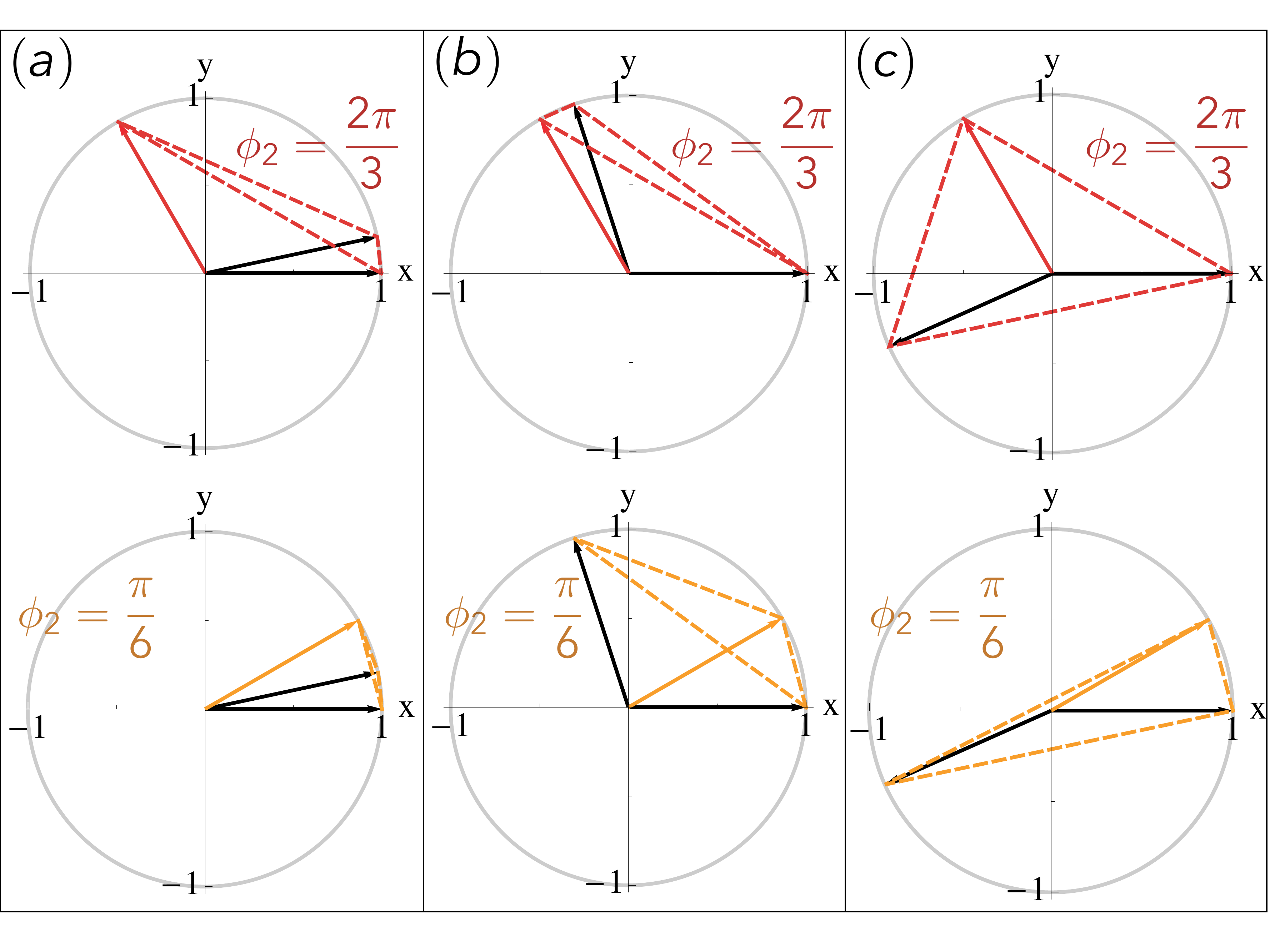}
\caption{Plot of the vectors of the three states $\vecund{r}_{\rho_{1}}=(1,0,0)$ (black fixed on $x$ axis), $\vecund{r}_{\rho_{3}}=(\cos\phi_{3},\sin\phi_{3},0)$ (black) and $\vecund{r}_{\rho_{2}}=(\cos\phi_{2},\sin\phi_{2},0)$ (red/dark gray at the top and orange/light gray at the bottom) on the $(x,y)$ Bloch plane, as well as the triangles formed by them (same color codes as $\vecund{r}_{\rho_{2}}$). The labels $a, b, c$ refer respectively to values of $\phi_{3}=\pi/15, 3\pi/5, 17\pi/15$, also highlighted in Fig.~\ref{f4}, while two values of $\phi_{2}= 2\pi/3,\pi/6$ (respectively red/dark gray and orange/light gray figures) are considered. By comparison with Fig.~\ref{f4} it is evident that $\mathbb{P}_{succ}\left(\mathcal{S}^{(3)}\right)$ is maximum when the triangle formed by the states contains the origin ($c$), while it is lower otherwise. In particular, when $\phi_{3}\leq\phi_{2}$ ($a, b$ top, $a$ bottom) the largest side of the triangle formed by the states is always $\vecund{r}_{\rho_{2}}-\vecund{r}_{\rho_{1}}$ and this determines completely the optimal success probability.}\label{f5}
\end{figure}

\chapter{Structured practical receiver schemes for communication and state discrimination}\label{ch:Imple}

\ifpdf
    \graphicspath{{Chapter4/Figs/Raster/}{Chapter4/Figs/PDF/}{Chapter4/Figs/}}
\else
    \graphicspath{{Chapter4/Figs/Vector/}{Chapter4/Figs/}}
\fi

In this chapter we consider the problems of optimal communication, Sec.~\ref{sec:commQ}, and state discrimination, Sec.~\ref{sec:disc}, from an implementation-oriented perspective, discussing three structured practical receiver structures for coherent states. In Sec.~\ref{sec:cohDisc} we describe a binary coherent-state discrimination scheme which improves on the Kennedy and Dolinar receivers, Sec.~\ref{subsec:discPract}, through the addition of a non-linear Gaussian element known as Probabilistic Amplifier (PA), proposed in~\cite{RalphStart,Caves,OptimalNLA}. In Sec.~\ref{sec:had} instead we consider a generalized family of multi-mode coherent-state communication protocols based on the Hadamard transform~\cite{Hadamard}, building on~\cite{guha1,guha2}; we evaluate the Holevo information of these protocols and the transmission rate achievable with a feasible decoder. Eventually, in Sec.~\ref{sec:adRec} instead we study the maximum transmission rate of a large class of realistic coherent-state adaptive decoders based on passive Gaussian interferometers and arbitrary single mode measurements, ruling out the possibility of achieving the Holevo bound with them and closing a conjecture made in~\cite{guhaDet1}. This chapter is based on~\cite{NOSTRODisc,NOSTROHad,NOSTROAdap}.

\section{Binary coherent-state discrimination via probabilistic amplification}\label{sec:cohDisc}
In this section we propose a practical scheme for the detection of binary coherent-state signals at low energies that employs a 
PA~\cite{RalphStart,AltAmp,ADagA,PhaseAmp,ExpRalph,ExpSciss,ExpZav,ExpQKD,AmpAsSq,CompareSqAmp,AmpSymp,Caves,OptimalNLA}, operated in its non-heralded version, as the underlying non-linear operation to improve the detection efficiency. This approach allows us to improve the statistics by keeping track of all possible outcomes of the PA stage (including failures). When compared with an optimized Kennedy receiver, the resulting discrimination success probability we obtain presents a gain of up to $\sim 1.88\%$ and it approaches the Helstrom bound appreciably faster than the Dolinar receiver, when employed in an adaptive strategy. We also notice that the advantages obtained can be ultimately 
associated with the fact that, in the high-gain limit, the non-heralded version of the PA induces a partial dephasing which preserves quantum coherence among low-photon-number states while removing it elsewhere. Hence a proposal to realize such transformation based on an optical-cavity implementation is presented and its discrimination probability is evaluated.

\subsection{The non-heralded-probabilistic-amplifier receiver}\label{sec2}
 Let us recall that an ordinary PA  performs  probabilistically the amplification of a coherent state $\ket{\alpha}$ into $\sim\ket{g\alpha}$, the failure events being highly probable but heralded by a triggering signal which allows one to discard them. In order to improve the outcome statistics of our receiver, we operate the PA device by considering its non-heralded version (NHPA), i.e., we act on the incoming signal with a standard PA machine 
with the only difference that all events, failures included, are accepted in the subsequent stages of the detection process. 
The NHPA of gain $g\geq 1$ and cutoff $n\in\mathbb{N}$ that we analyze here can be described as a CPTP map $\mathcal{A}_{g,n}$ characterized by two Kraus operators, see \eqq{kraus}, represented in the Fock basis as follows:
 \begin{align}\label{mSucc}
&\hat{M}_{S}=g^{-n}\sum_{k=0}^{N}g^{k}\dketbra{k}+\sum_{k=n+1}^{\infty}\dketbra{k}=\hat{\mathbf{1}}-\sum_{k=0}^{N}\left(1-g^{-(n-k)}\right)\dketbra{k},\\ 
 \label{mFail}&\hat{M}_{F}=\sum_{k=0}^{N}\left(1-g^{-2(n-k)}\right)^{1/2}\dketbra{k},
 \end{align}
which, according to the analysis presented in~\cite{Caves,OptimalNLA}, identify respectively success and failure in the amplification of a regular PA. We stress that this optimal theoretical form of the PA differs from those based  on conditional Gaussian operations~\cite{ExpRalph,ExpSciss,ExpZav,}  which, when employed in a non-heralded way, would result in a Gaussian measurement and thus would perform certainly worse than homodyne detection, as shown in~\cite{opGaussDet}. The action of $\mathcal{A}_{g,n}$ on an input coherent state $\ket{\alpha}$ is depicted in Fig.~\ref{Wigner}, in terms of the Wigner function of the corresponding output state, $W_{\alpha}(\gamma)$, whose computation is carried out in Appendix~\ref{app:wignerPA}. We notice that the state has a slightly non-Gaussian form: ultimately this is the key feature which allows one to
 improve the success probability of the scheme by reducing the difference between the vacuum-overlaps of the input states, see Sec.~\ref{subsec:discPract}. \\

\begin{figure}[t!]\center
\includegraphics[trim={0 1cm 0 .8cm},clip,scale=.43]{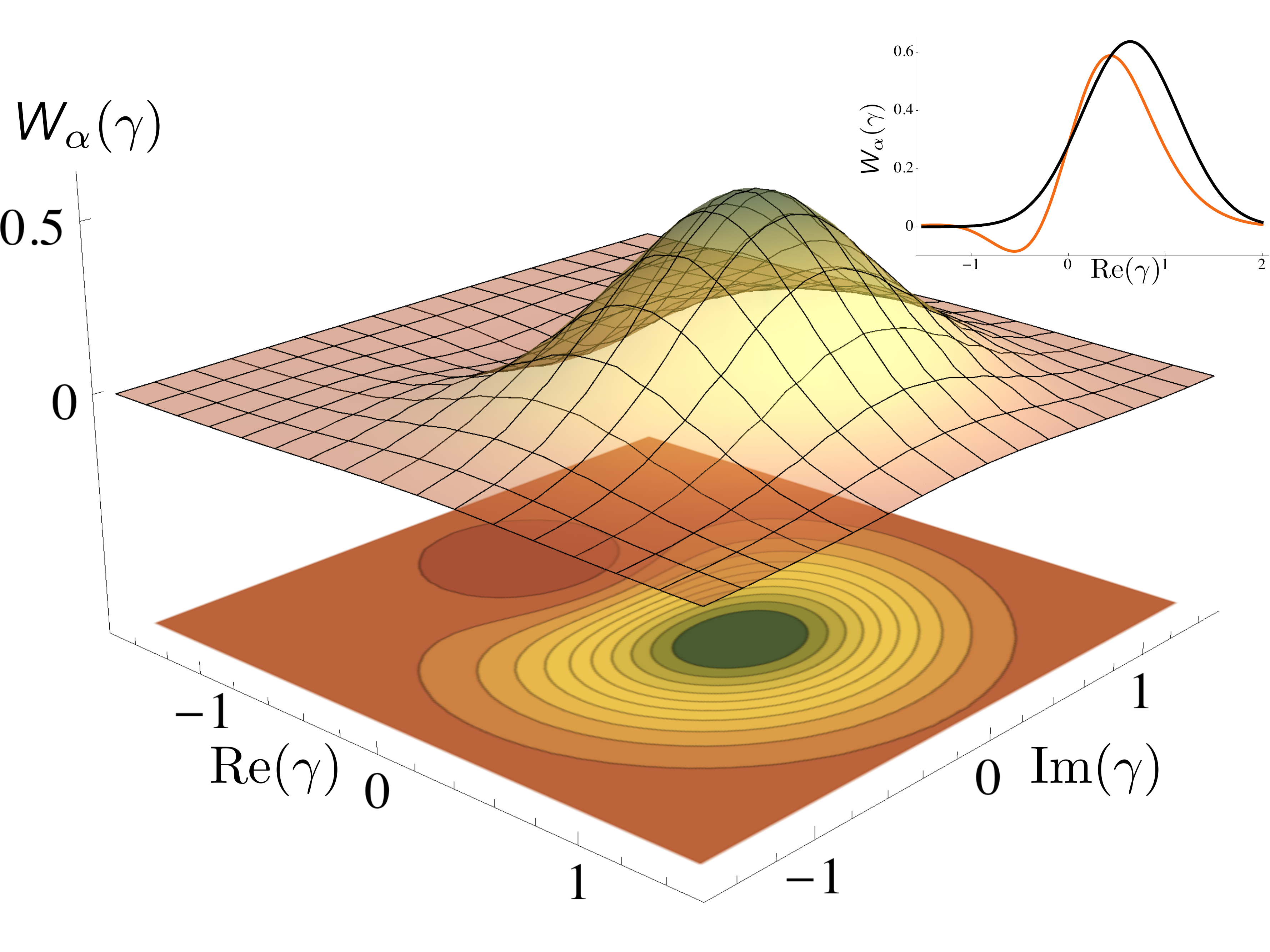}
\caption{Plot of the Wigner function $W_{\alpha}(\gamma)$ of the transformed coherent state $\mathcal{A}_{g,n}(\dketbra{\alpha})$ under the action of the NHPA. The inset shows a cut at $\operatorname{Im}(\gamma)=0$ of the same function (orange/light gray curve) and of the Wigner function of the input state $\dketbra{\alpha}$ (black curve). The former is far from an amplified version of the latter, still it exhibits a non-Gaussian profile, which reduces   the overlap~(\ref{OVERLAP}) improving the scheme success probability~(\ref{succprobAmp}). Both plots are drawn at the optimal gain value of the NHPA receiver for the chosen amplitude: $\alpha\simeq 0.5$, $g\simeq 28$, $n=2$.
 }\label{Wigner}
\end{figure}

Let us now suppose that we want to discriminate the two equiprobable\footnote{The case of non-equiprobable states is straightforward and exhibits the same qualitative behaviour as the equiprobable one.} coherent states $\ket{\pm\alpha}$, $\alpha>0$. Our decoding scheme is structured as shown in Fig.~\ref{contourplot}a:  we first apply a perfectly-nulling displacement for one of the two states, e.g., $\hat{D}(\alpha)$ to null $\ket{-\alpha}$. Next we make use of the NHPA, which leaves the nulled state in the vacuum, while sending the second state into a mixture of the target amplified state that is farther away from the vacuum and of a complex truncated state, being neither the original coherent state nor its desired amplified version. Eventually, we apply a second displacement operation $\hat{D}(-\beta)$, which we optimize in order to maximize the success probability of the protocol, and detect the signal with an OOD, employing the usual Kennedy inference. Let us note in passing that the second displacement is fundamental to get an improvement with respect to the optimized Kennedy strategy, as shown in Appendix~\ref{app:altDecs}. Accordingly, the input states are transformed as follows before entering the OOD:
\begin{align}
&\ket{\alpha_{-}}\rightarrow\hat{D}(-\beta)\mathcal{A}_{g,n}(\dketbra{0})\hat{D}^{\dag}(-\beta)=\dketbra{-\beta}, \\
&\ket{\alpha_{+}}\rightarrow\hat{D}(-\beta)\mathcal{A}_{g,n}(\dketbra{2\alpha})\hat{D}^{\dag}(-\beta), 
\end{align}
where in  the first term we used the fact that $\mathcal{A}_{g,n}$ leaves invariant the vacuum.
We now associate the event where no photons are detected to the arrival of $|\alpha_-\rangle$, while the others to the arrival of $|\alpha_+\rangle$. The success probability of the  protocol reads hence
\begin{equation}
P_{succ}(\alpha,g,n,\beta)=\oneover{2}\left(p_{0|-}+p_{1|+}\right)=\oneover{2}\left(1+p_{0|-}-p_{0|+}\right),\label{succprobAmp}
\end{equation}
where $p_{0|\pm}$ is the probability of detecting no photons from the transformed final state associated with the input $\ket{\pm\alpha}$, i.e., respectively, 
\begin{align} 
&p_{0|-}= \abs{\braket{0}{-\beta}}^{2}=e^{-\abs{\beta}^{2}},\label{OVERLAP1}\\
\label{OVERLAP}
&p_{0|+}=\langle \beta | \mathcal{A}_{g,n}(\dketbra{2\alpha})|\beta\rangle=\abs{\bra{\beta}\hat{M}_{S}\ket{2\alpha}}^{2}+\abs{\bra{\beta}\hat{M}_{F}\ket{2\alpha}}^{2},
 \end{align} 
 where
 \begin{align}
\abs{\bra{\beta}\hat{M}_{S}\ket{2\alpha}}^{2}&=e^{-(\abs{2\alpha}^{2}+\abs{\beta}^{2})}\abs{e^{2\alpha\beta^{*}}-\sum_{k\leq n}\frac{(2\alpha\beta^{*})^{k}}{k!}(1-g^{-(n-k)})}^{2},\\
\abs{\bra{\beta}\hat{M}_{F}\ket{2\alpha}}^{2}&=e^{-(\abs{2\alpha}^{2}+\abs{\beta}^{2})}\abs{\sum_{k\leq n}\frac{(2\alpha\beta^{*})^{k}}{k!}\left(1-g^{-2(n-k)}\right)^{1/2}}^{2},
\end{align}
The success probability of the NHPA, Eq.~(\ref{succprobAmp}), is a function of the amplification parameters $g,n$ and the final displacement $\beta$, as well as the states' amplitude $\alpha$ and has to be compared with the standard Kennedy scheme, \eqq{kennedy}, whose 
success probability can be recovered from the same expression by simply setting $g=1$, i.e., no amplification. Now, for a fixed value of the input amplitude $\alpha$, we may optimize \eqq{succprobAmp} with respect to three parameters: the amplifier's gain $g$, its internal cutoff degree $n$ and the displacement value $-\beta$. We thus define
\begin{align}\label{nhpaDef}
&\mathbb{P}_{succ}^{nhpa}(\alpha)=\max_{\beta,g,n}P_{succ}(\alpha,\beta,g,n),\\
\label{kenDef}&\mathbb{P}_{succ}^{(oken)}(\alpha)=\max_{\beta}P_{succ}(\alpha,\beta,1,n)\text{ for any } n.
\end{align}
In order to discuss the optimal case, let us observe some facts and properties of \eqq{succprobAmp}: 
\begin{enumerate}
\item Since we are treating a two-state discrimination, the best choice of the final displacement is obviously to be aligned to $\alpha$ in phase space and, since the latter can be chosen real w.l.o.g., $\beta\in\mathbb{R}$ too; 
\item We expect the success probability to be maximum in the region $\beta<0$, since in that case $p_{0|+}$ is the overlap between an output state of the NHPA similar to Fig.~\ref{Wigner} and a coherent state centered on the negative real axis, where the former has a negative quasi-probability; 
\item If we set $\beta=0$, thus applying no further displacement after the NHPA, then the behavior of the success probability is trivially upper bounded by that of the unoptimized Kennedy receiver, i.e., $P_{succ}(\alpha,0,g=1,n)$;
\item For $n=1$ and $\beta<0$ the success probability is a decreasing function of the gain, thus the optimal choice is not to amplify and the scheme simply resorts to the optimized Kennedy detector. 
\end{enumerate}
The last two points are proved in Appendix~\ref{app:altDecs}.
In light of the previous observations, we then study the function \eqq{succprobAmp} for $n=2$ and $\beta<0$ and observe the existence of a whole range of non-trivial values of the gain which increase the NHPA detector success probability above that of the optimized Kennedy detector. In particular, at $\alpha\simeq 0.32$, already for $g=3$ we have an increase of $\sim1.26\%$ with respect to \eqq{kenDef}. The maximum increase of $\sim 1.88\%$  is attained for $\alpha\simeq 0.29$ and $g\sim33$, while for $g\rightarrow\infty$ it lowers to $\sim 1.87\%$. 
\begin{figure}[t!]\center
\includegraphics[trim={0 0cm 0 0cm},clip,scale=.43]{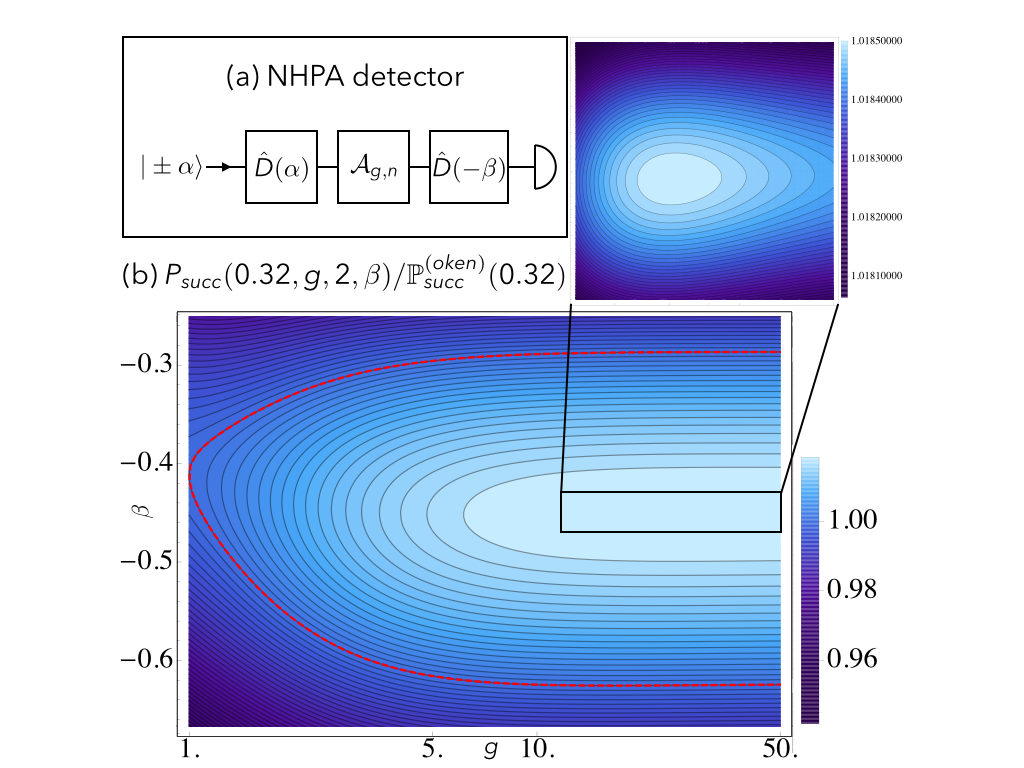}
\caption{(a) Schematic depiction of the NHPA receiver, where $\hat{D}(\alpha)$ is the first coherent-state displacement, $\mathcal{A}_{g,n}$ is the non-linear NHPA and $\hat{D}(-\beta)$ is the final optimized displacement, followed by an ordinary photon detector. The optimized Kennedy scheme is obtained by setting $g=1$, which amounts to removing $\mathcal{A}_{g,n}$. Other schemes described in the text rely on the substitution of $\mathcal{A}_{g,n}$ with simplified versions and/or on the introduction of an optimized squeezing operation $\hat{S}(r)$ before the final displacement. (b) Contour plot of the ratio $P_{succ}(\alpha,g,n,\beta)/{\mathbb P}_{succ}^{(oken)}(\alpha)$ between the 
success probabilities Eq.~(\ref{succprobAmp}) of the NHPA receiver and the one associated with the optimal Kennedy scheme, \eqq{kenDef}, for $n=2$ and $\alpha=0.32$ as a function of the gain $g$ and of the displacement $\beta<0$ for input amplitude (the optimal value of displacement for the Kennedy scheme being  $\beta^{op}_{Ken}\simeq 0.41$).  The red dashed line indicates  the points for which $P_{succ}(\alpha,g,n,\beta)={\mathbb P}_{succ}^{(oken)}(\alpha)$. The inset shows the optimal region for the NHPA receiver, in the range $\beta\in[-0.47,-0.43]$, $g\in[15,100]$.
 }\label{contourplot}
\end{figure} 
 
 \subsection{The partial-dephaser receiver and its cavity implementation}\label{sec3}
It is of primary importance to stress that replacing  $\mathcal{A}_{g,n}$  with an ordinary parametric amplifier or, more generally, a phase-insensitive Gaussian channel would not provide  the advantages  reported in the previous section: indeed in this case  the success probabilities would be worse than those attainable with the optimized Kennedy detector, as shown in  Appendix~\ref{app:altDecs}. Therefore the mere amplification of the incoming signals cannot account for the improved performance of the NHPA receiver. Still, as noted in the previous subsection by numerical analysis of the case $n=2$,  we observe that, fixing $\beta$ and $\alpha$, the advantage gained from the application of $\mathcal{A}_{g,2}$ is 
 almost optimal in the infinite-gain limit, i.e., for $g\rightarrow\infty$ (see Fig.~\ref{contourplot}b). 
In this regime the Kraus operators which define the action of the NHPA reduce to simple projectors 
  on the space of $2$ or more photons and its complementary, i.e.,
  \begin{equation}
  \hat{M}_{S}\rightarrow \hat{\Pi}_{\geq 2}=\sum_{k\geq 2}\dketbra{k},\quad\hat{M}_{F}\rightarrow \hat{\Pi}_{< 2}=\dketbra{0}+\dketbra{1}.
  \end{equation}
  Accordingly, $\mathcal{A}_{\infty,2}$ reduces to a partial dephasing channel that selectively removes the coherence among  such spaces while preserving any other form of quantum coherence in the system. This observation, along with the fact that, to our knowledge, the NHPA in the form \cite{Caves,OptimalNLA} has not been experimentally demonstrated yet, 
   leads us to replace $\mathcal{A}_{\infty,2}$ of Fig.~\ref{contourplot}a with a simplified version of such a map for which we propose a possible implementation. 
Specifically we consider a partial dephasing (PD) CPTP channel  ${\cal D}_n$  that is more destructive than  $\mathcal{A}_{\infty,n}$ as it only preserves coherence into the space formed by the first $n-1$ Fock states while inducing dephasing on the remaining ones, its
Kraus operators being hence
\begin{equation}
\hat{M}_{N}=\hat{\Pi}_{<n}=\sum_{k<n}\dketbra{k},\quad\hat{M}_{k}=\dketbra{k}\quad\forall k\geq n.
\end{equation} 
In Fig.~{\ref{kenCompare} we have tested the performance of this new detection scheme for the case $n=2$ observing only  a tiny decrement (of order $O(10^{-4})$)  of the associated
success probability  with respect to the one obtained in  the NHPA case, proving hence that, rather counterintuitively, the dephasing transformation  ${\cal D}_n$ is a useful resource for the discrimination problem we are considering here. \\
An implementation of ${\cal D}_2$ can be obtained by 
sending the incoming coherent signal $|\alpha\rangle$ inside an optical cavity, coupled to a two-level atom initialized into its ground state $\ket{G}$, via the Jaynes-Cummings hamiltonian \cite{nChuangBOOK}
\begin{equation}\label{ham}
\hat{H}_{JC}=\omega \hat{N}+\gamma(\hat{a}^{\dag}\hat{\sigma}_{-}+\hat{a}\hat{\sigma}_{+}),
\end{equation} 
where $\hat{N}=\hat{a}^{\dag}\hat{a}+\hat{Z}/2$ is a first integral of motion, obtained by combining the bosonic creation and annihiliation operators $\hat{a}^{\dag}$, $\hat{a}$ of the electromagnetic field inside the cavity and the atomic energy operator $\hat{Z}=\dketbra{E}-\dketbra{G}$, with $\ket{E}$  the excited atomic state; $\omega$ is the frequency of both the cavity and the atom at resonance; finally, the second term entering $\hat{H}_{JC}$ represents the cavity-atom coupling of strength $\gamma$, where 
\begin{equation}
\hat{\sigma}_{+}=\ketbra{E}{G},\quad\hat{\sigma}_{-}={\hat{\sigma}_{+}}^{\dag}=\ketbra{G}{E}
\end{equation}
are the atomic operators describing respectively excitation and decay of its quantum state. In Appendix~\ref{app:cavity} we provide a full treatment of the cavity implementation briefly discussed in the following . To induce the transformation ${\cal D}_2$,  we first let the coherent signal and the atom interact for a time $\tau$ chosen in such a way to induce a perfect Rabi oscillation. This guarantees that the joint cavity-atom state $\ket{1,G}$ of the input superposition $\ket{\alpha,G}$ is transferred to $\ket{0,E}$, thus encoding the zero- and one-photon-number cavity space in the atomic levels. This happens for the first time at $\tau=\pi/(2\gamma)$, leaving the system in the joint state
\begin{equation}
|\psi_{RABI}\rangle = e^{-\frac{\abs{\alpha}^{2}}{2}}( \ket{0,G}-i\alpha e^{-i\omega \tau}\ket{0,E} + |\Delta \rangle),   \label{firstRabi} 
\end{equation}
with $|\Delta \rangle$ a combination of terms which, on the optical part, posses at least one photon excitation, see Appendix~\ref{app:cavity}. 
 Next we abruptly decouple the two systems (say detuning the atom energy gap  with respect to the cavity frequency) while inducing a random perturbation on the cavity wavelength. Alternatively,
 we may assume  the optical signal to emerge from the cavity and to be  fed into a long waveguide that dephases the various Fock components of the propagated signals. 
 In both cases the net effect on  $|\psi_{RABI}\rangle$ can be described as an application of the operator $\hat{U}_{p}(\theta)$, \eqq{ps}, with $\theta$ being a random parameter we have to average over: the cavity states $\ket{k}$ simply acquire a phase factor $\exp(-i\theta k)$ but no phase will be added to the first two components of the global state \eqref{firstRabi}, containing the superposition we want to preserve.
After this stage, we apply a second Rabi oscillation in order to bring back the preserved atomic superposition onto the cavity states (e.g. by abruptly restoring the atom-field resonance condition or by
feeding the traveling signal back to the cavity). We describe this process with the same Jaynes-Cummings hamiltonian and interaction time $\tilde{\tau}=3\pi/(2g)$. The output field of the cavity, obtained by tracing out the atomic state and averaging over the random phase, can be written as
\begin{equation} \label{finCav}
\begin{aligned}
\hat{\rho} &= e^{-\abs{\alpha}^{2}}\Big(\dketbra{\alpha_{T}}+\frac{\abs{\alpha}^{4}}{2}\left(D\dketbra{1}+\alpha E\ketbra{2}{1}+\text{h.c.}\right)  \\
&+\sum_{k=2}^{\infty}\frac{\abs{\alpha}^{2k}}{k!}\left(D_{k}(\alpha)\dketbra{k}+\alpha E_{k}(\alpha)\ketbra{k+1}{k}+\text{h.c.}\right)\Big),
\end{aligned}
\end{equation} 
where $\ket{\alpha_{T}}=\ket{0}+\alpha \exp(-2i\omega\tau)\ket{1}$ is the superposition we aimed at preserving, apart from a phase which can be dropped out either by fine tuning the working frequency $\omega$ or by earlier compensation of the input coherent states; the various $D$, $E$ coefficients assume non-trivial values, having fixed $\tau$ in order to favor the desired Rabi transitions, see Appendix~\ref{app:cavity}. The main imperfection of this implementation with respect to the desired partial dephaser ${\cal D}_2$ 
is that  the final state $\hat{\rho}$  actually preserves some extra coherence between adjacent photon-number states. Nevertheless, if we employ this device instead of the ideal dephaser in our receiver setup, the top increase of the success probability previously obtained with the NHPA at $\alpha\simeq0.29$ is reduced to $\sim1.67\%$, i.e., the top performance is reduced of only $\sim 0.21\%$, and in the whole range $\alpha\geq1$ the worst loss of performance is around $\sim 0.46\%$, see Fig. \ref{kenCompare} for a comparison. Moreover the loss is further reduced when introduced the device into a Dolinar scheme with few adaptive steps (compare Fig. \ref{kenCompare} and \ref{dol2Step}).
 \begin{figure}[t!]\center
\includegraphics[trim={0 2.5cm 0 2.2cm},clip,scale=.43]{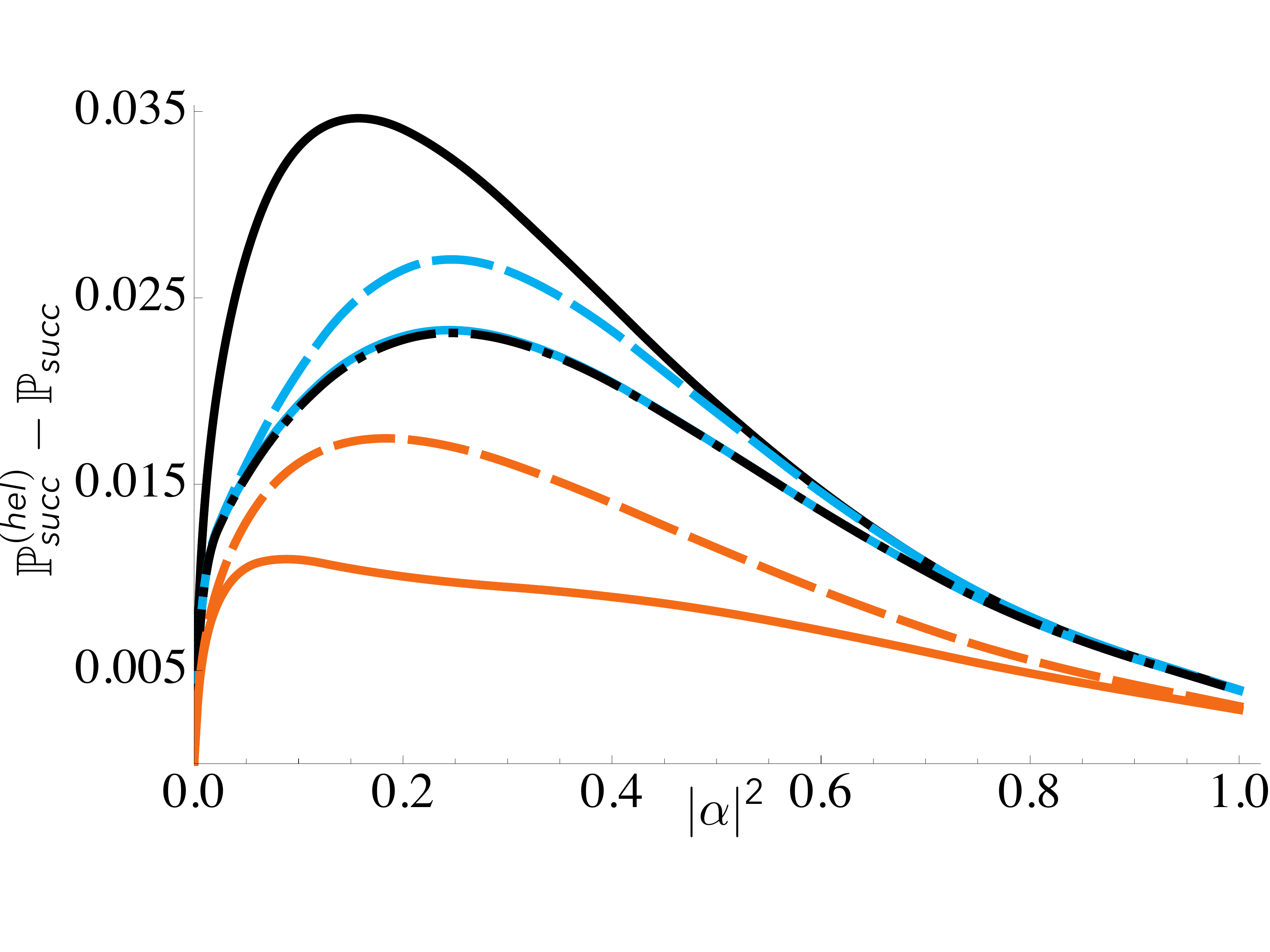}
\caption{Plot of the difference between the success probability of Helstrom, complementary to \eqq{helBound}, and the one of several Kennedy-like receivers, as a function of the input states' average photon number~$\abs{\alpha}^{2}$: optimized Kennedy scheme (black solid line), ${\cal A}_{\infty,2}$ scheme (black dot-dashed line), ${\cal D}_{2}$ scheme (cyan/light gray solid line), cavity implementation  (cyan/light gray dashed line), Takeoka-Sasaki (TS) scheme~\cite{opGaussDet} with squeezing and displacement (black dashed line), ${\cal A}_{\infty,2}$ plus TS scheme (orange/dark gray dashed line), ${\cal A}_{\infty,3}$ plus TS scheme (orange/dark gray solid line). 
 }\label{kenCompare}
\end{figure}

\subsection{Combination with other decoding schemes} \label{sec4}
A possible extension of our scheme can be obtained along the line proposed by Takeoka and Sasaki (TS)~\cite{opGaussDet}, where a standard Kennedy scheme has been improved by adding a single-mode squeezing operation $\hat{U}_{sq}(-r)$, \eqq{1msq},  before the last displacement transformation, 
resulting in an overall active (non phase-insensitive) Gaussian unitary before photon detection. 
When we apply the same method to our scheme, the probabilities~(\ref{OVERLAP1}) and (\ref{OVERLAP}) get replaced by
\begin{align}\label{sqAmp1}
p_{0|-}&=\abs{\braket{\beta,-r}{0}}^{2}, \\
\label{sqAmp2} p_{0|+}&=\bra{\beta, r} \mathcal{A}_{g,\infty}(\dketbra{2\alpha})\ket{\beta,r}=\abs{\sum_{k\geq n}\braket{2\alpha}{k}\braket{k}{\beta,r}}^{2}+\abs{\sum_{k> n}\braket{2\alpha}{k}\braket{k}{\beta,r}}^{2}
\end{align} 
where  $\ket{\beta,r}=\hat{U}_{sq}(r)\hat{D}(\beta)\ket{0}$ is a squeezed-displaced state and we considered directly the infinite-gain limit $\mathcal{A}_{g,\infty}$ where the NHPA behaves as a partial dephaser. The factor $\braket{2\alpha}{k}$ is straightforward to compute through \eqq{coherentstates}, while for $\braket{k}{\beta,r}$ we can apply the definitions Eqs.~(\ref{1msq},\ref{squeezedstates}) and its explicit expression is given in Appendix~\ref{app:altDecs}. We can optimize the success probability associated to Eqs.~(\ref{sqAmp1},\ref{sqAmp2}) with respect to both $\beta$ and $r$. At variance with the passive-Gaussian scheme, setting $n=2$ here turns out to be optimal only in the low-amplitude region $\abs{\alpha}^{2}\lesssim 0.1$. For higher amplitude values, an amplifier cutoff $n=3$ attains instead the most satisfying results, clearly surpassing all other receivers in performance (Fig. \ref{kenCompare}); this is probably due to the fact that the squeezing operation requires additional coherence terms between the zero and one-photon space and the two-photon- one in order to be properly optimized. In particular the optimal value of $r$ is negative, implying hence a squeezing of the $\hat{p}$ quadrature.\\

Since all the proposed detection schemes only require the insertion of an additional operation in a Kennedy-like receiver, it seems reasonable to study their extension to a Dolinar-like one, Sec.~\ref{subsec:discPract}. Accordingly, we now preliminary map
the input coherent states $\ket{\pm\alpha}$ on $N$ low-amplitude copies $\ket{\pm\alpha/\sqrt{N}}^{\otimes N}$ which we probe in sequence exploiting the information acquired at each stage to optimize the parameters (e.g. displacement, amplification, cutoff) of the detection that follows. 
In Fig.~\ref{dol2Step} we show the success probability of Dolinar-like detection schemes, for the simplest case $N=2$, taking the Helstrom bound as a reference. As it may be expected, since the proposed schemes outperform the Kennedy receiver, they also outperform the Dolinar one. The inset shows the same quantity as a function of the number of steps $N$ and at fixed amplitude. 
 \begin{figure}[t!]\center
\includegraphics[trim={0 3cm 0 2.6cm},clip,scale=.43]{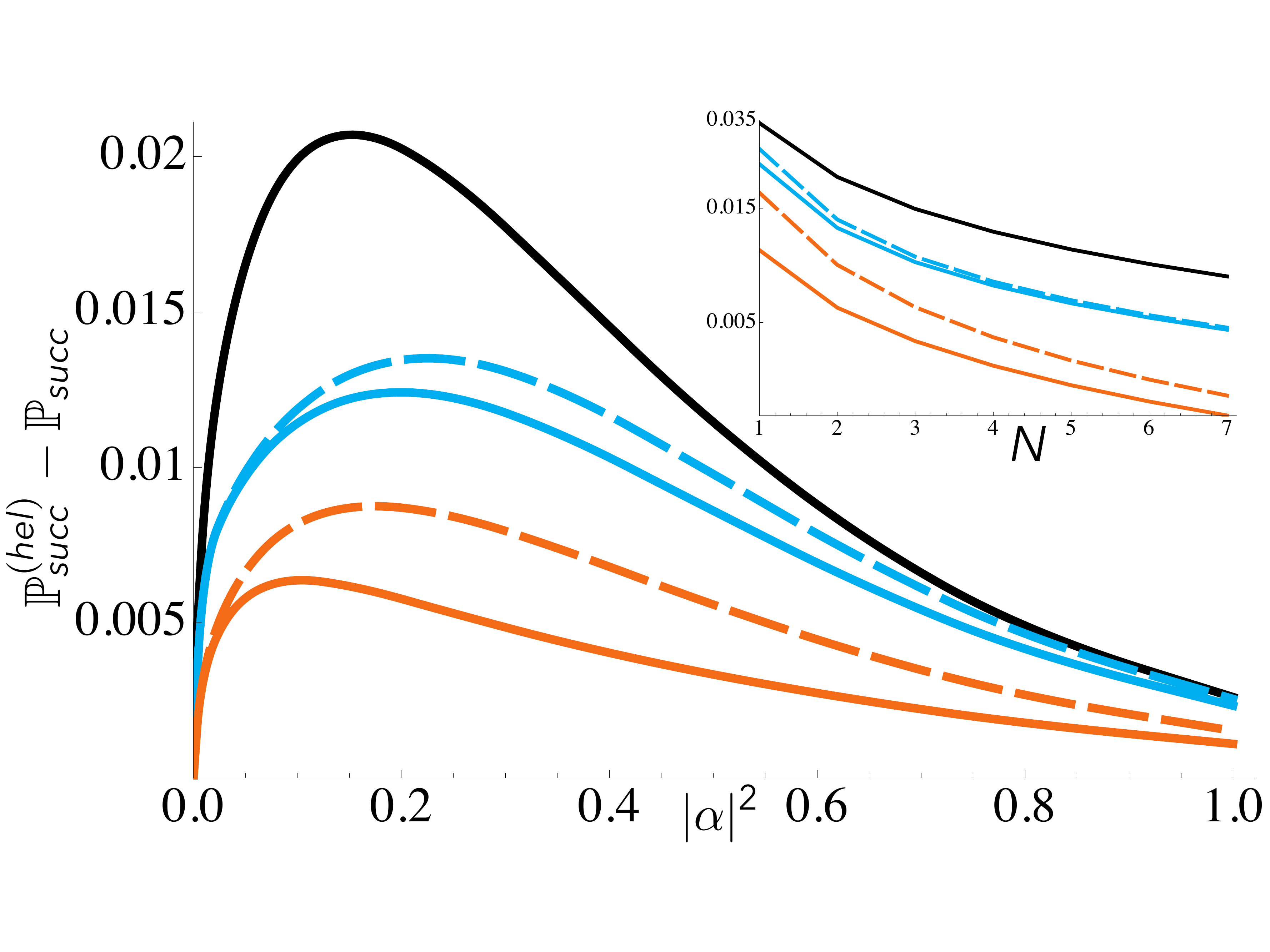}
\caption{Plot of the difference between the success probability of Helstrom and that of several Dolinar-like protocols, as a function of the input states average photon number~$\abs{\alpha}^{2}$, for two multiplexing steps: simple Dolinar scheme with Gaussian optimized displacement (black solid line), ${\cal D}_2$ scheme (cyan/light gray solid line), cavity implementation  (cyan/light gray dashed line), TS scheme, with Gaussian squeezing and displacement (black dashed line), ${\cal A}_{\infty,2}$  plus TS scheme (orange/dark gray dashed line),  ${\cal A}_{\infty,3}$ plus TS scheme (orange/dark gray solid line). 
The inset shows the same quantity (log-scale) as a function of the number of steps $N$ at fixed $\abs{\alpha}^{2}=0.2$, for the same protocols as in the main picture (only 
the case ${\cal A}_{\infty,3}$ is plotted for the dephaser plus TS scheme).}\label{dol2Step}
\end{figure}
 
\section{Multi-phase Hadamard receivers for communication on lossy bosonic channels}\label{sec:had}
In this section we present a practical scheme for transferring classical information over a lossy bosonic channel, Sec.~\ref{gaussChMeas}, by generalizing the proposal by Guha~\cite{guha1}, which we call the Hadamard receiver. These receivers are conceived in order to bridge the gap between single-mode practical detectors, Secs.~\ref{subsec:discPract},\ref{sec:cohDisc}, and multi-mode capacity-achieving optimal measurements, Secs.~\ref{subsec:commQ},\ref{sec:bisDec}, and they are among the first examples of structured joint-detection receivers for optical communications realizable in principle with current technology \cite{banaszek}. The scheme~\cite{guha1} relies on building a codebook (Hadamard code in the following) of $N$ separable codewords of length $N$ from the binary coherent alphabet $\ket{\pm\alpha}$ and transforming them at the receiver by a passive unitary gate $\hat{U}_{had}^{(N)}$  represented by a Hadamard matrix~\cite{Hadamard}. The output is a Pulse-Position-Modulation (PPM) code, which is easily read out by employing single-mode quantum state discrimination techniques, e.g., OOD. In the low-energy regime the scheme surpasses the   information rate obtained by separable detection techniques, jointly reading out larger and larger codewords as the energy gets lower. Further improvements can also be gained~\cite{guha2} by adding a second copy of the original Hadamard  codebook obtained from the latter by simply flipping the sign of the amplitudes of all the coherent states that form its codewords:
the new states behave as before under $\hat{U}_{had}^{(N)}$, producing two (phase-shifted) copies of the initial PPM code that can still be read out  by means of an adaptive Dolinar receiver, Sec.~\ref{subsec:discPract}. Building on these observations here  we analyze in detail the case of a codebook formed by 
 $M$  phase-shifted copies of the Hadamard code (the case $M=2$ corresponding to the model discussed in~\cite{guha2}): we call this a 
 Phase-Shift-Keying (PSK) Hadamard code of order $M$. 
 First, we compute the optimal information rate of this code by evaluating its associated Holevo information, \eqq{holevoInfo}, achievable with optimal joint-detection decoders of the kind discussed in Secs.~\ref{subsec:commQ},\ref{sec:bisDec}. We show that, for all choices of the number of phase modulations $M\geq 2$ and of the codewords' length $N$,  the optimal rate saturates the classical capacity of the channel, Theorem~\ref{ccapPI}, in the low-energy region. Accordingly, in such regime the PSK Hadamard codes are optimal  and any sub-performance resulting from their use is only a consequence of a lack of efficiency at the detection stage.   Next we  compute  the rates attainable when detecting PSK Hadamard codes with a modification of the Hadamard receiver of~\cite{guha2}, which we call PSK Hadamard receiver. It relies on the discrimination  of multiple symmetric coherent states of fixed amplitude, Sec.~\ref{sec:disc}, applied to the output PPM codewords. We then show that this practical Hadamard receiver, while suboptimal, allows for improvements of the transmission rate of up to $\sim6\%$ in an intermediate energy regime with respect to previous proposals~\cite{guha1,guha2}.

\subsection{PSK Hadamard codes}\label{had:sec2}
Here we formalize the  Hadamard  code proposed in~\cite{guha1} and its $M$-phase generalization. 
Suppose we want to transfer classical information through $N$ modes of the electromagnetic field, or $N$ distinct temporal pulses on a single field mode, traveling through a communication medium, e.g., an optical fiber or free space~\cite{cavesDrum}. The transmission line can be approximately represented by a quantum-limited PI lossy bosonic channel of loss $\eta$, see Sec.~\ref{gaussChMeas}.
As discussed in Theorem~\ref{ccapPI}, the classical capacity of this channel at constrained energy can be achieved by a simple Gaussian coherent-state source, \eqq{opSource}. Ultimately this is possible because when we encode information in the optical coherent states of the field, $\ket{\alpha}$, 
they reach the receiver-end of the channel as attenuated versions $\ket{\sqrt{\eta}~\alpha}$, while keeping their original purity.
Since the net action of the channel on these states is a rescaling of energy by $\eta$, we may take as a reference the \textit{received} energy without loss of generality, which is equivalent to setting $\eta=1$ in the following. \\
Let us now construct the Hadamard code that, as for the optimal Gaussian source, employs sequences of coherent states as codewords. At variance with the latter however such sequences are selected by extracting states from a binary coherent alphabet  $\ket{\pm\alpha}$,  with amplitude  $\alpha$
matching the average-energy constraint of the channel, i.e., $|\alpha|^2 = E$.   
The number of different $N$-long strings that we can generate in this way is equal to $2^N$ (specifically they are the vectors 
$\ket{\pm \alpha}_{0}\otimes\cdots\otimes\ket{\pm\alpha}_{N-1}$ where $\ket{\cdot}_{j}$ represents the state of the $j$-th communication mode). 
  Yet in constructing the Hadamard code we only select $N$ specific ones of them
 by picking those sequences  whose signs reproduce the columns of a Hadamard matrix~\cite{Hadamard}.  For instance for $N=2$  the Hadamard matrix, up to a normalization, is equal to 
 \begin{equation}
 H_{2}=\left(\begin{array}{cc}+1&+1\\+1&-1\end{array}\right);
 \end{equation} 
 therefore  from the set of four possible states we select $\ket{v_{0}(\alpha)}=\ket{+\alpha}_{0}\otimes\ket{+\alpha}_{1}$ and $\ket{v_{1}(\alpha)}=\ket{+\alpha}_{0}\otimes\ket{-\alpha}_{1}$ as the  first and second codeword of our Hadamard codebook. 
For arbitrary $N$, we recall that a Hadamard matrix of order $N=2^{i}$ with integer $i\geq0$, is a $N\times N$ matrix $H_{N}$ of elements $\pm 1$,  orthogonal up to a scaling factor of $\sqrt{N}$ and, for simplicity, symmetric (a permutation of columns or rows is sufficient to satisfy this further requirement), i.e., 
\begin{align}
&H_{N}H_{N}^{T}=H_{N}^{T}H_{N}=N\mathbf{1}_{N},\\
&H_{N}^{T}=H_{N}.
\end{align}
Such a matrix can be equivalently defined in terms of its elements as:
\begin{equation} \label{HadTerm}
(H_{N})_{j,k}=(-1)^{j\cdot k},\quad j\cdot k=\sum_{t=0}^{\log_{2}N}j_{t}k_{t},
\end{equation}
where we have defined $j\cdot k$ as the bitwise scalar product of the binary representations of $j,k=0,\cdots,N-1$.
Then a Hadamard code ${\cal H}_1^{(N)}(\alpha)$ of length $N$ and average mode-energy $E$ comprises the  $N$ codewords
\begin{equation} 
\ket{v_{k}(\alpha)}=\bigotimes_{j=0}^{N-1}\ket{(H_{N})_{j,k}~\alpha}_{j},
\end{equation} 
for all $k=0,\cdots,N-1$.
As shown in~\cite{guha1,guha2,banaszek}, this special set of states admits a passive unitary transformation $\hat{U}^{(N)}_{had}$, entirely implementable through passive interferometers, Sec.~\ref{gaussUni}, that transforms the received codewords $\ket{v_{k}(\alpha)}$ into equivalent PPM ones, i.e., 
\begin{equation}
\hat{U}^{(N)}_{had} \ket{v_{k}(\alpha)} = \ket{w_{k}(\alpha)}=\ket{\sqrt{N}~\alpha}_{k}\otimes\left(\bigotimes_{j\neq k}\ket{0}_{j}\right), \label{TRANS1}
\end{equation} characterized by a single high-energy pulse on one of the $N$ available modes, the vector $|0\rangle_j$ representing the vacuum state of the $j$-th mode. It is easy to show that the operator $\hat{U}^{(N)}_{had}$ has symplectic representation $H_{N}/\sqrt{N}$. Indeed we can express any received state as $\ket{v_{k}(\alpha)}=\hat{D}\left(\alpha~\underline{h}_{k}\right)\ket{0}^{\otimes N}$, with $\underline{h}_{k}=\left(\left(H_{N}\right)_{1,k},\cdots,\left(H_{N}\right)_{n-1,k}\right)^{T}$ and when applying the passive unitary we have:
\begin{equation}
\begin{aligned}\label{Transform}
\hat{U}^{(N)}_{had}\ket{v_{k}(\alpha)}&=\hat{D}\left(\frac{\alpha}{\sqrt{N}} H_{N}\underline{h}_{k}\right)\ket{0}^{\otimes n}=\hat{D}\left(\sqrt{N}\alpha(0,\cdots,0,1_{(k)},0,\cdots,0)\right)\ket{0}^{\otimes n}\\
&=\ket{w_{k}(\alpha)},
\end{aligned}
\end{equation}
where the second equality follows from the orthogonality property of $H_{N}$.
Note that the average received energy per mode $E$ is conserved by the transformation, although its total amount for $N$ modes is concentrated on a single pulse of energy $\mathcal{E}=NE$, whose position varies from codeword to codeword. Accordingly, the classical messages encoded into the elements of ${\cal H}_1^{(N)}(\alpha)$ can now be recovered by means of a readout strategy capable of detecting where such concentrated energy lies, e.g. OOD or a Dolinar receiver aimed at discriminating the coherent state $|\sqrt{N} \alpha\rangle$ from the vacuum  $|0\rangle$.\\

A refined version of the receiver proposed in~\cite{guha2} employs an enlarged code, obtained by adding those codewords $\ket{v_{k}(-\alpha)}$ with signs opposite to the ones in $\ket{v_{k}(\alpha)}$. More generally, we define a PSK Hadamard code of order $M$ 
by adding $M$ copies of  the Hadamard code, characterized by $M$ equally spaced phases  of the amplitude $\alpha$, i.e., the set 
\begin{equation} \label{PSKHAD} 
{\cal H}^{(N)}_M(\alpha) = \bigcup_{m=0}^{M-1}\; {\cal H}^{(N)}_1(\alpha_{m})=\left\{\ket{v_{k}(\alpha_{m})}:\alpha_{m}=e^{i\frac{2\pi}{M}m}\alpha\right\}_{k=0,m=0}^{N-1,M-1}.
\end{equation} 
 Since the property \eqq{TRANS1} holds for any $\alpha\in\mathbb{C}$, the elements of ${\cal H}^{(N)}_M(\alpha)$ get
 transformed by $\hat{U}^{(N)}_{had}$ in corresponding phase-shifted PPM codewords, i.e., the vectors
 \begin{equation}
\hat{U}^{(N)}_{had} \ket{v_{k}(\alpha_m)} = \ket{w_{k}(\alpha_m)}\;. \label{TRANS2}
\end{equation}
As we shall discuss explicitly in Sec.~\ref{secHadRate}, the readout of the input messages can benefit from this effect, the idea being to first
determine the value of $k$ by checking in which of the $N$ modes the energy has been concentrated, and then determine $m$ by using a single-mode detection scheme to read out the phase of $\alpha_m$.

\subsection{Optimal communication rates of PSK Hadamard codes}\label{secHadCap}
In this section we compute the optimal communication rate $R^{(N,M)}_{opt}(E)$ of a PSK Hadamard code ${\cal H}^{(N)}_M(\alpha)$ of order $M$ and average energy per mode  $E=|\alpha|^2$. This  is proportional to the Holevo information, \eqq{holevoInfo}, of the code itself, i.e.,
\begin{equation}\label{opRateHad}
R^{(N,M)}_{opt}(E)=\oneover{N}\chi\left(\left\{\ket{v_{k}(\alpha_{m})},\oneover{N}\right\}_{k=0,m=0}^{N-1,M-1}\right)=\oneover{N}S(\hat{\rho}_{had}),
\end{equation}
 where the average state of the code reads out
\begin{equation}
\hat{\rho}_{had}=\sum_{k=0}^{N-1}\sum_{m=0}^{M-1}\frac{\dketbra{v_{k}(\alpha_{m})}}{M N}
\end{equation}
and the Holevo information reduces to the entropy of such states because the code comprises only pure states.
 The quantity $R^{(N,M)}_{opt}(E)$ is the maximum rate of bits one can convey over the channel with the code  ${\cal H}^{(N)}_M(\alpha)$ with an optimal joint-decoding procedure, see Secs.~\ref{subsec:commQ},\ref{sec:bisDec}. To compute it we find it useful to exploit the unitary mapping Eq.~(\ref{TRANS2}). Accordingly, the eigenvalues, and hence the entropy, of $\hat{\rho}_{had}$ coincide with those
of the density matrix 
\begin{equation}\label{ENTROPY} 
\hat{\rho}_{PPM}= \sum_{k=0}^{N-1} \sum_{m=0}^{M-1}\frac{\dketbra{w_{k}(\alpha_{m})}}{M N}= \frac{1}{N} \sum_{k=0}^{N-1}\hat{\rho}^{loc}_{k}\otimes\left(\bigotimes_{j\neq k}\dketbra{0}_{j}\right), 
\end{equation} 
with
\begin{equation} 
\hat{\rho}^{loc}_{k}=\frac{1}{M} \sum_{m=0}^{M-1}\dketbra{\sqrt{N}\alpha_{m}}_{k}.
\end{equation} 
By construction the state $\hat{\rho}_{PPM}$  has non-zero support only on a $M N$-dimensional subspace of the whole Hilbert space, spanned by the linearly independent states  $\ket{w_{k}(\alpha_{m})}$: we compute the entropy of $\hat{\rho}_{PPM}$ by finding its spectrum in this subspace.
We proceed in two steps: first, for all $k$, we diagonalize $\hat{\rho}_{k}^{loc}$, then we diagonalize the resulting multimode superposition, which turns out to be already partially diagonal. The first step has already been carried out in~\cite{multiHel2}, when computing the optimal probability of discriminating the cyclic-symmetric source $\{(|\alpha_m\rangle,N^{-1})\}_{m=0}^{M-1}$, see also Secs.~\ref{subsec:discTheo},\ref{secHadRate}. Accordingly, we can write
\begin{equation}  \label{QUESTA}
\hat{\rho}_k^{loc}=\sum_{\ell=0}^{M-1}\frac{\lambda_{\ell}(\mathcal{E})}{M}\dketbra{d_{\ell}}_{k}, 
\end{equation} 
where the eigenvectors
\begin{equation}\label{mpskEigenvectors}
\ket{d_{\ell}}_k=\sum_{m=0}^{M-1}\frac{e^{-i\frac{2\pi}{M}\ell m}}{\sqrt{M\lambda_{\ell}(\mathcal{E})}}\ket{\sqrt{N}\alpha_{m}}_k\;,
\end{equation}
are obtained by Fourier transform of the original pulses $\ket{\sqrt{N} \alpha_{m}}_k$ and the eigenvalues
\begin{equation}\label{mpskEigen}
\lambda_{\ell}(\mathcal{E})=\sum_{h=0}^{M-1}\exp\left[-\left(1-e^{i\frac{2\pi}{M}h}\right)\mathcal{E}-i\frac{2\pi}{M}\ell h\right]
\end{equation}
depend on the fixed energy of the states, $\mathcal{E}=N E$ in our case.
It is important to note that the eigenvectors $\ket{d_{\ell}}_k$ for all allowed $\ell>0$ share the peculiar property of having a zero overlap with the vacuum state, i.e., 
\begin{equation}\label{OVER}
\braket{0}{d_{\ell}}_k=\delta_{\ell, 0}~e^{-\mathcal{E}/2}\sqrt{\frac{M}{\lambda_{0}(\mathcal{E})}}. 
\end{equation} 
This is due to the fact that the overlap of any PSK state with the vacuum depends only on its fixed energy and not on its specific phase, i.e., $\braket{0}{\sqrt{N}\alpha_{m}}=e^{-\mathcal{E}/2}$. Hence this overlap factors out of the sum in \eqq{mpskEigenvectors}, which then amounts to a simple discrete Fourier transform of a constant.
Replacing Eq.~(\ref{QUESTA}) into Eq.~(\ref{ENTROPY}) we can hence write 
\begin{equation}
\hat{\rho}_{PPM} = \sum_{\ell=0}^{M-1} \nu^{\ell}(\mathcal{E})  \sum_{k=0}^{N-1} \dketbra{e_{k}^{\ell}} \;, 
\end{equation} 
with 
\begin{align}
&\nu^{\ell}(\mathcal{E})={\lambda_{\ell}(\mathcal{E})}{M N},\\
&\ket{e_{k}^{\ell}}=\ket{d_{\ell}}_{k}\otimes\left(\bigotimes_{j\neq k}\ket{0}_{j}\right).\label{NEw1}
\end{align} 
These new multimode states share the PPM structure of the $\ket{w_{k}(\alpha_{m})}$ but with different single-mode pulses, one for each eigenvector of $\hat{\rho}^{loc}_k$. Most importantly, thanks to the property Eq.~(\ref{OVER}) they turn out to be partially orthogonal, i.e., 
\begin{equation}\braket{e_{k}^{\ell}}{e_{h}^{m}}=\begin{cases}
\braket{d_{\ell}}{d_{m}}_{k}=\delta_{\ell, m} &\text{if }k=h,\\
\braket{d_{\ell}}{0}_{k}\cdot\braket{0}{d_{m}}_{h}=\delta_{\ell, 0}\delta_{m, 0} \frac{M e^{-\mathcal{E}}}{\lambda_{0}(\mathcal{E})}&\text{if }k\neq h.
\end{cases}\label{eOverlaps}
\end{equation}
It is then advantageous to write   $\hat{\rho}_{PMM}$ as the sum of two operators, $\hat{\rho}_{PPM}=\hat{\rho}^{(0)}_{PPM}+\hat{\rho}^{(+)}_{PPM}$, with
\begin{align}\label{rho0}
\hat{\rho}_{PPM}^{(0)}&=\nu^{0}(\mathcal{E})
\sum_{k=0}^{N-1}\dketbra{e_{k}^{0}},\\ \label{rho+}
\hat{\rho}_{PPM}^{(+)}&=\sum_{\ell=1}^{M-1}\nu^{\ell}(\mathcal{E})\sum_{k=0}^{N-1}\dketbra{e_{k}^{\ell}}\;.
\end{align}
From Eqs.~(\ref{eOverlaps}-\ref{rho+}) we can deduce several facts: (i) $\hat{\rho}_{PPM}^{(0)}\hat{\rho}_{PPM}^{(+)}=\hat{\rho}_{PPM}^{(+)}\hat{\rho}_{PPM}^{(0)}=0$, i.e., the two operators are disjoint; (ii) the set of states $\ket{e_{k}^{\ell>0}}$ is an orthonormal basis of the space spanned by $\hat{\rho}^{(+)}_{PPM}$, of dimension  $(M-1)n$; (iii) $\hat{\rho}_{PPM}^{(+)}$ is diagonal in this basis, with eigenvalues $\nu^{\ell>0}(\mathcal{E})$ of multiplicity $N$ each. 
Accordingly, to determine the spectral structure of $\hat{\rho}_{PPM}$ we have to diagonalize $\hat{\rho}^{(0)}_{PPM}$. This can be carried out by diagonalizing the Gram matrix \cite{helstromBOOK,multiHel2} of the codewords $\ket{e_{k}^{0}}$, whose average state is proportional to $\hat{\rho}^{(0)}_{PPM}$ itself. Indeed define the Gram matrix of a set of $N$ codewords as the $N\times N$ hermitian matrix whose entries are the overlaps between the desired codewords, i.e., $(\Gamma^{(e)})_{kh}=\braket{e_{k}^{0}}{e_{h}^{0}}$ in our case. Since such matrix is diagonalizable, there exist a diagonal matrix $D=\operatorname{diag}(\mu_{0},\cdots,\mu_{N-1})$ and a unitary matrix $V$ that satisfy $D=V\Gamma^{(e)}V^{\dagger}$, with $\mu_{j}$ the eigenvalues of $\Gamma^{(e)}$. Then we can construct a new set of codewords 
\begin{equation} \label{fStates}
\ket{f_{j}}=\frac{1}{\sqrt{\mu_{j}}}\sum_{k=0}^{N-1}(V_{jk})^{*}\ket{e_{k}^{0}}
\end{equation} 
with two peculiar properties: (i) they form an orthonormal basis of the space spanned by the $\ket{e_{k}^{0}}$, since 
\begin{equation}
\braket{f_{j}}{f_{i}}=\frac{1}{\sqrt{\mu_{j}\mu_{i}}}\sum_{k,h=0}^{N-1}V_{jk}(\Gamma^{(e)})_{kh}(V^{\dagger})_{hi}=\delta_{j,i};
\end{equation}
 (ii) $\hat{\rho}^{(0)}_{PPM}$ is diagonal in this basis. The latter property can be easily shown by inverting the definition  \eqq{fStates}, i.e., $\ket{e_{k}^{0}}=\sum_{j=0}^{N-1}V_{jk}\sqrt{\mu_{j}}\ket{f_{j}}$, and inserting it in \eqq{rho0} to obtain 
\begin{equation} 
\hat{\rho}^{(0)}_{PPM}=\nu^{0}(\mathcal{E})\sum_{k,j,i=0}^{N-1}\sqrt{\mu_{j}\mu_{i}}V_{jk}(V_{ik})^{*}\dketbra{f_{j}}{f_{i}}=\nu^{0}(\mathcal{E})\sum_{j=0}^{N-1}\mu_{j}\dketbra{f_{j}}.
\end{equation}  
We can easily show that the Gram matrix $\Gamma^{(e)}$ of the codewords $\ket{e_{k}^{0}}$, $k=0,\cdots,N-1$ is diagonalized by an orthogonal matrix $V=H_{N}/\sqrt{N}$, proportional to the Hadamard matrix of order $N$. Indeed consider that
\begin{equation}\label{gram}
(V\Gamma^{(e)}V^{T})_{\ell k}=\frac{1}{N}\sum_{i,j=0}^{N-1}\left(H_{N}\right)_{\ell i}(\Gamma^{(e)})_{ij}\left(H_{N}\right)_{jk}=\delta_{\ell,k}+\frac{M e^{-\mathcal{E}}}{n\lambda_{0}(\mathcal{E})}\sum_{i,j\neq i}\left(-1\right)^{\ell\cdot i+k\cdot j},
\end{equation}
where we have applied the definition \eqq{HadTerm} of $H_{N}$ and employed the fact that the off-diagonal terms of $\Gamma^{(e)}$ are all equal, as implied by \eqq{eOverlaps}. The last term in the latter equation can be simplified further by computing the sum
\begin{equation}
\sum_{j=0}^{N-1}(-1)^{k\cdot j}=\prod_{t=1}^{\log_{2}N}\sum_{j_{t}=0,1}(-1)^{k_{t}j_{t}}=\prod_{t=1}^{\log_{2}N}\left(2\delta_{k_{t},0}\right)=N\delta_{k,0}.
\end{equation}
We have then:
\begin{equation}
\begin{aligned}
\sum_{i,j\neq i}\left(-1\right)^{\ell\cdot i+k\cdot j}&=\sum_{i=0}^{N-1}(-1)^{\ell\cdot i}\left(N\delta_{k,0}-(-1)^{k\cdot i}\right)=N^{2}\delta_{\ell,0}\delta_{k,0}-\prod_{t=1}^{\log_{2}N}\left(2\delta_{\ell_{t}+k_{t},0}\right)\\
&=N^{2}\delta_{\ell,0}\delta_{k,0}-N\delta_{\ell,k},
\end{aligned}
\end{equation}
where in the last equality we have used the fact that the binary sum $\ell_{t}+k_{t}$ can be zero also if both bits are equal to one, so that $\delta_{\ell_{t}+k_{t},0}=\delta_{\ell_{t},k_{t}}$. Inserting this expression in \eqq{gram} we obtain:
\begin{equation}
(V\Gamma^{(e)}V^{T})_{\ell k}=\delta_{\ell,k}\left(1+(N\delta_{\ell,0}-1)\frac{M e^{-\mathcal{E}}}{\lambda_{0}(\mathcal{E})}\right)=\delta_{\ell,k}\mu_{\lambda}(\mathcal{E}),
\end{equation}
confirming that $V=H_{N}/\sqrt{N}$ diagonalizes $\Gamma^{(e)}$, with eigenvalues
\begin{equation}\label{gramenvalues}
\mu_{\ell}(\mathcal{E})=1+(N\delta_{\ell,0}-1)\frac{M e^{-\mathcal{E}}}{\lambda_{0}(\mathcal{E})}.
\end{equation} 
Thus the average state \eqq{rho0} can be written in diagonal form as
\begin{align}\label{rho0Fin}
\hat{\rho}^{(0)}_{PPM}=\nu_{0}^{0}(\mathcal{E})\dketbra{f_{0}}+\nu_{+}^{0}(\mathcal{E})\sum_{k=1}^{N-1}\dketbra{f_{k}},
\end{align}
with eigenvalues 
\begin{equation}
 \nu_{0}^{0}(\mathcal{E})=\frac{1}{M N}[\lambda_{0}(\mathcal{E})+(n-1)M e^{-\mathcal{E}}],
\end{equation}  of multiplicity one and 
\begin{equation}
 \nu_{+}^{0}(\mathcal{E})=
\frac{1}{M N}[\lambda_{0}(\mathcal{E})-M e^{-\mathcal{E}}], 
\end{equation} 
of multiplicity $N-1$.
We are now able to compute the optimal rate \eqq{opRateHad}:
\begin{equation}\label{mpskHolCap}
\begin{aligned}
R^{(N,M)}_{opt}(E)&=\frac{S(\hat{\rho}_{PPM})}{N}=-\frac{1}{N}\Big(\nu_{0}^{0}\left(\mathcal{E}\right)\log_{2}\nu_{0}^{0}\left(\mathcal{E}\right)+(N-1)\nu_{+}^{0}\left(\mathcal{E}\right)\log_{2}\nu_{+}^{0}\left(\mathcal{E}\right) \\
&+N\sum_{\ell=1}^{M-1}\nu^{\ell}\left(\mathcal{E}\right)\log_{2}\nu^{\ell}\left(\mathcal{E}\right)\Big),
\end{aligned}
\end{equation}
where we recall that ${\mathcal{E}}=NE$. 
\begin{figure}[t!]\center
\includegraphics[trim={0 1.9cm 0 2cm},clip,scale=.43]{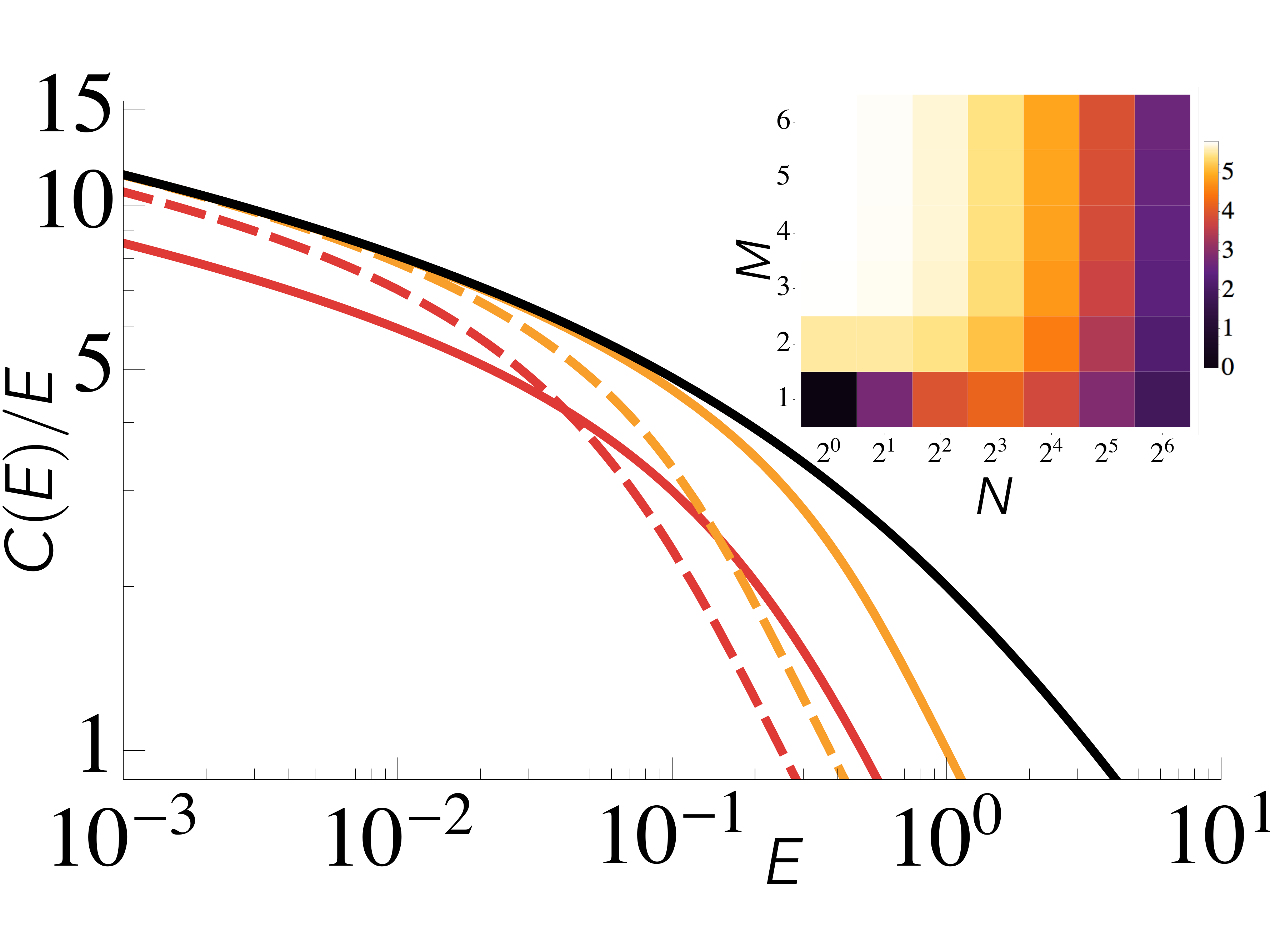}
\caption{Plot (log-log scale) of the optimal rate $R^{(N,M)}_{opt}(E)/E$ per average mode-energy of the PSK Hadamard code of order $M$,  for codewords of length $N=2, 2^{4}$ (respectively solid and dashed lines) and a number of phases $M=1, 4$ (respectively brown/dark-grey and orange/light-grey lines) and the classical capacity per average mode-energy, $C(E)/E$ (black solid line), as a function of the average mode-energy $E$. The quantity $R^{(N,M)}_{opt}(E)/E$ quickly increases at low energy values and for $M>1$ it even saturates the value of $C(E)/E$. For $N>1$ these features appear at lower energy. The inset shows $R^{(N,M)}_{opt}(E)/E$ at $E=0.05$, for several values of the couple $(N,M)$, each represented by a coloured tile. Darker colours are associated to lower values. The general behaviour on this plane seems to favour higher $M$ values at fixed $N$ and lower $N$ values at fixed $M$, except for $M=1$, where the capacity has its peak at $N^{*}(E)=2^{3}$. Also note that for $M>1$ and $N\lesssim N^{*}(E)$ the capacity is approximately constant and equal to its optimal value at that energy, i.e., the light-coloured region of the $(N,M)$ plane. }\label{f10}
\end{figure}  
In Fig.~\ref{f10} we plot $R^{(N,M)}_{opt}(E)/E$ (i.e., the value of  the optimal rate  of ${\cal H}^{(N)}_M(\alpha)$ per average mode-energy) for different values of $M$ and $N$. As a comparison we also report the value of the classical capacity of the channel, \eqq{ccapPI}, per average mode-energy, $C(E)/E$, that in this case is equal to:
\begin{equation} 
C(E)=(E+1)\log_{2}(E+1)-E\log_{2}E. \label{CCC}
\end{equation} 
We observe that  for $M>1$, the quantity $R^{(N,M)}_{opt}(E)/E$ considerably increases in the low-energy regime  where it asymptotically saturates the  upper bound $C(E)/E$; unfortunately this increase takes place at lower energy values for larger codewords' length $N$. 
More specifically, by taking a low-energy expansion, $\mathcal{E}\ll 1$ up to first order in $\mathcal{E}$, we can show that the optimal Hadamard rate $R_{opt}^{(N,M)}(E)$ of Eq.~(\ref{mpskHolCap}) attains the classical capacity Eq.~(\ref{CCC}) of the channel for any codewords' length $N$ and $M>1$. Let us then take the low-energy expansion of the optimal rate up to first order in $E$. The coefficients in \eqq{mpskEigen} behave as:
\begin{equation}
\lambda_{\ell}\left(\mathcal{E}\right)\simeq \sum_{h=0}^{M-1}\left(1-\left(1-e^{i\frac{2\pi}{M}h}\right)\mathcal{E}\right)e^{-i\frac{2\pi}{M}\ell h}=\left(1-\mathcal{E}\right)M\delta_{\ell,0}+\mathcal{E}M\delta_{\ell,1}.
\end{equation}
Hence the eigenvalues of $\hat{\rho}^{(+)}_{PPM}$, \eqq{rho+}, and $\hat{\rho}_{PPM}^{(0)}$, \eqq{rho0Fin}, yield respectively
\begin{align}
&\nu^{\ell>0}(\mathcal{E})\simeq \delta_{\ell,1}\frac{\mathcal{E}}{N},\\
&\nu_{0}^{0}\left(\mathcal{E}\right)\simeq\frac{(1-\mathcal{E})M+(N-1)M(1-\mathcal{E})}{MN}=1-\mathcal{E},\\
&\nu_{+}^{0}(\mathcal{E})\simeq\frac{(1-\mathcal{E})M-M(1-\mathcal{E})}{MN}=0.
\end{align}
Inserting the previous expressions in \eqq{mpskHolCap} then we obtain
\begin{equation}\label{hadCapApp}
R_{opt}^{(N,M)}(E)\simeq-\frac{(1-\mathcal{E})\log_{2}(1-\mathcal{E})+\mathcal{E}\log_{2}(\frac{\mathcal{E}}{N})}{N}\simeq E-E\log_{2}E,
\end{equation}
which coincides with the expansion of the classical capacity Eq.~(\ref{CCC}), meaning that the two quantities are the same at sufficiently low energy $\mathcal{E}\ll 1$, i.e., $E\ll1/N$. Also note that for $M=1$ there is no contribution from the eigenvalues $\nu^{\ell>0}$, so that the second term in \eqq{hadCapApp} is missing. This explains why  for this value, $R_{opt}^{(N,M)}(E)$ does not saturate the bound $C(E)$ even at high values of $N$.\\
Eventually, in the inset of Fig.~\ref{f10} we show the values of $R^{(N,M)}_{opt}(E)/E$, at fixed $E=0.05$, for several values of the couple $(N,M)$, each represented by a coloured tile (darker colours represent lower values). In general, to obtain higher capacities it is clearly convenient to increase $M$ at fixed $N$ and to decrease $N$ at fixed $M$ (the case $M=1$ representing an exception). In particular  for $M>1$,  $R^{(N,M)}_{opt}(E)/E$  seems to remain close to an optimal value for all $N$ below a critical threshold. We conclude that optimal values of the rate $R^{(N,M)}_{opt}(E)$ at fixed energy are obtained for large $M$ and small $N$ values, which motivates us to devise a detection scheme for this family of codes.

\subsection{The rate of the PSK Hadamard receiver}\label{secHadRate}
In this section we want to compute the rate of information transmitted per mode when reading out the PSK Hadamard code with the  PSK Hadamard receiver. As we mentioned in Sec.~\ref{had:sec2}, the ultimate advantage of this code is the unitary equivalence between low-energy pulses distributed on several modes by the sender and high-energy ones concentrated on a single mode by the receiver, i.e., respectively the $\ket{v_{k}(\alpha_m)}$ and $\ket{w_{k}(\alpha_m)}$ states of Eq.~(\ref{TRANS2}): the Hadamard transform is able to concentrate the scattered input energy, for a limited number of codewords. Bob takes advantage of this concentration with the readout operation, which is a separable technique repeated on each mode, labeled Vacuum or Pulse (VP) detection (see Fig.~\ref{f11}). Its purpose is first to determine whether a pulse is present on that mode then, in case of a positive answer, to determine which among the possible states $\{ \ket{\sqrt{N}\alpha_{m}}; m=0, 1, \cdots, M-1\}$ is the one that was selected by Alice. The simplest case $M=2$ was proposed and employed in~\cite{guha2}, while here we discuss its generalization to $M>2$.\\
The VP scheme can be implemented as shown in Fig.~\ref{f11}: first we split the unknown received state in two lower-energy copies, by means of a beam splitter, \eqq{bs}, of reflectivity $\eta_{1}=1/J$ (transmissivity $\tau_{1}=1-\eta_{1}$), then we measure the reflected copy with an OOD. Assuming no dark counts and perfect efficiency, the detector can click (``1'' in the figure) only if a pulse was present. In this case, the transmitted copy of the state is employed for PSK detection, which identifies the phase of the pulse among the $M$ possible ones. On the other hand if the detector does not click (``0'' in the figure) we can not exclude that the received state is the vacuum. Hence we repeat the initial procedure on the transmitted copy, sending it back to the beam splitter, measuring its reflected part on the OOD and applying the same decision rule. Note that we employ a beam splitter with rescaled reflectivity  $\eta_{2}=\eta_{1}/\tau_{1}$, so that the reflected part of the signal at this second step still carries a fraction $1/J$ of the input energy. We iterate this procedure of splitting the state and measuring its reflected part $N$ times, until either a click is registered at some iteration step, triggering PSK detection on the remaining part of the state and the identification of a pulse, or the OOD never clicks, in which case we guess that the received state was the vacuum.  By properly rescaling the reflectivity at each step  $p$ as $\eta_{j}=\eta_{j-1}/\tau_{j-1}$, we can assure to always extract and measure a fraction $1/J$ of the input energy, so that after $J$ extractions, i.e., no clicks, the energy is exhausted. 
\begin{figure}[t!]
\center
\includegraphics[trim={0cm 6cm 0cm 6.5cm},clip,scale=.43]{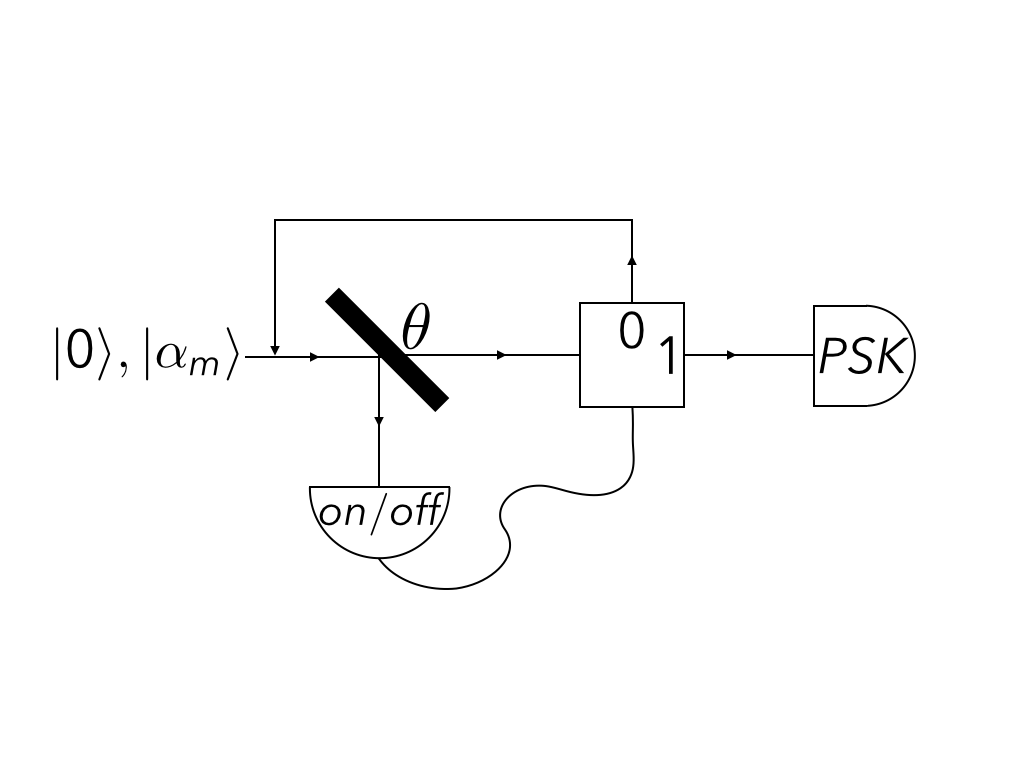}
\caption{Schematic depiction of the single-mode VP detection scheme, which determines whether or not a pulse was present on the given mode and in case of a positive answer determines its phase: the input state, $\ket{\alpha_{m}}$, $m=0,\cdots,M-1$ or the vacuum, is sent through a beam splitter of reflectivity $\eta_{1}=1/J$ (transmissivity $\tau_{1}=1-\eta_{1}$) and its reflected part is measured with an ideal OOD. If the detector clicks, i.e., ``$1$'', then a pulse must be present and the transmitted part of the state is sent through PSK detection, which identifies the pulse among the $M$ possible ones. If the detector does not click instead, i.e., ``0'', then the received state could be the vacuum, so its transmitted part is sent back to the beam splitter, repeating the same detection procedure with a rescaled reflectivity, $\eta_{2}=\eta_{1}/ \tau_{1}$, so that the fraction of energy incident on the OOD is always $1/N$. The scheme is iterated $N$ times, with reflectivity $\eta_{j}=\eta_{j-1}/\tau_{j-1}$ at the $j$-th OOD step, until either a click triggers PSK detection at some iteration step, identifying the phase of the pulse, or the entire energy is exhausted, in which case the receiver guesses that vacuum was sent. The conditional probabilities of detection in the limit $J\rightarrow\infty$, $1-\tau=1/N$, are given by \eqq{DDDlm}. An upper and lower bound on the PSK probability of detection are given in the text (see \eqq{multiHel} and Appendix~\ref{app:naiveProb}).}\label{f11}
\end{figure}\\
Accordingly the conditional probability of detecting a state $\ket{\alpha_{\ell}}$ if $\ket{\alpha_{m}}$ was sent, $P_{vp}^{(M)}(\ell|m;\mathcal{E},N)$, is given by the sum, over all values of the step index $j$, of the probability of detecting the first photon at that step and then switching to PSK detection of the $M$ pulses with remaining energy $\tau^{j}\mathcal{E}$. We have:
\begin{equation}\label{correctCondN}
P_{vp}^{(M)}(\ell|m;\mathcal{E},J)=\sum_{j=1}^{J}e^{-\frac{\mathcal{E}}{J}p}\left(1-e^{-\frac{\mathcal{E}}{J}}\right) P_{psk}^{(M)}\left(\ell|m
;\frac{\mathcal{E}}{J}(J-j)\right),
\end{equation}
where $e^{-\frac{\mathcal{E}}{J}}$ is the conditional probability of registering no click at any detection step if a pulse was present, while $P_{psk}^{(M)}\left(\ell|m;\tau^{j}\mathcal{E}\right)$ represents the PSK-detection probability of guessing $\ket{\alpha_{\ell}}$ if $\ket{\alpha_{m}}$ was sent, after the OOD clicked at the $j$-th step; its specific form will be discussed after computing the rate of the receiver. If instead the OOD never clicks, we guess that the vacuum was present, even though it may have actually been a pulse, making an error with probability
\begin{align}
\label{err0CondN}
P_{vp}^{(M)}(vac|m;\mathcal{E},J)=e^{-\mathcal{E}}.
\end{align}
The expression Eq.~(\ref{correctCondN}) can be evaluated as an integral in the limit of infinite splitting steps, $J\rightarrow\infty$.  It is sufficient to define the variable $x_{j}=j\mathcal{E}/J\in[\mathcal{E}/J,\mathcal{E}]$, whose increment is $\Delta x=\mathcal{E}/J$, infinitesimal in the large-$J$ limit. Hence the conditional probability Eq.~(\ref{correctCondN}) can be written at the first order in $\Delta x$ as
\begin{equation}
\begin{aligned}
P_{vp}^{(M)}(\ell|m;\mathcal{E},J)&\simeq\sum_{j=1}^{J} e^{-x_{j}} \Delta x~P_{psk}^{(M)}\left(\ell|m;\mathcal{E}-x_{j}\right)\\
&\underset{J\rightarrow\infty}{\longrightarrow}\int_{0}^{\mathcal{E}}dx e^{-x} P_{psk}^{(M)}\left(\ell|m;\mathcal{E}-x\right),
\end{aligned}
\end{equation}
which, after a change of variable $t=e^{-x}$, gives the final result:
\begin{equation}
\label{DDDlm}
P_{vp}^{(M)}\left(\ell|m;\mathcal{E}\right)=\lim_{J\rightarrow\infty} P_{vp}^{(M)}(\ell|m;\mathcal{E},J)=\int_{e^{-\mathcal{E}}}^{1}dt~P_{psk}^{(M)}(\ell|m;\mathcal{E}+\ln t).
\end{equation}
Eventually, the detection of a whole Hadamard codeword $\ket{w_{k}(\alpha_{m})}$ is carried out by applying VP to each mode. With this method a codeword with pulse on mode $k$ can be misinterpreted only as one of the other $M-1$ codewords living on the same mode, since all other modes are occupied by the vacuum, which with unit probability never clicks. The only additional source of error is when the pulse on the $k$-th mode does not click, in which case no codeword can be identified. Accordingly, the rate of the whole receiver is computed in terms of the mutual information, \eqq{mCInfo}, of the classical input/output random variables induced by the quantum encoding/decoding operations, respectively:  $x\in X=\{(k,m)\}_{k=0,m=0}^{N-1,M-1}$, determined by the index of the mode where the pulse is present, $k$, and the phase index of the pulse, $m$, and $y\in Y=X\cup\{err\}$, with an additional error outcome associated to the case where VP does not click on any mode. The input probability distribution is uniform, $P_{X}(x)=1/(NM)$, while the conditional output one is:
\begin{equation}
P_{Y|X}(y|x)=\begin{cases}
P_{vp}^{(M)}(vac|m_{x};\mathcal{E})&\text{ if }y=err,\\
P_{vp}^{(M)}(m_{y}|m_{x};\mathcal{E})&\text{ if }k_{y}=k_{x},\\
0&\text{ otherwise}.
\end{cases}
\end{equation}
The rate of the PSK Hadamard receiver with average received energy per mode $E$ is then given by:
\begin{equation}\begin{aligned}\label{hadRate}
R^{(N,M)}_{had}(E)&=\sum_{m_{x},m_{y}=0}^{M-1}\frac{P_{vp}^{(M)}(m_{y}|m_{x};\mathcal{E})}{MN}\log_{2}\left(\frac{M N~P_{vp}^{(M)}(m_{y}|m_{x};\mathcal{E})}{\sum_{m_{x'}=0}^{M-1} P_{vp}^{(M)}(m_{y}|m_{x};\mathcal{E})}\right).
\end{aligned}\end{equation}\\

Let us now discuss the specific form of $P_{psk}^{(M)}(\ell|m;\mathcal{E})$, which determines $P_{vp}^{(M)}\left(m_{y}|m_{x};\mathcal{E}\right)$ and hence the rate $R^{(N,M)}_{had}(E)$. For the case $M=2$, it is known that Dolinar detection attains the optimal success probability, see Sec.~\ref{sec:disc}. For the general case $M>2$, a variety of generalized-Dolinar schemes has been proposed, as also discussed in Sec.~\ref{subsec:discPract}. The large literature on the subject seems to suggest that this generalized Dolinar detection can get close to the PSK Helstrom bound on the discrimination of $M$ symmetrically-distributed coherent states of fixed amplitude, \eqq{cycSym}, in particular surpassing the performance obtained by classical detection techniques, but that it cannot attain the multi-state bound in general. In order to have both a best-case evaluation of the receiver's performance and a realistic, practical and fair comparison with the $M=2$ case, in the following we report both the detection efficiency one could reach by employing an optimal measurement that saturates the PSK Helstrom probability of discrimination~\cite{multiHel1,multiHel2} as the second part of the receiver, 
\begin{align}\label{multiHel}
P_{hel}^{(M)}(\ell|m;\mathcal{E})=\abs{\frac{1}{M}\sum_{j=0}^{M-1}e^{-i \frac{2\pi}{M} j(\ell-m)}\sqrt{\lambda_{j}(\mathcal{E})}}^{2},
\end{align}
with $\lambda_{j}(\mathcal{E})$ defined in \eqq{mpskEigen}, and the detection efficiency that can be attained by employing a very naive but realistic generalization of the Dolinar scheme, $P_{real}^{(M)}(\ell|m;\mathcal{E})$, based on the same splitting method of the VP detector and sequentially nulling one of the hypoteses, discarding it if a click is registered (see Appendix~\ref{app:naiveProb} for its detailed form). Hence we have both an upper and lower bound for the PSK probability of discrimination, i.e., $P_{hel}^{(M)}(\ell|m;\mathcal{E})\geq P_{psk}^{(M)}(\ell|m;\mathcal{E})\geq P_{real}^{(M)}(\ell|m;\mathcal{E})$; when substituting these bounds in \eqq{DDDlm}, we obtain two new expressions for the conditional probability of detection, 
\begin{align}
P_{vp\text{-}hel}^{(M)}(\ell|m;\mathcal{E})&=\int_{e^{-\mathcal{E}}}^{1}dt~P_{hel}^{(M)}(\ell|m;\mathcal{E}+\ln t),\\
P_{vp\text{-}real}^{(M)}(\ell|m;\mathcal{E})&=\int_{e^{-\mathcal{E}}}^{1}dt~P_{real}^{(M)}(\ell|m;\mathcal{E}+\ln t),\label{vpRealGen}
\end{align}
which, through equation \eqq{hadRate}, provide a correspondent upper and lower bound for the rate of the Hadamard receiver, i.e., $R^{(N,M)}_{hel}(E)\geq R^{(N,M)}_{had}(E)\geq R^{(N,M)}_{real}(E)$, with
\begin{equation}\label{helRate}
R^{(N,M)}_{hel}(E)=\sum_{m_{x},m_{y}=0}^{M-1}\frac{P_{vp\text{-}hel}^{(M)}(m_{y}|m_{x};\mathcal{E})}{MN}\log_{2}\left(\frac{M N~P_{vp\text{-}hel}^{(M)}(m_{y}|m_{x};\mathcal{E})}{\sum_{m_{x'}=0}^{M-1} P_{vp\text{-}hel}^{(M)}(m_{y}|m_{x};\mathcal{E})}\right),
\end{equation}
and $R^{(N,M)}_{real}(E)$ explicitly given in Appendix~\ref{app:naiveProb} for $M=3,4$.
 In particular the upper bound is saturated in the case $M=2$, while for $M>2$ there are optimized schemes which get close to it at low energy (see~\cite{marq} for a detailed description). \\
It is important to have a separable communication scheme to compare with, based on the same constellation of single-mode states. It seems reasonable to send one of the $M$ symmetric coherent states on each mode and read them out with generalized Dolinar detection. We call the latter a  PSK separable scheme and its rate is:
\begin{equation}
R^{(M)}_{sep}(E)=\sum_{m_{x},m_{y}=0}^{M-1}\frac{P_{hel}^{(M)}(m_{y}|m_{x};E)}{M}\log_{2}\left(\frac{M~P_{hel}^{(M)}(m_{y}|m_{x};E)}{\sum_{m_{x'}=0}^{M-1} P_{hel}^{(M)}(m_{y}|m_{x};E)}\right).
\end{equation} 
\begin{figure}[t!]\center
\includegraphics[trim={0cm 1cm 0cm 2.5cm},clip,scale=.43]{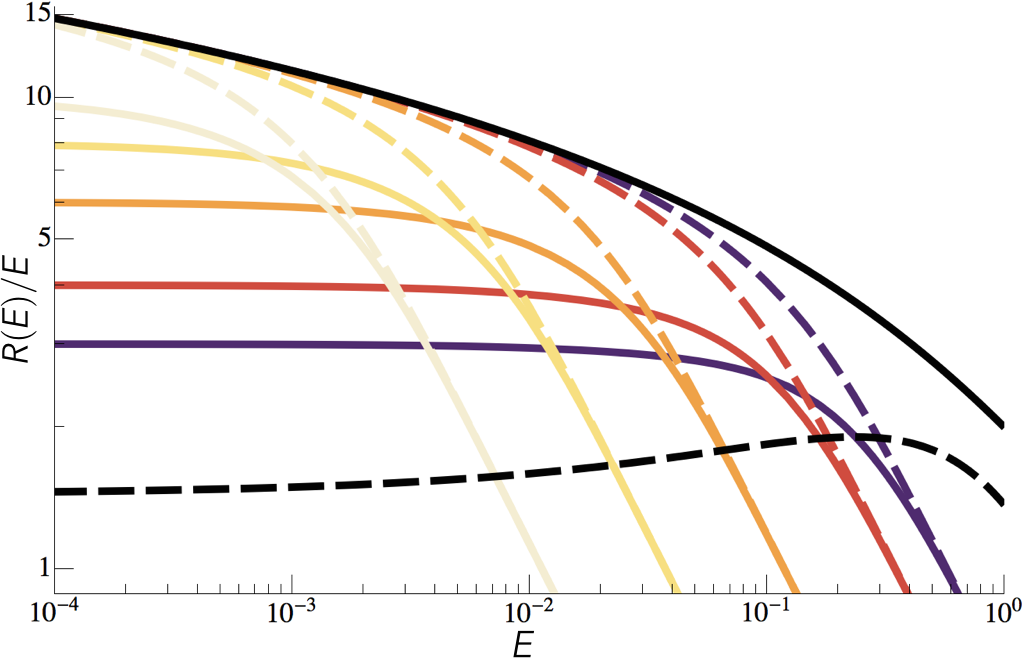}
\caption{Plot (log-log scale) of the upper-bound of the Hadamard rate per average mode-energy, $R_{hel}^{(N,M)}(E)/E$, for $M=3$ phases and codewords of length $N=2^{i}$, $i=3,4,6,8,10$ (solid lines with colours from purple/dark grey to light brown/light grey), the corresponding optimal rate per average mode-energy, $R_{opt}^{(N, M)}(E)/E$, for the same values of $n,M$ (dashed lines with colours from purple/dark grey to light brown/light grey), the separable rate per average mode-energy, $R_{sep}^{(M)}(E)/E$, for the same number of phases $M=3$ (black dashed line) and the classical capacity of the channel, $C(E)/E$ (black solid line), as a function of the average mode-energy $E$. Observe that the Hadamard rates rise over their separable counterpart at low energy, only to level off at a plateau later on; this effect is shifted towards lower energies and higher plateau values for larger codewords' length $N$. Note also that, for any $N$, the Hadamard rate attains its optimal value $R_{opt}^{(N,M)}(E)$ at high energy, but detaches towards the plateau when the latter attains the classical capacity of the channel.}\label{f12}
\end{figure}
In Fig.~\ref{f12} we show the Hadamard rate with Helstrom PSK detection per average mode-energy as a function of the average mode-energy, $R^{(N,M)}_{hel}(E)/E$ vs. E, for $M=3$ phases and several values of the codewords' length $N$, along with $R^{(M)}_{sep}(E)/E$ and $R^{(N,M)}_{opt}(E)/E$, also for $M=3$ and the same $N$ values, and $C(E)/E$. We observe that for any value of $N$ the Hadamard rates surpass the separable rate below a certain low-energy threshold, $E^{*}(N)<1$, then level off at a plateau for smaller energy. As $N$ increases, this energy threshold $E^{*}(N)$ lowers, while the plateau value grows. Moreover each Hadamard rate attains its optimal value, \eqq{mpskHolCap}, at high energy $E>E^{*}(N)$, but detaches from it and from the classical capacity of the channel when leveling off at the plateau. Still let us note that, if we take the envelope of $R^{(N,M)}_{hel}(E)$ for fixed $M$ and all allowed values of $N$ in a given interval $\mathcal{N}$, i.e., 
\begin{equation}\label{helEnv}
R^{(\mathcal{N},M)}_{hel}(E)=\max_{N\in \mathcal{N}}R^{(N,M)}_{hel}(E), 
\end{equation} 
we are able to get closer to $C(E)$ as the energy decreases.
\begin{figure}[t]\center
\includegraphics[trim={0cm 2.8cm 0cm 2.8cm},clip,scale=.43]{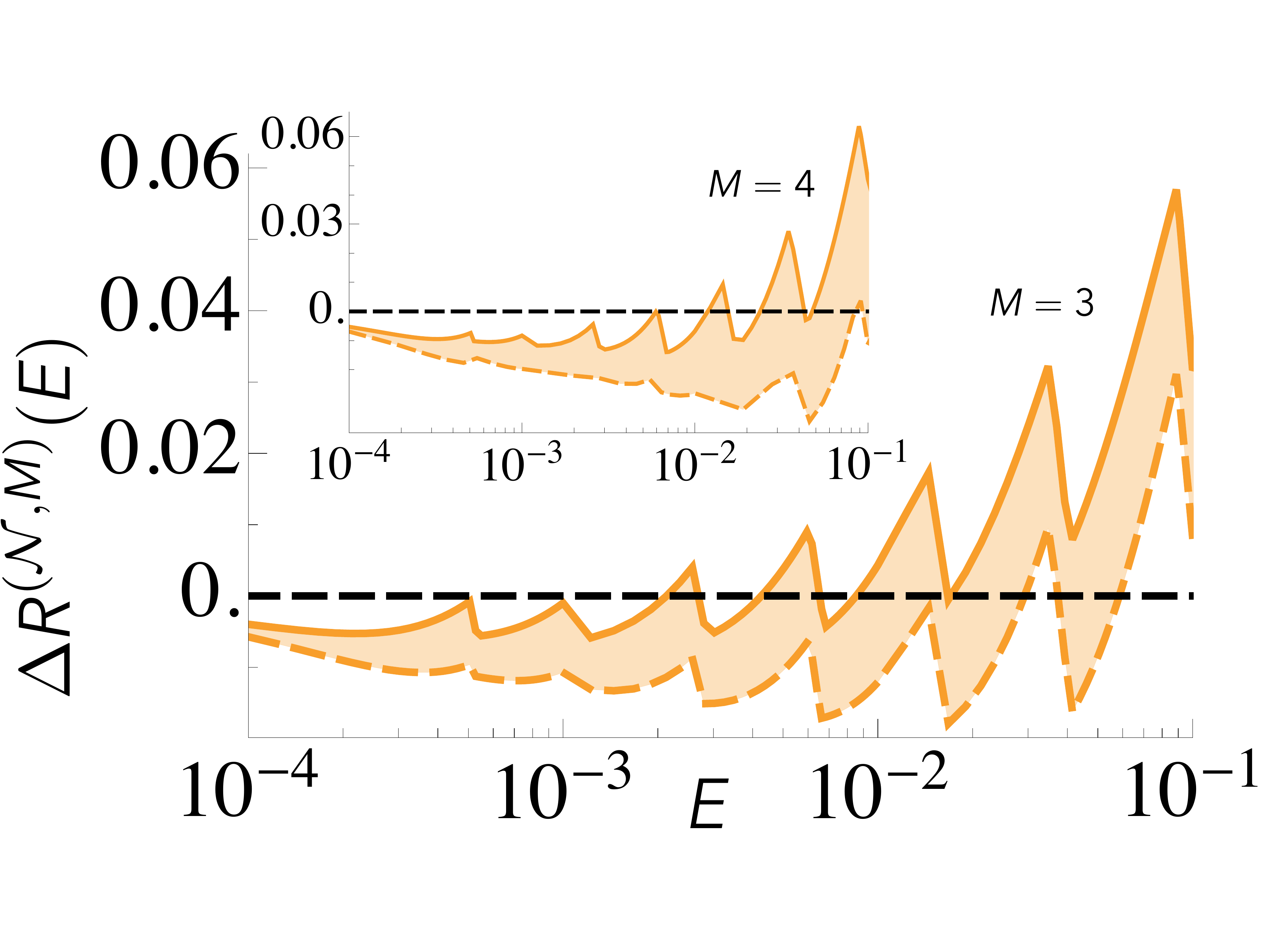}
\caption{Plot (log-linear scale) of $\Delta R_{hel}^{(\mathcal{N}, M)}(E)$ (orange/light-grey solid line) and $\Delta R_{real}^{(\mathcal{N},M)}(E)$ (orange/light-grey dashed line) as a function of the average energy per mode $E$ for $M=3$ and, in the inset, $M=4$. The shaded region between the two curves represents all rate values achievable by Hadamard coding and VP detection. The black dashed line of constant value $0$ represents the reference quantity $R_{hel}^{(\mathcal{N}, 2)}(E)\equiv R_{had}^{(\mathcal{N}, 2)}(E)$. Observe that the Hadamard rate with $M=3, 4$ phases beat both the that with $M=2$ phases in an interval of energy $E\sim[4\cdot10^{-3},10^{-1}]$, with a gain of up to $6\%$ over the rate at $M=2$. Hence the design of better PSK detection techniques, achieving the Helstrom bound for PSK states would provide access to the best communication rates so far in the low-energy regime.}\label{f13}
\end{figure}\\
\begin{figure}[t!]\center
\includegraphics[trim={0cm 2.8cm 0cm 2.8cm},clip,scale=.43]{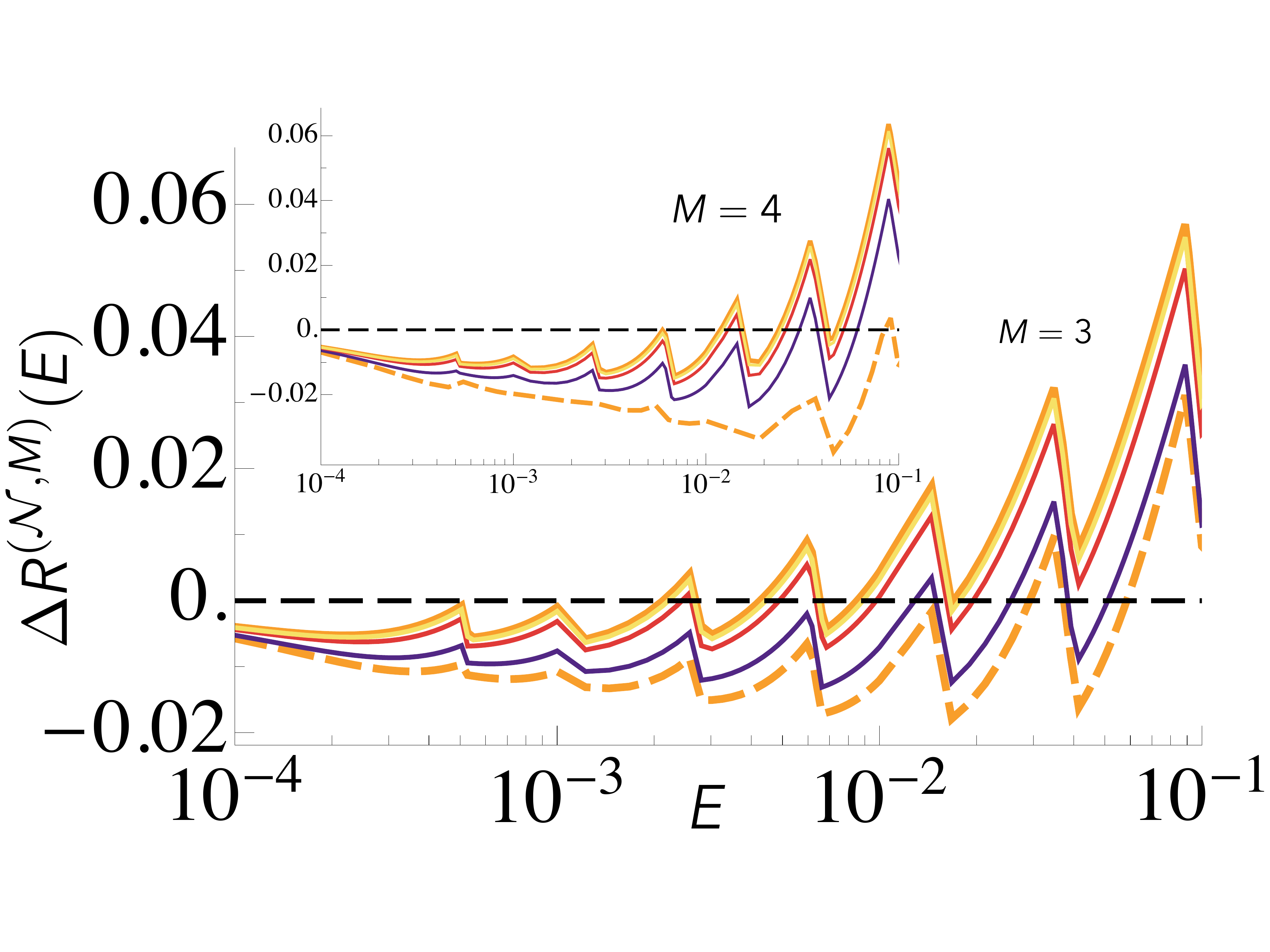}
\caption{Plot (log-linear scale) of the same rates of Fig.~\ref{f13} and of those achievable by the Helstrom-Hadamard receiver with a finite number of splitting steps $J=10, 30, 100$ (respectively from purple/dark-grey to yellow/light-grey solid lines), based on \eqq{correctCondN}. For both numbers of phases $M=3,4$, already $30$ splitting steps (red/middle lines) achieve a rate close to the one obtained in the infinite-$J$ limit (orange/light-grey solid curve).}\label{f14}
\end{figure}\\

Eventually we want to compare the performance of Hadamard rates for $M>2$  with that for the case $M=2$, discussed in~\cite{guha2}. In Fig.~\ref{f13} we show the difference between the upper-bound Hadamard envelopes \eqq{helEnv} for $M=3,4$ and that for $M=2$, relative to the latter, i.e., $\Delta R^{(\mathcal{N},M)}_{hel}(E)=\left[R^{(\mathcal{N},M)}_{hel}(E)-R^{(\mathcal{N},2)}_{hel}(E)\right]/R^{(\mathcal{N},2)}_{hel}(E)$, with $\mathcal{N}=\{2^{i},i=1,\cdots,10\}$. The same quantity is plotted also for the lower bound of the Hadamard rate, discussed in Appendix~\ref{app:naiveProb}, i.e., $\Delta R^{(\mathcal{N},M)}_{real}(E)$, so that the shaded region between the two curves represents all rates achievable by the Hadamard code with VP detection. We observe an advantage of the schemes with a higher number of phases in the energy region $E\in[4\cdot10^{-3},10^{-1}]$. For higher energy values the separable technique becomes superior, since the classical capacity allows for more codewords to be used, while for lower energy the Hadamard rate for $M=2$ performs better, since distinguishability of the signals becomes crucial. We conclude that the practical Hadamard receiver with more than two phases has relevant rate, with a gain of up to $6\%$ with respect to previous proposals, in a crossover amplitude region where separable techniques start performing worse than joint ones. In Fig.~\ref{f14} we show the same plots with the  addition of three curves, corresponding to the maximum rate achievable by the Hadamard receiver with a finite number of splitting steps $J=10, 30, 100$, i.e., employing the conditional probability \eqq{correctCondN} instead of \eqq{DDDlm}. It can be seen that a reasonable approximation of the infinite-$J$ limit is obtained already for $J=30$, independently of the number of phases $M=3,4$.

\section{The capacity of adaptive coherent-state decoders with passive interferometry}\label{sec:adRec}
In this last section we study a class of practical joint decoders for classical communication on quantum PI Gaussian channels, \eqq{PIG}, with coherent-state codewords. Given the general difficulty of implementing truly joint quantum measurements, Secs.~\ref{subsec:commQ},\ref{sec:bisDec}, this problem has strong practical relevance, and several decoder proposals so far, e.g., the Hadamard receivers of Sec.~\ref{sec:had}, are encompassed by the the general class of Adaptive Decoders (AD) depicted in Fig.~\ref{f15}a. The latter combines the available \emph{single-mode} technology, e.g., photodetectors and local transformations, with \emph{multi-mode passive} interferometers and \emph{classical feedforward control}. The rationale behind this choice is that introducing correlations between modes during the decoding procedure may increase the transmission rate of simpler separable measurements, getting closer to the structure of joint quantum measurements that seems to be ultimately necessary to achieve the capacity of PI Gaussian channels. On the contrary, here we prove that the maximum information transmission rate of such channels with coherent-state encoding and AD \emph{is equal to} that obtained with a separable Decoder (SD) employing the same measurement on each mode, as shown in Fig.~\ref{f15}b.
The general idea behind our proof is to map the quantum AD into an effective classical programmable channel with feedback to the encoder. Then we obtain our results by extending Shannon's feedback theorem~\cite{ShanFeed,covThomBOOK} to this kind of channels.\\
Our work gives several major contributions: i) it implies the conjecture by Chung \emph{et al.}~\cite{guhaDet,guhaDet1}, namely that adaptive passive Gaussian interactions, single-mode displacements and photodetectors do not increase the optimal transmission rate; ii) if the classical capacity of PI Gaussian channels, \eqq{ccapPI}, is achieved only by joint measurements, Secs.~\ref{subsec:commQ},\ref{sec:bisDec}, as the evidence suggests so far, then it cannot be achieved with our AD scheme; iii) it extends the results of Takeoka and Guha~\cite{takeokaGuha}, who considered only Gaussian measurements; iv) it extends the analysis made by Shor~\cite{shorAd} in the context of trine states to coherent states and passive interactions. Our results, though already envisaged in previous works on the subject, have strong relevance for future research on practical decoders: i) they extend the study of decoders by considering arbitrary single-mode manipulations before measurement, including non-Gaussian and non-unitary ones; ii) they exclude a decoding advantage of adaptive passive Gaussian interactions, which are the easiest to realize in practice, suggesting that more difficult interactions are necessary to achieve capacity. Eventually, let us note that the possibility of employing ancillary states is partially included in our AD scheme: this is the case if each ancilla is allowed to interact just with one mode before being measured; otherwise, i.e., if the ancillae can interact with several modes, the problem of determining the decoder's optimal rate remains open and could give a practical advantage over SDs, as stated in a new version of the conjecture by Chung et al~\cite{guhaDet1}. 
\begin{figure}[t!]\center
\includegraphics[trim={0cm 3.7cm 0cm 0cm},clip,scale=.43]{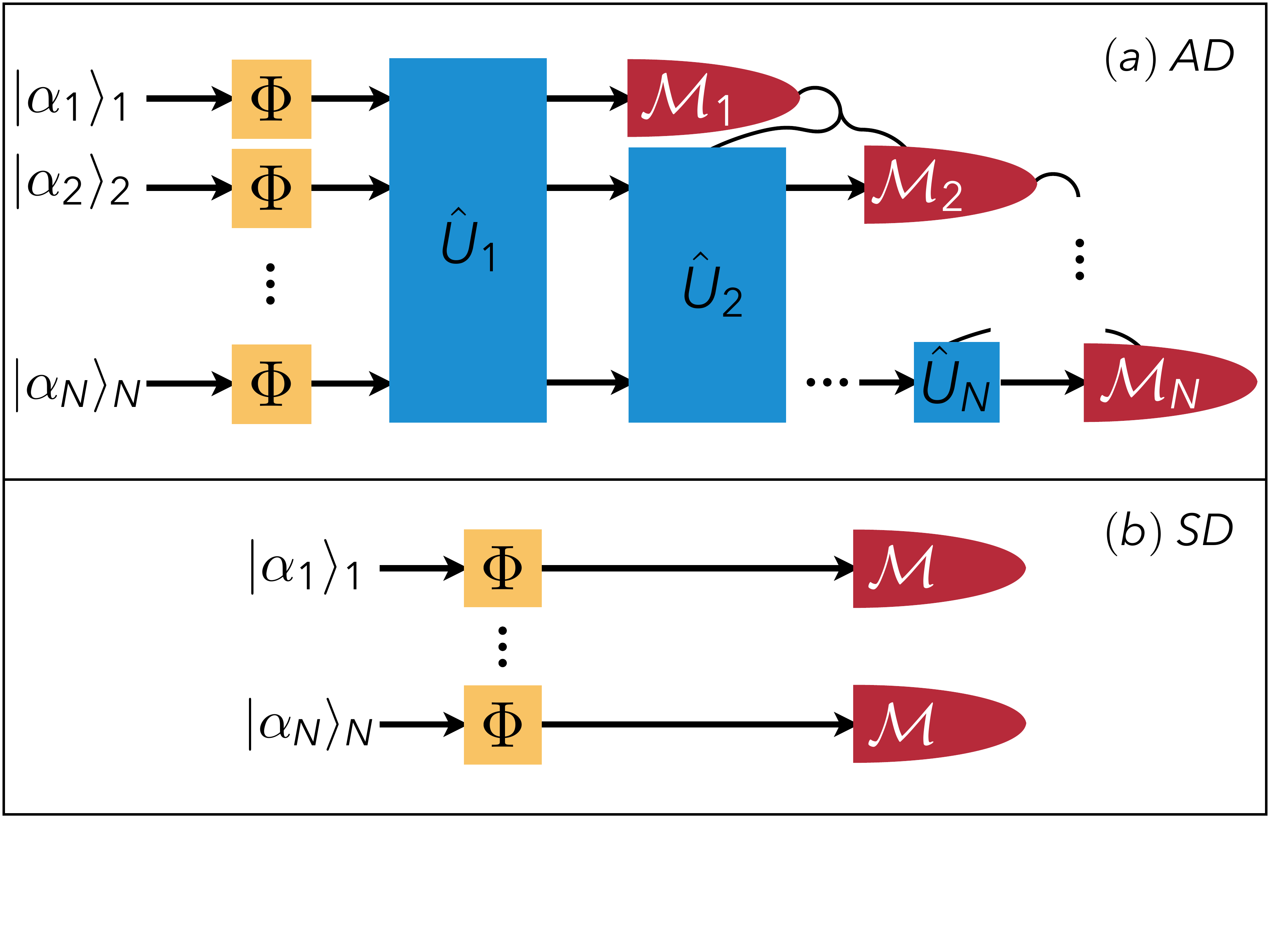}
\caption{Schematic depiction of the class of (a) Adaptive Decoders (AD) and (b) separable Decoders (SD) considered, whose maximum information transmission rate is proved to be equal. (a) In the AD case the sender, Alice, encodes the message into separable sequences of coherent states $\ket{\alpha_{1}}_{1}\otimes\cdots\otimes\ket{\alpha_{N}}_{N}$ and sends it to the receiver, Bob, with $N$ distinct uses of a quantum PI Gaussian channel $\Phi$ (yellow/light-gray boxes), \eqq{PIG}. Bob's AD comprises multi-mode passive Gaussian interferometers $\hat{U}_{j}$ (blue/gray boxes), \eqq{pass}, and arbitrary destructive single-mode measurements $\mathcal{M}_{j}$ (red/dark-gray shapes), \eqq{POVM}, adaptively dependent on the measurement results of previous modes and applied successively on the remaining modes. (b) In the SD case Alice uses the same encoding but Bob performs the same measurement $\mathcal{M}$ on each mode and cannot use adaptive procedures.}\label{f15}
\end{figure}

\subsection{The adaptive decoder}\label{ad}
Let us suppose that Alice wants to transmit a classical message, represented by the sequence of random variables $A_{(1,N)}$ with values $\vecund{\alpha}\in\mathbb{C}^{N}$, by making $N$ transmissions on a quantum PI Gaussian channel, $\Phi_{PI}$, \eqq{PIG}. She employs coherent-state codewords 
\begin{equation}
\ket{\vecund{\alpha}}=\ket{\alpha_{1}}_{1}\cdots\ket{\alpha_{N}}_{N}
\end{equation} 
with joint probability $p_{\vecund{\alpha}}=P_{A_{(1,N)}}\left(\vecund{\alpha};E\right)$ chosen to respect the maximum average-mode-energy constraint
 \begin{equation}\label{eneCon}
 \int d^{2N}p_{\vecund{\alpha}}\abs{\vecund{\alpha}}^{2}\leq N E.
 \end{equation}
 As discussed in Theorem~\ref{ccapPI}, this separable coherent-state encoding achieves the classical capacity of $\Phi_{PI}$ when its probability distribution is i.i.d. and Gaussian on each mode.\\
Let us also suppose that Bob has an AD that outputs the sequence of classical random variables $Y_{(1,N)}$, with values $\vecund{y}\in\mathcal{Y}^{N}$, where $\mathcal{Y}$ can be discrete or continuous, e.g., $\mathcal{Y}=\mathbb{R}$ for homodyne detection \eqq{homodyne}. The conditional probability distribution of the output variables, $P_{Y_{(1,N)}|A_{(1,N)}}\left(\vecund{y}|\vecund{\alpha}\right)$, can be computed as in Sec.~\ref{subsec:commQ} and it is determined by the specific decoding operations of the AD, Fig.~\ref{f15}a, comprising for all $j=1,\cdots,N$: 
\begin{enumerate}
\item A multi-mode passive Gaussian unitary $\hat{U}_{j}\left(y_{(1,j-1)}\right)$, i.e., a network of beam-splitters and phase-shifters conditioned on the outcomes of previous measurements, acting on the set of modes from the $j$-th to the $N$-th as
\begin{equation}
\hat{U}_{j}\left(y_{(1,j-1)}\right)\ket{\alpha_{(j,N)}}=\ket{U_{j}\left(y_{(1,j-1)}\right)\alpha_{(j,N)}},\label{pass}
\end{equation} 
where $U_{j}$ is the $(N-j+1)$-dimensional unitary matrix representing $\hat{U}_{j}$ in phase-space, see Sec.~\ref{gaussUni};
\item Single-mode operations and a final destructive measurement, altogether represented by a local POVM $\mathcal{M}_{j}\left(\lambda\left(y_{(1,j-1)}\right)\right)$ chosen among a set of possible POVMs that are labeled by the (discrete or continuous) index $\lambda\in\Lambda$ conditioned on the outcomes of previous modes. As discussed in Se.~\ref{sec:basics}, each POVM is defined by a collection of positive operators corresponding to the possible single-mode outcomes, 
\begin{equation}\label{POVM}
\mathcal{M}_{j}\left(\lambda\left(y_{(1,j-1)}\right)\right)=\left\{\hat{E}_{y_{j}}\left(\lambda\left(y_{(1,j-1)}\right)\right)\right\}_{y_{j}\in\mathcal{Y}},\hspace{-.5cm}
\end{equation} 
where the operators $\hat{E}_{y_{j}}$ sum up to the identity on the Hilbert space of a single mode.
\end{enumerate}
For our results to hold, a crucial assumption is that the single-mode POVMs completely destroy the measured state before any information is sent to the rest of the system; if instead Bob can perform partial measurements the ADs rate may increase, see~\cite{shorAd}. Let us also note that the generic set of allowed POVMs described above can be restricted case by case by properly choosing the $\hat{E}_{y}$. For example the simplest toolbox for optical-signal processing is that of the Kennedy receiver with POVMs of the form $\mathcal{M}_{ken}(\lambda)$, \eqq{kennedy}, where the index $\lambda\in\mathbb{C}$ is the amplitude of a phase-space displacement in this case. Since the latter depends adaptively on previous outcomes, let us also note that the AD with a single-mode Kennedy structure behaves similarly to a Dolinar receiver, Sec.~\ref{subsec:discPract}.

\subsection{The optimal rate} \label{thm}
The performance of a quantum decoder for the transmission of classical information can be evaluated by computing the mutual information of its classical input and output random variables, as discussed in Sec.~\ref{sec:commQ}. The latter is defined for our AD as $I\left(A_{(1,N)}:Y_{(1,N)}\right)$, \eqq{mCInfo}. The AD's optimal information transmission rate then is obtained by maximizing the mutual information over the input distribution with energy constraint $E$ and the decoding operations and regularizing it as a function of the number of uses $N$, i.e., 
\begin{equation}\label{opDef}
R_{AD}(E)=\lim_{N\rightarrow\infty}~\max_{\substack{ P_{A_{(1,N)}}\left(\vecund{\alpha};E\right), \\ \hat{U}_{j}\left(y_{(1,j-1)}\right), \\\lambda\left(y_{(1,j-1)}\right)\in\Lambda} }\frac{I\left(A_{(1,N)}:Y_{(1,N)}\right)}{N}.\hspace{.5cm}
\end{equation}
We want to compare the AD with the SD of Fig.~\ref{f15}b, comprising for each use of the channel $\Phi_{PI}$ only a single-mode POVM $\mathcal{M}(\lambda)$ chosen from the same set of those in the AD parametrized by $\lambda\in\Lambda$, \eqq{POVM}, but without any interaction or classical communication between modes. Obviously, the optimal rate of this SD is obtained by maximizing the mutual information of the single-mode input and output variables $A_{1}$ and $Y_{1}$ over the input distribution at constrained energy $E$ and the POVM's parameter, i.e.,
\begin{equation}\label{opSep}
R_{SD}(E)=\max_{P_{A_{1}}\left(\alpha_{1};E\right),~\lambda\in\Lambda}I\left(A_{1}:Y_{1}\right).
\end{equation}
In order to show that the optimization \eqq{opDef} reduces to \eqq{opSep}, we find it useful to consider a more general decoder comprising the AD and a classical feedback link from Bob to Alice, that certainly cannot decrease the optimal rate \eqq{opDef}. Exploiting this feedback and the phase-insensitive property of $\Phi_{PI}$,  \eqq{PIG}, Alice can always perform the $\hat{U}_{j}$ instead of Bob. Hence all the AD's interactions are represented by a classical feedback to the encoder, that rearranges the remaining sequences $\alpha_{(j,N)}\in A_{(j,N)}$ into new sequences $\beta_{(j,N)} \in B_{(j,N)}$ with $B_{j}=\left\{\beta_{j}\in\mathbb{C}\right\}$ for all modes $j=1,\cdots,N$, before the transmission on the channel. Crucially, each choice of $\hat{U}_{j}$ corresponds to a different rearrangement performed by the encoder in such a way that the total average-energy constraint \eqq{eneCon} is still respected by the joint probability distribution $P_{B_{(1,N)}}(\vecund{\beta};E)$ of the new messages $B_{(1,N)}$. 
\begin{figure}[t!]\center
\includegraphics[trim={0cm 9cm 0cm 9cm},clip,scale=.43]{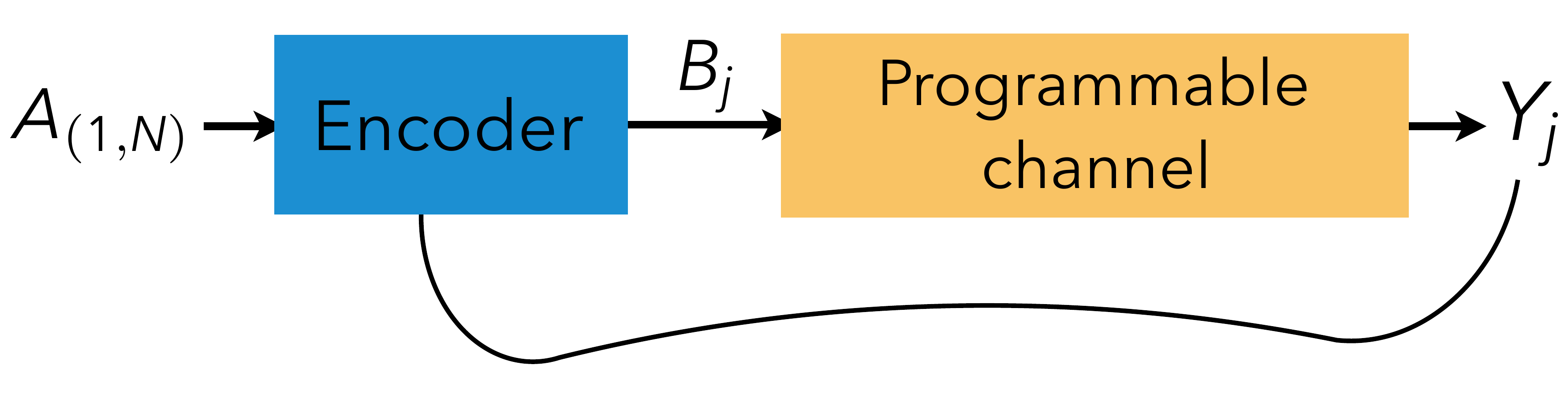}
\caption{Schematic depiction of the classical communication scheme induced by the quantum AD, Fig.~\ref{f15}a. The input sequence $\vecund{\alpha}\in A_{(1,N)}$ is encoded (blue/gray box) into single letters $\beta_{j}\in B_{j}$ that are sent one-by-one on a classical memory channel (yellow/light-gray box) with output $y_{j}\in Y_{j}$, for each $j=1,\cdots,N$. 
The adaptive passive interactions $\hat{U}_{j}$ of the quantum scheme here correspond to a classical feedback to the encoder, i.e., the encoding function that generates each $\beta_{j}$ depends on the input message $\vecund{\alpha}$ and on all previous results $y_{(1,j-1)}$. The single-mode PI Gaussian channel $\Phi_{PI}$ and adaptive POVM $\mathcal{M}_{j}$ employed on each letter $\beta_{j}$ instead correspond to several uses of a classical programmable channel, whose memory at each use depends only on previous outcomes through the parameter $\lambda$ characterizing the measurement, as in \eqq{cond}.}\label{f16}
\end{figure}\\
As a function of the encoded variables $\beta_{j}$, the rest of the AD scheme can then be rewritten as a single-mode classical \emph{programmable} channel, i.e., a channel with memory $\lambda$ that can be chosen adaptively depending on previous outcomes. The corresponding conditional probability at the $j$-th use then is
\begin{equation}\label{cond}
P_{Y_{j}|B_{j},Y_{(1,j-1)}}\big(y_{j}\big|\beta_{j},\lambda\left(y_{(1,j-1)}\right)\big)=\operatorname{Tr}\left[\hat{E}_{y_{j}}\left(\lambda\left(y_{(1,j-1)}\right)\right)\Phi_{PI}\left(\dketbra{\beta_{j}}\right)\right],
\end{equation}
where $\hat{E}_{y_{j}}(\lambda)$ are the elements of the POVM $\mathcal{M}_{j}(\lambda)$ as in \eqq{POVM}.

In light of the previous observations we can conclude that the AD of Fig.~\ref{f15}a, with additional classical communication from Bob to Alice, is equivalent to the classical programmable channel \eqq{cond} with feedback, as shown in Fig.~\ref{f16}. Hence the AD's optimal rate, \eqq{opDef}, is upper bounded by the feedback capacity of \eqq{cond}. Similarly, the capacity of the programmable channel without feedback for a single use is equal to the SDs optimal rate, \eqq{opSep}.
Eventually, the two classical capacities just defined are related via the following theorem, which is a generalization of Shannon's feedback capacity theorem, see \eqq{shanCap} and~\cite{ShanFeed,covThomBOOK}, to the class of programmable channels considered:
\begin{theorem}\label{feedMemo} 
The feedback capacity of a classical programmable channel is equal to its capacity without feedback and it is additive. \end{theorem}
\begin{proof}
Suppose we employ the channel to transmit a classical message $w\in W$ with probability distribution $P_{W}(w)$, outputting $y_{j}\in Y_{j}$ for each use $j$; the most general technique allows a feedback to the sender, who encodes the input message into a sequence of letters $\beta_{j}\in B_{j}$ through an encoding function $\beta_{j}=f\left(w, y_{(1,j-1)}\right)$ for each use $j$. If $\beta$ represents the complex amplitude of a signal we must impose a total average-energy constraint as in \eqq{eneCon}. The feedback capacity of this classical programmable channel at constrained total average-energy per mode $E$ is obtained by maximizing the mutual information over the input distribution, the encoding functions and the programmable parameters $\lambda\left(y_{(1,j-1)}\right)$ for each use:  
\begin{equation}\label{multiMode}
C^{fb}_{\infty}(E)=\lim_{N\rightarrow\infty}\max_{\substack{P_{W}\left(w\right),\\ f\left(w, y_{(1,j-1)}\right),\\ \lambda\left(y_{(1,j-1)}\right)\in\Lambda}}\frac{I\left(W:Y_{(1,N)}\right)}{N}.
\end{equation}
Similarly, for independent uses of the channel without feedback, the capacity at constrained average-energy $E$ can be defined as
\begin{equation}\label{singMode}
C_{1}(E)=\max_{\substack{P_{B_{1}}\left(\beta_{1};E\right),\\ \lambda\in\Lambda}}I\left(B_{1}:Y_{1}\right).
\end{equation}
Now let us note that $C^{fb}_{\infty}(E)\geq C_{1}(E)$, since among all adaptive schemes involved in the optimization \eqq{multiMode} there is one which employs no feedback and the same single-mode measurements that are optimal for \eqq{singMode}. To prove the opposite consider the following:
 \begin{equation}
 \begin{aligned}
I\big(W:Y_{(1,N)}\big)&=\sum_{j=1}^{N}I\left(B_{j}:Y_{j}\big|Y_{(1,j-1)}\right)\\
&\leq\ave{\sum_{j=1}^{N}C_{1}\left(E_{j}\left(y_{(1,j-1)}\right)\right)}_{P_{Y_{(1,N)}}\left(\vecund{y}\right)}\leq N C_{1}(E),\label{final}
\end{aligned}
\end{equation}
where the first equality follows form the chain rule of mutual information, \eqq{chainMInfo}, and the fact that conditioning over $W$ and $Y_{(1,j-1)}$ is equivalent to conditioning over $B_{j}$ and $Y_{(1,j-1)}$ thanks to the encoding functions, i.e., $H\left(Y_{j}\big|W,Y_{(1,j-1)}\right)=H\left(Y_{j}\big|B_{j},Y_{(1,j-1)}\right)$. The first inequality instead is obtained by employing the definition of \eqq{singMode} as an upper bound on each mutual information term in the sum and writing explicitly the average over the output distribution;  the last inequality follows from concavity of the classical capacity as a function of the energy and the total average-energy per mode constraint, i.e., $\sum_{j}E_{j}(y_{(1,j-1)})=NE$.
Eventually by plugging \eqq{final} into the definition \eqq{multiMode} we obtain the upper bound $C^{fb}_{\infty}(E)\leq C_{1}(E)$.
\end{proof}
This implies that the AD's optimal rate is not greater than the SD's one. Since the former is certainly not smaller than the latter, we conclude $R_{AD}(E)=R_{SD}(E)$. 

\subsection{Implications of the AD-SD equivalence}\label{end}
Our analysis implies that a broad class of adaptive decoders for coherent communication on PI Gaussian channels, including a majority of those most easily realizable with current technology, cannot beat the optimal single-mode-measurement rate of information transmission. This in turn seems to suggest that such decoders cannot achieve the classical capacity of PIe Gaussian channels; however there is no actual proof that joint decoders are really necessary for the task, so that this possibility remains open. In any case our result does not mean that block-coding techniques and adaptive receivers are completely useless for practical applications; indeed in general there may exist specific AD schemes that are more convenient to implement than SD ones and perform equally well, e.g., the Hadamard codes discussed in Sec~\ref{sec:had}.\\
Let us also note that, despite our result is very powerful in decoupling the AD's multi-mode structure for any kind of single-mode POVM, still the difficult optimization of the SD rate of Eq.~\eqq{opSep} is left if one wants an explicit expression of the rate for any set of POVMs. For example we can simplify this calculation for the set of single-mode receivers comprising a coherent displacement followed by any other kind of single-mode operation (the Kennedy receiver of \eqq{kennedy} belongs to this set). Indeed let us define the variance of a single-mode input probability distribution $P_{A}(\alpha)$ over coherent states as $V=\ave{|\alpha|^{2}}_{P_{A}(\alpha)}-|\ave{\alpha}_{P_{A}(\alpha)}|^{2}$; the energy is instead $E=\ave{|\alpha|^{2}}_{P_{A}(\alpha)}$. One can decide to put a constraint either on the energy or on the variance of the input signals and the former is stricter than the latter. It can then be shown that the net effect of the displacement in a coherent-state receiver is simply to enlarge the family of allowed input distributions from the energy- to the variance-constrained ones so that the optimal rate \eqq{opSep} can be computed on a shrunken set of allowed POVMs. \\
In particular, for the case of a single-mode Kennedy receiver, Eq.~(\ref{kennedy}), the SDs optimal rate for this has been computed in the low-energy limit $E\ll 1$ in~\cite{guhaDet,guhaDet1}, showing that it equals
\begin{equation}
R^{(oken)}_{SD}(E)=E\log\frac{1}{E}-E\log\log\frac{1}{E}+O(E).
\end{equation}
Moreover the same authors have shown that an AD scheme without unitaries has the same optimal rate and conjectured that also adaptive unitaries do not help. Our result exactly implies the validity of this conjecture for the particular choice of POVMs Eq.~(\ref{kennedy}). \\
Eventually our result intersects with those of ~\cite{takeokaGuha,shorAd}, expanding the set of adaptive receivers whose optimal rate is equal to that of separable ones. Indeed~\cite{takeokaGuha} compute the capacity of coherent communication with arbitrary adaptive Gaussian measurements, showing it is separable; here instead we considered a restricted interaction set, i.e., passive Gaussians, but an extended single-mode measurement one, i.e., arbitrary POVMs. As for~\cite{shorAd}, it is stated there that adaptive schemes based on partial single-mode measurements of all the modes may increase the optimal rate; here we considered only destructive single-mode measurements but included the simplest kind of interactions and still could not surpass separable decoding rates.  In particular, as previously noticed, our AD includes the use of ancillary systems if they interact with just one of the received modes, since it can be thought of as a part of the single-mode destructive measurements. Unfortunately the interaction of ancillary systems with multiple modes is not included, since it results in non-destructive measurements that could provide an advantage over SDs.

\chapter{Conclusions}\label{ch:conc}

\ifpdf
    \graphicspath{{Chapter4/Figs/Raster/}{Chapter4/Figs/PDF/}{Chapter4/Figs/}}
\else
    \graphicspath{{Chapter4/Figs/Vector/}{Chapter4/Figs/}}
\fi

In this thesis we have studied the problem of decoding classical information encoded on quantum states at the output of a quantum channel. We have treated the problem from an abstract perspective in Ch.~\ref{ch:Opt}, were we presented a method to decompose any quantum measurement into a sequence of easier nested measurements through a binary-tree search. Furthermore we showed that this decomposition can be used to build a capacity-achieving decoding protocol for classical communication on quantum channels and to solve the optimal discrimination of some sets of quantum states. These results clarified the structure of optimal quantum measurements, showing that it can be recast in a more operational fashion that could be more appealing to experimenters interested in increasing the performance of their decoding devices.\\
On the other hand, in Ch.~\ref{ch:Imple}, we considered a more practical approach and described three receiver structures for coherent states of the electromagnetic field with applications to single-mode state discrimination and multi-mode decoding at the output of a quantum channel. We treated the problem bearing in mind the technological limitations faced nowadays in the field of optical communications: we evaluated the performance of general decoding schemes based on such technology and reported increased performance of two schemes, the first one employing a non-Gaussian transformation and the second one employing a code tailored to be read out easily by the most common detectors. Eventually we characterized a large class of multi-mode adaptive receivers based on common technological resources, obtaining a no-go theorem for their capacity.\\
It turned out that there is still much work to do in order to simplify the theoretical predictions of the past and achieve reasonable decoder proposals. Conversely, from the practical perspective there is a strong need for novel optical-manipulation devices to go beyond the current decoding-rate performances.



\begin{appendices} 

  \chapter{The law of large numbers via Chernoff bound}\label{app:chernoff}
    In this appendix we compute an exponential bound for the law of large 
    numbers, which guarantees the convergence of the error probability of our protocol to zero.
   Indeed consider the small quantities $\epsilon, \epsilon_1, \epsilon_2$ 
    which appear in Secs.~\ref{sec:commQ},\ref{sec:bisDec}. These quantities 
    describe the high probability of finding respectively the average state $\hat{\rho}^{\otimes n}$ and the codeword states $\hat{\rho}_{\vecund{x}}$
    in their typical spaces, Eqs.~(\ref{typSpace},\ref{ctypSpace}).
    This is why they are connected, through the classical typical sets, Eqs.~(\ref{tSet},\ref{ctSet}), to 
    the law of large numbers. \\ 
    Consider for example the average state $\hat{\rho}_{\mathcal{S}}$ of the 
    source. We can easily prove that the probability of $N$ copies of the quantum state $\hat{\rho}_{\mathcal{S}}$ 
being in its $\delta-$typical space is equivalent to the probability of a random 
sample sequence $\vecund{x}$ of the corresponding classical source 
being in the classical $\delta-$typical space, i.e.,
\begin{equation}
\begin{aligned}
    \tr{\hat{\Pi}\hat{\rho}^{\otimes N}}&=\tr{\sum_{\vecund{j}\in T_{\delta,N}(\mathcal{J})}\ket{e_{\vecund{j}}}\bra{e_{\vecund{j}}} \sum_{\vecund{j}'\in T_{\delta,N}(\mathcal{J})}\lambda_{\vecund{j}'}\ket{e_{\vecund{j}'}}\bra{e_{\vecund{j}'} 
    }}\\
    &=\sum_{\vecund{j}\in T_{\delta,N}(\mathcal{J})}\lambda_{\vecund{j}}=Pr\left\{\vecund{j}\in T_{\delta,N}(\mathcal{J})\right\}.
\end{aligned}
\end{equation}  
A similar result is obtained for each codeword state $\hat{\rho}_{\vecund{x}}$, 
namely
\begin{equation}
  \tr{\hat{\Pi}_{\vecund{x}}\hat{\rho}_{\vecund{x}}}=Pr\left\{\vecund{j}\in T_{\delta,N}(\mathcal{J}|\vecund{x})\right\}.
\end{equation}
These probabilities can be bounded from above with the help of the law of large 
numbers. Consider for example the average typical space and the random variable $J$ determined by the labels of the spectrum of $\hat{\rho}_{\mathcal{S}}$, i.e., $j\in\mathcal{J}$. Then choose another random variable $Z$ taking values $z=z^{(j)}=-\log_2{\lambda_{j}}$.
 We also choose the same probability distribution both for $J$ and $Z$, i.e., $p_{z}=p_{z^{(j)}}=\equiv \lambda_j$. Then the law of large numbers states 
that, for any $\delta>0$, the probability that the average of $Z$ over $N$ 
extractions,
\begin{equation}
  \oneover{N}\sum_{n=1}^N z_n\equiv\bar{H}(\vecund{j}),
\end{equation} i.e., the sum of $N$ i.i.d. random variables, differs from its
expected value,
\begin{equation}
  \sum_{n=1}^N p_{z_n}z_n=H(J),
\end{equation}
for more than $\delta\geq0$ is lower than a quantity 
$\epsilon>0$, i.e.
\begin{equation}
 Pr\left\{\vecund{j}\in T_{\delta,N}(\mathcal{J})\right\}=Pr\left\{|\bar{H}(\vecund{j})-H(J)|\geq\delta\right\}\leq\epsilon.
\end{equation}
In usual derivations of this result the Chebyshev inequality is employed~\cite{covThomBOOK}, which 
gives a scaling behaviour $\epsilon\sim n^{-1}$. This is not sufficient for 
convergence of the error probability to $0$ for long sequences $N\rightarrow\infty$ 
in \eqq{dopolemma1}. Recalling also that the Chebyshev bound gives a dependence on the variance of the distribution, it is clear that such a scaling is a rough extimate, since the 
law of large numbers is known to be valid also for infinite-variance 
distributions. We therefore use the Chernoff bound~\cite{covThomBOOK} to obtain a faster, 
indeed exponential, convergence.\\
Consider first the Markov inequality~\cite{gallagerBOOK}, valid for any nonnegative random variable 
$t>0$ and $\delta>0$:
\begin{align}
  Pr\left\{t\geq\delta\right\}&=\sum_{t\geq\delta}p_t\leq\sum_{t\geq\delta}p_t\frac{t}{\delta}\\
  &\leq\oneover{\delta}\sum_t t p_t=\frac{\bar{t}}{\delta},
\end{align}
where we used a bar sign to indicate the average over the probability 
distribution of the random variable. The first inequality follows from 
introducing terms which certainly are less than one, given the constraint on the 
sum. The second inequality follows from adding positive terms to the sum, since
the random variable is positive. We now choose $t=e^{sw}$, with $w$ a new random 
variable\footnote{Note that the Chebyshev inequality can be obtained by choosing instead $t=(w-\bar{w})^2$.}
and $\delta=e^{sA}$, without loss of generality.
The Markov inequality then reads
\begin{equation}
  Pr\left\{e^{sw}\geq e^{sA}\right\}\leq e^{-sA}g_w(s)
\end{equation}
for any $s,A$, where we called $g_w(s)=\overline{\exp{(sw)}}$ the moment generating function of the random variable 
$w$, i.e.
\begin{equation}
  \overline{w^N}=\frac{d^N g_w(s)}{ds^N}\Big|_{s=0}.
\end{equation}
Now observe that the above inequality between exponentials has two 
different meanings depending on the sign of $s$, implying both
\begin{align}
  &Pr\left\{w\geq A\right\}\leq e^{-sA}g_w(s)\quad s>0\label{codapiu}\\
  &Pr\left\{w\leq A\right\}\leq e^{-sA}g_w(s)\quad s<0.\label{codameno}
\end{align}
These two relations give bounds on the tails of the $w$ probability 
distribution. In order to evaluate how tight such bounds are, we consider the specific case of $w$ being the sum of $N$ i.i.d random variables $x_n$, 
implying for the moment generating function
\begin{align}
  g_w(s)=\overline{\exp\left(s\sum_{n=1}^N 
  x_n\right)}=\prod_{n=1}^N\overline{e^{sx_n}}=\left(g_x(s)\right)^N,
\end{align} 
and take $A=na$, without loss of generality. The previous inequalities become
\begin{align}
  &Pr\left\{\oneover{n}\sum_{n=1}^Nx_n\geq a\right\}\leq \exp\left[-N\left(sa-\ln 
  g_x(s)\right)\right]\quad s>0\label{cp}\\
  &Pr\left\{\oneover{n}\sum_{n=1}^Nx_n\leq a\right\}\leq \exp\left[-N\left(sa-\ln 
  g_x(s)\right)\right]\quad s<0\label{cm}.
\end{align}
We now need to evaluate the behaviour of the coefficient function in the 
exponential: 
\begin{equation}
  h(s)=sa-\ln g_x(s).
\end{equation}
Consider first some properties of $\mu_x(s)=\ln g_x(s)$, following from the nature of the 
moment generating function:
\begin{itemize}
  \item $\mu_x(s=0)=0$, since $g_x(s=0)=1$;
  \item $\mu'_x(s=0)=g'_x(s=0)/g_x(s=0)=\bar{x}$, since $g_x(s=0)=\bar{x}$;
  \item it is convex \begin{align}
  \mu''_x(s)&=\frac{g''_x(s)}{g_x(s)}-\left(\frac{g'_x(s)}{g_x(s)}\right)^2\\
  &=\ave{x^2}_{e}-\ave{x}_{e}^2=\ave{\left(x-\ave{x}_e\right)^2}_e\geq 0
  \quad \forall s,
\end{align}
where we have indicated with \begin{equation}
\ave{f(x)}_e=\frac{\overline{f(x)e^{sx}}}{\overline{e^{sx}}}
\end{equation} 
the probability average with weight $e^{sx}$.
\end{itemize}
 From the previous properties, it follows that the slope of the function, 
starting at $\bar{x}$ at the origin, increases for $s>0$ and decreases for 
$s<0$. Expanding $\mu_x(s)$ for small $s$ at second order, we have for the 
coefficient function
\begin{equation}
  h(s)\simeq (a-\bar{x})s-\frac{s^2}{2}\mu''_x(0).\label{approxcoeff}
\end{equation}
This approximate function (for small $s$) is zero at
\begin{equation}
  s^*=\frac{2(a-\bar{x})}{\mu''_x(0)}.\label{zeropoint}
\end{equation}
Consider now the $s>0$ inequality \eqq{cp}. If $a>\bar{x}$, then the zero $s^*$ is 
positive and inside the range of validity of the inequality. Thus $h(s)>0$ for all $s<s^*$ in the range: the first inequality has a tight bound. Vice versa if $a<\bar{x}$, the zero $s^*$ 
is negative and $h(s)<0$ in the whole range of validity of the first inequality, making it useless.
\\ The situation is reversed when considering the $s<0$ inequality \eqq{cm}. In this case we 
need $s^*<0$, i.e., $a<\bar{x}$, and for any $s>s^*$ in the range the coefficient function will be positive 
again, providing a tight bound for the second inequality.
By calling
\begin{equation}
  h_p=\sup_{s>0}h(s),\quad h_m=\sup_{s<0}h(s),\qquad h_p,h_m>0, 
\end{equation}
the supremum of $h(s)$ in each region, we can thus rewrite the inequalities 
as tight bounds, taking respectively $a=\bar{x}+\delta>\bar{x}$ in the first inequality and $a=\bar{x}-\delta<\bar{x}$, 
in the second one, for any $\delta>0$:
\begin{align}
  &Pr\left\{\oneover{n}\sum_{n=1}^Nx_n-\bar{x}\geq \delta\right\}\leq e^{-Nh_p}\\
  &Pr\left\{\oneover{n}\sum_{n=1}^Nx_n-\bar{x}\leq -\delta\right\}\leq e^{-Nh_m}.
\end{align}
Eventually we sum the previous inequalities to obtain the law of large numbers 
with exponentially decreasing tails
\begin{equation}
  Pr\left\{\left|\oneover{n}\sum_{n=1}^Nx_n-\bar{x}\right|\geq \delta\right\}\leq 
  e^{-Nh_p}+e^{-Nh_m}=O(e^{-N})=\epsilon.
\end{equation}
Observe that the small quantity $\epsilon>0$ obtained in this way, exponentially decreasing with increasing $N$, 
also depends on the difference parameter $\delta$ that we 
chose, as of course is to be expected. Indeed this dependence is implicit in the 
definition of $h_p$,$h_m$: by choosing $\delta$, we set different values of $a$ (for both the $s>0$ and $s<0$ cases)
and this in turn varies the point $s^*$ \eqq{zeropoint}, i.e., the range of values of $s$ ($s<s^*$ or $s>s^*$) which can be 
chosen to maximize the coefficient functions. In particular, since the 
expression \eqq{approxcoeff} is a small-s expansion, we do not know what the 
absolute supremum of $h(s)$ is and where it is located\footnote{The function $\mu_x(s)$ is convex, but we do not know its behaviour at large $s$.}. Thus by varying the 
range of $s$ accessible through the tuning of $\delta$, we may happen to exclude 
this and other local supremum points, resulting in (possibly discontinuously) varying values of $h_p$,$h_m$. 
\\
In any case, for our purpose we need only the existence of a range of values 
$s$, both above and below zero and depending on $\delta$, where the coefficient function $h(s)$ is 
positive, and this is guaranteed by the properties of the $\mu_x(s)$ function, 
respectively when $a>\bar{x}$ for positive $s$ and when $a<\bar{x}$ for negative 
$s$.


  \chapter{Proofs of measurement Lemmas}\label{app:measLemmas}
    We give here the proofs of the measurement Lemmas of Sec.~\ref{subsec:distance}. A useful tool is a statement a little bit more general than \eqq{normLemma}:
    \begin{lemma}\label{trdist}
    For any unnormalized density operator $\hat{\omega}$, its trace-norm can be written as
   \begin{equation}\label{trdisteq}
   \norm{\hat{\omega}}_{1}=\max_{-\hat{\mathbf{1}}\leq\hat{\Gamma}\leq\hat{\mathbf{1}}}\tr{\hat{\Lambda}\hat{\omega}}.
   \end{equation}
    \end{lemma}
    \begin{proof}
      For a hermitian operator we can always write $\hat{\omega}=\hat{A}-\hat{B}$, where $\hat{A},\hat{B}$ are positive matrices with disjoint supports, 
      representing $\hat{\omega}$ respectively in the positive and negative part of 
      its support. Consider then the operator $\hat{\Lambda}=\hat{\Pi}_A-\hat{\Pi}_B$, with $\hat{\Pi}_A$ 
      and $\hat{\Pi}_B$ being projectors respectively on the support of $\hat{A}$ and of $\hat{B}$. 
      For this operator we can clearly state that 
      $-\hat{\mathbf{1}}\leq\hat{\Lambda}\leq\hat{\mathbf{1}}$. By construction we obtain thus an operator which saturates the bound 
      \eqq{trdisteq}:
\begin{equation}
      \begin{aligned}
        \tr{\hat{\Lambda}\hat{\omega}}&=\tr{\left(\hat{\Pi}_A-\hat{\Pi}_B\right)\hat{A}}-\tr{\left(\hat{\Pi}_A-\hat{\Pi}_B\right)\hat{B}}\\
        &=\tr{\hat{A}}+\tr{\hat{B}}=\tr{|\hat{\omega}|}=\norm{\hat{\omega}}_1.
      \end{aligned}
      \end{equation}
      In order to complete the proof, we need to show that $\hat{\Lambda}$ 
      is the maximizing operator among all possible
      $-\hat{\mathbf{1}}\leq\hat{\Gamma}\leq\hat{\mathbf{1}}$. First observe, by diagonalising $\hat{A}$ 
      and $\hat{B}$, that
      \begin{align}
        &\tr{\hat{\Lambda} 
        \hat{A}}=\sum_k\alpha_k\bra{a_k}\hat{\Lambda}\ket{a_k}\leq\sum_k\alpha_k\braket{a_k}{a_k}=\tr{\hat{A}}\\
        &\tr{\hat{\Lambda} 
        \hat{B}}=\sum_k\beta_k\bra{b_k}\hat{\Lambda}\ket{b_k}\geq\sum_k\beta_k(-\braket{b_k}{b_k})=-\tr{\hat{B}}.
      \end{align}
      Thus 
      \begin{align}
        \tr{\hat{\Lambda} \hat{\omega}}=\tr{\hat{\Lambda} \hat{A}}-\tr{\hat{\Lambda} \hat{B}}\leq 
        \tr{\hat{A}}+\tr{\hat{B}}=Tr|\hat{\omega}|=\norm{\hat{\omega}}_1.
      \end{align}
    \end{proof}
    
    \begin{proof}[Proof of Lemma \ref{appclose}]
      Consider that
      \begin{equation}
        2D(\hat{\rho},\hat{\sigma})=\norm{\hat{\rho}-\hat{\sigma}}_1=\max_{-\hat{\mathbf{1}}\leq\hat{\Gamma}\leq\hat{\mathbf{1}}}\tr{\hat{\Gamma}(\hat{\sigma}-\hat{\rho})}\geq\tr{\hat{\hat{E}}(\hat{\sigma}-\hat{\rho})},
      \end{equation}
      which follows from applying Lemma \ref{trdist} and from the fact that $0\leq \hat{\hat{E}}\leq\hat{\mathbf{1}}$ 
      surely is one of the operators included in the maximization procedure. The 
      result \eqq{appcloseq} is then easily obtained by separating the trace and rearranging terms in the 
      previous inequality.
    \end{proof}
    
    \begin{proof}[Proof of Lemma \ref{gentop}]
      Consider that
\begin{equation}\label{gentspezzato}
      \begin{aligned}
2D\left(\sqrt{\hat{E}}\hat{\rho}\sqrt{\hat{E}},\hat{\rho}\right)&=\norm{\hat{\rho}-\sqrt{\hat{E}}\hat{\rho}\sqrt{\hat{E}}}_1\leq\norm{\hat{\rho}-\sqrt{\hat{E}}\hat{\rho}}_1+
   \norm{\sqrt{\hat{E}}\hat{\rho}-\sqrt{\hat{E}}\hat{\rho}\sqrt{\hat{E}}}_1
  \\&=\norm{\left(\hat{\mathbf{1}}-\sqrt{\hat{E}}\right)\sqrt{\hat{\rho}}\cdot\sqrt{\hat{\rho}}}_1+
  \norm{\sqrt{\hat{E}}\cdot\hat{\rho}\left(\hat{\mathbf{1}}-\sqrt{\hat{E}}\right)}_1.
   \end{aligned}
   \end{equation}
     thanks to the triangular inequality for the trace distance, \eqq{triangTD}. Now for the first term we consider the spectral decomposition $\{\sqrt{\lambda_k},\ket{f_k}\}$ of $\sqrt{\hat{\rho}}$ and 
     use again the triangular inequality for the trace norm:
     \begin{equation}
     \begin{aligned}
      \norm{\left(\hat{\mathbf{1}}-\sqrt{\hat{E}}\right)\sqrt{\hat{\rho}}\sum_k\sqrt{\lambda_k}\ket{f_k}\bra{f_k}}_1 
      &\leq\sum_k\sqrt{\lambda_k}\norm{\left(\hat{\mathbf{1}}-\sqrt{\hat{E}}\right)\sqrt{\hat{\rho}}\ket{f_k}\bra{f_k}}_1\\
     &=\sum_k\sqrt{\lambda_k}Tr\sqrt{\ket{f_k}\bra{f_k}\sqrt{\hat{\rho}}\left(\hat{\mathbf{1}}-\sqrt{\hat{E}}\right)^2\sqrt{\hat{\rho}}\ket{f_k}\bra{f_k}}\\
     &=\sum_k\sqrt{\lambda_k}\sqrt{\bra{f_k}\sqrt{\hat{\rho}}\left(\hat{\mathbf{1}}-\sqrt{\hat{E}}\right)^2\sqrt{\hat{\rho}}\ket{f_k}}.
     \end{aligned}
     \end{equation}
    Apply then the Cauchy-Schwarz inequality
    \begin{equation}
      |\vecund{x}\cdot\vecund{y}|^2\leq|\vecund{x}|^2\cdot|\vecund{y}|^2,
    \end{equation}
    with $x_k=\sqrt{\lambda_k}$ and 
    $y_k=\sqrt{\bra{f_k}\sqrt{\hat{\rho}}\left(\hat{\mathbf{1}}-\sqrt{\hat{E}}\right)^2\sqrt{\hat{\rho}}\ket{f_k}}$,
to obtain 
\begin{equation}
\begin{aligned}
  \norm{\left(\hat{\mathbf{1}}-\sqrt{\hat{E}}\right)\sqrt{\hat{\rho}}\sum_k\sqrt{\lambda_k}\ket{f_k}\bra{f_k}}_1 
  &\leq\sqrt{\sum_k\lambda_k\sum_j\bra{f_j}\sqrt{\hat{\rho}}\left(\hat{\mathbf{1}}-\sqrt{\hat{E}}\right)^2\sqrt{\hat{\rho}}\ket{f_j}}\\
  &\leq\sqrt{\tr{\hat{\rho}\left(\hat{\mathbf{1}}-\sqrt{\hat{E}}\right)^2}},
\end{aligned}
\end{equation}
where we used the fact that $\tr{\hat{\rho}}=\sum_k\lambda_k\leq1$.
For the second term in \eqq{gentspezzato} use instead the spectral decomposition $\{\sqrt{\nu_k},\ket{e_k}\}$ of $\sqrt{\hat{E}}$ and proceed in a similar way as before:
\begin{equation}
\begin{aligned}
  \Big|\Big|\sum_k\sqrt{\nu_k}\ket{e_k}&\bra{e_k}\hat{\rho}\left(\hat{\mathbf{1}}-\sqrt{\hat{E}}\right)\Big|\Big|_1
\leq\sum_k\sqrt{\nu_k}\norm{\ket{e_k}\bra{e_k}\hat{\rho}\left(\hat{\mathbf{1}}-\sqrt{\hat{E}}\right)}_1\\
&=\sum_k\sqrt{\nu_k}\norm{\left(\hat{\mathbf{1}}-\sqrt{\hat{E}}\right)\hat{\rho}\ket{e_k}\bra{e_k}}_1=\sum_k\sqrt{\nu_k}\sqrt{\bra{e_k}\hat{\rho}\left(\hat{\mathbf{1}}-\sqrt{\hat{E}}\right)^2\hat{\rho}\ket{e_k}}\\
&\leq\sqrt{\sum_k\nu_k\sum_j\bra{e_j}\hat{\rho}\left(\hat{\mathbf{1}}-\sqrt{\hat{E}}\right)^2\hat{\rho}\ket{e_j}}\leq\sqrt{\tr{\hat{\rho}^2\left(\hat{\mathbf{1}}-\sqrt{\hat{E}}\right)^2}}\\
&\leq\sqrt{\tr{\hat{\rho}\left(\hat{\mathbf{1}}-\sqrt{\hat{E}}\right)^2}},
\end{aligned}
\end{equation}
where we used the triangular inequality, the invariance of the trace norm under hermitian conjugation, the Cauchy-Schwarz inequality, the 
fact that $\tr{\hat{E}}=\sum_k\nu_k\leq1$ and the property $\hat{\rho}^2\leq\hat{\rho}\leq\hat{\mathbf{1}}$. 
 The inequality \eqq{gentspezzato} then simply becomes
 \begin{equation}
   2D\left(\sqrt{\hat{E}}\hat{\rho}\sqrt{\hat{E}},\hat{\rho}\right)\leq2\sqrt{\tr{\hat{\rho}\left(\hat{\mathbf{1}}-\sqrt{\hat{E}}\right)^2}}\leq 2\sqrt{\tr{\hat{\rho}\left(\hat{\mathbf{1}}-\hat{E}\right)}},\label{quasifinale}
 \end{equation}
 since 
 \begin{equation}
   0\leq \hat{E}\leq\sqrt{\hat{E}}\leq\hat{\mathbf{1}}\Rightarrow 
   \left(\hat{\mathbf{1}}-\sqrt{\hat{E}}\right)^2=\hat{\mathbf{1}}+\hat{E}-2\sqrt{\hat{E}}\leq\hat{\mathbf{1}}-\hat{E}.
 \end{equation}
 Eventually we take the code average of \eqq{quasifinale} and use the 
 concavity of the square-root function and the hypotesis \eqq{gentopequno} to 
 obtain the thesis \eqq{gentopeqdue}:
 \begin{align}
   \ave{2D\left(\sqrt{\hat{E}}\hat{\rho}\sqrt{\hat{E}},\hat{\rho}\right)}&\leq2\sqrt{\ave{\tr{\hat{\rho}\left(\hat{\mathbf{1}}-\hat{E}\right)}}}
   \leq2\sqrt{\epsilon}.
 \end{align}
  \end{proof}


\chapter{Some cases of optimal discriminating measurement}\label{app:casesOpSol} 
In this appendix we study the function $\mathcal{\tiny{F}}_{\hat{Q}}(\hat{A},\hat{B},\hat{C})$ appearing in Eqs.~(\ref{DEFFQABC}) and discuss its optimization in the cases mentioned in Sec.~\ref{sec:opDisc}, proving Proposition~\ref{cases} and Remark~\ref{remark1}.
The optimization of $\mathcal{\tiny{F}}_{\hat{Q}}(\hat{A},\hat{B},\hat{C})$ is difficult because of the competing interests of the three terms composing it. Indeed each single term of \eqq{DEFFQABC} is maximized by a different operator $\hat{Q}$: the first one is maximum when $\hat{Q}$ is the projector on the positive support of $\hat{A}$; the second one is maximum when $\hat{Q}$ is the identity on the whole Hilbert space of the system; the third one is maximum when $\hat{Q}$ is zero. Hence we can solve the problem exactly only if the three operators exhibit specific properties. \\
We start by observing that the function is subadditive in all its arguments, i.e., for any set of operators $\{\hat{A}_{j}, \hat{B}_{j}, \hat{C}_{j}\}_{j=1,\cdots,n}$ it holds:
\begin{equation}
\mathcal{\tiny{F}}(\sum_{j}\hat{A}_{j},\sum_{j}\hat{B}_{j},\sum_{j}\hat{C}_{j})\leq\sum_{j}\mathcal{\tiny{F}}(\hat{A}_{j},\hat{B}_{j},\hat{C}_{j}).\label{subAdd}
\end{equation}
This follows from the subadditivity of the trace norm. 
We can now state some lemmas that help demonstrate Proposition~\ref{cases}. Throughout the Appendix, the notation $\hat{\mathbf{1}}_{X}$ represents the projector on the support of the operator $\hat{X}$.
\begin{lemma}\label{aPos}
Let us suppose that $\hat{A}$ is positive semidefinite, $\hat{B}$ has support inside the support of $\hat{A}$ and that $\hat{C}$ and $\hat{A}$ have orthogonal supports. Then 
\begin{equation}
\mathcal{\tiny{F}}\left(\hat{A},\hat{B},\hat{C}\right)=\tr{\hat{A}}+\norm{\hat{B}}_{1}+\norm{\hat{C}}_{1}.
\end{equation}
\end{lemma}
\begin{proof}
Consider the set of operators $\left\{\hat{A}_{j}=\hat{A}\delta_{j,1}, \hat{B}_{j}=\hat{B}\delta_{j,2}, \hat{C}_{j}=\hat{C}\delta_{j,3}\right\}_{j=1,2,3}$ and apply the subadditivity property \eqq{subAdd}: 
\begin{equation}
\mathcal{\tiny{F}}\left(\hat{A},\hat{B},\hat{C}\right)\leq\mathcal{\tiny{F}}(\hat{A},0,0)+\mathcal{\tiny{F}}(0,\hat{B},0)+\mathcal{\tiny{F}}(0,0,\hat{C})=\tr{\hat{A}}+\norm{\hat{B}}_{1}+\norm{\hat{C}}_{1}.\label{subAddApp}
\end{equation}
The latter inequality can be saturated under the hypotheses of this lemma by taking $\hat{Q}=\hat{\mathbf{1}}_{A}$.\end{proof}

\begin{lemma}\label{aNeg}
Let us suppose that $\hat{A}$ is negative semidefinite, $\hat{C}$ has support inside the support of $\hat{A}$ and that $\hat{B}$ and $\hat{A}$ have orthogonal supports. Then 
\begin{equation}
\mathcal{\tiny{F}}\left(\hat{A},\hat{B},\hat{C}\right)= \norm{\hat{B}}_{1}+\norm{\hat{C}}_{1}.
\end{equation}
\end{lemma}
\begin{proof}
Consider the same set of operators of Lemma~\ref{aPos} and apply again the subadditivity property \eqq{subAddApp}, then use the fact that $\hat{A}\leq0$: 
\begin{equation}
\mathcal{\tiny{F}}\left(\hat{A},\hat{B},\hat{C}\right)\leq\mathcal{\tiny{F}}(\hat{A},0,0)+\mathcal{\tiny{F}}(0,\hat{B},0)+\mathcal{\tiny{F}}(0,0,\hat{C})=\norm{\hat{B}}_{1}+\norm{\hat{C}}_{1}.
\end{equation}
The latter inequality can be saturated under the hypotheses of this lemma by taking $\hat{Q}=\hat{\mathbf{1}}_{B}$. 
\end{proof}
Hence we can prove Remark~\ref{remark1} and the first case of Proposition~\ref{cases}: let $\hat{A}=\hat{A}_{+}\oplus(- \hat{A}_{-})$ be the decomposition of $\hat{A}$ in terms of its positive and negative parts, with $\hat{A}_{\pm}\geq 0$, and suppose that $\hat{B}$, $\hat{C}$ have support respectively inside the support of $\hat{A}_{+}$, $\hat{A}_{-}$. Then consider the set of operators $\left\{\hat{A}_{j}=(-1)^{j+1}\hat{A}_{(-1)^{j+1}}, \hat{B}_{j}=\hat{B}\delta_{j,1}, \hat{C}_{j}=\hat{C}\delta_{j,2}\right\}_{j=1,2}$ and apply the subadditivity property \eqq{subAdd}, together with Lemmas~\ref{aPos}, \ref{aNeg}: 
\begin{equation}
\mathcal{\tiny{F}}\left(\hat{A},\hat{B},\hat{C}\right)\leq\mathcal{\tiny{F}}\left(\hat{A}_{+},\hat{B},0\right)+\mathcal{\tiny{F}}\left(-\hat{A}_{-},0,\hat{C}\right)=\operatorname{Tr}\left[\hat{A}_{+}\right]+\norm{\hat{B}}_{1}+\norm{\hat{C}}_{1},\label{spec}
\end{equation}
which is saturated by a measurement operator $\hat{Q}=\hat{\mathbf{1}}_{A_{+}}$. 
This expression is equivalent to that given in \eqq{value} under the current hypotheses, indeed in this case it holds 
\begin{equation}
\left(\hat{A}+\abs{\hat{B}}-\abs{\hat{C}}\right)_{+}=\left(\left(\hat{A}_{+}+\abs{\hat{B}}\right)\oplus\left(-\hat{A}_{-}-\abs{\hat{C}}\right)\right)_{+}=\hat{A}_{+}+\abs{\hat{B}},
\end{equation}
so that \eqq{spec} becomes
\begin{equation} 
\mathcal{\tiny{F}}\left(\hat{A},\hat{B},\hat{C}\right)=\operatorname{Tr}\left[\hat{A}_{+}+\abs{\hat{B}}\right]+\norm{\hat{C}}_{1}=\operatorname{Tr}\left[\left(\hat{A}_{+}+\abs{\hat{B}}-\abs{\hat{C}}\right)_{+}\right]+\norm{\hat{C}}_{1}.
\end{equation} 
As for the second and third cases of Proposition~\ref{cases}, let us first note that 
\begin{equation}\label{NQB}
\begin{aligned}
\norm{\sqrt{\hat{Q}}\hat{B}\sqrt{\hat{Q}}}_1 &\leq \norm{\sqrt{\hat{Q}}\hat{B}_+\sqrt{\hat{Q}}}_1 + \norm{\sqrt{\hat{Q}}\hat{B}_-\sqrt{\hat{Q}}}_1\\
&= \operatorname{Tr}\left[\hat{Q}\left(\hat{B}_+ + \hat{B}_-\right)\right] = \operatorname{Tr}\left[\hat{Q}\abs{\hat{B}}\right],
\end{aligned}
\end{equation}
where $\hat{B}_{\pm}$ are the positive and negative parts of $\hat{B}$ as defined above for $\hat{A}$, and analogously
\begin{equation} \label{NQC}
\norm{\sqrt{\hat{\mathbf{1}}-\hat{Q}}\hat{C}\sqrt{\hat{\mathbf{1}}-\hat{Q}}}_1 \le \tr{\left(\hat{\mathbf{1}}-\hat{Q}\right)\abs{\hat{C}}}.
\end{equation}
We then have
\begin{equation}\label{FQ2}
\mathcal{\tiny{F}}_{\hat{Q}}(\hat{A},\hat{B},\hat{C}) \leq \tr{\hat{Q}\left(\hat{A}+\abs{\hat{B}}-\abs{\hat{C}}\right)} + \norm{\hat{C}}_1\leq \tr{\left(\hat{A}+\abs{\hat{B}}-\abs{\hat{C}}\right)_+} + \norm{\hat{C}}_1.
\end{equation}
The second inequality in \eqq{FQ2} is saturated by taking $\hat{Q}=\hat{\mathbf{1}}_{(A+\abs{B}-\abs{C})_{+}}$. The inequalities~(\ref{NQB},\ref{NQC}) and hence the first inequality in \eqq{FQ2} are saturated in both the second and third cases of Proposition~\ref{cases}, though for different reasons:
\begin{itemize}
\item If $\hat{B}$ and $\hat{C}$ have a definite sign, then it holds $\hat{B}=\hat{B}_{+}$ or $\hat{B}=\hat{B}_{-}$, so that \eqq{NQB} is saturated and analogously \eqq{NQC};
\item If $\hat{A}$, $\hat{B}$, $\hat{C}$ all commute with each other, then Eqs.~(\ref{NQB},\ref{NQC}) are saturated by any operator $\hat{Q}$ which commutes with both $\hat{B}$ and $\hat{C}$. Eventually the choice $\hat{Q}=\hat{\mathbf{1}}_{(A+\abs{B}-\abs{C})_{+}}$ necessary to saturate \eqq{FQ2} satisfies this latter condition in the case considered.
\end{itemize}
Finally we note that the previous case of commuting operators, as well as further results, can also be derived by applying the symmetry property of the optimal success probabilities (\ref{prob4}, \ref{prob3}) to obtain recursive formulas, as discussed after Remark~\ref{recRem}, but still a full solution cannot be found in this way.


\chapter{Computation of the optimal success probability in the qubit case}\label{app:qubitOp}
In this appendix we derive the results (\ref{fQub}, \ref{defSign}) explicitly. 
As a preliminary recall that, for any three vectors $\vecund{a}, \vecund{b}, \vecund{c}\in\mathbb{R}^{3}$ and the vector of Pauli matrices $\vecund{\hat{\sigma}}$ it holds: 
\begin{align}
&\left(\vecund{a}\cdot\vecund{\hat{\sigma}}\right)\left(\vecund{b}\cdot\vecund{\hat{\sigma}}\right)=\left(\vecund{a}\cdot\vecund{b}\right)+i \left(\vecund{a}\times\vecund{b}\right)\cdot\vecund{\hat{\sigma}},\label{vec1}\\
&\left(\vecund{a}\times\left(\vecund{b}\times\vecund{c}\right)\right)=\left(\vecund{a}\cdot\vecund{c}\right)\vecund{b}-\left(\vecund{a}\cdot\vecund{b}\right)\vecund{c}.\label{vec2}
\end{align}
Moreover, given a positive operator $\hat{Q}$ on $\mathcal{H}_{2}$, the coefficients $c_{\sqrt{Q}}$, $\vecund{r}_{\sqrt{Q}}$ of its square root $\sqrt{\hat{Q}}$ can be expressed in terms of its coefficients $c_{Q}$, $\vecund{r}_{Q}$ as:
\begin{equation}
\begin{cases}
c_{Q}=\left(c_{\sqrt{Q}}\right)^{2}+\left(r_{\sqrt{Q}}\right)^{2}\\
r_{Q}=2c_{\sqrt{Q}}r_{\sqrt{Q}}
\end{cases}\leftrightarrow\quad
\begin{cases}
c_{\sqrt{Q}}=\frac{\sqrt{c_{Q}+r_{Q}}+\sqrt{c_{Q}-r_{Q}}}{2}\\
r_{\sqrt{Q}}=\frac{\sqrt{c_{Q}+r_{Q}}-\sqrt{c_{Q}-r_{Q}}}{2}.
\end{cases}\hspace{10pt}
\label{sqrtQ}
\end{equation}
with $\vecund{r}_{Q}\parallel\vecund{r}_{\sqrt{Q}}$.
In order to evaluate $\mathcal{\tiny{F}}_{Q}(A,B,C)$ we can compute its first two terms, while the third one is similar to the second one. Let us start with the product $\hat{Q}\hat{A}$: it is a generic operator with coefficients 
\begin{align}
c_{QA}&=c_{Q}c_{A}+\vecund{r}_{Q}\cdot\vecund{r}_{A}\label{cqa}\\
\vecund{r}_{QA}&=c_{Q}\vecund{r}_{A}+c_{A}\vecund{r}_{Q}+i\left(\vecund{r}_{Q}\times\vecund{r}_{A}\right),\label{rqa}
\end{align}
computed by applying \eqq{vec2}. Thus the first term of $\mathcal{\tiny{F}}_{Q}^{(2)}$ is simply $\tr{\hat{Q}\hat{A}}=2c_{QA}$. \\
As for the product $\sqrt{\hat{Q}}\hat{B}\sqrt{\hat{Q}}$, its first coefficient is simple: 
\begin{equation}
c_{\sqrt{Q}B\sqrt{Q}}=\oneover{2}\tr{\sqrt{\hat{Q}}\hat{B}\sqrt{\hat{Q}}}=c_{QB},
\end{equation}
easily obtained by relabelling \eqq{cqa}. The vector of coefficients instead is
\begin{equation}\label{rSqBSq}
\begin{aligned}
\vecund{r}_{\sqrt{Q}B\sqrt{Q}}&=c_{\sqrt{Q}B}\vecund{r}_{\sqrt{Q}}+c_{\sqrt{Q}}\vecund{r}_{\sqrt{Q}B}+i\left(\vecund{r}_{\sqrt{Q}B}\times\vecund{r}_{\sqrt{Q}}\right)\\
&=2c_{B}c_{\sqrt{Q}}\vecund{r}_{\sqrt{Q}}+\left(c_{\sqrt{Q}}\right)^{2}\vecund{r}_{B}+\left(\vecund{r}_{\sqrt{Q}}\cdot\vecund{r}_{B}\right)\vecund{r}_{\sqrt{Q}}-\left(r_{\sqrt{Q}}\right)^{2}\vecund{r}_{B}+\left(\vecund{r}_{\sqrt{Q}}\cdot\vecund{r}_{B}\right)\vecund{r}_{\sqrt{Q}}\\
&=c_{B}\vecund{r}_{Q}+2\left(\vecund{r}_{\sqrt{Q}}\cdot\vecund{r}_{B}\right)\vecund{r}_{\sqrt{Q}}+\left(\left(c_{\sqrt{Q}}\right)^{2}-\left(r_{\sqrt{Q}}\right)^{2}\right)\vecund{r}_{B},
\end{aligned}
\end{equation}
where we have first computed the product between $\sqrt{\hat{Q}}\hat{B}$ and $\sqrt{\hat{Q}}$, then substituted the expression for the former by relabelling once again the product $\hat{Q}\hat{A}$ and employed \eqq{sqrtQ}. We are interested in the absolute value of $\sqrt{\hat{Q}}\hat{B}\sqrt{\hat{Q}}$, i.e., the sum of the absolute value of its eigenvalues $\lambda^{(\pm)}_{\sqrt{Q}B\sqrt{Q}}=c_{\sqrt{Q}B\sqrt{Q}}\pm r_{\sqrt{Q}B\sqrt{Q}}$. Hence the only dependence of the final expression on \eqq{rSqBSq} is through its norm:
\begin{equation}\label{rMod}
\begin{aligned}
\left(r_{\sqrt{Q}B\sqrt{Q}}\right)^{2}&=\left(c_{B}r_{Q}\right)^{2}+4(\vecund{r}_{\sqrt{Q}}\cdot\vecund{r}_{B})^{2}\left(c_{\sqrt{Q}}\right)^{2}\\
&+\left(\left(c_{\sqrt{Q}}\right)^{2}-\left(r_{\sqrt{Q}}\right)^{2}\right)^{2}\left(r_{B}\right)^{2}+2c_{B}\left(\left(c_{\sqrt{Q}}\right)^{2}+\left(r_{\sqrt{Q}}\right)^{2}\right)\left(\vecund{r}_{Q}\cdot\vecund{r}_{B}\right)\\
&=\left(c_{Q}c_{B}+\vecund{r}_{Q}\cdot\vecund{r}_{B}\right)^{2}+\left(\left(r_{B}\right)^{2}-\left(c_{B}\right)^{2}\right)\left(\left(c_{Q}\right)^{2}-\left(r_{Q}\right)^{2}\right),
\end{aligned}
\end{equation}
which we have simplified by employing the relations \eqq{sqrtQ}. Eventually, we have to distinguish between two cases:
 \begin{itemize}
 \item If both $\hat{B}$ and $\hat{C}$ have definite sign then they can always be taken to be positive semidefinite, up to a relabeling $0\leftrightarrow 1$ of the second bits $k_{2}$ of the original states. Then we have 
 \begin{align}
 &\norm{\sqrt{\hat{Q}}\hat{B}\sqrt{\hat{Q}}}_{1}=\lambda^{(+)}_{\sqrt{Q}B\sqrt{Q}}+\lambda^{(-)}_{\sqrt{Q}B\sqrt{Q}}=2c_{\sqrt{Q}B\sqrt{Q}},\nonumber\\
 &\norm{\sqrt{\hat{\mathbf{1}}-\hat{Q}}\hat{C}\sqrt{\hat{\mathbf{1}}-\hat{Q}}}_{1}=2c_{\sqrt{\mathbf{1}-Q}C\sqrt{\mathbf{1}-Q}};
  \end{align}
 \item If instead $\hat{B}$ and $\hat{C}$ do not have a definite sign, then it must hold $\lambda^{(+)}_{\sqrt{Q}B\sqrt{Q}}\geq 0$ and $\lambda^{(-)}_{\sqrt{Q}B\sqrt{Q}}\leq 0$ and similar relations for $C$, so that 
  \begin{align}
 &\norm{\sqrt{\hat{Q}}\hat{B}\sqrt{\hat{Q}}}_{1}=\lambda^{(+)}_{\sqrt{Q}B\sqrt{Q}}-\lambda^{(-)}_{\sqrt{Q}B\sqrt{Q}}=2r_{\sqrt{Q}B\sqrt{Q}},\nonumber\\
 &\norm{\sqrt{\hat{\mathbf{1}}-\hat{Q}}\hat{C}\sqrt{\hat{\mathbf{1}}-\hat{Q}}}_{1}=2r_{\sqrt{\mathbf{1}-Q}C\sqrt{\mathbf{1}-Q}}.
  \end{align}
 \end{itemize}
 Let us also note that the third term $\norm{\sqrt{\hat{\mathbf{1}}-\hat{Q}}\hat{C}\sqrt{\hat{\mathbf{1}}-\hat{Q}}}_{1}$ can be expressed in terms of the coefficients of $\hat{Q}$ by observing that $c_{\mathbf{1}-Q}=(1-c_{Q})$ and $\vecund{r}_{\mathbf{1}-Q}=-\vecund{r}_{Q}$.\\
We can conclude that for $\hat{B}$ and $\hat{C}$ of non-definite sign
\begin{equation}
\mathcal{\tiny{F}}_{Q}^{(2)}\left(\hat{A},\hat{B},\hat{C}\right)=2(c_{QA}+r_{\sqrt{Q}B\sqrt{Q}}+r_{\sqrt{\mathbf{1}-Q}C\sqrt{\mathbf{1}-Q}}),
\end{equation} 
while for $\hat{B}$ and $\hat{C}$ of definite sign
\begin{equation}
\mathcal{\tiny{F}}_{Q}^{(2)}\left(\hat{A},\hat{B},\hat{C}\right)=2\left(c_{QA}+c_{QB}+c_{QC}\right),
\end{equation}
which give respectively Eqs.~(\ref{fQub}, \ref{defSign}) after inserting the values of the coefficients computed above, i.e., Eqs.~(\ref{cqa}, \ref{rMod}).


  \chapter{The Wigner function of a probabilistically-amplified coherent state}\label{app:wignerPA}
   In this appendix we compute the Wigner function $W^{(0)}$, \eqq{wignereco}, of the output state of the NHPA when a coherent state of amplitude $\alpha$ is employed as input. The function is named $W_{\alpha,out}(\beta)$, $\beta\in\mathbb{R}^{2}$ and it can be obtained by Fourier-transforming the corresponding characteristic function $\chi_{\alpha,out}(\beta)$, \eqq{charfunction}:
   \begin{equation}\label{charPA}
   \begin{aligned}
   \chi_{\alpha,out}(\beta)&=\tr{\mathcal{A}_{g,n}(\dketbra{\alpha})\hat{D}(\beta)}=\bra{\alpha}\hat{M}_{S}\hat{D}(\beta)\hat{M}_{S}^{\dag}\ket{\alpha}+\bra{\alpha}\hat{M}_{F}\hat{D}(\beta)\hat{M}_{F}^{\dag}\ket{\alpha}\\
   &=\bra{\alpha}\hat{D}(\beta)\ket{\alpha}+\sum_{k,j\leq n}s_{nk}s_{nj}\braket{\alpha}{k}\braket{j}{\alpha}\bra{k}\hat{D}(\beta)\ket{j}\\
   &-\sum_{k\leq n}s_{nk}\left[\braket{\alpha}{k}\bra{k}\hat{D}(\beta)\ket{\alpha}+\bra{\alpha}\hat{D}(\beta)\ket{k}\braket{k}{\alpha}\right]\\
   &+\sum_{k,j\leq n}f_{nk}f_{nj}\braket{\alpha}{k}\braket{j}{\alpha}\bra{k}\hat{D}(\beta)\ket{j}\\
   &=\chi_{\alpha,in}(\beta)+\sum_{k,j\leq n}(s_{nk}s_{nj}+f_{nk}f_{nj})\braket{\alpha}{k}\braket{j}{\alpha}d_{kj}(\beta)\\
   &-\sum_{k\leq n}s_{nk}\left[\braket{\alpha}{k}e_{k}(\beta,\alpha)+e^{*}_{k}(-\beta,\alpha)\braket{k}{\alpha}\right],
   \end{aligned}
   \end{equation} where in the second equality we have employed the definitions Eqs.~(\ref{mSucc},\ref{mFail}) of the Kraus operators of the NHPA map and set $s_{nk}=1-g^{-(n-k)}$, $f_{nk}=(1-g^{-2(n-k)})^{1/2}$, while in the third one we have defined $d_{kj}(\beta)=\bra{k}\hat{D}(\beta)\ket{j}$, $e_{k}(\beta,\alpha)=\bra{k}\hat{D}(\beta)\ket{\alpha}$, introduced the identity representation in the Fock basis and consider that the characteristic function of the input coherent state is $\bra{\alpha}\hat{D}(\beta)\ket{\alpha}$. The functions $d_{jk}(\beta)$, $e_{k}(\beta,\alpha)$ can be further manipulated using Eqs.~(\ref{dispaction},\ref{coherentstates}) to obtain:
   \begin{align}
   &d_{kj}(\beta)=\sum_{\ell=0}^{\min(k,j)}\frac{\sqrt{k!j!}}{\ell!(k-\ell)!(j-\ell)!}d_{kj\ell}(\beta),\\
   &e_{k}(\beta,\alpha)=e^{\frac{\beta\alpha^{*}-\beta^{*}\alpha}{2}}d_{k00}(\alpha+\beta),\\
   &d_{kj\ell}(\beta)=e^{-\frac{\abs{\beta}^{2}}{2}}\beta^{k-\ell}(-\beta^{*})^{j-\ell}.
   \end{align}
Now we take the complex Fourier transform of \eqq{charPA} with respect to $\beta$ to obtain the corresponding Wigner function. Let us note that  the only dependence on $\beta$ is through the functions $\chi_{\alpha,in}(\beta)$ and $d_{jk\ell}(\beta)$; the first straightforwardly gives the Wigner function of a coherent state $W_{\alpha,in}(\gamma)$, while the second one needs to be computed explicitly:
\begin{equation}
\begin{aligned}
\tilde{d}_{kj\ell}(\gamma)&=\int\frac{d^{2}\beta}{\pi^{2}}~d_{kj\ell}(\beta)e^{\gamma\beta^{*}-\gamma^{*}\beta}\\
&=\int\frac{d^{2}\vecund{r}_{\beta}}{(2\pi)^{2}}e^{-i\vecund{r}_{\beta}\Omega^{T}\vecund{s}_{\gamma}}e^{-\frac{q_{\beta}^{2}+p_{\beta}^{2}}{4}}\left(-q_{\beta}+ip_{\beta}\right)^{j-\ell}\left(q_{\beta}+ip_{\beta}\right)^{k-\ell}\\
&=\oneover{2\pi}\sum_{t=0}^{j-\ell}\sum_{u=0}^{k-\ell}\left(\begin{array}{c}j-\ell\\t\end{array}\right)\left(\begin{array}{c}k-\ell\\u\end{array}\right)(-1)^{j-\ell-t}(i)^{t+u} \\
&\cdot\int\frac{dq_{\beta}}{\sqrt{2\pi}}e^{-2iq_{\beta}p_{\gamma}}e^{-\frac{q_{\beta}^{2}}{4}}\left(q_{\beta}\right)^{j+k-2\ell-t-u}\int\frac{dp_{\beta}}{\sqrt{2\pi}}e^{2ip_{\beta}q_{\gamma}}e^{-\frac{p_{\beta}^{2}}{4}}\left(p_{\beta}\right)^{t+u}.
\end{aligned}
\end{equation}
Then it is sufficient t to apply the relation 
\begin{equation}
\int\frac{dx}{\sqrt{2\pi}}e^{-ixy}e^{-\frac{x^{2}}{2}}x^{t}=\left(\frac{-i}{\sqrt{2}}\right)^{t}e^{-\frac{y^{2}}{2}}H_{t}\left(\frac{y}{\sqrt{2}}\right),
\end{equation}
where $H_{t}$ is an Hermite polynomial of order $t$ to get:
\begin{equation}
\begin{aligned}
\tilde{d}_{kj\ell}(\gamma)&=\oneover{2\pi}e^{-\abs{\gamma}^{2}}\sum_{t=0}^{j-\ell}\sum_{u=0}^{k-\ell}\left(\begin{array}{c}j-\ell\\t\end{array}\right)\left(\begin{array}{c}k-\ell\\u\end{array}\right)\frac{(-1)^{j-\ell-t}(i)^{t+u}}{2^{\frac{j+k-2\ell}{2}}}\\
&\cdot H_{j+k-2\ell-t-u}\left(-p_{\gamma}\right)H_{t+u}\left(q_{\gamma}\right).
\end{aligned}
\end{equation}
The transform of $e_{k}(\beta,\alpha)$ with respect to $\beta$ then is
\begin{equation}
\tilde{e}_{k}(\gamma,\alpha)=\oneover{\sqrt{k!}}e^{\alpha\gamma^{*}-\alpha^{*}\gamma}\tilde{d}_{k00}\left(\gamma-\frac{\alpha}{2}\right).
\end{equation}
In conclusion, noting also that $W_{\alpha,in}(\gamma)=\tilde{d}_{000}(\gamma-\alpha)$ we have
 \begin{equation}\label{charPA}
   \begin{aligned}
   W_{\alpha,out}(\gamma)&=\tilde{d}_{000}(\gamma-\alpha)+\sum_{k,j\leq n}(s_{nk}s_{nj}+f_{nk}f_{nj})\braket{\alpha}{k}\braket{j}{\alpha}\tilde{d}_{kj}(\gamma)\\
 &-\sum_{k\leq n}s_{nk}\big[\braket{\alpha}{k}\braket{j}{\alpha}e^{\alpha\gamma^{*}-\alpha^{*}\gamma}\tilde{d}_{k00}\left(\gamma-\frac{\alpha}{2}\right)\\
 &+e^{-\alpha\gamma^{*}+\alpha^{*}\gamma}\tilde{d}_{0k0}\left(\gamma-\frac{\alpha}{2}\right)\braket{\alpha}{j}\braket{k}{\alpha}\big],
   \end{aligned}
   \end{equation}
   which can be computed numerically without any truncation of the sums.


\chapter{Secondary features of the amplifier decoding schemes} \label{app:altDecs}
In this appendix we prove some side observations related to the NHPA decoders discussed in Sec.~\ref{sec2}. We assume w.l.o.g. that $\alpha>0$, $\beta\in\mathbb{R}$. Let us start by proving the points 3 and 4 in the discussion of Sec.~\ref{sec2}:
\begin{enumerate}
\item If we set $\beta=0$ then the NHPA receiver perform better than the unoptimized Kennedy one, indeed
\begin{equation}
\begin{aligned}
P_{succ}(\alpha,0,g,n)&=1-\oneover{2}F\left((A_{g,n}(\ket{0}\bra{0}),A_{g,n}(\ket{2\alpha}\bra{2\alpha})\right)\\
&\leq 1-\oneover{2}F\left(\ket{0}\bra{0},\ket{2\alpha}\bra{2\alpha}\right)=P_{succ}(\alpha,0,g=1,n),
\end{aligned}
\end{equation}
where we have used the expression of the fidelity when one of its arguments is pure, \eqq{fidpure}, and the non-decreasing property under CPTP transformations of both its arguments, \eqq{nondec};
\item For $n=1$ and $\beta<0$ the NHPA success probability is a decreasing function of the gain, indeed in this case its expression is
\begin{equation}\label{psuccn1}
\begin{aligned}
P_{succ}(\alpha,\beta,g,1)&=\oneover{2}\Big[1-e^{-4\alpha^{2}-\beta^{2}}\left(e^{4\alpha\beta}+2(1-g^{-1})(1-e^{2\alpha\beta})\right)+e^{-\beta^{2}}\Big].
\end{aligned}
\end{equation}
We notice that the dependence on $g$ is via the term $e^{-4\alpha^{2}-\beta^{2}} (1-e^{2\alpha\beta})/g$. Accordingly  the probability is increasing for $\beta>0$, constant for $\beta=0$ and decreasing for $\beta<0$. 
Hence for $\beta= |\beta|$ positive we have 
\begin{equation}\label{psuccn11}
\begin{aligned}
P_{succ}(\alpha,\beta,g,1) &\leq P_{succ}(\alpha,\abs{\beta},\infty,1)\\
 & =  \oneover{2}\left[e^{-\beta^{2}} +1-e^{-4\alpha^{2}-\beta^{2}}\Big(e^{4\alpha|\beta|}+2\left(1-e^{2\alpha|\beta|}\right)\Big)\right],
\end{aligned}
\end{equation}
while for $\beta=- |\beta|$ negative we have 
\begin{equation}\label{psuccn12}
\begin{aligned}
P_{succ}(\alpha,\beta,g,1)\leq P_{succ}(\alpha,-\abs{\beta},1,1) =  \oneover{2}\left[e^{-\beta^{2}} +1-e^{-4\alpha^{2}-\beta^{2}}\; e^{-4\alpha|\beta|}\right].
 \end{aligned}
\end{equation}
Since for all $|\beta|$ we have 
$e^{-4\alpha|\beta|} \leq 1+ \left(1-e^{2\alpha|\beta|}\right)^2$,
it follows that 
\begin{equation}
 P_{succ}(\alpha,\abs{\beta},\infty,1) \leq P_{succ}(\alpha,-\abs{\beta},1,1).
 \end{equation}
Therefore we can conclude that the maximum of $P_{succ}(\alpha,\beta,g,1)$ is always achieved for $g=1$ by properly choosing the sign of the parameter $\beta$, i.e. no improvement can be gained via the action of the amplifier for $n=1$.
\end{enumerate}

Next we consider an alternative decoder were the NHPA is substituted with a generic PI Gaussian channel $\Phi_{PI}$, e.g., an ordinary parametric amplifier, showing that it gives no advantage over the optimized Kennedy detector. Indeed $\Phi_{PI}$ can be always decomposed as the concatenation of a quantum-limited  amplifier $\mathcal{A}_{\kappa,0}$ and attenuator $\mathcal{E}_{\eta,0}$,, as discussed in \eqq{ampattdeco}. By noticing that the latter has a dual channel $\mathcal{A}_{\kappa}^{*}=\kappa^{-1}\mathcal{E}_{\kappa^{-1}}$, \eqq{duality}, we can write the overlap between the coherent states $\ket{\alpha}$ and $\ket{\beta}$ after transformation of the former under $\Phi_{PI}$ as 
\begin{equation}
\begin{aligned} 
F\left(\dketbra{\beta}),\Phi_{PI}(\dketbra{\alpha})\right)&= F\left(\dketbra{\beta}),\left(\mathcal{A}_{\kappa,0}\circ\mathcal{E}_{\eta,0}\right)(\dketbra{\alpha})\right)\\
&=\kappa^{-1}F\left(\mathcal{E}_{\kappa^{-1}}(\dketbra{\beta}),\mathcal{E}_{\eta}(\dketbra{\alpha})\right)  \\
&=\kappa^{-1}F\left(\dketbra{\beta'},\dketbra{\sqrt{\eta}\alpha}\right), \nonumber 
\end{aligned}
\end{equation}
$\ket{\beta'}$ being an attenuated version  of $\ket{\beta}$ whose explicit value is irrelevant since it will be optimized. Thus, calling $\Delta_{\Phi_{PI},\alpha}$, $\Delta_{\alpha}$ the difference between the displaced-vacuum-overlaps of the two input states respectively with and without application of the channel $\Phi_{PI}$, we have $\Delta_{\Phi_{PI},\alpha}=\kappa^{-1}\Delta_{\sqrt{\eta}\alpha}\leq\Delta_{\alpha}$, i.e. no improvement in the success probability can be obtained by applying a PI Gaussian channel before the detection.\\

Eventually, we compute explicitly the overlap between a Fock state and a squeezed-displaced state needed to evaluate the success probability of the TS-based schemes, including the one with NHPA that we described in Sec.~\ref{sec2}.
\begin{equation}\label{sqdispOv}
\begin{aligned}
\braket{k}{\beta,r}&=\bra{k}\hat{U}_{sq}(r)\hat{D}(\beta)\ket{0}=\bra{k}\hat{D}(\tilde{\beta})\ket{r}=\sum_{\ell=0}^{\infty}\frac{c_{\ell}(r)}{\sqrt{(2\ell)!}}\bra{k}\hat{D}(\tilde{\beta})\ket{2\ell}\\
&=\sum_{\ell=0}^{\infty}\frac{c_{\ell}(r)}{\sqrt{(2\ell)!}}\bra{k}\left(\hat{D}(\tilde{\beta})\hat{a}^{\dag}\hat{D}^{\dag}(\tilde{\beta})\right)^{2\ell}\hat{D}(\tilde{\beta})\ket{0}\\
&=\sum_{\ell=0}^{\infty}\frac{c_{\ell}(r)}{\sqrt{(2\ell)!}}\bra{k}\left(\hat{a}^{\dag}-\tilde{\beta}^{*}\right)^{2\ell}\ket{\tilde{\beta}}\\
&=e^{-\frac{\abs{\tilde{\beta}}^{2}}{2}}\sum_{\ell=0}^{\infty}\frac{c_{\ell}(r)}{\sqrt{(2\ell)!}}\sum_{j=0}^{2\ell}\left(\begin{array}{c}2\ell\\j\end{array}\right)\left(-\tilde{\beta}^{*}\right)^{2\ell-j}\sum_{m=0}^{\infty}\frac{\tilde{\beta}^{m}}{\sqrt{m!}}\bra{k}\left(\hat{a}^{\dag}\right)^{j}\ket{m}\\
&=e^{-\frac{\abs{\tilde{\beta}}^{2}}{2}}\sum_{\ell=0}^{\infty}\sum_{j=0}^{\min(2\ell,k)}\frac{c_{\ell}(r)\sqrt{k!}}{(k-j)!\sqrt{(2\ell)!}}\left(\begin{array}{c}2\ell\\j\end{array}\right)\left(-\tilde{\beta}^{*}\right)^{2\ell-j}\tilde{\beta}^{k-j},
\end{aligned}
\end{equation}
where the second equality follows from commuting the squeezing and displacement operation, always possible up to redefining $\tilde{\beta}(\beta,r)$, which is irrelevant in our case since both $\beta$ and $r$ will be optimized; the third equality follows from the Fock representation of a squeezed state, \eqq{squeezedstates}; the third one from writing explicitly the Fock components of the expansion as powers of the photon-number increasing operator $\hat{a}^{\dag}$ and repeatedly introducing couples of operators $\hat{D}\hat{D}^{\dag}$ between them; the fifth equality follows from applying the displacement operation as in \eqq{dispaction}; the sixth one from expanding the binomial power and the coherent state $\ket{\tilde{\beta}}$ and the last one from computing the average $\bra{k}\left(\hat{a}^{\dag}\right)^{j}\ket{m}\propto\delta_{m,k-j}$. Let us note that \eqq{sqdispOv} can be computed numerically only by inserting a cutoff for the sum over $\ell$, as has been done to obtain the plots in Sec.~\ref{sec:cohDisc}.


\chapter{Approximate cavity implementation of a partial dephaser}\label{app:cavity}
In this appendix we report the full computation of the output state of the optical-cavity transformation described in Sec.~\ref{sec3}. Consider the state of the composite atom-cavity system after the injection of a coherent state $\ket{\alpha}$ into the cavity:
\begin{equation}
\ket{\alpha,G}=e^{-\frac{\abs{\alpha}^{2}}{2}}\sum_{k=0}^{\infty}\frac{\alpha^{k}}{\sqrt{k!}}\ket{k,G}.\label{in}
\end{equation}
In order to study the effect of the unitary evolution $\hat{U}_{t}=\exp(-i\hat{H}_{JC}t)$ with the Jaynes-Cummings hamiltonian, \eqq{ham}, we recall its spectral decomposition~\cite{nChuangBOOK}. In full generality we consider its off-resonance form that has an additional term $\delta\hat{Z}$ with respect to \eqq{ham}, with coefficient $\delta=(\omega-\omega_{0})/2$ proportional to the detuning between the cavity and atomic frequencies, respectively $\omega$, $\omega_{0}$. Considering the commutation property noticed in Sec.~\ref{sec3}, we can classify the eigenstates of the hamiltonian with respect to the average value of $\hat{N}$:
\begin{align}
&\ket{k,E}\rightarrow\ave{\hat{N}}_{k,E}=k+\oneover{2},\\
&\ket{k+1,G}\rightarrow\ave{\hat{N}}_{k+1,G}=k+1-\oneover{2}\equiv\ave{\hat{N}}_{k,E}.
\end{align}
Thus we can diagonalize the hamiltonian in the subspace \{$\ket{k,E}$, $\ket{k+1,G}$\} for any $k$. This restricted hamiltonian reads out
\[
\hat{H}_{JC}^{(k)}=\left(
\begin{array}{cc}
  \omega(k+\oneover{2})+\delta & g\sqrt{k+1}  \\
  \\
  g\sqrt{k+1} &  \omega(k+\oneover{2})-\delta     \\  
\end{array},
\right)
\]
its spectrum is
\begin{equation}
\lambda_{\pm}=\omega\left(k+\oneover{2}\right)\pm\Omega_{k},
\end{equation}
with $\Omega_{k}^{2}=\delta^{2}+g^{2}(k+1)$, while its eigenvectors can be written as
\begin{equation}\begin{aligned}
\ket{k,+}&=\cos\alpha_{k}\ket{k,E}+\sin\alpha_{k}\ket{k+1,E},\\
\ket{k,-}&=-\sin\alpha_{k}\ket{k,E}+\cos\alpha_{k}\ket{k+1,E},
\end{aligned}
\end{equation}
with
\begin{equation}\label{cosSen}
\begin{aligned}
\cos\alpha_{k}&=\frac{g\sqrt{k+1}}{\sqrt{g^{2}(k+1)+(\Omega_{k}-\delta)^{2}}},\quad\sin\alpha_{k}=\frac{\Omega_{k}-\delta}{\sqrt{g^{2}(k+1)+(\Omega_{k}-\delta)^{2}}}.
\end{aligned}
\end{equation}
The initial state \eqref{in} can be rewritten as
\begin{equation}
\ket{\alpha,G}=e^{-\frac{\abs{\alpha}^{2}}{2}}\left(\sum_{k=0}^{\infty}\frac{\alpha^{k+1}}{\sqrt{(k+1)!}}\ket{k+1,G}+\ket{0,G}\right),
\end{equation}
thus its evolution under the Jaynes-Cummings hamiltonian requires to study the evolution of states $\ket{k+1,G}$ and \ket{0,G}.
It is easily seen that the latter is eigenstate of $\hat{H}_{JC}$ with eigenvalue $-\omega_{0}/2$, while the former has to be expanded first in terms of $\ket{k,\pm}$, giving an evolution
\begin{equation}\label{evpsi2}
\begin{aligned}
\hat{U}_{t}\ket{k+1,G} &=\sin\alpha_{k}e^{-i\lambda_{+}t}\ket{k,+}+\cos\alpha_{k}e^{-i\lambda_{-}t}\ket{k,-}\\
&=e^{-i\omega\left(k+\oneover{2}\right)t}\left(e_{k}\ket{k,E}+d_{k}\ket{k+1,G}\right),
\end{aligned}
\end{equation}
with
\begin{equation}
\begin{aligned}
e_{k}&=-i\frac{g\sqrt{k+1}}{\Omega_{k}}\sin(\Omega_{k}t),\quad d_{k}=\cos(\Omega_{k}t)-i\frac{\delta}{\Omega_{k}}\sin(\Omega_{k}t),
\end{aligned}
\end{equation}
where we have used the relations $\cos(2\alpha_{k})=\delta/\Omega_{k}$, $\sin(2\alpha_{k})=g\sqrt{k+1}/\Omega_{k}$, derived from Eq.~(\ref{cosSen}).
Similarly, the state orthogonal to the latter evolves as
\begin{equation}
\hat{U}_{t}\ket{k,E} =e^{-i\omega\left(k+\oneover{2}\right)t}\left(d_{k}^{*}\ket{k,E}+e_{k}\ket{k+1,G}\right).\label{evpsi1}\\
\end{equation}\\

From now on we choose $\delta=0$, i.e., the condition of perfect resonance between cavity and atom assumed in Sec.~\ref{sec3}. This implies $\Omega_{k}=g\sqrt{k+1}$, $e_{k}=-i\sin(\Omega_{k}t)$ and $d_{k}=\cos(\Omega_{k}t)$. Moreover, we take $t=\tau=\pi/(2 g)$, which assures that a transition $\ket{1,G}\rightarrow\ket{0,E}$ will take place, since $e_{0}=-i$, $d_{0}=0$ with this choice. Note that the coefficients for higher values of $k>0$ are still different from zero. Hence we can write the evolved cavity-atom state at time $\tau$, up to an irrelevant global phase, as $\ket{\psi_{RABI}}=\hat{U}_{\tau}\ket{\alpha,G}$, see \eqq{firstRabi}, with the spurious terms equal to
\begin{equation}
\ket{\Delta}=\sum_{k=1}^{\infty}\frac{\alpha^{k+1}}{\sqrt{(k+1)!}}e^{-i\omega (k +1)\tau}\left(e_{k}\ket{k,E}+d_{k}\ket{k+1,G}\right).
\end{equation}
Next, as described in Sec.~\ref{sec3}, we proceed to decouple the cavity from the atom and applying a random dephasing field of frequency $\theta\gg\omega_{0}$, which cannot excite the atom and acts only on the cavity. The evolution for time $\tilde{\tau}=3\pi/(2g)$ is such that $\tilde{e}_{0}=-e_{0}$, $\tilde{d}_{0}=d_{0}=0$. Taking into account the fact that both kinds of states $\ket{k,E}$, $\ket{k+1,G}$ are now present in the global state and applying the transformation rules (\ref{evpsi2},\ref{evpsi1}) we obtain the final state
\begin{align}
&\ket{\Phi}_{CA}=e^{-\frac{\abs{\alpha}^{2}}{2}}\left(\ket{\Phi_{1}}_{C}\otimes\ket{G}_{A}+\ket{\Phi_{2}}_{C}\otimes\ket{E}_{A}\right),\\
&\begin{aligned}\ket{\Phi_{1}}&=\ket{\alpha_{T}}+\sum_{k=2}^{\infty}\frac{\alpha^{k}}{\sqrt{k!}}e^{-i \frac{2\pi}{g} \omega k-i\theta k}\left(e^{i\theta}e_{k-1}\tilde{e}_{k+1}+d_{k-1}\tilde{d}_{k+1}\right)\ket{k},
\end{aligned}\\
&\begin{aligned}\ket{\Phi_{2}}&=\sum_{k=1}^{\infty}\frac{\alpha^{k+1}}{\sqrt{(k+1)!}}e^{-i\frac{2\pi}{g}\omega (k+1)-i\theta k}\left(e_{k}\tilde{d}_{k}^{*}+e^{-i\theta}d_{k}\tilde{e}_{k}\right)\ket{k},
\end{aligned}
\end{align}
where $\ket{\alpha_{T}}$ is defined in Sec.~\ref{sec3}.
The output state of the field is obtained by tracing out the atomic component, $\hat{\rho}_{C}(\theta)=\dketbra{\Phi_{1}}+\dketbra{\Phi_{2}}$.\\
Eventually, we take the average over the random variable $\theta$. In this way, any term containing a non-zero phase factor $\exp(\theta(k-j))$, $k\neq j$, will vanish, leaving only the terms which are constant in $\theta$ in the superposition. With this procedure we aim at eliminating any coherence between the zero-one photon subspace and the rest of the photon-number states. Unfortunately this operation leaves some spurious terms of coherence between the states with one and two photons. The averaged state is \eqq{finCav} with coefficients given by 
\begin{equation}
\begin{aligned}
D&=\abs{e_{1}{\tilde{d}_{1}}^{*}}^{2}+\abs{\tilde{e}_{1}d_{1}}^{2},\quad E=\frac{e^{-i\frac{2\pi}{g}\omega}}{\sqrt{3}}e_{2}\tilde{d}_{2}^{*}d_{1}^{*}\tilde{e}_{1}^{*}\\
D_{k}(\alpha)&=\abs{e_{k-1}\tilde{e}_{k-1}}^{2}+\abs{d_{k-1}\tilde{d}_{k-1}}^{2}+\frac{\abs{\alpha}^{2}}{k+1}\abs{e_{k}\tilde{d}_{k}^{*}}^{2}+\abs{\tilde{e}_{k}d_{k}}^{2},\\
E_{k}(\alpha)&=\frac{e^{-i\frac{2\pi}{g}\omega}}{\sqrt{k+1}}d_{k}\tilde{d}_{k}e_{k-1}^{*}\tilde{e}_{k-1}^{*}+\frac{\abs{\alpha}^{2}}{k+1}\frac{e^{-i\frac{2\pi}{g}\omega}}{\sqrt{k+2}}e_{k+1}\tilde{d}_{k+1}^{*}d_{k}^{*}\tilde{e}_{k}^{*}.
\end{aligned}
\end{equation}
It can be seen that the spurious coherence between the one- and two-photon states is due to terms in $\ket{\Phi_{2}}$, originally associated to an excited atomic state. 


\chapter{Computation of the realistic PSK detection probability and Hadamard rate}\label{app:naiveProb}
In this appendix we compute the detection probability of PSK signals employing a naive yet realistic detection scheme. Better detection schemes can be developed by refining this one, as discussed in~\cite{marq}.
We employ a device similar to that of Fig. \ref{f11}, where the reflected part of the state after the beam-splitter is subject to a displacement of $-\alpha_{0}$ before photodetection. In this way we perfectly null the state $\ket{\alpha_{0}}$, if it was present, and can exclude its presence if a click is registered. When this happens we proceed with the PSK detection of $M-1$ signals. Specifically, for $M=3$ this second step will be an ordinary Dolinar detection, which is Helstrom-optimal, while for a higher number of pulse phases $M$, this procedure will establish a hierarchy of subsequent realistic PSK detections, excluding one state at each stage, all the way down to $M=2$ remaining states.  Accordingly, for $M=3$ and in the limit of infinite splitting steps $J\rightarrow\infty$, the conditional probability of guessing state $\ket{\alpha_{\ell}}$ if $\ket{\alpha_{m}}$ was sent, $P_{real}^{(M)}(\ell|m;\mathcal{E})$, is:
\begin{equation}
\begin{aligned}
P_{real}^{(3)}(0|0;\mathcal{E})&=1,\quad P_{real}^{(3)}(1|0;\mathcal{E})=0,\quad P_{real}^{(3)}(0|1;\mathcal{E})=\abs{\braket{0}{\alpha_{1}-\alpha_{0}}}^{2}=e^{-3\mathcal{E}},\\
P_{real}^{(3)}(1|1;\mathcal{E})&=\int_{0}^{\abs{\alpha_{1}-\alpha_{0}}^{2}}dx~e^{-x}P_{hel}^{(2)}\left(1|1;\frac{\abs{\alpha_{1}-\alpha_{2}}^{2}-x}{2}\right)\\
&=\int_{e^{-3\mathcal{E}}}^{1}dy~P_{hel}^{(2)}\left(1|1;\frac{3\mathcal{E}+\ln y}{2}\right),\\
P_{real}^{(3)}(2|1;\mathcal{E})&=\int_{e^{-3\mathcal{E}}}^{1}dy~P_{hel}^{(2)}\left(2|1;\frac{3\mathcal{E}+\ln y}{2}\right)=1-P_{real}^{(3)}(0|1;\mathcal{E})-P_{real}^{(3)}(1|1;\mathcal{E}),
\end{aligned}
\end{equation}
where the integrals have been obtained as in Sec.~\ref{secHadCap} up to a rescaling 
\begin{equation}
x\rightarrow x\cdot\abs{\alpha_{1}-\alpha_{0}}^{2}/\mathcal{E},
\end{equation} 
which takes into account the energy of the states after the nulling displacement. Given the symmetry of the states, the remaining conditional probabilities can be obtained by switching $1\leftrightarrow 2$ in the previous equations.
After substituting the previous expressions in \eqq{DDDlm}, we obtain the VP conditional probability of detection when this realistic PSK detection scheme is employed. Eventually the Hadamard rate for $M=3$ can be computed along the same lines of  Sec.~\ref{secHadRate} and it is equal to
\begin{equation}
\begin{aligned}
R^{(N,3)}_{real}(E)&=h\left(\frac{1-e^{-3\mathcal{E}}}{3N}\right)+2h\left(\frac{2-3e^{-\mathcal{E}}+e^{-3\mathcal{E}}}{6N}\right)-\frac{1}{3} h\left(1-e^{-\mathcal{E}}\right)\\
&-\frac{2}{3}\Bigg\{h\left(\frac{e^{-\mathcal{E}}-e^{-3\mathcal{E}}}{2}\right)+h\left(P_{vp\text{-}real}^{(3)}(1|1;\mathcal{E})\right)\\
&+h\left(\frac{2-3e^{-\mathcal{E}}+e^{-3\mathcal{E}}}{2}-P_{vp\text{-}real}^{(3)}(1|1;\mathcal{E})\right)\Bigg\},
\end{aligned}
\end{equation}
where $h(p)$ is defined as in Sec.~\ref{subsec:commC} and we have defined the conditional probability of correctly identifying the state $\ket{\alpha_{1}}$ with VP and realistic PSK detection for $M=3$ as in \eqq{vpRealGen}.\\
Similarly, for $M=4$ we have a two-stage hierarchy of realistic PSK detection where we first null the state $\ket{\alpha_{0}}$. If a click is registered, we are left with three states to discriminate and can resort to the case previously studied, though with a different symmetry between the states. In this case, we null the state $\ket{\alpha_{2}}$, which is equidistant from the other two remaining ones. Hence the conditional probability of realistic PSK detection reads out:
\begingroup
\allowdisplaybreaks
\begin{align}
P_{real}^{(4)}(0|0;\mathcal{E})&=1,\quad P_{real}^{(4)}(1,2,3|0;\mathcal{E})=0,\quad P_{real}^{(4)}(0|1;\mathcal{E})=e^{-2\mathcal{E}},\nonumber\\
P_{real}^{(4)}(0|2;\mathcal{E})&=\abs{\braket{0}{\alpha_{2}-\alpha_{0}}}^{2}=e^{-4\mathcal{E}},\quad P_{real}^{(4)}(1|2;\mathcal{E})=0,\nonumber\\
P_{real}^{(4)}(2|2;\mathcal{E})&=\int_{0}^{\abs{\alpha_{2}-\alpha_{0}}^{2}}dx~e^{-x}\cdot1=1-e^{-4\mathcal{E}},\nonumber\\
P_{real}^{(4)}(2|1;\mathcal{E})&=\int_{0}^{\abs{\alpha_{1}-\alpha_{0}}^{2}}dx~e^{-x} e^{-(\abs{\alpha_{1}-\alpha_{2}}^{2}-x)}=\int_{e^{-2\mathcal{E}}}^{1}dt~e^{-(2\mathcal{E}+\ln t)}=e^{-2\mathcal{E}}2\mathcal{E},\nonumber\\
P_{real}^{(4)}(1|1;\mathcal{E})&=\int_{0}^{\abs{\alpha_{1}-\alpha_{0}}^{2}}dx~e^{-x}\\
&\cdot\int_{0}^{(|\alpha_{1}-\alpha_{2}|^{2}-x)}dx'~e^{-x'}P_{hel}^{(2)}\left(1|1;\frac{(\abs{\alpha_{1}-\alpha_{3}}^{2}-x-x')}{2}\right)\nonumber\\
&={\displaystyle \int_{e^{-2\mathcal{E}}}^{1}}dt\int_{e^{-(2\mathcal{E}+\ln t)}}^{1}dt'P_{hel}^{(2)}\left(1|1;2\mathcal{E}+\ln(t t')\right),\nonumber\\
P_{real}^{(4)}(3|1;\mathcal{E})&=1-P_{real}^{(4)}(0|1;\mathcal{E})-P_{real}^{(4)}(2|1;\mathcal{E})-P_{real}^{(4)}(1|1;\mathcal{E}).\nonumber
\end{align}
\endgroup
As before, for symmetry reasons the remaining conditional probabilities can be obtained by switching the indexes $1\leftrightarrow3$.
Eventually the Hadamard rate for $M=4$ is given by
\begin{align}
R_{real}^{(4)}(E)&=h\left(\frac{3+4e^{-\mathcal{E}}-6e^{-2\mathcal{E}}-e^{-4\mathcal{E}}}{12N}\right)+h\left(\frac{3+8e^{-\mathcal{E}}-12e^{-2\mathcal{E}}+e^{-4\mathcal{E}}}{12N}\right)\nonumber\\
&+2h\left(\frac{1-4e^{-\mathcal{E}}+(3+2\mathcal{E})e^{-2\mathcal{E}}}{4N}\right)\nonumber\\
&-\frac{1}{4N}\left\{h\left(1-e^{-2\mathcal{E}}\right)+h\left(\frac{e^{-\mathcal{E}}-e^{-4\mathcal{E}}}{3}\right)+h\left(\frac{3-4e^{-\mathcal{E}}+e^{-4\mathcal{E}}}{3}\right)\right\}\\
&-\frac{1}{2N}\Bigg\{h\left(e^{-\mathcal{E}}-e^{-2\mathcal{E}}\right)+h\left(2(e^{-\mathcal{E}}-(1+\mathcal{E})e^{-2\mathcal{E}})\right)+h\left(P_{vp-real}^{(4)}(1|1;\mathcal{E})\right)\nonumber\\
&+h\left(1-4e^{-\mathcal{E}}+(3+2\mathcal{E})e^{-2\mathcal{E}}-P_{vp-real}^{(4)}(1|1;\mathcal{E})\right)\Bigg\},\nonumber
\end{align}
where we have defined the conditional probability of correctly identifying the state $\ket{\alpha_{1}}$ with VP and realistic PSK detection for $M=4$ as in \eqq{vpRealGen}.

\end{appendices}

\begin{spacing}{0.9}


\bibliographystyle{ieeetr} 
\cleardoublepage


\end{spacing}

\printthesisindex 

\end{document}